\documentclass[oneside]{scrbook}

\usepackage{color}
\usepackage{amsthm}
\usepackage{hyperref}
\usepackage{xhfill}

\usepackage{type1cm}      

\usepackage{graphicx}        
                             
\usepackage{multicol}        
\usepackage[bottom]{footmisc}

\usepackage{algorithm}
\usepackage{amsmath,amssymb}

\usepackage{tabularx}

\definecolor{light}{gray}{0.85}

 \newenvironment{trailer}[1]%
    {
    {{{\color{light}\rule{5mm}{1mm}}\hspace{2mm}}\textbf{#1}
    \xhrulefill{light}{1mm}\em \\}%
 }

\newenvironment{warning}[1]%
    {
    {{{\color{light}\rule{5mm}{1mm}}\hspace{2mm}}\textbf{#1}
    \xhrulefill{light}{1mm}\em \\}%
 }

    {
    {{{\color{light}\rule{5mm}{1mm}}\hspace{2mm}}\textbf{#1}
    \xhrulefill{light}{1mm}\em \\}%
 }

 \newenvironment{question}[1]%
     {
     {{{\color{light}\rule{5mm}{1mm}}\hspace{2mm}}\textbf{#1}
     \xhrulefill{light}{1mm}\em \\}%
  }

    {
    {{{\color{light}\rule{5mm}{1mm}}\hspace{2mm}}\textbf{#1}
    \xhrulefill{light}{1mm}\em \\}%
 }

\newenvironment{programcode}[1]%
    {
    {{{\color{light}\rule{5mm}{1mm}}\hspace{2mm}}\textbf{#1}
    \xhrulefill{light}{1mm}\em \\}%
 }

\newtheorem{ntheorem}{Theorem}[chapter]
\newtheorem{nexercise}[ntheorem]{Exercise}
\newtheorem{nlemma}[ntheorem]{Lemma}
\newtheorem{ndefinition}[ntheorem]{Definition}
\newtheorem{nproposition}[ntheorem]{Proposition}

\newtheorem{ncorollary}[ntheorem]{Corollary}
\newtheorem{nexample}[ntheorem]{Example}
\newtheorem{nremark}[ntheorem]{Remark}

\newcommand{\rk}{\operatorname{rk}}
\newcommand{\srk}{\operatorname{srk}}

\newcommand{\F}{\mathbb{F}}
\newcommand{\N}{\mathbb{N}}

\newcommand{\R}{\mathbb{R}}
\newcommand{\Z}{\mathbb{Z}}
\newcommand{\wt}{\operatorname{wt}}
\newcommand{\dH}{\operatorname{d}_{\text{H}}}
\newcommand{\ds}{d_{\text{S}}}
\newcommand{\di}{d_{\text{I}}}
\newcommand{\fdist}{d_{\text{G}}}
\newcommand{\dr}{d_{\text{R}}}
\newcommand{\dsr}{d_{\text{S-R}}}

\newcommand{\bc}{\mathbf{c}}

\newcommand{\be}{\mathbf{e}}

\newcommand{\bu}{\mathbf{u}}
\newcommand{\bv}{\mathbf{v}}
\newcommand{\bw}{\mathbf{w}}

\newcommand{\PG}{\operatorname{PG}}

\newcommand{\cA}{\mathcal{A}}
\newcommand{\cB}{\mathcal{B}}
\newcommand{\cC}{\mathcal{C}}
\newcommand{\cD}{\mathcal{D}}
\newcommand{\cF}{\mathcal{F}}
\newcommand{\cG}{\mathcal{G}}

\newcommand{\cL}{\mathcal{L}}
\newcommand{\cM}{\mathcal{M}}

\newcommand{\cP}{\mathcal{P}}
\newcommand{\cS}{\mathcal{S}}

\newcommand{\cU}{\mathcal{U}}
\newcommand{\cV}{\mathcal{V}}
\newcommand{\cW}{\mathcal{W}}

\newcommand{\qbin}[3]{\genfrac{[}{]}{0pt}{}{#1}{#2}_{#3}}
\newcommand{\spaces}[2]{\genfrac{[}{]}{0pt}{}{#1}{#2}}
 
\newcommand{\snumb}[3]{s_{#3}(#1,#2)} % r, i, q
\newcommand{\zv}{\mathbf{0}}

\newcommand{\rmc}{\texttt{RMC}}
\newcommand{\srmc}{\texttt{SRMC}}
\newcommand{\fdrm}{\texttt{FDRM}}
\newcommand{\cdc}{\texttt{CDC}}
\newcommand{\mdc}{\texttt{MDC}}
\newcommand{\mrd}{\texttt{MRD}}
\newcommand{\lmrd}{\texttt{LMRD}}
\def\llceil{\lceil\kern-3.2pt\lceil}
\def\rrceil{\rceil\kern-3.2pt\rceil}
\def\llfloor{\lfloor\kern-3.2pt\lfloor}
\def\rrfloor{\rfloor\kern-3.2pt\rfloor}
\def\leftllceil{\left\lceil\kern-3.2pt\left\lceil}
\def\rightrrceil{\right\rceil\kern-3.2pt\right\rceil}
\def\leftllfloor{\left\lfloor\kern-3.2pt\left\lfloor}
\def\rightrrfloor{\right\rfloor\kern-3.2pt\right\rfloor}
\def\bigllfloor{\bigg\lfloor\kern-3.2pt\bigg\lfloor}
\def\bigrrfloor{\bigg\rfloor\kern-3.2pt\bigg\rfloor}

\newcolumntype{L}[1]{>{\raggedright\arraybackslash}p{#1}}
\newcolumntype{C}[1]{>{\centering\arraybackslash}p{#1}}
\newcolumntype{R}[1]{>{\raggedleft\arraybackslash}p{#1}}

\makeindex             % used for the subject index
\begin{document}
\date{20.12.2025}
\title{Constructions and bounds for subspace codes}
\author{Sascha Kurz\\\footnotesize sascha.kurz@uni-bayreuth.de}
\publishers{\footnotesize\begin{flushleft}\textbf{Abstract} Subspace codes are the $q$-analog of binary block codes in the Hamming metric. Here the codewords are vector spaces over a finite field. They have e.g.\ applications in random 
linear network coding \cite{network_codes}, distributed storage \cite{codes_for_distributed_storage,raviv2017subspace}, and cryptography \cite{code_based_cryptography}. In this 
chapter we survey known constructions and upper bounds for subspace codes.\end{flushleft}}

\maketitle
\pagenumbering{roman}
\tableofcontents
\pagebreak
\pagenumbering{arabic}

\chapter{Introduction}
\label{sec_introduction}

An important and classical family of error-correcting codes are so-called \emph{block codes}. Given a non-empty \emph{alphabet} $\Sigma$ and a \emph{length} $n\in\N_{>0}$, 
a block code $C$ is a subset of $\Sigma^n$. The elements of $C$ are called \emph{codewords}. For $\bc,\bc'\in\Sigma^n$ the \emph{Hamming distance} is given by
\begin{equation}
  \dH(\bc,\bc')=\#\left\{1\le i\le n\,:\, c_i\neq c_i'\right\}\!,
\end{equation}   
i.e., the number of positions where the two codewords differ. With this, the \emph{minimum Hamming distance} of a block code $C$ is defined as 
\begin{equation}
  \dH(C)=\min\!\left\{\dH(\bc,\bc')\,:\, \bc,\bc'\in C, \bc\neq \bc'\right\}\!.
\end{equation}
By convention we formally set $\dH(C)=\infty$ if $\#C<2$, i.e., $\dH(C)>m$ for each integer $m$. If the alphabet $\Sigma$ is a finite field (or a ring), we 
can call a block code $C$ \emph{linear} if it is linearly closed and \emph{additive} if it is additively closed. While there is a lot of research on block codes with $\#\Sigma>2$,  
we want to consider the binary case $\Sigma=\F_2=\{0,1\}$ only. Alternatively we may represent the codewords $\bc\in\F_2^n$ as subsets of $\{1,\dots,n\}$ via 
$S_{\bc}:=\left\{1\le i\le n\,:\, c_i=1\right\}$, so that
\begin{equation}
  \dH(\bc,\bc')=\# S_{\bc}+\# S_{\bc'}-2\cdot\# \left(S_{\bc}\cap S_{\bc'}\right)\!.  
\end{equation} 
By $A(n,d)$ we denote the maximum possible cardinality of a binary block code $C$ with length $n$ and minimum Hamming distance 
at least $d$. The determination of $A(n,d)$ is an important problem that has achieved wide attention but is still widely open. I.e., except for a few special cases only 
lower and upper bounds for $A(n,d)$ are known, see e.g.\ \cite{agrell2001table,litsyn1998update,mounits2002improved,schrijver2005new}. For a vector $\bc\in\F_2^n$ the Hamming 
distance $\dH(\bc,\zv)$ between $\bc$ and the all-zero vector $\zv\in\F_2^n$ is called the \emph{Hamming weight} $\wt(\bc)$ of $\bc$, counting the number of non-zero entries. Note 
that additive and linear block codes indeed contain $\zv\in\F_2^n$ as a codeword, while this does not need to be the case in general. A block code $C$ where each codeword has the 
same Hamming weight, say $w$, 
is called \emph{constant weight} (block) code. The corresponding maximum possible cardinality is denoted by $A(n,d,w)$. For bounds and exact values for $A(n,d,w)$ we 
refer the reader e.g.\ to \cite{brouwer2006new,ostergard2010classification,polak2018semidefinite} and the citing papers.    

The aim of this chapter is the study of so-called \emph{subspace codes}. One way to introduce these codes is to consider them as $q$-analog of binary block codes, i.e., the 
codewords are subspaces of the vector space $\F_q^n$. 

\begin{trailer}{$q$-analogs}Many combinatorial structures are based on the subset lattice of some finite set $\cU$, which is mostly called {\lq\lq}universe{\rq\rq}. 
If we replace the subset lattice with the subspace lattice of a $\#\cU$-dimensional vector space $V$ over $\F_q$, then we obtain a $q$-analog, see e.g.\ \cite{andrews1999q,
barcucci1999some,etzion2013problems,thomas1987designs} for examples. The elements of $\cU$ correspond to the $1$-dimensional subspaces of $V$, $t$-subsets correspond 
to $t$-subspaces, and the union of two subsets corresponds to the sum of two subspaces. In Section~\ref{sec_preliminaries} we will introduce the 
$q$-binomial coefficient $\qbin{n}{k}{q}$ that corresponds to the binomial coefficient ${n\choose q}$. See also Section~\ref{sec_bounds_cdc} where we mention the 
$q$-Pochhammer symbol.   
\end{trailer}

Endowed with a suitable metric, see Section~\ref{sec_preliminaries} for details, the maximum possible sizes $A_q(n,d)$ of subspace codes in $\F_q^n$ with minimum distance at 
least $d$ can be studied. If all codewords of a subspace code $\cC$ have the same dimension, say $k$, we speak 
of a \emph{constant dimension code} and denote the corresponding maximum possible cardinality by $A_q(n,d;k)$. More precisely, here we want to survey known 
lower and upper bounds for $A_q(n,d)$ and $A_q(n,d;k)$ cf.~\cite{TableSubspacecodes}. Besides being a generalization of classical codes, another motivation comes from e.g.\ random linear network coding, 
see \cite{bassoli2013network,greferath2018network,network_codes}.      

The remaining part of this chapter is structured as follows. First we introduce necessary preliminaries in Section~\ref{sec_preliminaries}. Due to their close connection 
to constant dimension codes rank metric codes are discussed in Section~\ref{sec_rank_metric_codes}. In Section~\ref{sec_bounds_cdc} we survey upper bounds for $A_q(n,d;k)$ 
and lower bounds, i.e.\ constructions, in Section~\ref{sec_constructions_cdc}. The special parameters $A_2(7,4;3)$, i.e.\ the first open case where $A_2(n,d,k)$ has not been 
determined so far, is treated in Section~\ref{sec_fano}. In Section~\ref{sec_lower_bounds_cdc} we summarize the currently best known lower bounds for constant dimension codes 
for small parameters. Mixed dimension subspace codes are the topic of Section~\ref{sec_mdc}. We close with a few remarks on related topics in Section~\ref{sec_variants}. 
 
\chapter{Preliminaries}
\label{sec_preliminaries}

For a prime power $q>1$ let $\F_q$ be the finite field with $q$ elements. By $\F_q^n$ we denote the standard vector space of dimension $n\ge 1$ 
over $\F_q$. The set of all subspaces of $\F_q^n$, ordered by the incidence relation $\subseteq$, is called \emph{$(n-1)$-dimensional (coordinate)
projective geometry over $\F_q$} and denoted by $\PG(n-1,q)$, cf.~\cite{storme2021coding}. It 
forms a finite modular geometric lattice with \emph{meet} $U\wedge W = U\cap W$ and \emph{join} $U\vee W=U+W$. The graph theoretic distance 
\begin{equation}
  \ds(U,W)=\dim(U+W)-\dim(U\cap W)
\end{equation}  
in this lattice is called the \emph{subspace distance} between $U$ and $W$. By $\cP_q(n)$ we denote the set of all subspaces in $\F_q^n$ and by 
$\cG_q(n,k)$ the subset of those with dimension $0\le k\le n$, i.e., $\dot{\bigcup}_{k=0}^n \cG_q(n,k)=\cP_q(n)$. The elements of $\cG_q(n,k)$ are also called 
\emph{$k$-spaces} for brevity. Using geometric language, we also call $1$-, $2$-, $3$-, $4$-, and $(n-1)$-spaces \emph{points}, \emph{lines}, \emph{planes}, 
\emph{solids}, and \emph{hyperplanes}, respectively. An $(n-k)$-space is also called a subspace of \emph{codimension} $k$, i.e., a hyperplane has codimension $1$.  
A \emph{subspace code} $\cC$ is a subset of $\cP_q(n)$, where $n\ge 1$ is a suitable integer. If $\cC\subseteq \cG_q(n,k)$, i.e., all 
elements $U\in\cC$ have dimension $\dim(U)=k$, we speak of a \emph{constant dimension code} (\cdc). A subspace code $\cC$ that is not a constant dimension code is 
also called \emph{mixed dimension (subspace) code} (\mdc). 
\begin{nexercise}
  Verify that the subspace distance $\ds$ is a metric on $\cP_q(n)$ and satisfies
  \begin{eqnarray}
    \ds(U,W)&=&\dim(U)+\dim(W)-2\cdot\dim(U\cap W)\label{eq_subspace_distance_meet}\\ 
    &=& 2\cdot\dim(U+W)-\dim(U)-\dim(W)\label{eq_subspace_distance_join}.
  \end{eqnarray}
\end{nexercise}
The \emph{minimum subspace distance} $\ds(\cC)$ of a subspace code $\cC$ is defined as 
$$
  \ds(\cC)=\min\!\left\{\ds(U,W)\,:\, U,W\in \cC, U\neq W\right\}\!,
$$ 
where we formally set $\ds(\cC)=\infty$ if $\#\cC<2$, i.e., $\ds(\cC)>m$ for each integer $m$. The maximum possible cardinality of a subspace code in $\F_q^n$ with 
minimum subspace distance at least $d$ is denoted by $A_q(n,d)$. For constant dimension codes with codewords of dimension $k$ we denote the maximum possible 
cardinality by $A_q(n,d;k)$. Note that the subspace distance between two $k$-spaces satisfies $\ds(U,W)=2k-2\cdot\dim(U\cap W)=2\cdot\dim(U+W)-2k$, i.e., it 
is an even non-negative integer. For each subset $T\subseteq \{0,1,\dots,n\}$ we denote by $A_q(n,d;T)$ the maximum possible cardinality of a subspace 
code $\cC$ in $\F_q^n$ with $\ds(\cC)\ge d$ and $\dim(U)\in T$ for all $U\in\cC$, so that e.g.\ $A_q(n,d;k)=A_q(n,d;\{k\})$. Mostly we omit curly braces 
for one-element sets. If $\cC\subseteq \cG_q(n,k)$ with $d(\cC)\ge d$, then we also speak of an $(n,d;k)_q$--{\cdc}. From Equation~(\ref{eq_subspace_distance_meet}) 
we conclude that the dimension of the intersection of two codewords in $\cC$ is at most $k-d/2$ and also the minimum subspace distance is determined by the maximum 
dimension of the intersection of a pair of different codewords.\footnote{The same is true for the minimum dimension of the sum of two different codewords. The dimension 
of the sum of triples of codewords was e.g.\ considered in \cite{ballico2016higher} as another invariant of a {\cdc}.}

\begin{nexercise}
  Let $B$ be a non-degenerated bilinear form on $\F_q^n$ and $$U^\perp=\left\{x\in\F_q^n\,:\, B(x,y)=0\,\forall y\in W\right\}\!,$$ i.e., $U^\perp$ is the orthogonal 
  complement of $U$ with respect to $B$. Show $\dim(U^\perp)=n-\dim(U)$ and $\ds(U^\perp,W^\perp)=\ds(U,W)$ for all $U,W\in\cP_q(n)$.
\end{nexercise}

As an implication we remark
\begin{equation}
  A_q(n,d;T)=A_q(n,d;\{n-t\,:\, t\in T\})
\end{equation}
and
\begin{equation}
  A_q(n,d;k)=A_q(n,d;n-k),\label{eq_a_orthogonal_cdc}
\end{equation} 
so that we will mostly assume $2k\le n$. Under this assumption the maximum possible subspace distance between two $k$-spaces is $2k$, i.e., 
we have $A_q(n,d;k)=1$ if $d>2k$ and $0\le k\le n$.  
If $n<0$, $k<0$, or $k>n$, then we set $A_q(n,d;k)=0$, which allows us to omit explicit conditions on the parameters $n$, $d$, and $k$ in the following. For $A_q(n,d;T)$ we 
use the same type of conventions. Using geometric language, an $(n,2k;k)_q$--{\cdc} is also called \emph{partial spread} or \emph{partial $k$-spread}, to be more precise. 
Note that for a partial $k$-spread $\cC$ of cardinality at least $2$ we have $n\ge 2k$. 

Given a {\cdc} $\cC$ we also call $\cC^\perp:=\left\{U^\perp \,:\, U\in \cC\right\}$ the \emph{dual code}.    

As a representation for a $k$-space $U\in\cP_q(n)$ we use matrices $M\in\F_q^{k\times n}$ whose $k$ rows form a basis of $U$ and write $U=\langle M\rangle$. In 
this case we say that $M$ is a \emph{generator matrix} of $U$. If the underlying field is not clear from the context we more precisely write $\langle M\rangle_{\F_q}$ 
for the row span of $M$. 
\begin{ndefinition}
  Let $\cC$ be a subspace code in $\F_q^n$. We call a set of matrices $\cG$ a \emph{generating set} of $\cC$ if $\#\cC=\#\cG$ and 
  $\cC=\left\{\left\langle G\right\rangle\,:\, G\in\cG\right\}$.
\end{ndefinition}
In other words a generating set of a subspace code consist of a corresponding set of generator matrices.

For $U,W\in\cP_q(n)$ we have
$$
  \dim(U+W)=\rk\!\left(\begin{pmatrix}G_U\\G_W\end{pmatrix}\right),
$$ 
where $\rk(X)$ denotes the rank of a matrix $X$ and $G_U$, $G_W$ are generator matrices of $U$ and $W$, respectively. Inserting into Equation~(\ref{eq_subspace_distance_join}) 
gives
\begin{equation}
  \ds(U,W)=2\cdot \rk\!\left(\begin{pmatrix}G_U\\G_W\end{pmatrix}\right) -\dim(U)-\dim(W) \label{eq_subspace_distance_rank}.
\end{equation}

The number of $k$-spaces in $\F_q^n$ can be easily counted:\\[-3mm]
\begin{nexercise}
  Show that there are exactly $\prod_{i=0}^{k-1} \left(q^n-q^i\right)$ generator matrices (or ordered bases) for a $k$-space in $\F_q^n$ and that each such $k$-space admits 
  $\prod_{i=0}^{k-1} \left(q^k-q^i\right)$ different generator matrices, so that
  \begin{equation}
    \#\cG_q(n,k)=\prod_{i=0}^{k-1} \frac{q^{n-i}-1}{q^{k-i}-1}.\label{eq_count_k_spaces}
  \end{equation} 
\end{nexercise}
As further notation we set $\qbin{n}{k}{q}:=\#\cG_q(n,k)$, which is called \emph{$q$-binomial} or \emph{Gaussian binomial coefficient} since they are the $q$-analog 
of the binomial coefficient ${n\choose k}$ counting the number of $k$-element subsets of an $n$-element set.
\begin{nexercise}
  Consider $\qbin{n}{k}{q}$ as a function of $q$ on $\mathbb{R}_{>0}$ using Equation~(\ref{eq_count_k_spaces}) and show
  $$
    \lim\limits_{q\to 1} \qbin{n}{k}{q}={n\choose k}
  $$
  for all integers $0\le k\le n$.
\end{nexercise} 
\begin{nexercise}
  Show $\qbin{n}{k}{q}=\qbin{n}{n-k}{q}$ and $\qbin{n}{k}{q}=q^k\qbin{n-1}{k}{q}+\qbin{n-1}{k-1}{q}=\qbin{n-1}{k}{q}+q^{n-k}\qbin{n-1}{k-1}{q}$ whenever the 
  occurring Gaussian binomial coefficients are well defined.
\end{nexercise} 
For lower and upper bounds for $\qbin{n}{k}{q}$ we refer to the beginning of Section~\ref{sec_bounds_cdc}, see e.g.~Inequality~(\ref{ie_q_binomial_coefficient}).

Applying the Gaussian elimination algorithm to a generator matrix $G$ of a $k$-space $U$ gives a unique generator matrix $E(G)$ in 
\emph{reduced row echelon form}. Since $E(G)=E(G')$ for any two generator matrices $G$ and $G'$ of $U$, we will also directly write $E(U)$. By $v(G)\in\F_2^n$ 
or $v(U)\in\F_2^n$ we denote the characteristic vector of the pivot columns in $E(G)$ or $E(U)$, respectively. These vectors are also called \emph{identifying} 
or \emph{pivot vectors}. If $U\in\cG_q(n,k)$, then $\wt(v(U))=k$, i.e., the identifying vector of a $k$-space consists of $k$ ones (and $n-k)$ zeroes. Slightly 
abusing notation we use $\cG_1(n,k):=\left\{ v\in \F_2^n\,:\, \wt(v)=k\right\}$.

\begin{nexample}
  \label{ex_unique_generator_matrix}
For
$$
  U=\left\langle\begin{pmatrix}
  1&0&1&1&1&0&1&0&1\\
  1&0&0&1&1&1&1&1&1\\
  0&0&0&1&0&0&0&1&0\\ 
  0&0&0&0&0&1&1&0&1
  \end{pmatrix}\right\rangle\in\cG_2(9,4)
$$  
we have 
$$
  E(U)=
  \begin{pmatrix}
  1&0&0&0&1&0&0&0&0\\
  0&0&1&0&0&0&1&1&1\\
  0&0&0&1&0&0&0&1&0\\
  0&0&0&0&0&1&1&0&1
  \end{pmatrix}
$$
and $v(U)=1	0	1	1	0	1	0	0	0\in\F_2^{9}$.
\end{nexample}

Consider $M_U=E(U)$ and $M_W=E(W)$ in Equation~(\ref{eq_subspace_distance_rank}). Since the union of the pivot positions in $E(U)$ and $E(W)$ has cardinality 
$$
  \frac{\dH\!\big(v((U),v(W)\big)+\dim(U)+\dim(W)}{2}, 
$$  
we have have
$$
  2\cdot \rk\!\left(\begin{pmatrix}E(U)\\E(W)\end{pmatrix}\right)\ge \dH\!\big(v(U),v(W)\big)+\dim(U)+\dim(W), 
$$
so that applying Equation~(\ref{eq_subspace_distance_rank}) gives
\begin{equation}
  \ds(U,W)\ge \dH\!\big(v(U),v(W)\big)\label{ie_subspace_distance_hamming},
\end{equation}
cf.\ \cite[Lemma 2]{etzion2009error}.
\begin{nexercise}
  Let $q>1$ be a prime power, $\bu,\bw\in \F_2^n$, and 
  $$
    \dH(\bu,\bw)\le d\le \min\!\left\{\wt(\bu)+\wt(\bw),n-\frac{\wt(\bu)+\wt(\bw)-\dH(\bu,\bw)}{2}\right\} 
  $$
  with $d\equiv 0\pmod 2$ be arbitrary. Construct subspaces $U\in\cG_q\!\big(n,\wt(\bu)\big)$ and $W\in\cG_q\!\big(n,\wt(\bw)\big)$ 
  with $\ds(U,W)=d$.   
\end{nexercise}

Note that $v(U)$ depends on the ordering of the positions. By $\cS_n$ we denote the symmetric group on $\{1,\dots,n\}$. Let $\pi\in\cS_n$ be a permutation 
and $M\in \F_q^{k\times n}$ be a matrix. By $\pi M\in\F_q^{k\times n}$ we denote the matrix arising by permuting the columns of $M$ according to $\pi$. For a 
subspace $U\in\cG_q(n,k)$ we denote by $\pi U$ the $k$-space $\left\langle \pi E(U)\right\rangle$. Note that $\left\langle \pi G\right\rangle=\left\langle \pi E(U)\right\rangle$ 
for every generator matrix $G$ of $U$. 
\begin{nexercise}
  Show $\dim(U)=\dim(\pi U)$ and $\ds(U,W)=\ds(\pi U,\pi W)$ for all $U,W\in\cP_q(n)$ and $\pi\in\cS_n$. 
\end{nexercise}
\begin{nexample}
  Consider the two $2$-spaces
  $$
    U=
    \left\langle\begin{pmatrix}
    1 & 0 & 0 & 0\\ 
    0 & 1 & 0 & 0 
    \end{pmatrix}\right\rangle,\quad
    W=
    \left\langle\begin{pmatrix}
      1 & 0 & 2 & 1\\ 
      0 & 1 & 0 & 1
    \end{pmatrix}\right\rangle
  $$
  in $\cP_3(4)$. We have $v(U)=1100\in\F_2^4$ and $v(W)=1100\in\F_2^4$, so that $\dH\!\big(v(U),v(W)\big)=0<4=\ds(U,W)$. For the permutation $\pi=(13)(24)$ we 
  have $v(\pi U)=0011\in\F_w^4$ and $v(\pi W)=1100\in\F_2^4$, so that $\dH\!\big(v(\pi U),v(\pi W)\big)=4=\ds(U,W)$. 
\end{nexample}
\begin{nexercise}
  Let $U,W\in\cP_q(n)$ be arbitrary. Show the existence of a permutation $\pi\in\cS_n$ with $\ds(U,W)=\dH\!\big(v(\pi U),v(\pi W)\big)$.
\end{nexercise}
In other words, we have $\ds(U,W)\ge \dH\!\big(v(\pi U),v(\pi W)\big)$ for all $\pi\in\cS_n$ and there exists a permutation attaining equality.

\begin{ndefinition}
  Let $\cC\subseteq \cG_q(n,k)$ be a {\cdc}. The \emph{pivot structure} of $\cC$ is the subset $\cV:=\left\{v(U)\,:\, U\in\cC\right\}\subseteq \cG_1(n,k)$ 
  of binary vectors that are attained by pivot vectors of the codewords. By $A_q(n,d;k;\cV)$ we denote the maximum cardinality of a {\cdc} $\cC\subseteq\cG_q(n,k)$ 
  with minimum subspace distance at least $d$ whose pivot structure is a subset of $\cV$.
\end{ndefinition}

In order to describe specially structured subsets of $\cG_1(n,k)$ we denote by 
$$
  {n_1 \choose k_1},\dots, {n_l \choose k_l} 
$$
the set of binary vectors which contain exactly $k_i$ ones in positions $1+\sum_{j=1}^{i-1} n_j$ to $\sum_{j=1}^{i} n_j$ for all $1\le i\le l$. 
The cases of at least $k_i$ ones are denoted by ${n_i\choose {\ge k_i}}$ and the cases of at most $k_i$ ones are denoted by ${n_i\choose {\le k_i}}$. Also in this 
generalized setting we assume that the described set is a subset of $\cG_1(n,k)$, where $n=\sum_{i=1}^l n_i$ and $k=\sum_{i=1}^l k_i$, i.e.\ 
$$
  {n_1 \choose {\le k_1}},{{n-n_1} \choose {\ge k-k_1}} \subseteq \cG_1(n,k).
$$       
For two subsets $\cV,\cV'\subseteq\F_2^n$ we write $\dH(\cV,\cV')$ for the minimum Hamming distance $\dH(v,v')$ for arbitrary $v\in\cV$ and $v'\in \cV'$. 
\begin{nexercise}
  \label{exercise_distance_analysis_1}
  Let $\cV={m \choose k},{{n-m}\choose 0}$ and $\cV'={m \choose {\le k-d/2}},{{n-m}\choose {\ge d/2}}$ be two subsets of $\cG_1(n,k)$. Verify $\dH(\cV,\cV')=d$. 
\end{nexercise} 

Our counting formula for $k$-spaces in Equation~(\ref{eq_count_k_spaces}) can be refined to prescribed pivot vectors. To this end, let the 
\emph{Ferrers tableaux} $T(U)$ of $U$ arise from $E(U)$ by removing the zeroes from each row of $E(U)$ left to the pivots and afterwards removing 
all pivot columns. If we then replace all remaining entries by dots we obtain the \emph{Ferrers diagram} $\cF(U)$ of $U$ which only depends on the 
identifying vector $v(U)$. 
\begin{nexample}
  For the subspace $U$ from Example~\ref{ex_unique_generator_matrix} we have   
  $$
  T(U)=
  \begin{pmatrix}
   0&1&0&0&0\\
    &0&1&1&1\\
    &0&0&1&0\\
    & &1&0&1
  \end{pmatrix}
  \quad\text{and}\quad
  \cF(U)=
  \begin{array}{lllll}
    \bullet & \bullet & \bullet & \bullet & \bullet \\
            & \bullet & \bullet & \bullet & \bullet \\
            & \bullet & \bullet & \bullet & \bullet \\
            &         & \bullet & \bullet & \bullet 
  \end{array}.
  $$
\end{nexample}
The partially filled $k\times(n-k)$ matrix $T(U)$ contains all essential information to describe the codeword $U$. The entries in $T(U)$ have no further restrictions besides 
being contained in $\F_q$, which is reflected by the notation $\cF(U)$. Indeed, every different choice gives a different $k$-dimensional subspace in $\F_q^n$. So, the pivot 
vector $v(U)$ and the Ferrers diagram $\cF(U)$ of $U$ both partition $\cG_q(n,k)$ into specific classes. As indicated before, these classes are not preserved by permutations 
of the coordinates. If $n$ is given, $v(U)$ and $\cF(U)$ can be converted into each other.\footnote{The only issue occurs for pivot vectors $v(U)$ starting with a sequence of zeroes 
corresponding to the same number of leading empty columns in the Ferrers diagram. The latter, or their number, may not be directly visible.} So, we also write $v(\cF)$ for a given 
Ferrers diagram and $\cF(\bu)$ for a given vector $\bu\in\F_2^n$.

Denoting the number of dots in $\cF(\bu)$ by $\#\cF(\bu)$ we can state that the number of 
$\wt(\bu)$-spaces in $\F_q^n$ is given by $q^{\#\cF(\bu)}$.    

\begin{nexercise}
  Show that for $\bu\in\F_2^n$ we have $\#\cF(\bu)=\sum_{i=1}^n u_i\cdot \sum_{j=i+1}^n \left(1-u_j\right)$.
\end{nexercise}

For two $k$-spaces with the same pivot vector Equation~(\ref{eq_subspace_distance_rank}) can be used to relate the subspace distance with the rank distance 
of the corresponding generator matrices:
\begin{nlemma}
  \label{lemma_dist_subspace_rank}
  (\cite[Corollary 3]{silberstein2011large}) For $U,W\in\cG_q(n,k)$ with $v(U)=v(W)$ we have $\ds(U,W)=2\dr(E(U),E(W))$.
\end{nlemma}

\medskip

As we will see later on, a different kind of codes is closely related to subspace codes. For two matrices $U,W\in\F_q^{m\times n}$ the \emph{rank distance} is 
defined as $\dr(U,W)=\rk(U-W)$. As observed e.g.\ in \cite{gabidulin1985theory}, $\dr$ is indeed a metric on the set of $(m\times n)$ matrices over $\F_q$ with values in $\{0,1,\dots,\min\{m,n\}\}$. 
A subset $\cM\subseteq \F_q^{m\times n}$ is called a \emph{rank metric code} ({\rmc}) and by $\dr(\cM):=\min\left\{\dr(A,B)\,:\, A,B\in\cM, A\neq B\right\}$ we denote 
the corresponding \emph{minimum rank distance}. As a shorthand, we speak of an $(m\times n,d)_q$--{\rmc}. We call $\cM$ \emph{additive} if it is additively 
closed and \emph{linear} if it forms a subspace of $\F_q^{m\times n}$. In Section~\ref{sec_rank_metric_codes} we will summarize more details on {\rmc}s that actually 
are part of the preliminaries and relevant for the later sections.

\medskip

For the sake of completeness, we mention a few standard notations that we are using in the following. The \emph{sum} of two sets $A$ and $B$ is given by 
$A+B:=\{a+b\,:\, a\in A,b\in B\}$. For $a\in A$ we also use the abbreviation $a+B$ for $\{a\}+B$.  
\begin{ndefinition} \textbf{(Packings and partitions)}\\
  \label{definition_packing}
  A \emph{packing} $\cP = \left\{P_1,\dots,P_s\right\}$ of a set $X$ is a set of subsets $P_i \subseteq X$ such that $P_i \cap P_j = \emptyset$ for all $1 \le i < 
  j\le s$, i.e., the subsets $P_i$ are pairwise disjoint. The number of elements $s$ is also called the \emph{cardinality} $\#\cP$ of $\cP$. If additionally $\cup_{i=1}^s P_i=X$, 
  then we speak of a \emph{partition}.  
\end{ndefinition}

For packings or partitions of {\cdc}s or {\rmc}s we will need a stronger condition than pairwise disjointness in some applications.
\begin{ndefinition} \textbf{($\mathbf{d}$-packings and $\mathbf{d}$-partitions of codes)}\\
\label{definition_d_packing} 
  A packing $\cP = \left\{P_1,\dots,P_s\right\}$ of a {\cdc} $\cC$ is called \emph{$d$-packing} if $\ds(\cP_i)\ge d$ (and $\cP_i\subseteq \cC$) for all $1\le i\le s$. Similarly, a 
  packing $\cP = \left\{P_1,\dots,P_s\right\}$ of a {\rmc} $\cM$ is called \emph{$d$-packing} if $\dr(\cP_i)\ge d$ (and $\cP_i\subseteq \cM$) for all $1\le i\le s$. 
  If the packings are partitions, then we speak of a $d$-partition in both cases.  
\end{ndefinition} 

\chapter{Rank metric codes}
\label{sec_rank_metric_codes}
Since rank metric codes ({\rmc}s) are closely related to subspace codes, we summarize several facts on ranks of matrices and rank metric codes that will be frequently used 
later on in this chapter. For a broader overview we refer to e.g.\ \cite{gabidulin2021rank} and the references mentioned therein.

Via Equation~(\ref{eq_subspace_distance_rank}) the subspace distance between two spaces $U,W\in \F_q^n$ is linked to the ranks of certain matrices. I.e., if 
$G_U$ and $G_W$ are generator matrices of $U$ and $W$, respectively, then we have
\begin{equation}
  \ds(U,W)=2\rk\left(\begin{pmatrix}G_U\\G_W\end{pmatrix}\right)-\rk(G_U)-\rk(G_W).
\end{equation}
So, we summarize a few equations and inequalities for the rank of a matrix. First note that the operations of the Gaussian elimination algorithm do not change 
the rank of a matrix, which also holds for column permutations.
\begin{nexercise}
  Show that for compatible matrices we have
  \begin{eqnarray*}
    \!\!\!\!\!&&\rk(M)=\rk(M^\perp);\\
    \!\!\!\!\!&&\rk(M) \le \rk\left(\begin{pmatrix}M&M'\end{pmatrix}\right)\le \rk(M)+\rk(M');\\
    \!\!\!\!\!&&|\rk(M)-\rk(M')|\le\dr(M,M')=|\rk(M-M')|\le \rk(M)+\rk(M');\\
    \!\!\!\!\!&&\rk\left(\begin{pmatrix}M_{1,1} & M_{1,2} & \dots & M_{1,l}\\ 
                            \mathbf{0} & M_{2,2} & \dots & M_{2,l}\\\vdots&\ddots&\ddots&\vdots\\ \mathbf{0}&\dots&\mathbf{0}&M_{l,l}\end{pmatrix} \right) 
    =\sum_{i=1}^l \rk(M_{i,i}) \text{ for } l\ge 1.
  \end{eqnarray*}
\end{nexercise}

\begin{nlemma}
  \textbf{(Singleton bound for rank metric codes -- e.g.\ \cite{gabidulin1985theory})}\\
  \label{let_MRD_size}
  Let $m,n\ge d$ be positive integers, $q>1$ a prime power, and $\cM\subseteq \F_q^{m\times n}$ be a rank metric
  code with minimum rank distance $d$. Then, $\# \cM\le q^{\max\{n,m\}\cdot (\min\{n,m\}-d+1)}$.
\end{nlemma}
Codes attaining this upper bound are called \emph{maximum rank distance} (\mrd) codes. More precisely,  $(m \times n, d)_q$--{\mrd} codes. They exist for all 
(suitable) choices of parameters, which remains true if we restrict to linear rank metric codes, see e.g.\ the survey \cite{sheekey2019mrd}. If $m<d$ or $n<d$, 
then only $\# \cM=1$ is possible, which can be achieved by a zero matrix and may be summarized to the single upper bound
\begin{equation}
  \# \cM\le \left\lceil q^{\max\{n,m\}\cdot (\min\{n,m\}-d+1)}\right\rceil=:A_q^R(m\times n,d). 
\end{equation}

\begin{trailer}{Delsarte--Gabidulin codes \cite{cooperstein1998external,delsarte1978bilinear,gabidulin1985theory,roth1991maximum}}A \emph{linearized polynomial} (over $\F_{q^n}$) is a 
polynomial of type $f_0x + f_1 x^q +\dots+ f_{n-1}x^{q^{n-1}}$ with coefficients $f_i\in \F_{q^n}$. The $q$-degree of a non-zero linearized polynomial is the
maximum $i$ such that $f_i \neq 0$. A rank metric code can be described as a set of linearized polynomials. By $\cL_{k,q,n}$ we denote the set of linearized polynomials of $q$-degree 
at most $k-1$ over $\F_{q^n}$. Now $\dim_{\F_q}\left(\cL_{k,q,n}\right) = nk$, and since every non-zero element of $\cL_{k,q,n}$ has nullity at most $k-1$ it has a rank 
of at least $n-k +1$. Thus, $\cL_{k,q,n}$  gives an $(n\times n,n-k+1)_q$--{\mrd} code. Via puncturing or shortening, see e.g.\  \cite{sheekey2019mrd},  $(m\times n,d)_q$--{\mrd} codes 
can be obtained for the cases $m\neq n$. One might say that Delsarte--Gabidulin codes are the rank metric analogue of Reed-Solomon codes. 
\end{trailer}

%% An $(m\times n,d)_q$--{\rmc} $\cM$ of cardinality $A_q^R(m\times n,d)$ is also called $(m\times n,d)_q$--{\mrd} code. 
In \cite[Section IV.A]{silva2008rank} {\rmc}s were related 
to {\cdc}s via a so-called \emph{lifting construction}, cf. Subsection~\ref{subsec_lifting}. Given a matrix $M\in\F_q^{k\times m}$ its lifting is the $k$-space 
$\left\langle \begin{pmatrix}I_k & M\end{pmatrix}\right\rangle\in\cG_q(k+m,k)$. By lifting a given {\rmc} $\cM$ we understand the {\cdc} $\cC$ arising as the union of 
the liftings of the elements of $\cM$. If $U$ arises from lifting $M$ and $U'$ arises from lifting $M'$, then we have $\ds(U,U')=$
\begin{eqnarray*}
  &&2\rk\left(\begin{pmatrix} I_k& M\\ I_k & M'\end{pmatrix}\right)-\rk\left(\begin{pmatrix}I_k& M\end{pmatrix}\right)-\rk\left(\begin{pmatrix}I_k& M'\end{pmatrix}\right) 
  =2\rk\left(\begin{pmatrix} I_k& M\\ \mathbf{0} & M-M'\end{pmatrix}\right) -2k\\  
  &=& 2\rk(I_k)+2\rk(M-M')-2k=2\dr(M,M'),
\end{eqnarray*}
cf.~Lemma~\ref{lemma_dist_subspace_rank}, so that $\ds(\cC)=2\dr(\cM)$. A {\cdc} obtained from lifting an {\mrd} code is called \emph{lifted {\mrd}} ({\lmrd}) code yielding:
\begin{ntheorem}\textbf{(Lifted {\mrd} code -- \cite{silva2008rank})}\\
\label{theorem_lifted_mrd}
$$
  A_q(m+k,d;k)\ge A_q^R(k\times m,d/2)=q^{\max\{m,k\}\cdot (\min\{m,k\}-d/2+1)}.
$$
\end{ntheorem}   

\medskip

In some applications the ranks of the codewords of a {\rmc} have to lie in some set $R\subseteq \N_0$. Each 
$(m\times n,d)_q$--{\rmc} $\cM$, where $\rk(M)\in R$ for each $M\in \cM$, is called $(m\times n,d;R)_q$--{\rmc}. 
The corresponding maximum possible cardinality is denoted by $A_q^R(m\times n,d;R)$. For a non-negative integer 
$l$ we also use the notations $\le l$ and $[0,l]$ for the set $R=\{0,\dots,l\}$. More generally, we also write $[a,b]$ for 
the interval of integers $\{a,a+1,\dots,b-1,b\}$. 

The number of matrices of given rank $r$ in $\F_q^{m\times n}$ is well known and its determination can be traced back at least to \cite{landsberg1893ueber}. Clearly, these 
numbers yield the exact values of $A_q^R(m\times n,1;R)$ for minimum rank distance $1$.
\begin{nproposition}
  $$
    A_q^R(m\times n,1;R)=\sum_{r\in R} \qbin{m}{r}{q}\!\!\cdot\prod_{i=0}^{r-1} \left(q^n-q^i\right)=\sum_{r\in R} \qbin{n}{r}{q}\!\!\cdot\prod_{i=0}^{r-1} \left(q^m-q^i\right).
  $$
\end{nproposition}  
\begin{ncorollary}
  $$A_q^R(m\times n,1;\le 1)=\frac{\left(q^n-1\right)\left(q^m-1\right)}{q-1}+1.$$
\end{ncorollary}

If a {\mrd} code $\cM$ is additive, then its rank distribution is completely determined by its parameters:  
\begin{ntheorem}\textbf{(Rank distribution of additive {\mrd} codes -- \cite[Theorem 5.6]{delsarte1978bilinear}, \cite[Theorem 5]{sheekey2019mrd})}\\
  \label{theorem_rank_distribution_additive_mrd}
  The number of codewords of rank $r$ in an additive $(m \times n,d)_q$--{\mrd} code is given by $a_q(m\times n,d;r):=$ 
  \begin{equation}
    \qbin{\min\{n,m\}}{r}{q}\sum_{s=0}^{r-d} (-1)^sq^{{s\choose 2}}\cdot\qbin{r}{s}{q}\cdot\left(q^{\max\{n,m\}\cdot(r-d-s+1)}-1\right)
  \end{equation}
  for all $d\le r\le \min\{n,m\}$. 
\end{ntheorem}
Clearly, there is a unique codeword of rank strictly smaller than $d$ -- the zero matrix, which has to be contained in any additive rank metric code.  
%% All other codewords have a rank of at least $d$. 
This may be different for non-additive {\mrd} codes.   
\begin{nexample}
  \label{example_a_q_4_4_2_2}
  For $n=m=4$ and $d=2$ the rank distribution of an additive $(4\times 4,2)_q$--{\mrd} is given by
  \begin{eqnarray*}
    a_q(4\times 4,2;0) &=& 1,\\ 
    a_q(4\times 4,2;1) &=& 0,\\
    a_q(4\times 4,2;2) &=& %% \qbin{4}{2}{q}\cdot\left(q^4-1\right)=
                     q^8 + q^7 + 2q^6 + q^5 - q^3 - 2q^2 - q - 1\\
                 &=& \left(q^2+q+1\right)\left(q^2+1\right)^2(q+1)(q-1),\\
    a_q(4\times 4,2;3) %%&=& \qbin{4}{3}{q}\cdot\left(\,\left(q^8-1\right)-\qbin{3}{1}{q}\cdot\left(q^4-1\right)\,\right)\\ 
                 &=& q^{11} + q^{10} - q^8 - 3q^7 - 3q^6 - q^5 + q^4 + 2q^3 + 2q^2 + q\\
                 &=& \left(q^3 - q - 1\right)\left(q^2 + 1\right)^2(q + 1)^2(q - 1)q,\text{ and}\\  
    a_q(4\times 4,2;4) &=& q^{12} - q^{11} - q^{10} + 2q^7 + q^6 - q^4 - q^3\\
                 &=& \left(q^5 - q^4 - q^3 + q + 1\right)\left(q^2 + 1\right)(q+1)(q-1)q^3.
  \end{eqnarray*}
  Of course, these five terms add up to $A_q^R(4\times 4,2)=q^{12}$.
\end{nexample}

\begin{nlemma}
  \label{lemma_lb_restricted_rank_additive_mrd}
  For each $R\subseteq \N_0$ we have
  $$
    A_q^R(m\times n,d;R)\ge \sum_{r\in R} a_q(m\times n,d;r). 
  $$
\end{nlemma}
The easy observation in Lemma~\ref{lemma_lb_restricted_rank_additive_mrd} is implicitly contained in e.g.\ \cite{xu2018new}.
\begin{nexample}
  \label{example_A_q_R_4_4_2_2}
  From Example~\ref{example_a_q_4_4_2_2} and Lemma~\ref{lemma_lb_restricted_rank_additive_mrd} we directly compute
  \begin{eqnarray*}
    A_q^R(4\times 4,2;0)     &\ge& 1,\\
    A_q^R(4\times 4,2;\le 1) &\ge& 1,\\ 
    A_q^R(4\times 4,2;\le 2) &\ge& q^8 + q^7 + 2q^6 + q^5 - q^3 - 2q^2 - q,\\
    A_q^R(4\times 4,2;\le 3) &\ge& q^{11} + q^{10} - 2q^7 - q^6 + q^4 + q^3 ,\text{ and}\\ 
    A_q^R(4\times 4,2;\le 4) &\ge& q^{12},
  \end{eqnarray*}
  i.e., $A_2(4\times 4,2;0)\ge 1$, $A_2(4\times 4,2;\le 1)\ge 1$, $A_2(4\times 4,2;\le 2)\ge 526$, $A_2(4\times 4,2;\le 3)\ge 2776$, and $A_2(4\times 4,2;\le 4)\ge 4096$.
\end{nexample}
\begin{nexercise}
  Let $m,n,d$ be positive integers and $R\subseteq \N_0$. Show
  \begin{enumerate}
    \item[(1)] $A_q(m\times n,d;0)=1$;
    \item[(2)] $A_q(m\times n,d;R)\le 1$ if $R\subseteq \left[0,\left\lfloor\tfrac{d-1}{2}\right\rfloor\right]$;
    \item[(3)] $A_q(m\times n,d;R')\le A_q(m\times n,d;R)$ if $R'\subseteq R$; and
    \item[(4)] $A_q(m\times n,d;R)=A_q(m\times n,d)$ if $[0,n]\subseteq R$.
  \end{enumerate}   
\end{nexercise}

In order to exploit the inequality $\dr(M,M')\ge |\rk(M)-\rk(M')|$ we define a metric $d$ on subsets of non-negative integers. Specializing the usual 
metric on $\R$ we set $d(s,s')=|s-s'|$ for all $s,s'\in\N_0$. With this, we set $d(S)=\min\{d(s,s'),\:\, s,s'\in S, s\neq s'\}$ and $d(S,S'):=\min\{d(s,s')\,:\,s\in S,s'\in S\}$
for any two arbitrary subsets $S,S'\subseteq\N_0$. Actually we use the two later constructs for any metric, i.e., we also use the notations $\ds(\cC,\cC')$ and $\dr(\cM,\cM')$ 
for the minimum subspace distance between two subspaces from two different {\cdc}s and for the minimum rank-distance between two matrices from two different {\rmc}s.
\begin{nlemma}
  Let $\cM$ be an $(m\times n,d;R)_q$--{\rmc} and $\cM'$ be an $(m\times n,d;R')_q$--{\rmc}. If $d(R,R')\ge d\ge 1$, then $\cM\cup\cM'$ is an 
  $(m\times n,d;R\cup R')_q$--{\rmc} of cardinality $\#\cM+\#\cM'$.
\end{nlemma}
\begin{nexample}
  The union of a $(4\times 3,2;\le 1)_q$--{\rmc} and a $(4\times 3,2;3)_q$--{\rmc} is a $(4\times 3,2;\le 3)_q$--{\rmc}.
\end{nexample}  
$(m\times n,d;R)_q$--{\rmc}s with $R=\{r\}$ are also called \emph{constant rank codes} and their relation to constant dimension codes has e.g.\ been studied in 
\cite{constant_rank_codes,gadouleau2010constant}.
\begin{nlemma}\cite[Proposition 3]{gadouleau2010constant}
\label{lemma_crc_from cdc}
$$
  A_q^R(m\times n,d_1/2+d_2/2;r)\ge \min\left\{A_q(m,d_1;r),A_q(n,d_2,r)\right\}
$$
\end{nlemma} 
\begin{nexample}  
  From Lemma~\ref{lemma_crc_from cdc} we can conclude
  $$
    A_q^R(4\times 4,2;\le 1)\ge A_q^R(4\times 4,2;1)\ge A_q(4,2;1)=\qbin{4}{1}{q}=q^3+q^2+q+1 
  $$  
  and
  $$
    A_q^R(4\times 3,2;1)\ge \min\left\{A_q(4,2;1),A_q(3,2;1)\right\}=\qbin{3}{1}{q}=q^2+q+1.
  $$
\end{nexample}  
\begin{nproposition}\cite[Corollary 4]{gadouleau2010constant}
  \label{proposition_exact_crc_value_1}
  If $1\le r\le \min\{m,n\}$, then we have
  $$
    A_q^R(m\times n,r+1;r)=\qbin{\min\{m,n\}}{r}{q}=A_q(\min\{m,n\},2;r).
  $$
\end{nproposition}

Further lower bounds for $A_q^R(m\times n,d;r)$ can be concluded from the pigeonhole principle. To this end we use the following partitioning result for {\mrd} codes. 
\begin{nlemma}\textbf{(Parallel {\mrd} codes -- \cite[Lemma 5]{etzion2012codes})}\\
  \label{lemma_parallel_mrd}  
  For $d'>d>0$ there exists an $(n\times m,d)_q$--{\mrd} code $\cM$ that can be partitioned in $\alpha:=A_q^R(n\times m,d)/A_q^R(n\times m,d')$ {\rmc}s $\cM_i$ 
  with $\dr(\cM_i)\ge d'$ for $1\le i\le \alpha$.
\end{nlemma}
Let $\cM$ be a linear $(n\times m,d)_q$--{\mrd} code that contains a linear $(n\times m,d')_q$--{\mrd} $\cM'$ as a subcode. 
With this, the set $\left\{ M+\cM' \,:\, M\in\cM\right\}$ is such a partition described in Lemma~\ref{lemma_parallel_mrd}, cf.~Lemma~\ref{lemma_cosets}. 
In terms of Definition~\ref{definition_d_packing} we also speak of a $d'$-partition of $\cM$.
\begin{nexercise}
  \label{exercise_parallel_rmc}
  Prove the following statements in order to deduce Lemma~\ref{lemma_parallel_mrd}.
  \begin{enumerate}
    \item[(1)] Let $\cM$ be an $(n\times m,d)_q$--{\rmc}. For each matrix $M\in\F_q^{n\times m}$ also $M+\cM$ is an $(n\times m,d)_q$--{\rmc} with the same cardinality $\#\cM$.
    \item[(2)] Let $\cM$ be an additive $(n\times m,d)_q$--{\rmc} and $M,M'\in\F_q^{n\times m}$ be arbitrary matrices. We have $M+\cM=M'+\cM$ iff $M'-M\in\cM$ and 
               $(M+\cM)\cap (M'+\cM)=\emptyset$ otherwise. 
    \item[(3)] Let $\cM$ be an $(n\times m,d)_q$--{\rmc} that contains an additive $(n\times m,d')_q$--{\rmc} as a subcode, where $d'\ge d$. Then, 
               $\left\{ M+\cM' \,:\, M\in\cM\right\}$ is a set of $(n\times m,d')_q$--{\rmc}s $\cM_1,\dots,\cM_s$, where $s\ge \#\cM/\#\cM'$ and $\dr(\cM_i,\cM_j)\ge d$ 
               for all $1\le i<j\le s$. Moreover, $\cup_{i=1}^s \cM_i$ is an $(n\times m,d)_q$--{\rmc} of cardinality $s\cdot \#\cM'$.
    \item[(4)] Use the Delsarte--Gabidulin {\mrd}-codes to show that for any positive integers $m$ and $n$ there exists a chain of linear $m\times n$--{\mrd}-codes 
                $\cM_1\subseteq \cM_2\subseteq\dots$ such that $\cM_i$ has minimum rank distance $i$ for all $1\le i\le min\{n,m\}$.                
  \end{enumerate}  
\end{nexercise}
\begin{nremark}
Note that there are examples of {\mrd} codes with minimum rank distance $d$ which cannot be extended to an {\mrd} code with minimum rank distance $d+1$, see   
e.g.\ \cite[Section 1.6]{sheekey2016new} and \cite[Example 34]{byrne2017covering}. In \cite[Theorem 9]{sheekey2019binary} it was shown that every binary additive 
{\mrd} code with minimum rank distance $n-1$ contains a binary additive {\mrd} code with minimum rank distance $n$ as a subcode.
\end{nremark}

\begin{nlemma}
  \label{lemma_lb_restricted_rank_additive_mrd_average}
  For each $R\subseteq \N_0$ we have
  $$
    A_q^R(m\times n,d;R)\ge \max_{1\le d'\le d} \frac{A_q^R(m\times n,d)}{A_q^R(m\times n,d')}\cdot \sum_{r\in R} a_q(m\times n,d';r). 
  $$
\end{nlemma}
\begin{proof}
  Let $\cM'$ be a linear $(n\times m,d')_q$--{\mrd} code that contains a linear $(n\times m,d)_q$--{\mrd} $\cM$ as a subcode. By $\cM_1,\dots,\cM_\alpha$ we denote the 
  $\alpha:=A_q^R(m\times n,d')/A_q^R(m\times n,d)$ cosets $M+\cM$ of $\cM$ in $\cM'$. By the pigeonhole principle there exists an index $1\le i\le \alpha$ such that
  $\#\{M\in\cM_i\,\rk(M)\in R\}\ge\tfrac{1}{\alpha}\cdot \#\{M\in \cM'\,:\, \rk(M)\in R\}$.  
\end{proof}

\begin{nexample}
  From Theorem~\ref{theorem_rank_distribution_additive_mrd} we compute $a_q(4\times 4,1;1) = q^7 + q^6 + q^5 + q^4 - q^3 - q^2 - q - 1$, so that
  $$
    A_q^R(4\times 4,2;1)\ge \left\lceil\frac{a_q(4\times 4,1;1)}{q^4}\right\rceil= q^3 + q^2 + q^1+ \left\lceil\frac{q^4-q^3-q^2-q}{q^4}\right\rceil=\qbin{4}{1}{q}.
  $$
  Due to Proposition~\ref{proposition_exact_crc_value_1} this lower bound is tight. Note that $\rk(M'-M)\le \rk(M)+\rk(M')$ implies $A_q^R(4\times 4,2;\le 1)=A_q^R(4\times 4,2;1)$. For 
  $A_q^R(4\times 4,2;\le 2)$ and $A_q^R(4\times 4,2;\le 3)$ Lemma~\ref{lemma_lb_restricted_rank_additive_mrd_average} yields a weaker lower bound than  
  Lemma~\ref{lemma_lb_restricted_rank_additive_mrd}.  
  %% 
  %% \begin{eqnarray*}
  %%   a_q(4\times 4,1;0) &=& 1\\ 
  %%   a_q(4\times 4,1;1) &=& q^7 + q^6 + q^5 + q^4 - q^3 - q^2 - q - 1 \\
  %%   a_q(4\times 4,1;2) &=& q^{12} + q^{11} + 2q^{10} - q^8 - 3q^7 - 3q^6 - q^5 + 2q^3 + q^2 + q\\
  %%   a_q(4\times 4,1;3) &=& q^{15} + q^{14} - q^{12} - 3q^{11} - 2q^{10} + 2q^8 + 3q^7 + q^6 - q^4 - q^3 \\
  %%   a_q(4\times 4,1;4) &=& q^{16} - q^{15} - q^{14} + 2q^{11} - q^8 - q^7 + q^6 \\ 
  %%   a_q(4\times 4,2;0) &=& 1\\ 
  %%   a_q(4\times 4,2;1) &=& 0\\
  %%   a_q(4\times 4,2;2) &=& q^8 + q^7 + 2q^6 + q^5 - q^3 - 2q^2 - q - 1\\
  %%   a_q(4\times 4,2;3) &=& q^{11} + q^{10} - q^8 - 3q^7 - 3q^6 - q^5 + q^4 + 2q^3 + 2q^2 + q\\
  %%   a_q(4\times 4,2;4) &=& q^{12} - q^{11} - q^{10} + 2q^7 + q^6 - q^4 - q^3\\
  %% \end{eqnarray*}
\end{nexample}
Removing the coset $\mathbf{0}+\cM=\cM$ from the consideration yields a slightly different variant of Lemma~\ref{lemma_lb_restricted_rank_additive_mrd_average}:
\begin{nlemma}
  \label{lemma_avg_lb_rk_i}
  For each $R\subseteq \N_0$ we have $A_q^R(m\times n,d;R)\ge$
  $$
    \max_{1\le d'< d}  \frac{1}{A_q^R(m\times n,d')/A_q^R(m\times n,d)-1}\cdot \sum_{r\in R} \Big(a_q(m\times n,d';r)-a_q(m\times n,d;r)\Big). 
  $$
\end{nlemma}
\begin{ncorollary}(Cf.~\cite[Proposition 2.4]{liu2020parallel})
  If $m\le n$ and $r<d$, then we have
  $$
    A_q^R(m\times n,d;\le r)\ge \max_{1\le d'< d} \frac{1}{q^{d-d'}-1}\cdot \sum_{1\le i\le r} a_q(m\times n,d';i).
  $$
\end{ncorollary} 

\begin{nexample}
  \label{example_A_R_5_5_2_le_3}
  We compute
  \begin{eqnarray*}
    a_q(5\times 5,1;0) &=& 1,\\ 
    a_q(5\times 5,1;1) &=& q^9 + q^8 + q^7 + q^6 + q^5 - q^4 - q^3 - q^2 - q - 1,\\
    a_q(5\times 5,1;2) &=& q^{16} + q^{15} + 2q^{14} + 2q^{13} + q^{12} - q^{11} - 2q^{10} - 4q^9 - 4q^8 \\ 
                       & & - 2q^7 - q^6 + q^5 + 2q^4 + 2q^3 + q^2 + q,\\
    a_q(5\times 5,1;3) &=& q^{21} + q^{20} + 2q^{19} + q^{18} - 3q^{16} - 4q^{15} - 5q^{14} - 3q^{13} + 3q^{11} \\ 
                       & & + 5q^{10} + 4q^9 + 3q^8 - q^6 - 2q^5 - q^4 - q^3,\\ 
    a_q(5\times 5,1;4) &=& q^{24} + q^{23} - q^{21} - 2q^{20} - 3q^{19} - 2q^{18} + q^{17} + 3q^{16} + 4q^{15} \\ 
                       & & + 3q^{14} + q^{13} - 2q^{12} - 3q^{11} - 2q^{10} - q^9 + q^7 + q^6,\\ 
    a_q(5\times 5,1;5) &=& q^{25} - q^{24} - q^{23} + q^{20} + q^{19} + q^{18} - q^{17} - q^{16} - q^{15},\\ 
                       & & + q^{12} + q^{11} - q^{10},\\
    a_q(5\times 5,2;0) &=& 1,\\                                                                                    
    a_q(5\times 5,2;1) &=& 0,
    \end{eqnarray*}
    \begin{eqnarray*}
    a_q(5\times 5,2;2) &=& q^{11} + q^{10} + 2q^9 + 2q^8 + 2q^7 - 2q^4 - 2q^3 - 2q^2 - q - 1,\\
    a_q(5\times 5,2;3) &=& q^{16} + q^{15} + 2q^{14} + q^{13} - 3q^{11} - 4q^{10} - 6q^9 - 4q^8 - 2q^7 + q^6 \\ 
                       & &  + 3q^5 + 4q^4 + 3q^3 + 2q^2 + q,\\ 
    a_q(5\times 5,2;4) &=& q^{19} + q^{18} - q^{16} - 2q^{15} - 3q^{14} - 2q^{13} + q^{12} + 3q^{11} + 5q^{10} \\ 
                       & & + 4q^9 + 2q^8 - q^7 - 2q^6 - 3q^5 - 2q^4 - q^3, \text{ and}\\ 
    a_q(5\times 5,2;5) &=& q^{20} - q^{19} - q^{18} + q^{15} + q^{14} + q^{13} - q^{12} - q^{11} - 2q^{10} + q^7 + q^6.                                              
  \end{eqnarray*}  
  So, choosing $d'=1$ in Lemma~\ref{lemma_avg_lb_rk_i} gives $A_q^R(5\times 5,2;\le 3)$
  \begin{eqnarray*}
     &\ge& \frac{1}{q^5-1} \cdot \sum\limits_{r=1}^3 \left(a_q(5\times 5,1;r)-a_q(5\times 5,2;r)\right) \\
    &=& \left(q^4 + q^3 + q^2 + q + 1\right)\cdot \left(q^9 + q^7 - q^6 - q^5 - q^4 - q^3 + q^2 + q + 1\right)\cdot q^3\\ 
    &=& q^{16} + q^{15} + 2q^{14} + q^{13} - 2q^{11} - 3q^{10} - 3q^9 - q^8 + q^7 + 2q^6 + 3q^5 + 2q^4 + q^3. 
  \end{eqnarray*}
\end{nexample}
We remark that Lemma~\ref{lemma_lb_restricted_rank_additive_mrd} gives only
$$
  A_q^R(5\times 5,2;\le 3) \ge q^{11} + q^{10} + 2q^9 + 2q^8 + 2q^7 - 2q^4 - 2q^3 - 2q^2 - q.
$$

\medskip

Sometimes we want to control the possible ranks of submatrices of the elements in a {\rmc}. By suitably choosing the {\rmc}s $\cM_i$ this is e.g.\ possible via:
\begin{nlemma}\textbf{(Product construction for rank metric codes)}\\
  \label{lemma_product_construction_rmc}
  Let $l\ge 1$ and $\bar{n}=\left(n_1,\dots,n_l\right)\in\N^l$. For $1\le i\le l$ let $\cM_i$ be a $(k\times n_i,d)_q$--{\rmc}. With this, 
  $$
    \cM=\left\{\begin{pmatrix}M_1&\dots&M_l\end{pmatrix}\,:\, M_i\in\cM_i\,\forall 1\le i\le l\right\}
  $$
  is a $(k\times n,d)_q$--{\rmc} with cardinality $\#\cM=\prod_{i=1}^l \#\cM_i$, where $n=\sum_{i=1}^l n_i$.
\end{nlemma}  
\begin{proof}
  It suffices to show $\dr(\cM)\ge d$. To this end let $M=\begin{pmatrix}M_1\dots&M_l\end{pmatrix}$ and $M'=\begin{pmatrix}M_1'\dots&M_l'\end{pmatrix}$ 
  be two different codewords in $\cM$. Since $M\neq M'$, there exists an index $1\le i\le l$ with $M_i\neq M_i'$, so that
  $\dr(M,M')=\rk\left(\begin{pmatrix}M_1-M_1'&\dots&M_l-M_l'\end{pmatrix}\right)\ge \rk\left(M_i-M_i'\right)=\dr(M_i,M_i')\ge \dr(\cM_i)\ge d$.
\end{proof}
As abbreviation we write $\cM=\cM_1\times \dots\times \cM_l$ for an {\rmc} obtained by the product construction. Another variant can be used to combine several 
{\rmc}s to an {\rmc} with a larger minimum rank distance.
\begin{nlemma}\textbf{(Diagonal concatenation of rank metric codes})\\
  \label{lemma_concatenation_rmc}
   Let $\cM_1$ be a $(k_1\times n_1,d_1)_q$--{\rmc}, $\cM_2$ be a $(k_2\times n_2,d_2)_q$--{\rmc}, and $M_1^1,\dots,M_1^{s_1}$, $M_2^1,\dots,M_2^{s_2}$ arbitrary 
   enumerations of $\cM_1$ and $\cM_2$, respectively. Then. 
  $$
    \cM=\left\{\begin{pmatrix}M_1^i&\mathbf{0}_{k_1\times n_2}\\ \mathbf{0}_{k_2\times n_1} & M_2^i\end{pmatrix}\,:\, 1\le i\le \min\{s_1,s_2\}\right\}
  $$
  is a $\left((k_1+k-2)\times (n_1+n_2),d_1+d_2\right)_q$--{\rmc} with cardinality $\#\cM=\min\left\{\#\cM_1,\right.$ $\left.\#\cM_2\right\}$.
\end{nlemma}   
\begin{proof}
  Let $G=\begin{pmatrix}M_1&\mathbf{0}\\\mathbf{0} & M_2\end{pmatrix}$ and $G'=\begin{pmatrix}M_1'&\mathbf{0}\\\mathbf{0} & M_2'\end{pmatrix}$ be two different 
  elements in $\cM$. By construction, $G\neq G'$ implies $M_1\neq M_1'$ and $M_2\neq M_2'$, so that $\dr(G,G')=$
  $$
    \rk(G-G')=\rk\left(\begin{pmatrix}M_1'-M_1&\mathbf{0}\\\mathbf{0}&M_2'-M_2\end{pmatrix}\right)=\rk(M_1'-M_1)+\rk(M_2'-M_2)\ge d_1+d_2,
  $$ 
  i.e., $\dr(\cM)\ge d_1+d_2$.
\end{proof}
We remark that the iterative application of Lemma~\ref{lemma_concatenation_rmc} results in a 
$\left(k\times n,d\right)_q$--{\rmc} $\cM$ with cardinality $\min\{ \#\cM_i\,:\, 1\le i\le \}$ 
given $(k_i\times n_i,d_i)_q$--{\rmc}s $\cM_i$ for $1\le i\le l$, where $l\ge 1$, $n=\sum_{i=1}^l n_i$, $d=\sum_{i=1}^l d_i$, and  $k=\sum_{i=1}^l k_i$.

\begin{trailer}{Sum-rank metric codes}
In the following we want to consider restrictions on the ranks of different submatrices of a rank metric code. It turns out that those restrictions fit into 
the framework of sum-rank metric codes that were already used for space-time coding, see e.g.\ \cite{el2003design,space_time_coding}. For positive integers 
$t$, $m_1,\dots,m_t$, $n_1,\dots,n_t$ consider the product of $t$ matrix spaces
$$
  \Pi:= \bigoplus_{i=1}^t \F_{q}^{m_i\times n_i}
$$
and define the \emph{sum-rank} of an element $X=(X_1,\dots,X_t)\in\Pi$ as
\begin{equation}
  \srk(X):=\sum_{i=1}^t \rk(X_i).
\end{equation}
\begin{nexercise}
  Show that the sum-rank induces a metric on $\Pi$ via $(X,Y)\mapsto \srk(X-Y)$.
\end{nexercise}
\begin{ndefinition}
  A subset $\cM\subseteq \Pi$ is called a \emph{sum-rank metric code} ({\srmc}) and by $\dsr(\cM):=\min\left\{\dsr(A,B)\,:\, A,B\in\cM, A\neq B\right\}$ we denote 
  the corresponding \emph{minimum sum-rank distance}. We call $\cM$ \emph{additive} if it is additively closed and \emph{linear} if it forms a subspace of $\Pi$. 
  By $A_q^r(m_1\times n_1,\dots, m_t\times n_t,d)$ we denote the corresponding maximum possible cardinality for minimum sum-rank distance $d$. If we additionally 
  require that the sum-ranks of the elements in $\cM$ have to be contained in a set $R\subset\N_0$, then we denote the corresponding maximum possible 
  cardinality by $A_q^r(m_1\times n_1,\dots, m_t\times n_t,d;R)$.
\end{ndefinition}
\end{trailer}

In the following we will state two explicit construction for {\srmc}s and refer to e.g.~\cite{byrne2021fundamental} for further results.
%%\begin{ndefinition}
%%   An $\left(m_1\times n_1, m_2\times n_2,d;R\right)_q$--{\rmc} is a set $\cM\subseteq \F_q^{m_1\times n_1}\times \F_q^{m_2\times n_2}$ of pairs of matrices such that 
%%   $\rk(M_1)+\rk(M_2)\in R$ for all $\left(M_1,M_2\right)\in \cM$ and $\dr(M_1,M_1')+\dr(M_2,M_2')\ge d$ for all $\left(M_1,M_2\right),\left(M_1',M_2'\right)\in \cM$ 
%%   with $\left(M_1,M_2\right)\neq\left(M_1',M_2'\right)$. By $A_q^r(m_1\times n_1,m_2\times n_2,d;R)$ we denote the corresponding maximum possible cardinality. 
%%\end{ndefinition}

\begin{nlemma}
  \label{lemma_pair_rmc_construction_product}
  Let $\cM_1$ be an $\left(m_1\times n_1,d;R_1\right)_q$--{\rmc} and $\cM_2$ be an $\left(m_2\times n_2,d;\right.$ $\left.R_2\right)_q$--{\rmc}. Then, there exists an 
  $\left(m_1\times n_1,m_2\times n_2,d;R_1+R_2\right)_q$--{\srmc} with cardinality $\#\cM=\#\cM_1\cdot \#\cM_2$. 
\end{nlemma}
\begin{proof}
  Let $\cM=\left\{\left(M_1,M_2\right)\,:\, M_1\in\cM_1,M_2\in\cM_2\right\}$, so that $\#\cM=\#\cM_1\cdot\#\cM_2$. Consider arbitrary elements $\left(M_1,M_2\right), \left(M_1',M_2'\right) 
  \in \cM$ with $\left(M_1,M_2\right)\neq\left(M_1',M_2'\right)$. If $M_1\neq M_1'$, then we have 
  \begin{eqnarray*}
    \dsr\big((M_1,M_2),(M_1',M_2')\big)&=&\dr(M_1,M_1')+\dr(M_2,M_2')\\ 
    &\ge& \dr(M_1,M_1') \ge \dr(\cM_1)\ge d.
  \end{eqnarray*}   
  If $M_1=M_1'$, then we  
  have $M_2\neq M_2'$ and 
  \begin{eqnarray*}
    \dsr\big((M_1,M_2),(M_1',M_2')\big)&=& \dr(M_1,M_1')+\dr(M_2,M_2')\\ 
    &\ge& \dr(M_2,M_2')\ge \dr(\cM_2)\ge d.
  \end{eqnarray*}  
\end{proof}

\begin{nlemma}
  \label{lemma_pair_rmc_construction_concatenation}
  Let $\cM_1$ be an $\left(m_1\times n_1,d_1;R_1\right)_q$--{\rmc} and $\cM_2$ be an $\big(m_2\times n_2,d_2;$ $R_2\big)_q$--{\rmc}. Then, there exists an 
  $\left(m_1\times n_1,m_2\times n_2,d_1+d_2;R_1+R_2\right)_q$--{\srmc} with cardinality $\#\cM=\min\left\{\#\cM_1,\#\cM_2\right\}$. 
\end{nlemma}
\begin{proof}
  Let $M_1^1,\dots,M_1^s$ be an arbitrary numbering of the elements of $\cM_1$ and $M_2^1,\dots,M_2^r$ be an arbitrary numbering of the elements of $\cM_2$. With this we set 
  $$
    \cM=\left\{\left(M_1^i,M_2^i\right)\,:\, 1\le i\le \min\{s,r\}\right\},
  $$ 
  so that $\#\cM=\min\left\{\#\cM_1,\#\cM_2\right\}$. Let $\left(M_1,M_2\right)\in\cM$ be an 
  arbitrary element. By construction we have $\rk(M_1)+\rk(M_2)\in R_1+R_2$. Let $\left(M_1',M_2'\right)\in\cM$ be another element with $\left(M_1,M_2\right)\neq \left(M_1',M_2'\right)$. 
  Then, we have $M_1\neq M_1'$ and $M_2\neq M_2'$, so that 
  \begin{eqnarray*}
    \dsr\big((M_1,M_2),(M_1',M_2')\big)&=& \dr(M_1,M_1')+\dr(M_2,M_2')\\ 
    &\ge& \dr(\cM_1)+\dr(\cM_2)\ge d_1+d_2.
  \end{eqnarray*}  
\end{proof}

\begin{nlemma}
  \label{lemma_combine_pair_rmc}
  For $M_1,M_1'\in \F_q^{m_1\times n_1}$ and $M_2,M_2'\in\F_q^{m_2\times n_2}$ we have 
  $$
    \dr(M_1,M_1')+\dr(M_2,M_2')\ge \left|\rk(M_1)-\rk(M_1')\right|+  \left|\rk(M_2)-\rk(M_2')\right|.
  $$
\end{nlemma}

\begin{nexample}
  \label{example_pair_rmc}
  Applying Lemma~\ref{lemma_pair_rmc_construction_concatenation} to a $(3\times 3,1;0)_q$--{\rmc} and a $(3\times 3,2;0)_q$--{\rmc} yields a $(3\times 3,3\times 3,3,0)_q$--{\srmc} 
  $\cM_1$ of cardinality $1$. Applying Lemma~\ref{lemma_pair_rmc_construction_concatenation} to a $(3\times 3,1;1)_q$--{\rmc} and a $(3\times 3,2;2)_q$--{\rmc} yields a 
  $(3\times 3,3\times 3,3,3)_q$--{\srmc} $\cM_2$ of cardinality $$\min\left\{A_q^R(3\times 3,1;1),A_q^R(3\times 3,2;2)\right\}\ge \qbin{3}{1}{q}\!\!\!\cdot\left(q^3-1\right)=
  q^5 + q^4 + q^3 - q^2 - q - 1.$$ Applying Lemma~\ref{lemma_pair_rmc_construction_product} to a $(3\times 3,3;3)_q$--{\rmc} and a $(3\times 3,3;0)_q$--{\rmc} yields a 
  $(3\times 3,3\times 3,3,3)_q$--{\srmc} $\cM_3$ of cardinality $q^3\cdot 1=q^3$. From Lemma~\ref{lemma_combine_pair_rmc} we conclude that $\cM=\cM_1\cup\cM_2\cup\cM_3$ 
  is a   $(3\times 3,3\times 3,3,\le 3)_q$--{\srmc}, so that $A_q^R(3\times 3,3\times 3,3,\le 3)\ge q^5 + q^4 + 2q^3 - q^2 - q$, i.e., $A_2^R(3\times 3,3\times 3,3,\le 3)\ge 58$ for $q=2$. 
\end{nexample}
We remark that Example~\ref{example_pair_rmc} will be explicitly used in the construction for a {\cdc} considered in Example~\ref{example_using_pair_rmc}.

\chapter{Upper bounds for constant dimension codes}
\label{sec_bounds_cdc}
In this section we want to survey upper bounds for $A_q(n,d;k)$ and variants thereof.  For other surveys on upper bounds for constant dimension codes we refer 
e.g.\ to \cite{heinlein2017asymptotic,khaleghi2009subspace}. A relatively good and simple upper bound is given by the so-called \emph{anticode bound}, see 
Theorem~\ref{thm_anticode}:
\[
  A_q(n,d;k)\le \frac{\qbin{n}{k}{q}}{\qbin{\max\{k,n-k\}+d/2-1}{d/2-1}{q}}. 
\]
Using the probabilistic approach one can show that the upper bound can be asymptotically attained if either $k$ or $n-k$ stays constant when $n$ tends to infinity. 
\begin{ntheorem}\textbf{(Asymptotic value -- \cite[Theorem~4.1]{MR829351}, cf.~\cite{blackburn2012asymptotic})}
 $$\lim\limits_{n\to\infty}  \frac{\qbin{n}{k}{q}}{\qbin{\max\{k,n-k\}+d/2-1}{d/2-1}{q} \cdot A_q(n,d;k)}=
 \lim\limits_{n\to\infty}  \frac{\qbin{n}{k}{q}}{\qbin{\max\{k,n-k\}+d/2-1}{d/2-1}{q} \cdot A_q(n,d;n-k)}=1$$
\end{ntheorem}
Besides the two exact values $A_2(6,4;3)=77$ and $A_2(8,6;4)=257$ obtained by extensive integer linear programming computations \cite{heinlein2019classifying,hkk77}, all currently best known upper 
bounds for $A_q(n,d;k)$ are obtained from the the improved Johnson bound in Theorem~\ref{thm:johnson_improved} if $d\neq 2\cdot \min\{k,n-k\}$. A key ingredient is a classification 
result for the possible cardinalities of $q^{k-1}$-divisible multisets of points and the necessary background is briefly presented in Subsection~\ref{subsec_ub_divisible_codes}. 
Since the (improved) Johnson bound is recursive, the case of partial $k$-spreads, i.e.\ $A_q(n,2\cdot \min\{k,n-k\};k)$, has to be treated seperately. 
In principle all currently known upper bounds for partial $k$-spreads are implied by non-existence results for $q^{k-1}$-divisible sets of points, see e.g.\ \cite{kurz2021divisible} for a survey. 
We summarize the major parametric upper bounds in Subsection~\ref{subsec_bounds_partial_spreads}. Before we survey $q$-analogs of the classical upper bounds for binary constant weight codes 
in Subsection~\ref{subsec_bounds_constant_weight_q_analog} and other approaches from the literature in Subsection~\ref{subsec_bounds_other}, we state estimates for 
$q$-binomial coefficients and the ratio between the best known lower and upper bounds.

\medskip
    
A relatively good and simple lower bound is given by the construction of lifted {\mrd} codes in Theorem~\ref{theorem_lifted_mrd}:
\begin{equation}
  A_q(n,d;k)\ge q^{\max\{k,n-k\}\cdot\left(\min\{k,n-k\}-\tfrac{d}{2}+1\right)}\!.
\end{equation}  
In \cite[Lemma 4]{koetter2008coding} the bounds $1< q^{-k(n-k)}\cdot\qbin{n}{k}{q}<4$ were shown. The corresponding proof itself and 
associated remarks actually give a refined upper bound.

\begin{trailer}{$\mathbf{q}$-Pochhammer symbol}The $q$-analog of the \emph{Pochhammer symbol} is the \emph{$q$-Pochhammer symbol}
\begin{equation}
  (a;q)_n:=\prod_{i=0}^{n-1} \left(1-aq^i\right)
\end{equation}  
with $(a;q)_0=1$ by definition. In the theory of basic hypergeometric series (or $q$-hypergeometric series), the $q$-Pochhammer symbol plays the role that the ordinary 
Pochhammer symbol plays in the theory of generalized hypergeometric series. It can be extended to an infinite product $(a;q)_\infty=\prod_{i=0}^{\infty} \left(1-aq^i\right)$. 
Setting $a=q$ this is an analytic function of $q$ in the interior of the unit disk and can also be considered as a formal power series in $q$, whose reciprocal is the 
generating function of integer partitions, see e.g.\ \cite[Chapter 15]{van2001course}.  
\end{trailer}

Here we specialize the $q$-Pochhammer symbol to $(1/q;1/q)_n=\prod_{i=1}^{n}\left(1-1/q^i\right)$ and state the bounds
\begin{equation}
  \label{ie_q_binomial_coefficient}
  1\le \frac{\qbin{n}{k}{q}}{q^{k(n-k)}} \le \frac{1}{(1/q;1/q)_k}<\frac{1}{(1/q;1/q)_\infty}\le \frac{1}{(1/2;1/2)_\infty}\approx 3.4627, 
\end{equation} 
see \cite[Section 5]{heinlein2017asymptotic}.
\begin{nexercise}
  Show that the sequence $(1/q;1/q)_\infty$ is monotonically increasing with $q$ and approaches $(q-1)/q$ for large $q$.
\end{nexercise}
\begin{nexercise}
  \label{exercise_q_binomial_coefficient_limit}
  Show $\lim\limits_{a\to\infty} \frac{\qbin{a+b}{b}{q}}{q^{ab}}=\frac{1}{(1/q;1/q)_b}$ for each $b\in \N_{\ge 0}$.
\end{nexercise}   
\begin{table}[htp]
  \begin{center}\tiny
    \begin{tabular}{c|ccccccccccccccc}
    $q$ & 2 & 3 & 4 & 5 & 7 & 8 & 9 & 11 & 16 & 32 & 64 & 128 & 256 & 512 \\
    \hline 
    $1/(1/q;1/q)_\infty$ & 3.46 & 1.79 & 1.45 & 1.32 & 1.20 & 1.16 & 1.14 & 1.11 & 1.07 & 1.03 & 1.02 & 1.01 & 1.004 & 1.002\\   
    \end{tabular}
    \caption{Approximate values of $1/(1/q;1/q)_\infty$ for selected field sizes.}
  \end{center}
\end{table}

In \cite{koetter2008coding} it was observed that the Singleton bound implies $A_q(n,d;k)\le 4\cdot q^{(n-k)\cdot (k-d/2+1)}$ for $2k\le n$, 
i.e., {\lmrd} codes are at most a factor of four ($2$ bits) distant to optimal codes. Using the $q$-Pochhammer symbol this estimate can be 
slightly improved. A further improvement is possible using the anticode bound.  
\begin{nproposition}(\cite[Proposition 7]{heinlein2017asymptotic})\\
  %%For $k \le n-k$ the ratio of the 
  The size of an {\lmrd} code divided by the size of the Singleton bound in Theorem~\ref{thm_singleton_bound} converges 
  for $n \rightarrow \infty$ monotonically decreasing to 
  $$
    (1/q;1/q)_{k-d/2+1}> (1/q;1/q)_{\infty}\ge (1/2;1/2)_\infty > 0.288788.
  $$  
\end{nproposition}

\begin{nproposition}(\cite[Proposition 8]{heinlein2017asymptotic})\\
  \label{prop_ratio_lmrd_anticode}
  $\!\!\!\!\!\!$
  %%For $k \le n-k$ the ratio of the 
  The size of an {\lmrd} code divided by the size of the anticode bound in Theorem~\ref{thm_anticode} converges 
  for $n \rightarrow \infty$ monotonically decreasing to 
  $$
    \frac{(1/q;1/q)_{k}}{(1/q;1/q)_{d/2-1}}\ge \frac{q}{q-1}\cdot (1/q;1/q)_{k}
    \ge 2\cdot (1/2;1/2)_\infty > 0.577576.
  $$
\end{nproposition}
The largest gap of the latter estimate is attained for $d=4$ and $k=\left\lfloor n/2\right\rfloor$.  

\begin{ncorollary}
$$\lim\limits_{q\to\infty}  \frac{\qbin{n}{k}{q}}{\qbin{\max\{k,n-k\}+d/2-1}{d/2-1}{q} \cdot A_q(n,d;k)}=1$$
\end{ncorollary}

\medskip

While we prefer to state most statements on bounds for $A_q(n,d;k)$ without restrictions for the involved parameters, we mention some easy relations. Due to 
$A_q(n,d;k)=A_q(n,d;n-k)$, see Equation~(\ref{eq_a_orthogonal_cdc}), we may assume $2k\le n$. Since the codewords of an $(n,d;k)_q$--{\cdc} are contained in $\cG_q(n,k)$, 
we have $A_q(n,d;k)\le \qbin{n}{k}{q}$. For minimum subspace distance $d=2$ 
this upper bound is tight, i.e., $\cC=\cG_q(n,k)$ is an $(n,2;k)_q$--{\cdc} with cardinality $\qbin{n}{k}{q}$. So, for $d$ it suffices to consider 
even values between $4$ and $2k$ only, i.e.\ we especially can assume $k\ge 2$ and $n\ge 4$. Since the maximum size of a code with certain parameters is always an integer and some of the
latter upper bounds can produce non-integer values, we may always round them down. To ease the notation we will mostly omit the final rounding step.

\section{q-analogs of upper bounds for binary constant weight codes}
\label{subsec_bounds_constant_weight_q_analog}
Here we want to study the $q$-analogs of the classical upper bounds for binary constant weight codes before we briefly discuss other approaches from the literature 
in Subsection~\ref{subsec_bounds_other}.

\begin{trailer}{Grassmann graph}The vertices of the \emph{Grassmann graph} $J_q (n ,k)$, named after Hermann G\"unther Gra\ss mann, are the $\qbin{n}{k}{q}$ 
$k$-spaces in $\F_q^n$ where two vertices are adjacent when their intersection is $(k-1)$-dimensional. Grassmann graphs are $q$-analogs of \emph{Johnson graphs} 
and \emph{distance-regular}\footnote{A distance-regular graph is a regular graph such that for any two vertices $v$ and $w$, the number of vertices at distance $j$ 
from $v$ and at distance $k$ from $w$ depends only upon $j$, $k$, and the distance $i$ between $v$ and $w$.}. 
\end{trailer} 

Note that $\dim(U\cap W)\ge k-t$ is equivalent to $\ds(U,W)\le m-k+2t$. The fact that the Grassmann graph is distance-regular implies a sphere-packing bound. 
To this end we count $k$-dimensional subspaces having a {\lq\lq}large{\rq\rq} intersection with a fixed $m$-dimensional subspace:
\begin{nexercise}
  \label{exercise_gsphere}
  Show that for integers $0\le t\le k\le n$ and $k-t\le m\le n$ we have
  \[
    \# \left\{U\in \spaces{V}{k} \mid \dim(U\cap W)\ge k-t \right\}=
    \sum_{i=0}^{t} q^{(m+i-k)i} \qbin{m}{k-i}{q} \qbin{n-m}{i}{q}\!\!,
  \]
  where $V=\F_q^n$, $W\le V$, and $\dim(W)=m$. 
\end{nexercise} 
\begin{ntheorem}\textbf{(Sphere-packing bound -- \cite[Theorem~6]{koetter2008coding})}\\
  \label{thm_sphere_packing}
  \[
    A_q(n,d;k)\le \frac{\qbin{n}{k}{q}}{\sum\limits_{i=0}^{\left\lfloor (d/2-1)/2\right\rfloor} q^{i^2} \qbin{k}{i}{q} \qbin{n-k}{i}{q}} 
  \]
\end{ntheorem}  
We remark, that we can obtain the denominator of the formula of Theorem~\ref{thm_sphere_packing} by setting $m=k$, $2t=d/2-1$ in 
Exercise~\ref{exercise_gsphere} and applying $\qbin{k}{k-i}{q}=\qbin{k}{i}{q}$. The right hand side is symmetric with respect to orthogonal 
complements, i.e., the mapping $k\mapsto n-k$ leaves it invariant.

By defining a puncturing operation one can decrease the dimension of the ambient space and the codewords. Since the minimum distance decreases 
by at most two, we can iteratively puncture $d/2-1$ times, so that $A_q(n,d;k)\le \qbin{n-d/2+1}{k-d/2+1}{q}= \qbin{n-d/2+1}{v-k}{q}$ 
since $A_q(v',2;k')=\qbin{v'}{k'}{q}$. Considering either the code or its dual code gives: 
\begin{ntheorem}\textbf{(Singleton bound -- \cite[Theorem~9]{koetter2008coding})}
  \label{thm_singleton_bound}  
  $$
    A_q(n,d;k)\le \qbin{n-d/2+1}{\max\{k,n-k\}}{q} 
  $$    
\end{ntheorem}
In \cite[Example 2.2.8]{weiss2023linear} a \emph{Singleton polynomial} was used to conclude the anticode bound from Theorem~\ref{thm_anticode}, which is different to the 
bound in Theorem~\ref{thm_singleton_bound}.

\begin{trailer}{Comparison between the Sphere-packing and the Singleton bound}Referring to \cite{koetter2008coding} the authors of \cite{khaleghi2009subspace}  
state that the Singleton bound is always stronger than the sphere packing bound for non-trivial codes. However, for $q=2$, $n=8$, $d=6$, and $k=4$, the sphere-packing 
bound gives an upper bound of $200787/451\approx 445.20399$ while the Singleton bound gives an upper bound of $\qbin{6}{4}{2}=651$. For $q=2$, $n=8$, $d=4$, and $k=4$ 
it is just the other way round, i.e., the Singleton bound gives $\qbin{7}{3}{2}=11811$ and the sphere-packing bound gives $\qbin{8}{4}{2}=200787$. For $d=2$ both bounds 
coincide and for $d=4$ the Singleton bound is always stronger than the sphere-packing bound since $\qbin{n-1}{k}{q}<\qbin{n}{k}{q}$. The asymptotic bounds 
\cite[Corollaries 7 and 10]{koetter2008coding}, using normalized parameters, and \cite[Figure 1]{koetter2008coding} suggest that there is only a small range of 
parameters where the sphere-packing bound can be superior to the Singleton bound.
%%
%% Für die Parameter unserer HP ist Sphere < Singleton in
%% 
%% q n d k
%% 
%% 2 6 6 3
%% 2 8 6 4
%% 2 10 6 5
%% 2 12 6 6
%% 2 14 6 7
%% 2 16 6 8
%% 2 18 6 9
\end{trailer}
   
\begin{nexercise}
  Show that the sphere-packing bound is strictly tighter than the Singleton bound iff $q=2$, $n=2k$, and $d=6$.
\end{nexercise}

\begin{trailer}{Anticode bounds}Given an arbitrary metric space $X$, an \emph{anticode} of diameter $e$ is a subset whose elements have pairwise distance at most $e$. 
For every association scheme, which applies to the $q$-Johnson scheme in our situation, the anticode bound of Delsarte \cite{delsarte1973algebraic} can be applied.  
As a standalone argument we go along the lines of \cite{ahlswede2009error} and consider bounds for codes on transitive graphs. By double-counting the number of pairs 
$(a,g)\in A\cdot\operatorname{Aut}(\Gamma)$, where $g(a)\in B$, we obtain:
\begin{nlemma}
  (\cite[Lemma~1]{ahlswede2009error}, cf.~\cite[Theorem~1']{ahlswede2001perfect})\\ Let $\Gamma=(V,E)$ be a graph that admits a transitive group of automorphisms $\operatorname{Aut}(\Gamma)$ and 
  let $A,B$ be arbitrary subsets of the vertex set $V$. Then, there exists a group element $g\in\operatorname{Aut}(\Gamma)$ such that
  \[
    \frac{|g(A)\cap B|}{|B|}\ge \frac{|A|}{|V|}.
  \]  
\end{nlemma} 
\begin{ncorollary}
  \label{ncor_ahlswede}
  (\cite[Corollary~1]{ahlswede2009error}, cf.~\cite[Theorem~1]{ahlswede2001perfect})\\ Let $\cC_D\subseteq \cG_q(n,k)$ be a code with (injection or graph) 
  distances from $D=\{d_1,\dots,d_s\}\subseteq \{1,\dots,v\}$. Then, for an arbitrary subset $\cB\subseteq \cG_q(n,k)$ there exists a code 
  $\cC_D^*(\cB)\subseteq\cB$ with distances from $D$ such that
  \[
    \frac{\left|\cC_D^*(\cB)\right|}{\left|\cB\right|}\ge\frac{\left|\cC_D\right|}{\qbin{n}{k}{q}}.
  \]  
\end{ncorollary}
\end{trailer}

If $\cC_D\subseteq \cG_q(n,k)$ is a {\cdc} with minimum injection distance $d$, i.e., $D=\{d,\dots,v\}$, and $\cB$ is an anticode with diameter $d-1$, we have 
$\# \cC_D^*(\cB)=1$, so that we obtain Delsarte's anticode bound
\begin{equation}
  \#\cC_D\le \frac{\qbin{n}{k}{q}}{\# \cB}.
\end{equation}

The set of all elements of $\cG_q(n,k)$ which contain a fixed $(k-d/2+1)$-dimensional subspace is an anticode of diameter $d-2$ with 
$\qbin{n-k+d/2-1}{d/2-1}{q}$ elements. By duality, the set of all elements of $\cG_q(n,k)$ which are contained in a fixed 
$(k+d/2-1)$-dimensional subspace is also an anticode of diameter $d-2$ with $\qbin{k+d/2-1}{k}{q}=\qbin{k+d/2-1}{d/2-1}{q}$ elements. Frankl and Wilson proved 
in \cite[Theorem~1]{MR867648} that these anticodes have the largest possible size, which implies:
\begin{ntheorem}
  \label{thm_anticode}
  \textbf{(Anticode bound -- \cite[Theorem~5.2]{wang2003linear})}
  \[
    A_q(n,d;k)\le \frac{\qbin{n}{k}{q}}{\qbin{\max\{k,n-k\}+d/2-1}{d/2-1}{q}} 
  \]
\end{ntheorem} 
Codes whose size attain the anticode bound are called \emph{Steiner structures}. The reduction to Delsarte's anticode bound can e.g.\ be found in \cite[Theorem~1]{MR2810308}.

Since the sphere underlying the proof of Theorem~\ref{thm_sphere_packing} is also an anticode, Theorem~\ref{thm_sphere_packing} is implied by 
Theorem~\ref{thm_anticode}. For $d=2$ both bounds coincide. In \cite[Section 4]{xia2009johnson} Xia and Fu verified that the anticode bound is 
always stronger than the Singleton bound for the ranges of parameters considered by us. In \cite[Example 2.2.8]{weiss2023linear} a \emph{Singleton polynomial} 
was used to conclude the anticode bound, so that one may speak of Singleton (type) bound. Other names used in the literature are Wang-Xing-Safavi-Naini bound and packing bound.   

Mimicking a classical bound of Johnson on binary error-correcting codes with respect to the Hamming distance, see \cite[Theorem~3]{johnson1962new} 
and also \cite{tonchev1998codes}, the following upper bound was obtained:
\begin{ntheorem}
  \label{thm_johnson_I}
  \textbf{(Johnson type bound I -- \cite[Theorem~2]{xia2009johnson})}\\ If $\left(q^k-1\right)^2>\left(q^n-1\right)\left(q^{k-d/2}-1\right)$, then
  \[
    A_q(n,d;k)\le \frac{\left(q^k-q^{k-d/2}\right)\left(q^n-1\right)}{\left(q^k-1\right)^2-\left(q^n-1\right)\left(q^{k-d/2}-1\right)}.
  \]
\end{ntheorem} 
However, the required condition of Theorem~\ref{thm_johnson_I} is rather restrictive and can be simplified considerably.
\begin{nproposition}
  \label{prop_johnson_I}
  (\cite[Proposition 1]{heinlein2017asymptotic})\\
  For $0\le k<n$, the bound in Theorem~\ref{thm_johnson_I} is applicable iff $d=2\min\{k,n-k\}$ and $k\ge 1$. Then, it is equivalent to
  \[
    A_q(n,d;k)\le \frac{q^n-1}{q^{\min\{k,n-k\}}-1}.
  \]
\end{nproposition}
In other words, Theorem~\ref{thm_johnson_I} is equivalent to a rather simple upper bound for partial spreads, see Subsection~\ref{subsec_bounds_partial_spreads}.

Let $\cC$ be a ${\cdc}$ in $\PG(n-1,q)$. For each point $P$ and each hyperplane $H$ in $\PG(n-1,q)$ consider the subcodes 
$\cC_P:=\left\{U\in\cC\,:\, P\le U\right\}$ and $\cC_H:=\left\{U\in\cC\,:\, U\le H\right\}$. A little counting argument gives:
\begin{ntheorem}
  \label{thm_johnson_II}
  \textbf{(Johnson type bound II -- \cite[Theorem~3]{xia2009johnson}, \cite[Theorem~4,5]{MR2810308})} 
  \begin{eqnarray}
  A_q(n,d;k) &\le& %%\left\lfloor 
  \frac{[n]_q A_q\!(n\!-\!1,d;k\!-\!1)}{[k]_q} =\frac{q^n\!-\!1}{q^k\!-\!1} \cdot A_q\!(n\!-\!1,d;k\!-\!1) %%\right\rfloor
  \label{ie_j_2}\\
  A_q(n,d;k) &\le& %%\left\lfloor 
  \frac{[n]_q A_q\!(n\!-\!1,d;k\!-\!1)}{[n\!-\!k]_q} =\frac{q^n-1}{q^{n-k}-1} \cdot A_q(n-1,d;k) %%\right\rfloor
  \label{ie_j_o}
  \end{eqnarray}
\end{ntheorem}

\begin{trailer}{Type II Johnson bounds for binary constant weight codes}In \cite[Inequality~(5)]{johnson1962new} the upper bounds $A(n,d;w) \le \lfloor n/w \cdot A(n-1,d;w-1) \rfloor$ 
and $A(n,d;w) \le \lfloor n/(n-w) \cdot A(n-1,d;w) \rfloor$ for binary constant weight codes were obtained. Of course both bounds can be applied iteratively. However, 
the optimal choice of the corresponding inequalities is unclear, see e.g.\ \cite[Research Problem 17.1]{macwilliams1977theory}. The bounds in Theorem~\ref{thm_johnson_II} are the $q$-analog of 
the mentioned bounds for constant weight codes.
\end{trailer}

While e.g.\ the authors of \cite{MR2810308,khaleghi2009subspace} stated that the optimal choice of Inequality~(\ref{ie_j_2}) or 
Inequality~(\ref{ie_j_o}) is unclear too, there is now an explicit answer for {\cdc}s:
\begin{nproposition} (\cite[Proposition 3]{heinlein2017asymptotic})
  \label{prop_optimal_johnson}
  For $k \le n/2$ we have
  \[
    \left\lfloor \frac{q^n-1}{q^k-1} A_q(n-1,d;k-1) \right\rfloor
    \le
    \left\lfloor \frac{q^n-1}{q^{n-k}-1} A_q(n-1,d;k) \right\rfloor,
  \]
  where equality holds iff $n=2k$.
\end{nproposition}
\begin{nexercise}
  Consider the dual code to show that Inequality~(\ref{ie_j_2}) and Inequality~(\ref{ie_j_o}) are equivalent. 
\end{nexercise}

Knowing the optimal choice between Inequality~(\ref{ie_j_2}) and Inequality~(\ref{ie_j_o}), we can iteratively apply Theorem~\ref{thm_johnson_II} 
in an ideal way (initially assuming $k\le n/2$):   
\begin{ncorollary}
\textbf{(Implication of the Johnson type bound II)}
\label{cor_johnson_opt}
\[
A_q(n,d;k)
\le
\left\lfloor \frac{q^{n}\!-\!1}{q^{k}\!-\!1} \left\lfloor \frac{q^{n\!-\!1}\!-\!1}{q^{k\!-\!1}\!-\!1} \left\lfloor \ldots 
\left\lfloor \frac{q^{n\!-\!k\!+\!d/2\!+\!1}\!-\!1}{q^{d/2\!+\!1}\!-\!1} A_q(n\!-\!k\!+\!d/2,d;d/2) \right\rfloor 
\ldots \right\rfloor \right\rfloor \right\rfloor
\]
\end{ncorollary}
We remark that this upper bound is commonly stated in an explicit version, where 
$A_q(n\!-\!k\!+\!d/2,d;d/2)\le \left\lfloor\frac{q^{n-k+d/2}-1}{q^{d/2}-1}\right\rfloor$ is inserted, see e.g.\ \cite[Theorem~6]{MR2810308}, 
\cite[Theorem~7]{khaleghi2009subspace}, and \cite[Corollary~3]{xia2009johnson}. However, better bounds for partial spreads are available now, see 
Subsection~\ref{subsec_bounds_partial_spreads}.

\begin{trailer}{Comparison of the Johnson bound with the previous bounds}It is shown in \cite{xia2009johnson} that the Johnson bound of Theorem~\ref{thm_johnson_II} improves 
on the anticode bound in Theorem~\ref{thm_anticode}, see also~\cite{MR3063504}. To be more precise, removing the floors in the upper bound of 
Corollary~\ref{cor_johnson_opt} and replacing $A_q(n-k+d/2,d;d/2)$ by $\frac{q^{n-k+d/2}-1}{q^{d/2}-1}$ gives 
\begin{equation}
  \prod_{i=0}^{k-d/2} \frac{q^{n-i}-1}{q^{k-i}-1}
  =
\frac{\prod_{i=0}^{k-1} \frac{q^{n-i}-1}{q^{k-i}-1}}{\prod_{i=k-d/2+1}^{k-1} \frac{q^{n-i}-1}{q^{k-i}-1}}
  =\frac{\qbin{n}{k}{q}}{\qbin{n-k+d/2-1}{d/2-1}{q}},
\end{equation}
which is the right hand side of the anticode bound for $k\le n-k$. So, all upper bounds mentioned so far are (weakly) dominated by 
Corollary~\ref{cor_johnson_opt}, if we additionally assume $k\le n-k$. We will slightly improve upon Theorem~\ref{thm_johnson_II} in 
Theorem~\ref{thm:johnson_improved} where we replace the possible rounding down by a tighter variant based on divisible multisets of points.  
\end{trailer}

\section{Other upper bounds for constant dimension codes}
\label{subsec_bounds_other}

As a possible improvement of known upper bounds \cite[Theorem~3]{ahlswede2009error} was mentioned in \cite[Theorem~8]{khaleghi2009subspace}, cf.~\cite[Theorem 8]{heinlein2017asymptotic}.
\begin{ntheorem}
  \label{thm_ahlswede}
  \textbf{(Ahlswede and Aydinian bound -- \cite[Theorem~3]{ahlswede2009error})}\\
  For integers $0\le t< r\le k$, $k-t\le m\le n$, and $t\le n-m$ we have
  \[
    A_q(n,2r;k)\le \frac{\qbin{n}{k}{q} A_q(m,2r-2t;k-t)}{\sum_{i=0}^t q^{i(m+i-k)}\qbin{m}{k-i}{q}\qbin{n-m}{i}{q}}.
  \]  
\end{ntheorem}
As Theorem~\ref{thm_ahlswede} has quite some degrees of freedom, we partially discuss the optimal choice of parameters.
For $t=0$ and $m\le v-1$, we obtain $A_q(n,d;k)\le \qbin{n}{k}{q}/\qbin{m}{k}{q}\cdot A_q(m,d;k)$, which is 
the $(n-m)$-fold iteration of Inequality~(\ref{ie_j_o}) of the Johnson bound (without rounding). Thus, $m=n-1$ is the best choice 
for $t=0$, yielding a bound that is equivalent to Inequality~(\ref{ie_j_o}). For $t=1$ and $m=n-1$ the bound can be rewritten to 
$A_q(n,d;k)\le  A_q(n-1,d-2;k-1)$. For $t> n-m$ the bound remains 
valid but is strictly weaker than for $t=n-m$. Choosing $m=n$ gives the trivial bound $A_q(n,2r;k)\le A_q(m,2r-2t;k-t)$. 
For the range of parameters $2\le q\le 9$, $4\le n\le 100$ and $4\le d\le 2k\le n$, where $q$ is a prime power and $d$ is even, the situation is as follows. 
If $d\neq 2k$, there are no proper improvements with respect to Theorem~\ref{thm_johnson_II}. For the case $d=2k$ we have some improvements compared to 
most easy upper bound $A_q(n,2k;k)\le \lfloor(q^n-1)/(q^k-1)\rfloor$ while the tightest known upper bounds for partial spreads, see Subsection~\ref{subsec_bounds_partial_spreads}, 
are not improved.  

\begin{question}{Research problem}Verify that the upper bounds of Theorem~\ref{thm_ahlswede} are implied by other known upper bounds or find specific parameters where  
this is not the case. 
\end{question}

\begin{trailer}{Linear programming bounds}Every association scheme gives rise to a linear programming upper bound, see e.g.\ \cite{delsarte1973algebraic}. For linear 
codes this relation can be expressed via the so-called MacWilliams identities. General introductions can e.g.\ be found in \cite{delsarte1998association,sloane1975introduction}. Explicit parametric 
upper bounds can be commonly obtained via this approach. Examples for linear codes are given in e.g.\ \cite{bierbrauer2007direct} and \cite[Section 15.3]{bierbrauer2016introduction}. 
For binary block and constant weight codes we refer e.g.\ to \cite{mounits2007new}. The Delsarte linear programming bound for the $q$-Johnson scheme was obtained in 
\cite{delsarte1978hahn}.  
%% However, numerical computations indicate that it is not better than the anticode bound in Theorem~\ref{thm_anticode}, see \cite{MR3063504}. In \cite{zhang2011linear} it was 
%% shown that the anticode bound is implied by the Delsarte linear programming bound. 
In \cite{MR3063504} it was shown that a semidefinite programming 
formulation\footnote{Due to the property of the symmetry group of $(\mathbb{F}_q^n,d_S)$, i.e., two-point homogeneous, the symmetry reduced version of the semidefinite 
programming formulation of the maximum clique problem formulation collapses the Delsarte linear programming bound for the $q$-Johnson scheme.}, that is equivalent to the Delsarte 
linear programming bound, implies the anticode bound of Theorem~\ref{thm_anticode} (cf.~\cite{zhang2011linear}), the sphere-packing bound of Theorem~\ref{thm_sphere_packing}, the Johnson type I bound of 
Theorem~\ref{thm_johnson_I}, and the Johnson type II bound of Theorem~\ref{thm_johnson_II}. Numerical computations indicate that it is not better than the anticode bound. Indeed, 
it was finally shown in \cite{schmidt2025linear,weiss2023linear} that the optimal solution of the linear program in Theorem~\ref{lin_prog_bound_cdc} is given by the anticode bound. 
In \cite[Example 2.2.8]{weiss2023linear} a \emph{Singleton polynomial} was used to conclude the anticode bound.  
\end{trailer}

\begin{ntheorem}\textbf{(Linear programming bound for {\cdc}s -- e.g.~\cite[Proposition 3]{zhang2011linear})}\\
  \label{lin_prog_bound_cdc}
  For integers $0 \le k \le n$ and $2 \le d \le \min\{k,n-k\}$ such that $d$ is even, we have  
  \begin{eqnarray}
    A_q(n,d;k)&\le& \max \Big\{ 1+\sum_{i=d/2}^k x_i \,\Big\vert\, \sum_{i=d/2}^k -Q_j(i) x_i \le u_j \,\forall j=1, 2, \ldots, k \text{ and }\notag\\ 
     && x_i \ge 0 \,\forall i=d / 2, d/2+1, \ldots, k \Big\}
\end{eqnarray}
 with
\begin{equation}
 u_j=\qbin{n}{j}{q}-\qbin{n}{j-1}{q},
\end{equation}
\begin{equation}
 v_i=q^{i^2}\qbin{l}{i}{q}-\qbin{n-1}{i}{q},
\end{equation}
\begin{equation}
 E_i(j)=\sum_{m=0}^i (-1)^{i-m} q^{{{i-m}\choose 2}+jm}\qbin{k-m}{k-1}{q}\qbin{k-j}{m}{q}\qbin{n-k-j+m}{m}{q}\text{ and}
\end{equation}
\begin{equation}
 Q_j(i)=\frac{u_j}{v_i}E_i(j).
\end{equation}
\end{ntheorem} 

%% \begin{nremark}
%% \label{remark_lin_prog_bound_cdc}
%% Using \texttt{Maple} and exact arithmetic, we have checked that for all $2\le q\le 9$, $4\le n\le 19$, $2\le k\le n/2$, $4\le d\le 2k$ the optimal value of the Delsarte linear programming 
%% bound is indeed the anticode bound Theorem~\ref{thm_anticode}. Given the result from \cite{zhang2011linear} it remains to construct a feasible solution of the 
%% Delsarte linear programming formulation whose target value equals the anticode bound. Such a feasible solution can also be constructed recursively. 
%% To this end, let $x_0$, \dots, $x_{k-1}$ denote a primal solution for the parameters of $A_q(n-1, d; k-1)$, then $z_0$, \dots, $z_k$ 
%% is a feasible solution for the parameters of $A_q(n, d;k)$ setting $z_i=x_i \cdot  \qbin{k}{1}{q} \ \qbin{k-i}{1}{q}$ for all
%% $0 \le i \le k-1$ and $z_k= \qbin{n}{k}{q}/\qbin{n-k+d/2-1}{d/2-1}{q} - z_0 -  \dots -z_{k-1}$. For the mentioned parameter space this conjectured 
%% primal solution is feasible with the anticode bound as target value. 
%% \end{nremark}

%% \begin{question}{Research problem}Verify that the optimal solution of the linear program in Theorem~\ref{lin_prog_bound_cdc} is given by the anticode bound, see 
%% Remark~\ref{remark_lin_prog_bound_cdc}, or give an explicit counter example. 
%% \end{question}

The iterated application of the Johnson bound of Theorem~\ref{thm_johnson_II} rounded down to integers in each iteration can improve upon the anticode bound.  
In Subsection~\ref{subsec_ub_divisible_codes} and Subsection~\ref{subsec_bounds_partial_spreads} we will present further upper bounds that improve upon the anticode 
or Johnson bound. Adding corresponding constraints to our linear programming formulation of Theorem~\ref{lin_prog_bound_cdc} of course gives tighter bounds.    

\begin{question}{Research problem}Find additional inequalities for the linear programming approach and improve at least one of the known upper bounds 
for $A_q(n,d;k)$.
\end{question}

As mentioned in the introduction, semidefinite programming bounds for $A(n,d)$ and $A(n,d;w)$ were quite successful in recent years, see e.g.\ \cite{vallentin2021semidefinite}. 
The same is true for {\mdc}s, i.e., upper bounds for $A_q(n,d)$, see \cite{MR3063504,heinlein2020new}. For {\cdc}s currently no improvement via semidefinite programming is known, 
see the blog entry \begin{center}\url{https://ratiobound.wordpress.com/2018/10/11/}.\end{center}
For related literature into this direction we refer to \cite{dunkl1978addition,liang2020terwilliger}.

Another rather general technique to obtain upper bounds for the maximum clique sizes of a graph is to use $p$-ranks of adjacency matrices.
\begin{nlemma}(E.g.~\cite[Lemma 1.3]{ihringer2018new})\\
  Let $G$ be a graph with adjacency matrix $A$ and $Y$ be a clique of $G$, then
  $$
    |Y|\le\left\{
    \begin{array}{rcl}
      \operatorname{rank}_p(A)+1 && \text{if $p$ divides $|Y|-1$,}\\
      \operatorname{rank}_p(A)   & & \text{otherwise.}
    \end{array}
    \right.
  $$
\end{nlemma}
Some numerical experiments suggest that the resulting upper bounds are rather weak for {\cdc}s. We e.g.~have
$A_2(4,4;2) \le 5$,
$A_2(5,4;2) \le 19$,
$A_2(6,4;2) \le 49$,
$A_2(6,4;3) \le 223$, and
$A_2(6,6;3) \le 19$.

\begin{trailer}{Integer linear programming formulations for $\mathbf{A_q(n,d;k)}$}The exact determination of $A_q(n,d;k)$ can be formulated as 
an integer linear program (ILP). To this end we introduce binary variables $x_K\in\{0,1\}$ for each $k$-space $K\in\cG_q(n,k)$ and 
maximize their sum $\sum_{K\in\cG_q(n,k)} x_K$ subject to the constraints 
\begin{equation}
  \label{ie_ilp_1}
  \sum_{K\in \cG_q(n,k)\,:\, S\le K} x_K\le 1
\end{equation} 
for all $S\in \cG_q(n,k-d/2+1)$, which guarantee the minimum subspace distance. This ILP can be solved directly for rather small 
parameters only. However, it was the basis for the determination of $A_2(6,4;3)=77$ and the classification of the corresponding five  
optimal isomorphism types in \cite{hkk77}. The determination of $A_2(8,6;4)=257$ and the classification of the corresponding two optimal 
isomorphism types required a tailored approach with relaxations to subconfigurations, see \cite{heinlein2019classifying} for the details.\footnote{The 
intermediate upper bound $A_2(8,6;4) \le 272$ was determined in \cite{heinlein2017new}.} 

The inequalities~(\ref{ie_ilp_1}) can be tightened by adding
\begin{equation}
  \label{ie_ilp_2}
  \sum_{K\in \cG_q(n,k)\,:\, W\le K} x_K\le A_q(n-w,d;k-w)
\end{equation} 
for all $w\in \cG_q(n,w)$, where $w\in\{1,\dots,k-1\}$, and by
\begin{equation}
  \label{ie_ilp_3}
  \sum_{K\in \cG_q(n,k)\,:\,  K\le A} x_K\le A_q(a,d;k)
\end{equation} 
for all $A\in \cG_q(n,a)$, where $a\in\{k+1,\dots,n-1\}$. The LP relaxations of inequalities (\ref{ie_ilp_2}) and (\ref{ie_ilp_3}) yields the 
following explicit bounds:
\begin{eqnarray}
  A_q(n,d;k) \le \frac{\qbin{n}{w}{q}}{\qbin{k}{w}{q}}A_q(n-w,d;k-w) && \forall w \in \{1,\ldots,k-d/2\},\\
  A_q(n,d;k) \le \frac{\qbin{n}{w}{q}}{\qbin{k}{w}{q}} && \forall w \in \{k-d/2+1, \ldots, k-1\},\\
  A_q(n,d;k) \le \frac{\qbin{n}{a}{q}}{\qbin{n-k}{a-k}{q}} && \forall a \in \{k+1, \ldots, k+d/2-1\},\text{ and }\quad\\
  A_q(n,d;k) \le \frac{\qbin{n}{a}{q}}{\qbin{n-k}{a-k}{q}}A_q(a,d;k) && \forall a \in \{k+d/2,\ldots,n-1\}.
 \end{eqnarray}

We remark that the ILP approach can also be used to 
construct {\cdc}'s of large cardinality. To restrict the search space typically a subgroup of the automorphism group of the {\cdc} is prescribed, 
see e.g.\ \cite{paper_axel}. 
\end{trailer}

If the presence of certain automorphisms is assumed, then for many cases improved upper bounds can be concluded from the LP relaxation. It is also 
possible to deduce parametric bounds from this approach, see e.g.\ \cite[Section 10]{phd_heinlein}. 

We close this overview mentioning that {\cdc}s containing a lifted {\mrd} code as subcode allow tighter upper bounds on their cardinality, see 
\cite{etzion2012codes,heinlein2019new,kurz2020generalized}. We remark that many of the currently best known constructions for {\cdc}s involve 
a lifted {\mrd} as a subcode, see Section~\ref{sec_constructions_cdc}. In \cite[Section 4]{kurz2021interplay} the underlying techniques have been extended 
to infer upper bounds for {\cdc}s arising from other specific constructions from the literature. 

\begin{question}{Research problem}Provide more specialized upper bounds for subcodes appearing in constructions for {\cdc}s in the literature (or 
Section~\ref{sec_constructions_cdc}).
\end{question}

\section{Upper bounds for partial spreads}
\label{subsec_bounds_partial_spreads}
Assume, as before, $k\le n-k$. An $(n,2k;k)_q$--{\cdc} is also called \emph{partial spread} or \emph{partial $k$-spread} to be more 
precise. Those {\cdc}s attain the maximum possible subspace distance, which is equivalent to the geometric description that the pairwise 
intersection of the $k$-spaces is trivial, i.e., $0$-dimensional. In principle all currently known upper bounds for partial $k$-spreads are 
implied by non-existence results for $q^{k-1}$-divisible sets of points, see Lemma~\ref{lemma_partial_spread_div_bound}. However, determining 
the possible lenghts seems to be a hard problem, see e.g.\ \cite{kurz2021divisible} for a survey. Here we just list several parametric bounds 
from the literature.

\smallskip

Applying the Johnson bound of Theorem~\ref{thm_johnson_II} to the parameters 
of a partial spread yields
$$
  A_q(n,2k;k)\le \frac{[n]_q}{[k]_q} \cdot A_q(n-1,2k;k-1)=\frac{[n]_q}{[k]_q}
$$ 
since $A_q(n-1,2k;k-1)=1$. An easy direct geometric justification comes from the fact that $\PG(n-1,q)$ contains $[n]_q$ points and each $k$-space contains 
$[k]_q$ points. Spelling out the $q$-factorials and rounding down we obtain
\begin{equation}
  \label{ie_ps_trivial}
  A_q(n,2k;k)\le \left\lfloor\frac{q^n-1}{q^k-1}\right\rfloor.
\end{equation}
In the following we review improved classical bounds for partial spreads from the literature. Other surveys can e.g.\ be found in \cite{honold2018partial,storme2021coding}. 
In the subsequent Subsection~\ref{subsec_ub_divisible_codes} we will briefly introduce a contemporary 
approach based on $q^{k-1}$-divisible (multi-) sets of points. It will turn out that all upper bounds of this subsection can be obtained from non-existence results for 
$q^{k-1}$-divisible sets of points in $\PG(n-1,q)$, where $n$ is assumed to be sufficiently large.

An $(n,2k;k)_q$--{\cdc} of cardinality $[n]_q/[k]_q$ is called a \emph{$k$-spread} (or just \emph{spread}). A handy existence criterion is
known from the work of Segre in 1964.
\begin{ntheorem}\textbf{(Existence of spreads -- \cite[\S VI]{segre1964teoria})}\\
  \label{thm_spread} $\PG(n-1,q)$ contains a $k$-spread iff $k$ is a divisor of $n$.
\end{ntheorem}
\begin{nexercise}
  Write $n=tk+r$ with $1\le r\le k-1$ and $t\ge 2$. Verify
  $$
    A_q(n,2k;k)\leq\left\lfloor\frac{q^n-1}{q^k-1}\right\rfloor
  =\frac{q^{tk+r}-q^r}{q^k-1}+\left\lfloor\frac{q^r-1}{q^k-1}\right\rfloor
  =\sum_{s=0}^{t-1}q^{sk+r}=q^r\qbin{t}{1}{q^k}\!\!\!\!.
  $$
\end{nexercise}
\begin{ndefinition}\textbf{(Deficiency of partial $k$-spreads in $\PG(n-1,q)$ -- cf.~\cite{beutelspacher1975partial})}\\
  The number $\sigma$ defined by $$A_q(tk+r;2k;k)=\sum_{s=0}^{t-1}q^{sk+r}-\sigma,$$ where $0\le r\le k-1$ and $t\ge 2$, is called the
  \emph{deficiency} of the partial $k$-spreads of maximum possible size in $\PG(tk+r-1,q)$.\footnote{This makes sense also for $r=0$:
  Spreads are assigned deficiency $\sigma=0$.} 
\end{ndefinition}
\begin{warning}{Deficiency of a partial $k$-spread $\cP$ in $\PG(n-1,q)$}If $\cP$ is a partial $k$-spread in $\PG(n-1,q)$, where $n=tk+r$ 
with $0\le r\le k-1$ and $t\ge 2$, then the deficiency of $\cP$ is defined as $\sum_{s=0}^{t-1}q^{sk+r}-\#\cP$ in several papers. I.e.\ the 
value $\sigma$ is just a lower bound for the deficiency of a given partial spread and there is some interest in the possible deficiencies of 
inclusion-maximal partial spreads.\end{warning}

\begin{ntheorem}(\cite{beutelspacher1975partial,beutelspacher1976correction}, cf.~\cite[Theorem~2.7(a)]{eisfeldt})\\
  \label{thm:lb_def}
  The deficiency of a maximal $k$-spread in $\PG(n-1,q)$, where $k$ does not divide $n$, is at least $q-1$.
\end{ntheorem}

We remark that we indeed have 
\begin{equation}
  \label{eq_lb_partial_spread}
  A_q(tk+r,2k;k)\ge \sum_{s=0}^{t-1}q^{sk+r}-(q^r-1)
\end{equation}
for all $k,t\ge 2$ and $0\le r\le k-1$, see e.g.\ \cite{beutelspacher1975partial} or Exercise~\ref{exercise_lb_partial_spreads}. So, the cases {\lq\lq}$r=0${\rq\rq} and 
{\lq\lq}$r=1${\rq\rq} are completely resolved.  

\begin{ntheorem}(\cite[Theorem 4.3]{kurzspreads})
  We have 
  \begin{equation}
    A_2(tk+2,2k;k)\le \sum_{s=0}^{t-1}2^{sk+2}-\left(2^2-1\right)
  \end{equation}
  for all $k\ge 4$, $t\ge 2$.
\end{ntheorem}

\begin{ntheorem}\textbf{($\mathbf{k}$ sufficiently large, the asymptotic case -- \cite[Theorem 5]{nastase2016maximum})}\\
  \label{thm_ps_asymptotic}
  We have 
  \begin{equation}
    A_q(tk+r,2k;k)\le \sum_{s=0}^{t-1}q^{sk+r}-\left(q^r-1\right)
  \end{equation}
  for all $k>[r]_q$, $t\ge 2$.
\end{ntheorem}

\begin{ntheorem}(\cite[Theorem 2.9]{kurz2017packing},\cite[Theorem 9]{honold2018partial},\cite[Corollary 7]{honold2018partial})\\
  \label{thm_ps_asymptotic_gen}
  For integers $r\ge 1$, $t\ge 2$, $u\ge 0$, and $z\ge 0$ with $k=[r]_q+1-z+u>r$ we have
  \begin{equation}
    A_q(tk+r,2k;k)\le \sum_{s=0}^{t-1}q^{sk+r}-\left(q^r-1\right) \,+\, z(q-1).
  \end{equation}
\end{ntheorem}
Setting $z=0$ in Theorem~\ref{thm_ps_asymptotic_gen} gives Theorem~\ref{thm_ps_asymptotic}.

For a long time the best upper bound for partial spreads was given by Drake and Freeman:
\begin{ntheorem}{(\cite[Corollary 8]{nets_and_spreads})}
  \label{thm_partial_spread_4}
  If $n=kt+r$ with $0<r<k$ and $t\ge 2$, then 
  $$
    A_q(n,2k;k)\le \sum_{i=0}^{t-1} q^{ik+r} -\left\lfloor\theta\right\rfloor-1
    =q^r\cdot \frac{q^{kt}-1}{q^k-1}-\left\lfloor\theta\right\rfloor-1
    =\frac{q^{n}-q^r}{q^k-1}-\left\lfloor\theta\right\rfloor-1,
  $$
  where $2\theta=\sqrt{1+4q^k(q^k-q^r)}-(2q^k-2q^r+1)$.
\end{ntheorem}
\begin{nexample}
  If we apply Theorem~\ref{thm_partial_spread_4} with $q=5$, $n=16$, $k=6$, and $r=4$, then we obtain $\theta\approx 308.81090$ and $A_5(16,12;6)\le 9765941$.
\end{nexample}

\newcommand{\uu}{\lambda}
\begin{ntheorem}{(\cite[Theorem 10]{honold2018partial},\cite[Theorem 2.10]{kurz2017packing})}
  \label{theorem_parametric_ps_bound_2}
  For integers $r\ge 1$, $t\ge 2$, $y\ge \max\{r,2\}$, $z\ge 0$ with $\uu=q^{y}$, $y\le k$,  
  $k=[r]_q+1-z>r$, $n=kt+r$, and  $l=\frac{q^{n-k}-q^r}{q^k-1}$, we have 
  \begin{equation}
   A_q(n,2k;k)\le 
       lq^k+\left\lceil \uu -\frac{1}{2}-\frac{1}{2}
    \sqrt{1+4\uu\left(\uu-(z+y-1)(q-1)-1\right)} \right\rceil.
  \end{equation}   
\end{ntheorem}
Using Theorem~\ref{theorem_parametric_ps_bound_2} with $q=5$, $k=6$, $n=15$, $r=3$, $z=17$, and $y=5$ gives $A_5(15,12;6)\le 1953186$. Choosing $y=t$ we obtain 
Theorem~\ref{thm_partial_spread_4}. Theorem~\ref{theorem_parametric_ps_bound_2} also covers \cite[Theorems 6,7]{nastase2016maximumII} and yields improvements in a few 
instances, e.g.\ $A_3(15,12;6)\le 19695$.

A few further parametric upper bounds have been mentioned in \cite{kurz2017packing}. For $t\ge 2$ we have
\begin{itemize}
  \item $2^4l+1\le A_2(4t+3,8;4)\le 2^4l+4$, where $l=\frac{2^{4t-1}-2^3}{2^4-1}$;
  \item $2^6l+1\le A_2(6t+4,12;6)\le 2^6l+8$, where $l=\frac{2^{6t-2}-2^4}{2^6-1}$;
  \item $2^6l+1\le A_2(6t+5,12;6)\le 2^6l+18$, where $l=\frac{2^{6t-1}-2^5}{2^6-1}$;
  \item $3^4l+1\le A_3(4t+3,8;4)\le 3^4l+14$, where $l=\frac{3^{4t-1}-3^3}{3^4-1}$;
  \item $3^5l+1\le A_3(5t+3,10;5)\le 3^5l+13$, where $l=\frac{3^{5t-2}-3^5}{3^3-1}$;
  \item $3^5l+1\le A_3(5t+4,10;5)\le 3^5l+44$, where $l=\frac{3^{5t-1}-3^4}{3^5-1}$;
  \item $3^6l+1\le A_3(6t+4,12;6)\le 3^6l+41$, where $l=\frac{3^{6t-2}-3^4}{3^6-1}$;
  \item $3^6l+1\le A_3(6t+5,12;6)\le 3^6l+133$, where $l=\frac{3^{6t-1}-3^5}{3^6-1}$;
  \item $3^7l+1\le A_3(7t+4,14;7)\le 3^7l+40$, where $l=\frac{3^{7t-3}-3^4}{3^7-1}$;
  \item $4^4l+1\le A_4(4t+2,8;4)\le 4^4l+6$, where $l=\frac{4^{4t-2}-4^2}{4^4-1}$;
  \item $4^5l+1\le A_4(5t+3,10;5)\le 4^5l+32$, where $l=\frac{4^{5t-2}-4^3}{4^5-1}$;
  \item $4^6l+1\le A_4(6t+3,12;6)\le 4^6l+30$, where $l=\frac{4^{6t-3}-4^3}{4^6-1}$;
  \item $4^6l+1\le A_4(6t+5,12;6)\le 4^6l+548$, where $l=\frac{4^{6t-1}-4^5}{4^6-1}$;
  \item $4^7l+1\le A_4(7t+4,14;7)\le 4^7l+128$, where $l=\frac{4^{7t-3}-4^4}{4^7-1}$;
  \item $5^5l+1\le A_5(5t+2,10;5)\le 5^5l+7$, where $l=\frac{5^{5t-3}-5^2}{5^5-1}$;
  \item $5^5l+1\le A_5(5t+4,10;5)\le 5^5l+329$, where $l=\frac{5^{5t-1}-5^4}{5^5-1}$;
  \item $5^6l+1\le A_5(6t+3,8;4)\le 5^6l+61$, where $l=\frac{5^{6t-3}-5^3}{5^6-1}$;
  \item $5^6l+1\le A_5(6t+4,8;4)\le 5^6l+316$, where $l=\frac{5^{6t-2}-5^4}{5^6-1}$;
  \item $7^5l+1\le A_7(5t+4,10;5)\le 7^5l+1246$, where $l=\frac{7^{5t-1}-7^2}{7^5-1}$;
  \item $7^6l+1\le A_7(6t+2,8;4)\le 7^6l+15$, where $l=\frac{7^{6t-4}-7^3}{7^6-1}$;
  \item $8^4l+1\le A_8(4t+3,8;4)\le 8^4l+264$, where $l=\frac{8^{4t-1}-8^3}{8^4-1}$;
  \item $8^5l+1\le A_8(5t+2,10;5)\le 8^5l+25$, where $l=\frac{8^{5t-3}-8^2}{8^5-1}$;
  \item $8^6l+1\le A_8(6t+2,8;4)\le 8^6l+21$, where $l=\frac{8^{6t-4}-8^3}{8^6-1}$;
  \item $9^3l+1\le A_9(3t+2,6;3)\le 9^3l+41$, where $l=\frac{9^{3t-1}-9^2}{9^3-1}$;
  \item $9^5l+1\le A_9(5t+3,10;5)\le 9^5l+365$, where $l=\frac{9^{5t-2}-9^3}{9^5-1}$.     
\end{itemize}
Actually, each improved upper bound for $A_q(n,2k;k)$ for specific parameters implies 
a parametric series of upper bounds.
\begin{nlemma}(\cite[Lemma 4]{honold2018partial})\\
  For fixed $q$, $k$ and $r$ the deficiency $\sigma$ is a non-increasing function of $n = kt + r$.
\end{nlemma}

\section{Upper bounds based on divisible multisets of points}
\label{subsec_ub_divisible_codes}
A \emph{multiset} $\cM$ of points in $\PG(n-1,q)$ is a mapping $\cM\colon\cG_q(n,1)\to\N_0$. For each point $P\in\cG_q(n,1)$ the integer $\cM(P)$ is called the 
\emph{multiplicity} of $P$ and it counts how often point $P$ is contained in the multiset. If $\cM(P)\in\{0,1\}$ for all $P\in\cG_q(n,1)$ we also speak of a set instead 
of a multiset (of points). We call a multiset of points $\Delta$-divisible iff the corresponding linear code $C$ is 
$\Delta$-divisible, i.e., if the weights of all codewords in $C$ are divisible by $\Delta$. Note that this is equivalent to 
\begin{equation}
  \label{eq_divisible_multiset}
  \cM(H)\equiv \#\cM\pmod \Delta
\end{equation}
for every hyperplane $H$, where $\cM(H)$ is the sum of the multiplicities of the points contained in $H$ and $\#M$ is the sum of the multiplicities over all points. 
The set of points of a $k$-space, the multiset of points of a multiset of $k$-spaces, and the set of holes of a partial $k$-spread are $q^{k-1}$-divisible. Here we briefly 
state upper bounds for $A_q(n,d;k)$ that can be concluded from non-existence results of $\Delta$-divisible multisets of points. For an introduction we refer e.g.\ 
to  \cite{heinlein2019projective,honold2018partial}.

For each integer $r$ and each dimension $1\le i\le r+1$ the $q^{r+1-i}$-fold repetition of an $i$-space in $\PG(v-1,q)$ is a $q^r$-divisible multiset of points of 
cardinality $q^{r+1-i}\cdot [i]_q$. So, for a fixed prime power $q$, a non-negative integer $r$, and $i\in\{0,\ldots,r\}$, we define
\begin{equation}
	    \snumb{r}{i}{q}
	    := q^i\cdot[r-i+1]_q
	    = \frac{q^{r+1}-q^i}{q-1}
	    =\sum_{j=i}^r q^{j}
	    =q^i + q^{i+1} + \ldots + q^r
\end{equation}
and state:
\begin{nlemma}
	\label{lemma:snumb_card}
	For each $r\in\N_0$ and each $i\in\{0,\ldots,r\}$ there is a $q^r$-divisible multiset of points of cardinality 
	$\snumb{r}{i}{q}$.
\end{nlemma}
As a consequence of Lemma~\ref{lemma:snumb_card} all integers $n = \sum_{i=0}^r a_i \snumb{r}{i}{q}$ with $a_i\in\mathbb{N}_0$ are realizable 
cardinalities of $q^r$-divisible multisets of points. Note that the number $\snumb{r}{i}{q}$ is divisible by $q^i$, but not by $q^{i+1}$. This property allows us to create 
kind of a positional 
system upon the sequence of base numbers
\[
	S_q(r) := (\snumb{r}{0}{q}, \snumb{r}{1}{q},\ldots, \snumb{r}{r}{q})\text{.}
\]
\begin{nexercise}
Show that each integer $n$ has a unique \emph{$S_q(r)$-adic expansion}
\begin{equation}
	\label{eq:sqadic}
	n = \sum_{i=0}^r a_i \snumb{r}{i}{q}
\end{equation}
with $a_0,\ldots,a_{r-1}\in\{0,\ldots,q-1\}$ and \emph{leading coefficient} $a_r\in\mathbb{Z}$.
\end{nexercise} 
\begin{programcode}{Algorithm}
%%\label{alg_s_q_adic_representation}
\noindent
$\!$\textbf{Input:} $n\in\mathbb{Z}$, field size $q$, exponent $r\in\N_0$\\
\textbf{Output:} representation $n=\sum\limits_{i=0}^r a_i \snumb{r}{i}{q}$ with $a_0,\ldots,a_{r-1}\in\{0,\ldots,q-1\}$ and $a_r\in\mathbb{Z}$\\
%\For{$i\gets0$ \KwTo $r$}
%{$e_i\gets  0$}
$m\gets n$\\
For {$i\gets 0$ To $r-1$}\\
{
\hspace*{0.6cm}$a_i\gets m\bmod q$\\
\hspace*{0.6cm}$m\gets \frac{m-a_i\cdot[r-i+1]_q}{q}$\\
}
$a_r\gets m$\\
% Magma-Code:
% function s(r,i,q) return q^i * ExactQuotient(q^(r-i+1)-1,q-1); end function;
% function expand(n,r,q) res := []; for i:=0 to r-1 do a:=n mod q; n := ExactQuotient(n - a*s(r-i,0,q),q); Append(~res,a); end for; Append(~res,n); return res; end function;
\end{programcode}
Here $m\bmod q$ denotes the remainder of the division of $m$ by $q$.
\begin{nexample}
The $S_2(2)$-adic expansion of $n=11$ is given by $11=1\cdot 7+0\cdot 6+1\cdot 4$ and the $S_2(2)$-adic expansion of $n=9$ is given by $1\cdot 7+1\cdot 6-1\cdot 4$, 
i.e., the leading coefficient is $-1$. 
\end{nexample}
\begin{nexercise}
  \label{exercise_s_q_r_adic_expansion_137}
  Compute the $S_3(3)$-adic expansion of $n=137$ and determine the leading coefficient.
\end{nexercise}
\begin{ntheorem}\textbf{(Possible lengths of divisible codes -- \cite[Theorem 1]{kiermaier2020lengths})}\\
  \label{thm_characterization_div}
  For $n\in\Z$ and $r\in\N_0$ the following statements are equivalent:
  \begin{enumerate}
  \item[(i)]\label{thm:characterization_div:card_multiset} There exists a $q^r$-divisible multiset of points of cardinality $n$ over $\F_q$.   
  \item[(ii)]\label{thm:characterization_div:card_code} There exists a full-length $q^r$-divisible linear code of length $n$ over $\F_q$.
  \item[(iii)]\label{thm:characterization_div:n_strong} The leading coefficient of the $S_q(r)$-adic expansion of $n$ is non-negative.
  \end{enumerate}
\end{ntheorem}
So, the $S_q(r)$-adic expansion of $n$ provides a certificate not only for the existence, but remarkably also for  
the non-existence of a $q^r$-divisible multiset of size $n$. As computed in Exercise~\ref{exercise_s_q_r_adic_expansion_137}, the leading coefficient of 
the $S_3(3)$-adic expansion of $n=137$ is $-2$, so that there is no $27$-divisible ternary linear code of effective length $137$. If $q=p^m$ is a proper 
prime power then also the possible cardinalities of $p^r$-divisible multisets of points over $\F_q$ are of interest when $r$ is not divisible by $m$. 
To that end \cite[Theorem 1]{kiermaier2020lengths}) was completed in \cite{kurz2023lengths}. 

\begin{trailer}{Sharpened rounding}
\vspace*{-6mm}
\begin{ndefinition}
	\label{def:divisible_gauss_bracket}
  For $a\in\Z$ and $b\in\Z\setminus\{0\}$ let $\llfloor a/b \rrfloor_{q^r}$ be the maximal $n\in\Z$ such that there exists a $q^r$-divisible $\F_q$-linear code of effective length $a-nb$.
  If no such code exists for any $n$, we set $\llfloor a/b \rrfloor_{q^r} = -\infty$.
  Similarly, let $\llceil a/b\rrceil_{q^r}$ denote the minimal $n\in\Z$ such that there exists a $q^r$-divisible $\F_q$-linear code of effective length $nb-a$. 
  If no such code exists for any $n$, we set $\llceil a/b\rrceil_{q^r} = \infty$.
\end{ndefinition}
\end{trailer}

Note that the symbols $\llfloor a/b \rrfloor_{q^r}$ and $\llceil a/b \rrceil_{q^r}$ encode the four values $a$, $b$, $q$ and $r$. Thus, the fraction $a/b$ is a formal fraction and the power 
$q^r$ is a formal power, i.e.\ we assume $1530/14\neq 765/7$ and $2^2\neq 4^1$ in this context.
\begin{nexercise}
  Compute $\llfloor 765/7 \rrfloor_{2^2}$ and $\llfloor 1530/14 \rrfloor_{4^1}$. Verify
  \[
		    \llfloor 0/b\rrfloor_{q^r} = \llceil 0/b\rrceil_{q^r} = 0
		\]
		and
		\begin{eqnarray*}
		  && \ldots \leq \llfloor a/b\rrfloor_{q^2} \leq \llfloor a/b\rrfloor_{q^1} \leq \llfloor a/b \rrfloor_{q^0} = \left\lfloor \tfrac{a}{b} \right\rfloor \\
		  &&  \leq a/b \leq \lceil a/b\rceil = \llceil a/b \rrceil_{q^0} \leq \llceil a/b\rrceil_{q^1} \leq \llceil a/b\rrceil_{q^2} \leq \ldots
		\end{eqnarray*}
\end{nexercise}

\begin{nlemma}(\cite[Lemma 13]{kiermaier2020lengths})\\
  \label{lem:pack_cover}
  Let $k \in \Z_{\geq 1}$ and $\mathcal{U}$ be a multiset of $k$-spaces in $\PG(n-1,q)$.
  \begin{enumerate}
    \item[(i)]\label{lem:pack_cover:pack} If every point in $\cP$ is covered by at most $\lambda$ elements of $\mathcal{U}$, then 
    \[
	\#\cU\le \llfloor\lambda[n]_q/[k]_q\rrfloor_{q^{k-1}}\text{.}
    \]
    \item[(ii)]\label{lem:pack_cover:cover} If every point in $\cP$ is covered by at least $\lambda$ elements in $\mathcal{U}$, then 
    \[
	\#\cU\ge \llceil\lambda[n]_q/[k]_q\rrceil_{q^{k-1}}\text{.}
    \]
  \end{enumerate}  
\end{nlemma}

\begin{trailer}{An improvement of the Johnson bound from Theorem~\ref{thm_johnson_II}}
Instead of rounding down the right hand side of Inequality~(\ref{ie_j_2}) we can use the sharpened rounding from Definition~\ref{def:divisible_gauss_bracket}:  
\begin{ntheorem}(\cite[Theorem 12]{kiermaier2020lengths})\\
  \label{thm:johnson_improved}
  \[
  A_q(n,d;k)
	  \le \leftllfloor \frac{[n]_q\cdot A_q(n-1,d;k-1)}{[k]_q}\rightrrfloor_{q^{k-1}} \!\!\!\!\!\!\!\!\text{.}
  \]
\end{ntheorem}
  With $n'=n-k+d/2$, the iterated application of Theorem~\ref{thm:johnson_improved} yields
  \begin{eqnarray*}
  &A_q(n,d;k) \le \bigllfloor \frac{[n]_q}{[k]_q}\cdot\bigllfloor \frac{[n-1]_q}{[k-1]_q}\cdot
  \bigllfloor \cdots \\ & \bigllfloor\frac{[n'+1]_q}{[d/2+1]_q}\cdot A_q(n',d;d/2) 
  \bigrrfloor_{q^{d/2-1}} \cdots \bigrrfloor_{q^{k-3}} \bigrrfloor_{q^{k-2}}\bigrrfloor_{q^{k-1}}\text{.}\footnotemark
  \end{eqnarray*}
  \footnotetext{Expressions of the form $\llfloor\frac{a}{b}\cdot c\rrfloor_{q^r}$ should be read as $\llfloor\frac{a\cdot c}{b}\rrfloor_{q^r}$.}
\end{trailer}
 
\begin{nexample}
So far, the best known upper bound on $A_2(9,6;4)$ has been given by the Johnson bound~\eqref{ie_j_2}, using $A_2(8,6;3)=34$ from 
\cite{spreadsk3}:
\[
	A_2(9,6;4) \leq \left\lfloor\frac{[9]_2}{[4]_2}\cdot A_2(8,6;3)\right\rfloor
	= \left\lfloor\frac{2^9-1}{2^4-1}\cdot 34\right\rfloor = 1158\text{.}
\]
To improve that bound by Theorem~\ref{thm:johnson_improved}, we are looking for the largest integer $n$ such that a $q^{k-1}$-divisible multiset of size
\[
	M(n)
	= [9]_2\cdot A_2(8,6;3) - n \cdot [4]_2
	= 17374 - 15n
\]
exists.

This question can be investigated with Theorem~\ref{thm_characterization_div}.
We have $S_2(3) = (15, 14, 12, 8)$.
The $S_2(3)$-adic expansion of $M(1157) = 17374 - 15\cdot 1157 = 19$ is $1\cdot 15 + 0\cdot 14 + 1\cdot 12 + (-1)\cdot 8$.
As the leading coefficient $-1$ is negative, there is no $8$-divisible multiset of points of size $19$ by Theorem~\ref{thm_characterization_div}.
The $S_2(3)$-adic expansion of $M(1156) = 34$ is $0\cdot 15 + 1\cdot 14 + 1\cdot 12 + 1\cdot 8$.
As the leading coefficient $1$ is non-negative, there exists a $8$-divisible multiset of points of size $34$.
Thus, we have
\[
	A_2(9,6;4)\le \leftllfloor\frac{[9]_2}{[4]_2}\cdot A_2(8,6;3)\rightrrfloor_{2^3}
	= \llfloor 17374 / 15\rrfloor_{2^3}
	=1156\text{,}
\]
which improves the original Johnson bound~\eqref{ie_j_2} by $2$.
\end{nexample}
 
\begin{nlemma}(\cite[Lemma 17]{kiermaier2020lengths})
The improvement of Theorem~\ref{thm:johnson_improved} over the original Johnson bound~\eqref{ie_j_2} is at most $(q-1)(k-1)$.
\end{nlemma} 

The sharpened rounding in Theorem~\ref{thm:johnson_improved} can also be evaluated parametric in the field size $q$.
\begin{nproposition}(\cite[Proposition 2]{kiermaier2020lengths})
	\label{prop:v11d6k4}
  For all prime powers $q\ge 2$ we have
  \begin{align*}
	  &A_q(11,6;4)  \le q^{14}+q^{11}+q^{10}+2q^7+q^6+q^3+q^2-2q+1 \\
	  & = (q^2-q+1)(q^{12}+q^{11}+q^8+q^7+q^5+2q^4+q^3-q^2-q+1)\text{.}
  \end{align*}
\end{nproposition} 

As a refinement of the sharpened rounding from Definition~\ref{def:divisible_gauss_bracket} we introduce:
\begin{ndefinition}
  \label{def:divisible_gauss_bracket_point_multiplicity}
  For $a\in\Z$ and $b\in\Z\setminus\{0\}$ let $\llfloor a/b \rrfloor_{q^r\!\!,\lambda}$ be the maximal $n\in\Z$ such that there exists a $q^r$-divisible 
  multisets of points in $\PG(v-1,q)$ for suitably large $v$ with maximum point multiplicity at most $\lambda$ and cardinality $a-nb$. If no such multiset 
  exists for any $n$, we set $\llfloor a/b \rrfloor_{q^r\!\!,\lambda} = -\infty$.
\end{ndefinition}

With this we can sharpen the almost trivial upper bound (\ref{ie_ps_trivial}) for partial spreads, see e.g.\ 
\cite{heinlein2019projective,honold2018partial} 
for the details.
\begin{nlemma}
  \label{lemma_partial_spread_div_bound}
  Let $\cU$ be a set of $k$-spaces in $\PG(v-1,q)$, where $1\le k\le v$, with pairwise trivial intersection. Then, we have
  \begin{equation}
    \label{ie_partial_spread_div_bound}
    \#\cU\le \llfloor [v]_q/[k]_q\rrfloor_{q^{k-1}\!\!,1}\text{.}
  \end{equation}   
\end{nlemma}
So, for $2\le k\le n/2$ we obtain the upper bound $A_q(n,2k;k)\le \llfloor [n]_q/[k]_q\rrfloor_{q^{k-1}\!\!,1}$. In contrast to 
$\llfloor a/b \rrfloor_{q^r}$ there is no known efficient algorithm to evaluate $\llfloor a/b \rrfloor_{q^r\!\!,\lambda}$ in general. In other words, 
the determination of the possible cardinalities of $q^r$-divisible multisets of points with maximum point multiplicity $\lambda$ is a hard open problem, 
see e.g.~\cite{honold2019lengths}. For a survey of partial results for $\lambda=1$ we refer to \cite{heinlein2019projective}. 
\begin{nexample}
  In e.g.\ \cite{kurz2020no131} it was shown that no $2^4$-divisible set of $131$ points exists in $\PG(v-1,2)$. This implies $A_2(13, 10; 5)\le 259$ 
  since a partial $5$-spread in $\PG(12,2)$ of cardinality $260$ would give a $2^4$-divisible set of $131$ holes (i.e.\ uncovered points). With this, 
  Theorem~\ref{thm:johnson_improved} e.g.\ yields $A_2(14, 10; 6) \le 67 349$. 
\end{nexample}
Nevertheless, several parametric bounds for $q^r$-divisible sets of points (where $\lambda=1$) are known, see  
\cite{honold2018partial}. And indeed, all upper bounds for partial spreads presented in Subsection~\ref{subsec_bounds_partial_spreads} can be deduced 
from Lemma~\ref{lemma_partial_spread_div_bound}. Some partial results for $q^r$-divisible multisets of points with restricted point multiplicity larger 
than $1$ have been obtained in \cite{korner2023lengths}.

\begin{trailer}{The tightest known upper bounds for {\cdc}s}Assume $k\le n-k$. All currently known upper bounds for partial $k$-spreads are 
implied by $A_q(n,2k;k)\le \llfloor [n]_q/[k]_q\rrfloor_{q^{k-1}\!\!,1}$, see Lemma~\ref{lemma_partial_spread_div_bound}, and non-existence results 
for $q^{k-1}$-divisible sets of points. For $d<2k$ all currently known upper bounds for $A_q(n,d;k)$ are implied by the improved Johnson bound in 
Theorem~\ref{thm:johnson_improved} except $A_2(6,4;3)=77$ and $A_2(8,6;4)=257$, which are obtained via extensive ILP computations, see \cite{hkk77} 
and \cite{heinlein2019classifying}, respectively.

In \cite{heinlein2019generalized} it was observed that also a combinatorial relaxation of a {\cdc} $\cC\subset\cG_2(8,4)$ with minimum subspace distance $6$ 
has a maximum possible cardinality strictly less than $289$, which is the upper bound for $A_2(8,6;4)$ that can be obtained by Theorem~\ref{thm:johnson_improved}. 
Possibly the notion of generalized vector space partitions from \cite{heinlein2019generalized} allows further theoretical insights.  
\end{trailer} 
 
 \begin{warning}{The dominance relation between the upper bounds is just a snapshot}The clear picture on the dominance between the different known 
 upper bounds for {\cdc}s might just reflect our fragmentary knowledge and may change with time. While we currently do not know a single upper bound 
 for $A_q(n,2k;k)$ that cannot be obtained via a non-existence result for $q^{k-1}$-divisible sets of points, there are indeed known criteria to show 
 that certain $q^{k-1}$-divisible sets of points cannot coincide with the set of holes of a partial $k$-spread. 
 \end{warning}
 
\begin{question}{Research problem}Find a computer-free proof of $A_2(6,4;3)<81$ or $A_2(8,6;4)<289$.
\end{question} 
 
\chapter{Constructions for constant dimension codes}
\label{sec_constructions_cdc}

In this section we want to review lower bounds for $A_q(n,d;k)$, i.e., constructions for constant dimension codes. Our aim will to be to make the underlying 
ideas as clearly as possible, to show up the relations between different constructions from the literature, and to highlight potential for further improvements. 
To this end, we introduce a classification scheme to get a quick, rough picture of the different constructions. We will also try to decompose the, sometimes 
quite involved constructions, into smaller and easier components. While we want to trace the evolution of different constructions and their successive improvements, 
we will also have a closer look at the underlying distance analysises and possibilities to add further codewords. In some cases we so obtain improvements over the 
existing literature.

Common components are constant dimension codes (of smaller size), abbreviated by \texttt{C}, and rank metric codes, abbreviated by \texttt{R}. 
A \emph{matrix description} of a subspace code $\cV$ is a dissection of a rectangle into sub rectangles  describing the structure of a generating set for $\cV$, i.e., 
the structure of generator matrices for codewords in $\cV$. As an example we consider the following matrix description for $\cV$:
\begin{center}
  \begin{tabular}{|C{2cm}|C{2cm}|} 
    \hline
    \texttt{C} & \texttt{R}\\ 
    \hline
  \end{tabular}
\end{center}  
The meaning is that we assume the existence of a {\cdc} $\cC$ and a {\rmc} $\cM$ so that 
$$
  \left\{ \begin{pmatrix}A & M\end{pmatrix}\,:\, A\in\cG, M\in\cM  \right\}
$$
is a generating set of $\cV$, where $\cG$ is a generating set of $\cC$. Note that we need matrices representing the constant dimension codes in the components, since 
we want to end up with a generating set of matrices in the end. The fact that the matrices in $\cG$ and $\cM$ must have the same number of rows is indicated by common vertical 
border edge between the two cells. However, we do not assume that the rectangle dissection is true to scale. I.e., while the two cells have the same width, we do not 
assume that the matrices in $\cG$ and $\cM$ have the same number of columns. Of course the parameters of $\cC$ and $\cM$ determine the parameters of $\cV$. E.g.\ 
we are interested in a lower bound for the minimum distance and the cardinality of $\cV$ as well as whether $\cV$ is a {\cdc}. The details then are subject to a theorem. 
In our example the construction principle is called \emph{Construction D} in \cite{silberstein2015error} and the details can be found in Theorem~\ref{theorem_construction_d}.  

By $\texttt{0}$ we denote a rectangular all-zero matrix and by $\texttt{I}$ a unit matrix, which gives us the extra condition that the corresponding rectangle has to 
be a square in the dissection. Since an identity matrix generates a {\cdc} of cardinality $1$, we can specialize our example to:
\begin{center}
  \begin{tabular}{|C{2cm}|C{2cm}|} 
    \hline
    \texttt{I} & \texttt{R}\\ 
    \hline
  \end{tabular}
\end{center}  
This construction is known under the name of \emph{lifted {\mrd} codes} assuming that the involved {\rmc} is of maximum possible size, see 
Theorem~\ref{theorem_lifted_mrd}.
  
Another, almost trivial, specialization of our initial matrix description is:
\begin{center}
  \begin{tabular}{|C{2cm}|C{2cm}|} 
    \hline
    \texttt{C} & \texttt{0}\\ 
    \hline
  \end{tabular}
\end{center}  
Since we may permute columns arbitrary, it is equivalent to the description:  
\begin{center}
  \begin{tabular}{|C{2cm}|C{2cm}|} 
    \hline
    \texttt{0} & \texttt{C}\\ 
    \hline
  \end{tabular}
\end{center}

Such a subcode will be useful if combined with others only. So, we will also consider the combination of different matrix descriptions by listing them 
one underneath the other. An example, corresponding to the \emph{linkage construction} in Theorem~\ref{theorem_linkage}, is given by:
\begin{center}
  \begin{tabular}{|C{2cm}|C{2cm}|} 
    \hline
    \texttt{C} & \texttt{R}\\ 
    \hline
    \hline
    \texttt{0} & \texttt{C}\\ 
    \hline
  \end{tabular}
\end{center}
Here we align the vertical lines such that they reflect the relationship between the matrix sizes involved in the different subcodes. As an example, the 
\emph{improved linkage construction}, see Theorem~\ref{theorem_improved_linkage}, is described by:
\begin{center}
  \begin{tabular}{|C{2cm}|C{2cm}|} 
    \hline
    \texttt{C} & \texttt{R}\\ 
    \hline
   \end{tabular}\\
   \begin{tabular}{|C{1.5cm}|C{2.5cm}|} 
    \hline
    \texttt{0} & \texttt{C}\\ 
    \hline
  \end{tabular}
\end{center}
I.e., the length of the second {\cdc} can be strictly larger than the length of the used {\rmc}.

While those matrix descriptions are useful, not all constructions from the literature can be described that way. 

For other surveys on constructions for constant dimension codes we refer e.g.\ to \cite{horlemann2018constructions,khaleghi2009subspace}.

\section{Lifting, linkage, and related constructions}
\label{subsec_lifting}
In this subsection we briefly survey the so-called linkage construction with its different variants. The starting point is the same as for lifted {\mrd} 
codes. Instead of a $k\times k$ identity matrix $I_{k}$ (or $I_{k\times k}$) we can also use any matrix of full row rank $k$ as a prefix for the matrices from a rank metric code.

\begin{ntheorem}\textbf{(Lifting construction / Construction D -- \cite[Theorem~37]{silberstein2015error})}\\
  \label{theorem_construction_d}
  Let $\cC$ be an $(n_1,d;k)_q$--{\cdc} and $\cM$ be a $(k\times n_2,d/2)$--{\rmc}. Then
  $$
    \cW:=\left\{ \left\langle \begin{pmatrix}G&M\end{pmatrix}\right\rangle \,:\, G\in\cG, M\in \cM\right\},
  $$
  where $\cG$ is a generating set of $\cC$, is an $(n_1+n_2,d;k)_q$--{\cdc} with cardinality $\#\cW=\#\cC\cdot\#\cM$.
\end{ntheorem}
\begin{proof}
  For all $G\in \cG$ and all $M\in \cM$ we have $k\ge \rk(\begin{pmatrix}G&M\end{pmatrix}) \ge \rk(G)=k$, so that $\dim(W)=k$ for all $W\in\cW$, i.e., $\cW$ is a 
  {\cdc} with codewords of dimension $k$. 
  
  Now let $G,G'\in\cG$,  $M,M'\in\cM$ be arbitrary, $U=\langle G\rangle$, $U'=\langle G'\rangle$, $W=\left\langle \begin{pmatrix}E(U)&A\end{pmatrix}\right\rangle$, 
  and $W'=\left\langle \begin{pmatrix}E(U)'&A'\end{pmatrix}\right\rangle$. 
  If $G\neq G'$, then we have $U\neq U'$ so that 
  $$
    \ds(W,W')=2\cdot \rk\!\left(\begin{pmatrix} G & M\\ G' & M'\end{pmatrix}\right)-2k 
             \ge 2\cdot \rk\!\left(\begin{pmatrix} G\\ G'\end{pmatrix}\right)-2k=\ds(U,U')\ge d.  
  $$
  If $G=G'$, then we have $U=U'$ and $M\neq M'$ so that
  \begin{eqnarray*}
    \ds(W,W') &=& 2\cdot \rk\!\left(\begin{pmatrix} G & M\\ G & M'\end{pmatrix}\right)-2k=2\cdot \rk\!\left(\begin{pmatrix} G & M\\ \zv_{k\times m} & M'-M\end{pmatrix}\right)-2k\\ 
     &=&2\rk(G)+2\rk(M'-M)-2k=2\dr(M,M')\ge d.
  \end{eqnarray*}
\end{proof}

This generalized lifting idea was called \emph{Construction D} in \cite[Theorem~37]{silberstein2015error}, cf.~\cite[Theorem~5.1]{gluesing2015cyclic}. Note that if $\cC$ contains 
two codewords $U,U'$ with distance $\ds(U,U')=d$ and $\cM$ contains an element $M$ with $\rk(M)\le d/2$, which is the case if $\#\cM>1$, then we have 
$\ds(W,W')=d$ for $W=\left\langle\begin{pmatrix} U,M\end{pmatrix}\right\rangle$, $W'=\left\langle\begin{pmatrix} U',M\end{pmatrix}\right\rangle$. If $\cM$ contains 
two elements $M,M'$ with distance $\dr(M,M')=d/2$ and $\cC$ at least one element $U$, then we have $\ds(W,W')=d$ for $W=\left\langle\begin{pmatrix} U,M\end{pmatrix}\right\rangle$, 
$W'=\left\langle\begin{pmatrix} U,M'\end{pmatrix}\right\rangle$. So, the assumptions on the minimum distances of $\cC$ and $\cM$ are tight, i.e., they cannot be further relaxed 
besides degenerated and uninteresting special cases. Moreover, the parameter $m$ is the only degree of freedom that we have if we want to end up with an $(n,d;k)_q$--{\cdc} 
in the end, i.e., the formulation is as general as possible (assuming the corresponding matrix description). 

Choosing $\cC$ and $\cM$ as large as possible and using the parameterization $m=n_1$ and $n=n_1+n_2$, we conclude:
\begin{ncorollary}(C.f.~\cite[Theorem~37]{silberstein2015error})
  \label{cor_construction_d}
  \begin{equation}
    A_q(n,d;k)\ge A_q(m,d;k)\cdot A_q^R(k\times(n-m),d/2)
  \end{equation}
\end{ncorollary}
We find it convenient to split \cite[Theorem~37]{silberstein2015error} into Theorem~\ref{theorem_construction_d} and Corollary~\ref{cor_construction_d} since we will use 
Theorem~\ref{theorem_construction_d} in other contexts where we assume further conditions for $\cM$. The matrix description of construction D in Theorem~\ref{theorem_construction_d} 
is given by  
\begin{center}
  \begin{tabular}{|C{2cm}|C{2cm}|} 
    \hline
    \texttt{C} & \texttt{R}\\ 
    \hline
  \end{tabular}
\end{center}
Directly from the construction we read off:
\begin{nlemma}
  \label{lemma_pivot_structure_construction_d}
  The pivot structure of a {\cdc} obtained via construction D in Theorem~\ref{theorem_construction_d} is a subset of $\left({n_1 \choose k},{{n_2}\choose 0}\right)$.
\end{nlemma}
\begin{ncorollary}
  \label{cor_construction_d_pivot_structure}
  \begin{equation}
    A_q\!\left(n,d;k;{m \choose k},{{n-m}\choose 0}\right)\ge A_q(m,d;k)\cdot A_q^R(k\times (n-m)n,d/2)
  \end{equation}
\end{ncorollary}
%% \begin{nexercise} CORRECT AT ALL ???
%%   Show $A_q\!\left(n,d;k;{m \choose k},{{n-m}\choose 0}\right)= A_q(m,d;k)\cdot M(q,k,n-m,d/2)$.
%% \end{nexercise}
Besides being recursive, the lower bound in Corollary~\ref{cor_construction_d_pivot_structure} is very explicit and the only subtlety is a good choice of the 
free parameter $m$. Since the parameter space is rather small one may simply loop over all $1\le m\le n-1$. 

In \cite{kurz2019note} it was analyzed which codewords can be added to a subcode obtained via construction~D in Theorem~\ref{theorem_construction_d} without violating 
the minimum subspace distance. 
\begin{nlemma}
  \label{lemma_construction_d_additional_codewords}
  Let $\cC$ be a {\cdc} obtained via construction D in Theorem~\ref{theorem_construction_d} with parameters $\left(n_1,n_2,d,k\right)$ and $U\in\cG_q(n_1,k)$ with generator matrix 
  $G$ and pivot vector $v$. We have $\ds(\cC,U)\ge d$, i.e.\ $\cC\cup \{U\}$ is an $(n_1+n_2,d;k)_q$--{\cdc}, if one of the following equivalent conditions is satisfied:
  \begin{enumerate}
    \item[(a)] $\dH\big(\left({n_1 \choose k},{{n_2}\choose 0}\right),v\big)\ge d$;
    \item[(b)] at least $d/2$ of the $k$ ones in $v$ are contained in the last $n_2$ positions;
    \item[(c)] $\rk(G_1)\le k-d/2$, where $G_1\in\F_q^{k\times n_1}$, $G_2\in\F_q^{k\times n_2}$ with $G=\begin{pmatrix}G_1&G_2\end{pmatrix}$; and.
    \item[(d)] $\dim(U\cap E_2)\ge d/2$, where $E_2$ is the $n_2$-space spanned by the unit vectors $\be_i$ with $n_1+1\le i\le n_1+n_2$.
  \end{enumerate}  
\end{nlemma}
While the listed conditions are only sufficient in general, in some sense, they are indeed also necessary if our only information on $\cC$ is its matrix description or the 
pivot structure from Lemma~\ref{lemma_pivot_structure_construction_d}.    
\begin{ncorollary}
  \label{cor_general_linkage_structure}
  $$
    A_q(n,d;k) \ge A_q\!\left(n,d;k;{m \choose k},{{n-m}\choose 0}\right) + A_q\!\left(n,d;k;{m \choose {\le k-d/2}},{{n-m}\choose {\ge d/2}}\right)
  $$  
\end{ncorollary}
See e.g.~Exercise \ref{exercise_distance_analysis_1} for the corresponding distance analysis.

While the lower bound in Corollary~\ref{cor_general_linkage_structure} is very handy and indeed an essential ingredient for many good constructions in the literature, the second 
summand gives no hint how to construct corresponding subcodes.
\begin{ntheorem}
  \textbf{(Linkage construction -- \cite[Corollary~39]{silberstein2015error}, \cite[Theorem 2.3]{gluesing2016construction})}\\
  \label{theorem_linkage}
  Let $\cC_1$ be an $(n_1,d;k)_q$--{\cdc}, $\cC_2$ be an  $(n_2,d;k)_q$--{\cdc}, and $\cM$ be a $(k\times n_2,d/2)$--{\rmc}. Then,  
  $\cW:=\cW_1 \cup \cW_2$ is an $(n_1+n_2,d;k)$--{\cdc} of cardinality $\#\cC_1\cdot\#\cM+\#\cC_2$, where
  $$
    \left\{ \begin{pmatrix} G&M\end{pmatrix}\,:\,G\in \cG_1, M\in\cM\right\}
  $$  
  is a generating set of $\cW_1$,
  $$
    \left\{ \begin{pmatrix} 0_{k\times n_1} &G'\end{pmatrix}\,:\,G'\in \cG_2\right\}
  $$
  is a generating set of $\cW_2$, and $\cG_1,\cG_2$ are generating sets of $\cC_1$, $\cC_2$, respectively.
\end{ntheorem}
The matrix description of the linkage construction is given by:
\begin{center}
  \begin{tabular}{|C{2cm}|C{2cm}|} 
    \hline
    \texttt{C} & \texttt{R}\\ 
    \hline
    \hline
    \texttt{0} & \texttt{C}\\ 
    \hline
  \end{tabular}
\end{center}
The properties of the subcodes $\cW_1$ and $\cW_2$ may be directly concluded from Theorem~\ref{theorem_construction_d}. The {\lq\lq}linkage property{\rq\rq} 
$\ds(\cW_1,\cW_2)\ge d$ follows e.g.\ from Lemma~\ref{lemma_construction_d_additional_codewords}.(d) and $d\le 2k$. The latter also implies the 
observation $$A_q\!\left(n,d;k;{m \choose {\le k-d/2}},{{n-m}\choose {\ge d/2}}\right)\ge A_q(n-m,d;k).$$ 

\begin{nexample}
  For $n_1=4$, $n_2=4$, $d=6$, and $k=4$ choose $\cC_1=\cC_2=\left\{\langle I_{4}\rangle\right\}$, and $\cM$ as a $(4\times 4,3)_q$--{\mrd} code in 
  Theorem~\ref{theorem_linkage}. Since $\#\cC_1=\#\cC_2=1$ and $\#\cM=q^8$ we have $\#\cW_1=q^8$, $\#\cW_2=1$, and $\#\cW=q^8+1$, so that 
  $A_q(8,6;4)\ge q^8+1$. We remark that this is still the best known lower bound for all field sizes $q$ and that $A_2(8,6;4)=2^8+1=257$ was shown in 
  \cite{heinlein2019classifying}.
\end{nexample}

We remark that the verbal comparison of \cite[Theorem 2.3]{gluesing2016construction}), \cite[Corollary~39]{silberstein2015error}, and other similar 
variants in the literature with Theorem~\ref{theorem_linkage} are a bit involved due to different parameterizations and additional conditions that exclude 
cases where other constructions with competing code sizes are known. 
\begin{nexercise}
  Show:
  \begin{itemize}
    \item[(a)] if $n_1<k$, then $\#\cW_1=0$; if $n_2<k$, then $\#\cW_2=0$;
    \item[(b)] if $2k\le n_1+n_2\le 3k-1$, then the optimal choice is $n_1=k$, so that $\cW_1$ is an {\lmrd} code, cf.\ the additional condition $3k\le n_1+n_2$ 
               in \cite[Corollary~39]{silberstein2015error} noting that for $2k>n_1+n_2$ one may consider the orthogonal code;    
    \item[(c)] if $\cC_1$, $\cC_2$, and $\cM$ have minimum distance $d_1$, $d_2$, and $d_r$, respectively, then we have $d_1\ge d$, $d_2\ge d$, and 
               $d_r\ge d/2$ for $d=\min\{d_1,d_2,2d_r\}$, cf.~\cite[Theorem 2.3]{gluesing2016construction}.
  \end{itemize}  
\end{nexercise}    

\begin{ncorollary}
  $$
    A_q(n,d;k)\ge A_q(m,d;k)\cdot A_q^R(k\times(n-m);d/2)+A_q(n-m,d;k)
  $$  
\end{ncorollary}

Since the matrix descriptions of two subcodes in Theorem~\ref{theorem_linkage} are just column permutations of  
\begin{center}
  \begin{tabular}{|C{2cm}|C{2cm}|} 
    \hline
    \texttt{C} & \texttt{R}\\ 
    \hline
  \end{tabular}
\end{center}
we can use Lemma~\ref{lemma_construction_d_additional_codewords}.(d) to directly conclude a sufficient condition for the addition of further codewords 
to a {\cdc} constructed via the linkage construction:
\begin{nlemma}
  \label{lemma_additional_codewords_linkage_construction}
  Let $\cC$ be a {\cdc} obtained via the linkage construction in Theorem~\ref{theorem_linkage} with parameters $\left(n_1,n_2,d,k\right)$, $E_2$ be the $n_2$-space 
  spanned by the unit vectors $\be_i$ with $n_1+1\le i\le n_1+n_2$, and $E_1$ be the $n_1$-space spanned by the unit vectors $\be_i$ with $1\le i\le n_1$. 
  If $\dim(U\cap E_1)\ge d/2$ and $\dim(U\cap E_2)\ge d/2$ for $U\in\cG_q(n_1+n_2,k)$, then $\cC\cup\{U\}$ is an $(n_1+n_2,d;k)_q$--{\cdc}.  
\end{nlemma}  

Since we actually have $\ds(\cW_1,\cW_2)\ge 2k$ in Theorem~\ref{theorem_linkage} it can be easily improved if $d<2k$:
\begin{ntheorem}
  \textbf{(Improved linkage construction -- \cite[Theorem 18]{heinlein2017asymptotic})}\\
  \label{theorem_improved_linkage}
  Let $\cC_1$ be an $(n_1,d;k)_q$--{\cdc}, $\cC_2$ be an  $(n_2+k-d/2,d;k)_q$--{\cdc}, and $\cM$ be a $(k\times n_2,d/2)$--{\rmc}. Then,  
  $\cW:=\cW_1 \cup \cW_2$ is an $(n_1+n_2,d;k)$--{\cdc} of cardinality $\#\cC_1\cdot\#\cM+\#\cC_2$, where
  $$
    \left\{ \begin{pmatrix} G_1&M\end{pmatrix}\,:\,G_1\in \cG_1, M\in\cM\right\}
  $$  
  is a generating set of $\cW_1$,
  $$
    \left\{ \begin{pmatrix} 0_{k\times (n_1-k+d/2)} &G_2\end{pmatrix}\,:\,G_2\in \cG_2\right\}
  $$
  is a generating set of $\cW_2$, and $\cG_1,\cG_2$ are generating sets of $\cC_1$, $\cC_2$, respectively.
\end{ntheorem}
The matrix description of the improved linkage construction is given by:
\begin{center}
  \begin{tabular}{|C{2cm}|C{2cm}|} 
    \hline
    \texttt{C} & \texttt{R}\\ 
    \hline
  \end{tabular}\\
  \begin{tabular}{|C{1.5cm}|C{2.5cm}|}  
    \hline
    \texttt{0} & \texttt{C}\\ 
    \hline
  \end{tabular}
\end{center}
The {\lq\lq}linkage property{\rq\rq} $\ds(\cW_1,\cW_2)\ge d$ follows e.g.\ from Lemma~\ref{lemma_construction_d_additional_codewords}.(b). 
\begin{ncorollary}
  $$
    A_q(n,d;k)\ge A_q(m,d;k)\cdot A_q^R(k\times(n-m);d/2)+A_q(n-m+k-d/2,d;k)
  $$  
\end{ncorollary}
Clearly, the lower bounds that can be obtained with Theorem~\ref{theorem_improved_linkage} are at least as large as those from Theorem~\ref{theorem_linkage}. 

%% \begin{nexercise} WRONG!!!!!
%%   Let $\bv\in\cG_1(n,k)$. Show that we have
%%   $$
%%     \dH\!\big( \left({m\choose k},{{n-m}\choose 0}\right)\cup\left({{m-k+d/2}\choose 0},{{n-m+k-d/2}\choose k}\right),\bv\big)<d,
%%   $$
%%   i.e., the obvious super set for the pivot structure for codes obtained via the improved linkage construction in Theorem~\ref{theorem_improved_linkage} does 
%%   not permit the addition of further codewords. 
%% \end{nexercise}
Also using Lemma~\ref{lemma_construction_d_additional_codewords}.(d), we can adjust Lemma~\ref{lemma_additional_codewords_linkage_construction} to the 
improved linkage construction: 
\begin{nlemma}
  \label{lemma_additional_codewords_improved_linkage_construction}
  Let $\cC$ be a {\cdc} obtained via the improved linkage construction in Theorem~\ref{theorem_improved_linkage} with parameters $\left(n_1,n_2,d,k\right)$, 
  $E_2$ be the $n_2$-space spanned by the unit vectors $\be_i$ with $n_1+1\le i\le n_1+n_2$, and $E_1$ be the $n_1-k+d/2$-space spanned by the unit vectors 
  $\be_i$ with $1\le i\le n_1-k+d/2$. If $\dim(U\cap E_1)\ge d/2$ and $\dim(U\cap E_2)\ge d/2$ for $U\in\cG_q(n,k)$, then $\cC\cup\{U\}$ is an $(n_1+n_2,d;k)_q$--{\cdc}.  
\end{nlemma}  
\begin{nexercise}
  Let $\cW$ be a $(12,6;4)_q$--{\cdc} constructed via the improved linkage construction in Theorem~\ref{theorem_improved_linkage} with $m=6$. Determine 
  all $\bv\in\cG_1(12,6)$ such that for every $U\in\cG_q(12,6)$ with pivot vector $\bv$ we have $\ds(\cW,U)\ge 4$.    
\end{nexercise}
%% 2 Einsen in den ersten beiden Positionen und mindestens 2 Einsen in den letzten sechs Positionen.

A different variant of the linkage construction exploits Lemma~\ref{lemma_construction_d_additional_codewords}.(c), i.e., we ensure that the generator matrices of 
the additional codewords have rank at most $k-d/2$ in their first $n_1$ columns to deduce the {\lq\lq}linkage property{\rq\rq} $\ds(\cW_1,\cW_2)\ge d$:  
\begin{ntheorem}\textbf{(Generalized linkage construction -- \cite[Lemma 4.1 with $\mathbf{l=2}$]{cossidente2021combining})}\\
  \label{theorem_generalized_linkage} 
  Let $\cC_1$ be an $(n_1,d;k)_q$--{\cdc}, $\cC_2$ be an  $(n_2,d;k)_q$--{\cdc}, $\cM_1$ be a $(k\times n_2,d/2)$--{\rmc}, and $\cM_2$ be a 
  $(k\times n_1,d/2;\le k-d/2)$--{\rmc}. Then,  
  $\cW:=\cW_1 \cup \cW_2$ is an $(n_1+n_2,d;k)$--{\cdc} of cardinality $\#\cC_1\cdot\#\cM_1+\#\cC_2\cdot\cM_2$, where
  $$
    \left\{ \begin{pmatrix} G_1&M_1\end{pmatrix}\,:\,G_1\in \cG_1, M_1\in\cM_1\right\}
  $$  
  is a generating set of $\cW_1$,
  $$
    \left\{ \begin{pmatrix} M_2 &G_2\end{pmatrix}\,:\,G_2\in \cG_2, M_2\in\cM_2\right\}
  $$
  is a generating set of $\cW_2$, and $\cG_1,\cG_2$ are generating sets of $\cC_1$, $\cC_2$, respectively.
\end{ntheorem}
The matrix description of the generalized linkage construction is given by
\begin{center}
  \begin{tabular}{|C{2cm}|C{2cm}|} 
    \hline
    \texttt{C} & \texttt{R}\\ 
    \hline
    \hline
    \texttt{R} & \texttt{C}\\ 
    \hline
  \end{tabular}
\end{center}
so that the linkage construction is contained as a special subcase. See also \cite[Theorem 2]{he2020construction}.

\begin{ncorollary}
  We have $A_q(n,d;k)\ge$
  $$
     A_q(m,d;k)\cdot A_q^R(k\times(n-m);d/2)+A_q(n-m,d;k)\cdot A_q^R(k\times m,d/2;\le k-d/2).
  $$  
  The right hand side can be attained as the cardinality of an $(n,k;d)_q$--{\cdc} $\cW$ constructed by the generalized linkage construction  
  in Theorem~\ref{theorem_generalized_linkage}.
\end{ncorollary}

Using Lemma~\ref{lemma_construction_d_additional_codewords}.(d) we can directly conclude a sufficient condition for the addition of further codewords 
to a {\cdc} constructed via the generalized linkage construction:
\begin{nlemma}
  \label{lemma_additional_codewords_generalized_linkage_construction}
  Let $\cC$ be a {\cdc} obtained via the generalized linkage construction in Theorem~\ref{theorem_generalized_linkage} with parameters $\left(n_1,n_2,d,k\right)$, $E_2$ be the 
  $n_2$-space spanned by the unit vectors $\be_i$ with $n_1+1\le i\le n_1+n_2$, and $E_1$ be the $n_1$-space spanned by the unit vectors $\be_i$ with $1\le i\le n_1$. 
  If $\dim(U\cap E_1)\ge d/2$ and $\dim(U\cap E_2)\ge d/2$ for $U\in\cG_q(n_1+n_2,k)$, then $\cC\cup\{U\}$ is an $(n_1+n_2,d;k)_q$--{\cdc}.  
\end{nlemma}  

Theorem~\ref{theorem_generalized_linkage} has a lot of predecessors in the literature that cover special subcases and also alternative proofs.  
%% We discuss them in Subsection~\ref{subsec_subcases_generalized_linkage}. 
As indicated, Theorem~\ref{theorem_generalized_linkage} is just a special case of \cite[Lemma 4.1]{cossidente2021combining}. 
In Subsection~\ref{subsec_variants_generalized_linkage} we will consider variants and generalizations of Theorem~\ref{theorem_generalized_linkage}. However, for none of 
these an explicit strict improvement over Theorem~\ref{theorem_generalized_linkage} is known. See also e.g.\ \cite{chen2020new,li2019construction} for further 
variations of the linkage construction.    

\subsection{Variants of the generalized linkage construction}
\label{subsec_variants_generalized_linkage}
In its original formulation of the generalized linkage construction in \cite[Lemma 4.1]{cossidente2021combining}, the approach was extended to $l\ge 2$ subcodes $\cW_i$.  
Here we decompose the result into a few sub statements. Combining Construction D (Theorem~\ref{theorem_construction_d}) with the product construction for rank metric 
codes (Lemma~\ref{lemma_product_construction_rmc}) yields: 
\begin{nlemma} 
  \label{lemma_construction_d_generalized}
  Let $l\ge 2$ and $\bar{n}=\left(n_1,\dots,n_l\right)\in\N^l$. For $2\le i\le l$ let $\cM_i$ be a $(k\times n_i,d)_q$--{\rmc} and $\cC$ be an 
  $(n_1,d;k)_q$--{\cdc} with representation set $\cG$. With this, let
  $$
    \left\{ \begin{pmatrix} G& M_2&\dots&M_l\end{pmatrix} \,:\, G\in \cC, M_i\in \cM_i \forall 2\le i\le l\right\}
  $$
  a generating set and $\cW$ be the generated subspace code. Then, $\cW$ is an $(n,d;k)_q$--{\cdc} with cardinality 
  $\#\cW=\#\cC\cdot\prod\limits_{i=2}^l \#\cM_i$, where $n=\sum\limits_{i=1}^l n_i$.
 \end{nlemma}    
%% \begin{proof}
%%   Since $\rk(G)=k$ for all $G\in\cG$ we have $\dim(W)=k$ for all $w\in\cW$, i.e., $\cW$ is a {\cdc}. It remains to show $\ds(\cW)\ge d$. 
%%   To this end let $W,W'\in\cW$ be two different codewords and $\begin{pmatrix}G&M_2&\dots&M_l\end{pmatrix}$, $\begin{pmatrix}G'&M_2'&\dots&M_l'\end{pmatrix}$ 
%%   be the corresponding elements in the generating set. If $G\neq G'$, then we have
%%   \begin{eqnarray*}
%%     \ds(W,W') &=& 2\cdot\rk\left(\begin{pmatrix}G & M_2&\dots&M_l\\G' & M_2'&\dots&M_l'\end{pmatrix}\right)-2k \ge 
%%     2\cdot \rk\left(\begin{pmatrix}G\\G'\end{pmatrix}\right)-2k\\ 
%%     &=& 2\cdot \rk\left(\begin{pmatrix}G\\ G-G'\end{pmatrix}\right)-2k=\ds(U,U')\ge d,
%%    \end{eqnarray*}  
%%    where $U=\left\langle G\right\rangle$, $U'=\left\langle G'\right\rangle$ with $U,U'\in\cC$.
%%    
%%    If $G=G'$ then there exists an index $2\le i\le l$ with $M_i\neq M_i'$, so that
%%    \begin{eqnarray*}
%%     \ds(W,W') &=& 2\cdot\rk\left(\begin{pmatrix}G & M_2&\dots&M_l\\G' & M_2'&\dots&M_l'\end{pmatrix}\right)-2k \ge 
%%     2\cdot\rk\left(\begin{pmatrix}G&M_i\\G&M_i'\end{pmatrix}\right)-2k\\ 
%%     &=& 2\cdot \rk\left(\begin{pmatrix}G&M_i\\ 0_{k\times n_1} &M_i'-M_i\end{pmatrix}\right)-2k=2\cdot \rk(G)+2\dr(M_i,M_i')-2k \\
%%     &=& 2\dr(M_i,M_i')\ge \dr(\cM_i)\ge d.
%%   \end{eqnarray*}    
%% \end{proof}
The corresponding matrix description is given by
\begin{center}
  \begin{tabular}{|C{2cm}|C{2cm}|C{2cm}|C{2cm}|} 
    \hline
    \texttt{C} & \texttt{R} & \dots & \texttt{R}\\ 
    \hline
  \end{tabular}\\
\end{center} 
where the unique {\cdc}-component may be permuted to each of the $l\ge 2$ positions.  

\begin{ntheorem}
  \label{theorem_generalized_linkage_grid}
  Let $l\ge 2$ and $\bar{n}=\left(n_1,\dots,n_l\right)\in\N^l$. For $1\le i\le l$ let $\cC_i$ be an $(n_i,d;k)_q$--{\cdc} and $\cG_i$ a corresponding representation set. 
  For $1\le j<i\le l$ let $\cM_i^j$ be a $(k\times n_i,d;\le k-d/2)_q$--{\rmc} and for $1\le i<j\le l$ let $\cM_i^j$ be a $(k\times n_i,d)_q$--{\rmc}. 
  With this, let 
  $$
    \left\{ \begin{pmatrix} M_i^1 & \dots & M_i^{i-1} & G_i & M_i^{i+1} & \dots M_i^l\end{pmatrix} \,:\,  G_i\in\cG_i, M_i^j\in\cM_i^j \,\forall 1\le j\le l, j\neq i \right\}
  $$   
  be a generating set for the subcode $\cW_i$, where $1\le i\le l$. Then, $\cW=\cup_{i=1}^l \cW_i$ is an $(n,k;d)_q$--{\cdc}, 
  where $n=\sum_{i=1}^l n_i$.
\end{ntheorem}
\begin{proof}
  For $1\le i\le l$ the subcode $\cW_i$ is an $(n,d;k)_q$--{\cdc} with cardinality $$\prod_{j=1}^{i-1} \#\cM_i^j\cdot\#\cC_i\cdot\prod_{j=i+1}^l  \cM_i^j$$ by 
  Lemma~\ref{lemma_construction_d_generalized}. Let $$H=\begin{pmatrix} M_1&\dots M_{i-1}&G& M_{i+1}&\dots M_l\end{pmatrix}$$ be an arbitrary element in the generating set 
  of the subcode $\cW_i$ and $$H'=\begin{pmatrix} M_1'&\dots M_{i-1}'&G'& M_{i+1}'&\dots M_l'\end{pmatrix}$$ be an arbitrary element in the generating set of the subcode 
  $\cW_j$, where $1\le i<j\le l$ are arbitrary. Set $\bar{H}=\begin{pmatrix} G & M_j\end{pmatrix}$ and $\bar{H}'=\begin{pmatrix}M_i'&G'\end{pmatrix}$ and note 
  $\rk(H)=\rk(H')=\rk(\bar{H})=\rk(\bar{H}')=k$, so that $\ds(\left\langle H\right\rangle,\left\langle H'\right\rangle)\ge\ds(\left\langle \bar{H}\right\rangle,\left\langle \bar{H}'\right\rangle)$. 
  Since $\rk(M_i')\le k-d/2$ we can apply Lemma~\ref{lemma_construction_d_additional_codewords}.(c) to deduce $\ds(\cW_i,\cW_j)\ge d$, so that 
  $\ds(\cW)\ge d$.   
\end{proof}
The corresponding matrix description is given by
\begin{center}
  \begin{tabular}{|C{1.5cm}|C{1.5cm}|C{1.5cm}|C{1.5cm}|C{1.5cm}|} 
    \hline
    \texttt{C} & \texttt{R} & \texttt{R}& $\dots$ & \texttt{R}\\
    \hline
    \hline
    \texttt{R} & \texttt{C} & \texttt{R}& $\dots$ & \texttt{R}\\ 
    \hline
   \end{tabular}\\
   \begin{tabular}{C{1.5cm}C{1.5cm}C{1.5cm}C{1.5cm}C{1.5cm}}
     $\vdots$ & $\ddots$ & $\ddots$ & $\ddots$ & $\vdots$\\ 
   \end{tabular}
   \begin{tabular}{|C{1.5cm}|C{1.5cm}|C{1.5cm}|C{1.5cm}|C{1.5cm}|} 
    \hline
    \texttt{R} & $\dots$ & \texttt{R} & \texttt{C} & \texttt{R}\\  
    \hline 
    \hline
    \texttt{R} & $\dots$ & \texttt{R} & \texttt{R} & \texttt{C}\\  
    \hline
  \end{tabular}\\
\end{center} 
\begin{ncorollary}
  $$
    A_q(n,d;k)\ge \sum_{i=1}^l \left(\prod_{j=1}^{i-1} A_q^R(k\times n_j,\tfrac{d}{2};k-\tfrac{d}{2})\right)\cdot A_q(n_i,d;k)\cdot\left(\prod_{j=i+1}^l A_q^R(k\times n_j,\tfrac{d}{2})\right)
  $$  
  %%The right hand side can be attained as the cardinality of an $(n,k;d)_q$--{\cdc} $\cW$ constructed by the generalized linkage grid construction 
  %%in Theorem~\ref{theorem_generalized_linkage_grid} with $\bar{n}=\left(n_1,\dots,n_l\right)$.  
\end{ncorollary}
We remark that in the original formulation of \cite[Lemma 4.1]{cossidente2021combining} the rank metric codes $\cM_i^j$; where $1\le j\le l$ and $j\neq i$, are 
assumed to be subcodes of a $(k\times n_i,d/2)_q$--{\rmc} $\cM_i$, which is not necessary and may make a difference if $l\ge 3$ only. However, currently 
none of the best known codes uses Theorem~\ref{theorem_generalized_linkage_grid} or \cite[Lemma 4.1]{cossidente2021combining} with $l\ge 3$. 
Actually, the parameter $l$ in Theorem~\ref{theorem_generalized_linkage_grid} can be recursively reduced to $2$, so that we finally end up with Theorem~\ref{theorem_linkage}:
\begin{nexercise}
  Let $\cW$ be an $(n,d;k)_q$--{\cdc} constructed via Theorem~\ref{theorem_generalized_linkage_grid} with $l\ge 3$. Set 
  \begin{itemize}
    \item $\widehat{n}_i=n_i$ for all $1\le i\le l-2$, $\widehat{n}_{l-1}=n_{l-1}+n_l$; 
    \item $\widehat{\cC}_i=\cC_i$ for all $1\le i\le l-2$;
    \item $\widehat{\cM}_i^j=\cM_i^j$ for all $1\le i,j\le l-2$, $i\neq j$;
    \item $\widehat{\cM}_i^{l-1}=\cM_i^{l-1}\times \cM_i^l$ for all $1\le i\le l-2$;
    \item $\widehat{\cC}_{l-1}$ to the {\cdc} obtained from the generalized linkage construction in Theorem~\ref{theorem_generalized_linkage} using 
          $\cC_{l-1}$, $\cC_{l}$, $\cM_l^{l-1}$, and $\cM_{l-1}^l$; and
    \item $\widehat{\cM}_{l-1}^j=\cM_h^j$ for all $1\le j\le l-2$, where $h\!\in\! \{l\!-\!1,l\}$ maximizes $\#\cM_h^1\times\dots\times \cM_h^{l-2}$.   
  \end{itemize}  
  Show that we can apply Theorem~\ref{theorem_generalized_linkage_grid} with the above components to obtain a {\cdc} $\widehat{\cW}$ with $\#\widehat{\cW}\ge \#\cW$. 
\end{nexercise}

In principle it is not necessary that the matrix description of the generalized linkage construction has a grid-like structure.
\begin{ntheorem}  (\cite[Theorem 26]{heinlein2020generalized})
  \label{theorem_non_grid_generalized_linkage}
  Let $\cC_1$ be an $(n_1,d;k)_q$--{\cdc}, $\cC_2$ be an  $(n_2+t,d;k)_q$--{\cdc}, $\cM_1$ be a $(k\times n_2,d/2)$--{\rmc}, and $\cM_2$ be a 
  $(k\times (n_1-t),d/2;\le k-d/2-t)$--{\rmc}. Then,  
  $\cW:=\cW_1 \cup \cW_2$ is an $(n_1+n_2,d;k)$--{\cdc} of cardinality $\#\cC_1\cdot\#\cM+\#\cC_2\cdot\#\cM_2$, where
  $$
    \left\{ \begin{pmatrix} G_1 &M_1\end{pmatrix}\,:\,G_1\in \cG_1, M_1\in\cM\_1\right\}
  $$  
  is a generating set of $\cW_1$,
  $$
    \left\{ \begin{pmatrix} M_2 &G_2\end{pmatrix}\,:\,G_2\in \cG_2,M_2\in\cM_2\right\}
  $$
  is a generating set of $\cW_2$, and $\cG_1,\cG_2$ are generating sets of $\cC_1$, $\cC_2$, respectively.
\end{ntheorem}
The corresponding matrix description is given by
\begin{center}
  \begin{tabular}{|C{2cm}|C{2cm}|} 
    \hline
    \texttt{C} & \texttt{R}\\ 
    \hline
  \end{tabular}\\  
  \begin{tabular}{|C{1.5cm}|C{2.5cm}|}
    \hline
    \texttt{R} & \texttt{C}\\ 
    \hline
  \end{tabular}
\end{center}
so that Theorem~\ref{theorem_non_grid_generalized_linkage} generalizes the improved linkage construction in Theorem~\ref{theorem_improved_linkage}. However, currently  
no single case where Theorem~\ref{theorem_non_grid_generalized_linkage} yields strictly larger codes than Theorem~\ref{theorem_improved_linkage} and 
Theorem~\ref{theorem_generalized_linkage} is known. 
\begin{ncorollary}
  \begin{eqnarray*}
    A_q(n,d;k)&\ge& A_q(m,d;k)\cdot A_q^R(k\times(n-m),d/2)\\ 
    &&+A_q(n-m+t,d;k)\cdot A_q^R(k\times (m-t),d/2;\le k-d/2)
  \end{eqnarray*}  
\end{ncorollary}

\section{The Echelon--Ferrers construction and their variants}
\label{subsec_echelon_ferrers}

The basis for the \emph{Echelon--Ferrers} or \emph{multilevel construction} from \cite{etzion2009error} is Inequality~(\ref{ie_subspace_distance_hamming}), i.e.\ 
$\ds(U,W)\ge \dH\!\big(v(U),v(W)\big)$.

\begin{ntheorem}\textbf{(Multilevel construction -- \cite[Theorem 3]{etzion2009error})}\\ \label{thm_EF}
  Let $\cS\subseteq\cG_1(n,k)$ with $\dH(\cS)\ge d$. If $\cC_v\subseteq\cG_q(n,k)$ is an 
  $(n,d;k)_q$--{\cdc} whose codewords have pivot vector $v$ for each $v\in \cS$, then $\cC=\cup_{v\in\cS} \cC_v$ is an $(n,d;k)_q$--{\cdc} with cardinality $\sum_{v\in \cS} \#\cC_v$.
\end{ntheorem}
Suitable choices for the $\cC_v$ are also discussed in e.g.\ \cite{etzion2009error} and we will do so in a moment, see Example~\ref{example_EF}. The set $\cS$ is a binary code with minimum 
Hamming distance $d$ and sometimes called \emph{skeleton code}. By $A_q(n,d;k;v)$ we denote the maximum possible cardinality $M$ of an $(n,d;k)_q$-{\cdc} where 
all codewords have pivot vector $v$, so that Theorem~\ref{thm_EF} gives the lower bound  
\begin{equation}
  A_q(n,d;k)\ge \sum_{v\in \cS} A_q(n,d;k;v),
\end{equation}
where $\dH(\cS)\ge d$. Actually the notion $A_q(n,d;k;v)$ is a special case of our notion $A_q(n,d;k;\cV)$ for arbitrary subsets $\cV\subseteq\cG_1(n,k)$. 
And so also Theorem~\ref{thm_EF} can be generalized: 
\begin{ntheorem}(\cite[Theorem 2.3]{kurz2021interplay})\\
  \label{thm_generalized_skeleton_code}
  Let $\cV_1,\dots,\cV_s$ be subsets of $\cG_1(n,k)$ with $\dH(\cV_i,\cV_j)\ge d$ for all $1\le i<j\le s$. If $\cC_{\cV_i}\subseteq\cG_q(n,k)$ is an 
  $(n,d;k)_q$--{\cdc} with pivot structure $\cV_i$ for each $1\le i\le s$, then $\cC=\cup_{1\le i\le s} \cC_{\cV_i}$ is an $(n,d;k)_q$--{\cdc} with cardinality 
  $\sum_{1\le i\le s} \#\cC_{\cV_i}$. 
\end{ntheorem} 

We call $S=\left\{\cV_1,\dots,\cV_s\right\}$ a \emph{generalized skeleton code}, see \cite{kurz2021interplay}. For constructions that fit into the context of 
Theorem~\ref{thm_generalized_skeleton_code} we refer e.g.\ to \cite{he2020improving,kurz2021interplay}.

Given a Ferrers diagram $\cF$ with $m$ dots in the rightmost column and $l$ dots 
in the top row, we call a rank-metric code $C_{\cF}$ a \emph{Ferrers diagram rank-metric} (\fdrm) 
code if for any codeword $M\in \F_q^{m\times l}$ of $C_{\cF}$ all entries not in $\cF$ are zero. 
By $\dr(C_{\cF})$ we denote the minimum rank distance, i.e., the minimum of the rank distance between 
pairs of different codewords. 

\begin{ndefinition}\label{definition_lifted_ferrers}
  (\cite{silberstein2015error})\\
  Let $\cF$ be a Ferrers diagram and $C_{\cF}\subseteq \F_q^{k\times(n-k)}$ be an {\fdrm} code. The 
  corresponding \emph{lifted {\fdrm} code} $\cC_{\cF}$ is given by
  $$
    \cC_{\cF}=\left\{U\in\cG_q(n,k)\,:\, \cF(U)=\cF, T(U)\in C_{\cF}\right\}\!.
  $$ 
\end{ndefinition}
\begin{nlemma}(\cite[Lemma 4]{etzion2009error})\\\label{lemma_FDRM_CDC_equivalence} 
  Let $C_{\cF}\subseteq \F_q^{k\times(n-k)}$ be an {\fdrm} code with minimum rank distance $\delta$, then the lifted 
  {\fdrm} code $\cC_{\cF}\subseteq \cG_q(n,k)$ is an $(n,2\delta;k)_q$--{\cdc} of cardinality $\#C_{\cF}$.
\end{nlemma}

\begin{nexample}
\label{example_EF}
For the Ferrers diagram 
$$
\cF=
  \begin{array}{llll}
    \bullet & \bullet & \bullet & \bullet \\
    \bullet & \bullet & \bullet & \bullet \\
    \bullet & \bullet & \bullet & \bullet \\
  \end{array}  
$$
over $\F_2$ a linear {\fdrm} code with minimum rank distance $\dr=3$ and cardinality $16$ is given by 
$$
C_{\cF}=\left\langle
\begin{pmatrix}
  0 & 1 & 0 & 0\\ 
  0 & 0 & 1 & 0\\ 
  0 & 0 & 0 & 1\\ 
\end{pmatrix},
\begin{pmatrix}
  1 & 0 & 0 & 0\\ 
  0 & 0 & 1 & 1\\ 
  0 & 0 & 1 & 0\\ 
\end{pmatrix},
\begin{pmatrix}
  0 & 1 & 1 & 0\\ 
  1 & 0 & 0 & 1\\ 
  0 & 1 & 0 & 0\\ 
\end{pmatrix},
\begin{pmatrix}
  0 & 0 & 0 & 1 \\
  1 & 1 & 0 & 1 \\
  1 & 0 & 1 & 0\\
\end{pmatrix}\right\rangle\subseteq\F_2^{3\times 4}.
$$
Via lifting we obtain a {\cdc} with pivot structure $\left\{(1,1,1,0,0,0,0)\right\}$ showing $$A_2(7,6;3;(1,1,1,0,0,0,0))\ge 16.$$ Since 
$\dH\big((1,1,1,0,0,0,0),(0,0,0,1,1,0,1)\big)$ $=6$ we have 
$$
  A_2(7,6;3)\ge A_2(7,6;3;(1,1,1,0,0,0,0))+A_2(7,6;3;(0,0,0,1,1,0,1)).
$$
The Ferrers diagram for pivot vector $(0,0,0,1,1,0,1)$ is $\begin{array}{l}\bullet\\\bullet\end{array}$ with e.g.\ $\left\{\begin{pmatrix}0\\1\end{pmatrix}\right\}$ as a 
possible {\fdrm} code. The corresponding lifted codeword has generator matrix
$$
  \begin{pmatrix}
    0 & 0 & 0 & 1 & 0 & 0 & 0 \\ 
    0 & 0 & 0 & 0 & 1 & 1 & 0 \\ 
    0 & 0 & 0 & 0 & 0 & 0 & 1  
  \end{pmatrix}.
$$ 
Since $A_2(7,6;3)=17$, see e.g.\ the partial spread bound in Theorem~\ref{thm:lb_def}, we have $$A_2(7,6;3;(1,1,1,0,0,0,0))=16$$ and $A_2(7,6;3;(0,0,0,1,1,0,1))=1$.  
\end{nexample}

Lifted {\fdrm} codes $\cC_{\cF}$ are exactly the subcodes $\cC_v$ needed in the Echelon-Ferrers construction in Theorem~\ref{thm_EF}.
In \cite[Theorem 1]{etzion2009error} a general upper bound for (linear) {\fdrm} codes was given. Since the bound is also 
true for non-linear {\fdrm} codes, as observed by several authors, denoting the pivot vector corresponding to 
a given Ferrers diagram $\cF$ by $v(\cF)$ and using Lemma~\ref{lemma_FDRM_CDC_equivalence}, we can rewrite the upper bound to:

\begin{ntheorem}
  \label{thm_upper_bound_ef}
  $$
    A_q(n,d;k;v(\cF))\le q^{\min\!\left\{\nu_i\,:\,0\le i\le d/2-1\right\}},
  $$
  where $\nu_i$ is the number of dots in $\cF$, which are neither contained in the first $i$ rows nor contained in the 
  last $\tfrac{d}{2}-1-i$ columns. 
\end{ntheorem}  
If we choose a minimum subspace distance of $d=6$, then we obtain $$A_2(9,6;4;101101000)\le 2^7$$ due to
$$
  \begin{array}{lllll}
    \circ & \circ & \circ & \bullet & \bullet \\
          & \circ & \circ & \bullet & \bullet \\
          & \circ & \circ & \bullet & \bullet \\
          &       & \circ & \bullet & \bullet 
  \end{array}
  \quad  
  \begin{array}{llllll}
    \bullet & \bullet & \bullet & \bullet & \bullet \\
        & \circ & \circ & \circ & \bullet \\
        & \circ & \circ & \circ & \bullet \\
        &     & \circ & \circ & \bullet 
  \end{array}
  \quad  
  \begin{array}{llllll}
    \bullet & \bullet & \bullet & \bullet & \bullet \\
        & \bullet & \bullet & \bullet & \bullet \\
        & \circ & \circ & \circ & \circ \\
        &     & \circ & \circ & \circ 
  \end{array}
  .
$$
where the non-solid dots are those that are neither contained in the first $i$ rows nor contained in the last $\tfrac{d}{2}-1-i$ columns for $1\le i\le 3$.

While it is conjectured that the upper bound from Theorem~\ref{thm_upper_bound_ef} (and the corresponding bound for {\fdrm} codes) can always be attained, this 
problem is currently solved for specific instances like e.g.\ rank-distances $\delta=2$ only. For more results see e.g.\ \cite{antrobus2019state,antrobus2019maximal,
etzion2016optimal,liu2019constructions} and the references mentioned therein.

\begin{nexample}
  \label{example_generalized_skeleton_code}
  We choose a generalized skeleton code $\cS$ with vertices {\tiny 
  $\left(\!{4\choose 0},\!{7\choose 4}\!\right)$, $0 0 0 1 0 0 0 0 1 1 1$, $0 0 0 1 0 1 0 0 0 1 1$,  
    $0 0 0 1 1 0 0 0 0 1 1$,\, 
    $0 0 0 1 1 0 0 0 1 1 0$,\,
    $0 0 1 0 0 0 0 1 0 1 1$,\, 
    $0 0 1 0 0 0 0 1 1 0 1$,\, 
    $0 0 1 0 0 0 0 1 1 1 0$,\, 
    $0 0 1 0 0 1 0 0 1 0 1$,\, 
    $0 0 1 0 0 1 0 0 1 1 0$,\, 
    $0 0 1 0 0 1 0 1 0 0 1$, 
    $0 0 1 0 1 0 0 0 1 0 1$, 
    $0 0 1 1 0 0 0 0 1 1 0$, 
    $0 0 1 1 0 1 0 1 0 0 0$, 
    $0 1 1 0 0 0 1 0 0 0 1$, 
    $1 0 0 0 0 1 0 1 1 0 0$, 
    $1 0 0 0 1 0 0 1 0 0 1$, 
    $1 0 0 1 1 1 0 0 0 0 0$ 
    $1 0 1 0 0 0 0 0 0 1 1$, and 
    $1 0 1 0 0 1 1 0 0 0 0$}, so that
    $$
      A_q(11,4;4)\ge q^{21}+q^{17}+2q^{15}+3q^{14}+4q^{13}+q^{12}+q^{11}+q^9+2q^7+2q^6+q^5+A_q(7,4;4),
    $$
    see \cite[Proposition 3.1]{kurz2021interplay}.
\end{nexample}  

While the upper bound from Theorem~\ref{thm_upper_bound_ef} can always be attained for minimum subspace distance $d=4$, the determination of a {\lq\lq}good{\rq\rq}
(generalized) skeleton code is still a tough discrete optimization problem.\footnote{Note that it generalizes the computation of $A(n,d;k)$.} In \cite{kurz2020lifted} 
several new (generalized) skeleton codes improving the previously best known lower bounds for $A_q(n,d;k)$ are given. We remark that it is also possible 
to compute upper bounds for the cardinalities of {\cdc}s that can be obtained by the Echelon--Ferrers construction and to perform those computations parametric in 
the field size $q$, see \cite{feng2020bounds}. There are many other papers with explicitly determine (generalized) skeleton codes and heuristic algorithms to 
compute them, see the citations of \cite{etzion2009error}. For greedy-type approaches we refer to e.g.\ \cite{he2019hierarchical,shishkin2016new,shishkin2014cardinality}.

For the case of partial spreads, i.e.\ for $d=2k\le n$, the determination of a good skeleton code for the Echelon--Ferrers construction is rather easy. Note that 
the condition $\dH(v,v')\ge d=2k$ for $v,v'\in\cG_1(n,k)$ means that the ones of $v$ and those of $v'$ have to be disjoint, so that $A(n,2k;k)\le \left\lfloor n/k\right\rfloor$. 
By choosing $v^i\in\cG_1(n,k)$ such that the $k$ ones are in positions $(i-1)k+1,\dots,ik$ for $1\le i\le \left\lfloor n/k\right\rfloor$ the upper bound can 
be attained and all corresponding Ferrers diagrams are rectangular, so that we can use {\mrd} codes.
\begin{nexercise}
  \label{exercise_lb_partial_spreads}
  Show $A_q(n,2k;k)\ge \frac{q^n-q^k(q^{(n \bmod k)}-1)-1}{q^k-1}$ for $2k\le n$.
\end{nexercise}  
We remark that a more general construction, along similar lines and including explicit formulas for the respective cardinalities, has 
been presented in \cite{skachek2010recursive}, see also~\cite{gabidulin2021bounds}. For another approach how to select the skeleton codes via 
so-called lexicodes see \cite{silberstein2011large}.

Consider the following three Ferrers diagrams
$$
  \begin{array}{llllll}
    \circ & \circ & \circ & \bullet & \bullet & \bullet \\
          &       & \circ & \bullet & \bullet & \bullet \\
          &       &       & \bullet & \bullet & \bullet 
  \end{array},\quad
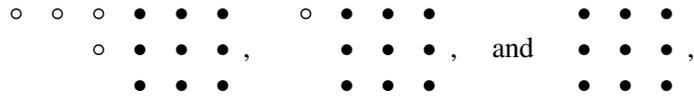
  \begin{array}{llll}
    \circ   & \bullet & \bullet & \bullet \\
            & \bullet & \bullet & \bullet \\
            & \bullet & \bullet & \bullet 
  \end{array},\quad\text{and}\quad
  \begin{array}{lll}
    \bullet & \bullet & \bullet \\
    \bullet & \bullet & \bullet \\
    \bullet & \bullet & \bullet 
  \end{array},
$$
where we have marked a few special dots by non-solid circles. For minimum rank distance $\dr=3$ corresponding {\fdrm} or lifted {\fdrm} codes can have a cardinality of at most 
$q^3$ in all three cases (and this upper bound can indeed be attained). So, we can remove the non-solid circles from the diagrams without decreasing the upper bound. Or, 
framed differently, we can used this free extra positions to add a few more codewords. The single non-solid circle in the middle diagram is called a \emph{pending dot}, see 
\cite{etzion2012codes} for the details. This notion was generalized to so-called \emph{pending blocks} and the four non-solid circles in the leftmost diagram form such a pending 
block. For details we refer to \cite{silberstein2013new,silberstein2015error,trautmann2013constructions}. 

Explicit series of constructions using pending dots are e.g.\ given by the following two theorems.
\begin{ntheorem}\textbf{(Construction 1 -- \cite[Chapter IV, Theorem 16]{etzion2012codes})}\\
\label{theorem_construction_1}
$$
  A_q(n,2(k-1);k) \ge q^{2(n-k)} + A_q(n-k,2(k-2);k-1)
$$ 
if $q^2+q+1 \ge s$ with $s=n-4$ if $n$ is odd and $s=n-3$ else.
\end{ntheorem}

\begin{ntheorem}\textbf{(Construction 2 --\cite[Chapter IV, Theorem 17]{etzion2012codes})}
$$
  A_q(n,4;3) \ge q^{2(n-3)} + \sum_{i=1}^{\alpha} q^{2(n-3-(q^2+q+2)i)}
$$ 
if $q^2+q+1 < s$ with $s=n-4$ if n is odd and $s=n-3$ else and $\alpha = \left\lfloor \frac{n-3}{q^2+q+2} \right\rfloor$
\end{ntheorem}

Explicit series of constructions using pending blocks are e.g.\ given by the following two theorems.
\begin{ntheorem}\textbf{(Construction A -- \cite[Chapter III, Theorem 19, Corollary 20]{silberstein2015error})}\\
Let $n\geq \frac{k^2+3k-2}{2}$ and $q^2+q+1\geq \ell$, where $\ell= n-\frac{k^2+k-6}{2}$ for odd $n-\frac{k^2+k-6}{2}$ (or $\ell= n-\frac{k^2+k-4}{2}$ for even $n-\frac{k^2+k-6}{2}$). Then
$A_q(n, 2k-2; k) \ge q^{2(n-k)}+\sum_{j=3}^{k-1} q^{2(n-\sum_{i=j}^k i)}+\qbin{n-\frac{k^2+k-6}{2}}{2}{q}$.
\end{ntheorem}

\begin{ntheorem}\textbf{(Construction B -- \cite[Chapter IV, Theorem 26, Corollary 27]{silberstein2015error})}\\
Let $n\geq 2k+2$. Then we have $A_q(n,4;k) \ge$ 
$$\sum_{i=1}^{\lfloor\frac{n-2}{k}\rfloor -1}\left( q^{(k-1)(n-ik)}+ \frac{(q^{2(k-2)}-1)(q^{2(n-ik-1)}-1)}{(q^4-1)^2}q^{(k-3)(n-ik-2)+4}\right).$$
\end{ntheorem}

%% \begin{ntheorem}\textbf{(Two-pivot-blocks Construction -- \cite{etzion2016optimal})}\\
%% Let $q \ge 2$ be a prime power and $2 \le d/2 \le k \le v-k$ integers.
%% If additionally $d \le k+1$, then
%% \[A_q(n,d;k)\ge q^{(v-k)(k-d/2+1)} \frac{q^{(d/2)^2(M+1)}-1}{q^{(d/2)^2}-1}q^{-(d/2)^2M}\]
%% with $M=\lceil 2(v-k)/d \rceil$.
%% \end{ntheorem}

\section{The coset construction}
\label{subsec_coset_construction}

The starting point for the so-called \emph{coset construction} introduce in \cite{heinlein2017coset} was \cite[Construction III]{etzion2012codes} leading to the lower bound 
$A_2(8,4;4)\ge 4797$. The corresponding generator matrices have the form
$$
  \begin{pmatrix}
    G_1 & \varphi_H(M) \\ 
    \mathbf{0} & G_2
  \end{pmatrix}  
$$
where $G_1\in\F_q^{k_1\times n_1}$ and $G_2\in\F_q^{k_2\times n_2}$ are generator matrices of $\left(n_1,d;k_1\right)_q$- and $\left(n_2,d;k_2\right)_q$--{\cdc}s, respectively. 
The matrix $M\in\F_q^{k_1\times(n_2-k_2)}$ is an element of a $\left(k_1\times (n_2-k_2),d/2\right)_q$--{\rmc} and the function $\varphi_{G_2}$ maps $M$ into 
$\F_q^{k_1\times n_2}$ by inserting $k_2$ additional zero columns at a set $S$ of positions where corresponding submatrix of $G_2$ has rank $k_2$.  
\begin{ndefinition}
  Let $M\in \F_q^{k\times n}$ be arbitrary and $S$ a subset of $\{1,\dots,n\}$. By $M|_S$ we denote the restriction of $M$ to the columns of $M$ with indices in $S$.
\end{ndefinition}
For one-element subsets we also use the abbreviation $M|_i=M|_{\{i\}}$.
\begin{nexample}
  For $M=\begin{pmatrix}1 & 0 & 1 & 0 & 1\\ 1 & 1 & 1 & 0 & 0\end{pmatrix}\in\F_2^{2\times 5}$ and $S=\{1,3,5\}$ we have $M|_S=\begin{pmatrix}1 & 1 & 1\\ 1 & 1 & 0\end{pmatrix}$.  
\end{nexample}
\begin{ndefinition}
  \label{def_embedding_function}
  Let $G\in\F_q^{k_2\times n}$ of rank $k_2$ and $M\in\F_q^{k_1\times (n-k_2)}$ be arbitrary. We call function $\varphi\colon \F_q^{k_1\times(n-k_2)}\to\F_q^{k_1\times n}$ an 
  \emph{embedding function compatible with $G$} if there exists a subset $S\subseteq\{1,\dots,n\}$ of cardinality $k_2$ such that $\varphi(M)|_S=\mathbf{0}_{k_1\times k_2}$ and 
  $\rk\!\left(G|_S\right)=\rk(G)=k_2$.  
\end{ndefinition}
In order to indicate the dependence on $H$ we typically denote embedding functions compatible with $G$ by $\varphi_G$. As an abbreviation for the function value 
$\varphi_G\!(M)$ we also write $M\!\!\uparrow_G$ or $M\!\!\uparrow$, whenever $G$ is clear from the context or secondary. A feasible and typical choice for $\varphi_G$  
is to choose the index set $S$ as the set of the pivot positions in $E(G)$.
\begin{nexample}
  For  
  $G=\begin{pmatrix}
  0 & 1 & 1 & 0 & 1 & 0 \\
  0 & 1 & 1 & 1 & 1 & 0 \\
  0 & 0 & 0 & 0 & 0 & 1 
  \end{pmatrix}$
  and
  $M=\begin{pmatrix}
  1 & 0 & 0 \\
  0 & 1 & 0 \\
  1 & 0 & 1 \\
  0 & 1 & 1 
  \end{pmatrix}$ we have 
  $$E(G)=\begin{pmatrix}
  0 & 1 & 1 & 0 & 1 & 0 \\
  0 & 0 & 0 & 1 & 0 & 0 \\
  0 & 0 & 0 & 0 & 0 & 1 
  \end{pmatrix}$$
  so that $v(G)=010101$ and $S:=\left\{1\le i\le 6\,:\, v(G)|_i=1\right\}=\{2,4,6\}$. For the embedding function $\varphi_G$ compatible with $H$ defined via the index set 
  $S$ we have
  $$
    \varphi_G(M)=\begin{pmatrix}
     1 & 0 & 0 & 0 & 0 & 0 \\
     0 & 0 & 1 & 0 & 0 & 0 \\
     1 & 0 & 0 & 0 & 1 & 0 \\
     0 & 0 & 1 & 0 & 1 & 0  
    \end{pmatrix}.
  $$  
\end{nexample} 
\begin{nlemma}
  Let $G\in\F_q^{k_2\times n}$ with $\rk(G)=k_2$ and $\varphi_G\colon \F_q^{k_1\times(n-k_2)}\to\F_q^{k_1\times n}$ an embedding function compatible with $G$. Then, we 
  have 
  \begin{equation}
    \rk\left(\begin{pmatrix}\varphi_G(M)\\G\end{pmatrix}\right)=\rk(G)+\rk(M)=k_2+\rk(M)
  \end{equation}
  for all $M\in \F_q^{k_1\times(n-k_2)}$
  and
  \begin{eqnarray}
    \rk\left(\begin{pmatrix}\sum\limits_{i=1}^l \lambda_i\cdot\varphi_G(M_i)\\G\end{pmatrix}\right) &=&
    \rk(G)+\rk\left(\sum_{i=1}^l \lambda_i\cdot M_i\right)\notag\\ 
    &=& k_2+\rk\left(\sum_{i=1}^l \lambda_i\cdot M_i\right)
  \end{eqnarray}
  for all $l\in \N$, and $\lambda_i\in \F_q$, $M_i\in\F_q^{k_1\times(n-k_2)}$ with $1\le i\le l$.
\end{nlemma}
\begin{proof}
  Let $S\subseteq\{1,\dots,n\}$ be the subset in Definition~\ref{def_embedding_function} corresponding to $\varphi_G$ and $[n]\backslash S=\{1,\dots,n\}\backslash S$. Note 
  that we have $\varphi_G(M)|_S=\mathbf{0}_{k_1\times k_2}$ and $\varphi_G(M)|_{[n]\backslash S}=M$ for all $M\in\F_q^{k_1\times(n-k_2)}$. Since $\rk(G|_S)=\rk(G)=k_2$ we have
  $$
    \rk\left(\begin{pmatrix}\varphi_G(M)\\G\end{pmatrix}\right) 
    =\rk\left(\begin{pmatrix}\mathbf{0}_{k_1\times k_2} & M\\ G|_S & G|_{[n]\backslash S}\end{pmatrix}\right)=
    \rk(M)+\rk(G|_S)=\rk(G)+\rk(M),
  $$ 
  i.e., the first equation is valid (using $\rk(G)=k_2$).
  
  Set $M=\sum_{i=1}^l \lambda_i M_i\in \F_q^{k_1\times(n-k_2}$ and $M'=\sum_{i=1}^l \varphi_G(M_i)\in\F_q^{k_1\times n}$. 
  Since $\varphi_G(M)=M'$ the second equation directly follows from the first. 
\end{proof}

\begin{nlemma}\textbf{(Product construction for constant dimension codes)}
  \label{lemma_coset_part}
  Let $\cC_1$ be an $\left(n_1,d;k_1\right)_q$--{\cdc}, $\cC_2$ be an $\left(n_2,d;k_2\right)_q$--{\cdc}, $\cM$ be a $(k_1\times (n_2-k_2),d/2)_q$--{\rmc}, and 
  $\cG_1$, $\cG_2$ be generating sets of $\cC_1$, $\cC_2$, respectively. For each $G_2\in \cG_2$ we denote by $\varphi_{G_2}$ an embedding function 
  $\F_q^{k_1\times(n_2-k_2)}\to\F_q^{k_1\times n_2}$ compatible with $G_2$. With this,
  $$
    \left\{\begin{pmatrix} G_1 & \varphi_{G_2}(M) \\ \mathbf{0}_{k_2\times n_1} & G_2\end{pmatrix} \,:\, G_1\in\cG_1, M\in \cM, G_2\in\cG_2\right\}
  $$  
  is the generating set of an $\left(n_1+n_2,d;k_1+k_2\right)_q$--{\cdc} $\cW$ with cardinality $\#\cC_1\cdot\#\cM\cdot\#\cC_2$.
\end{nlemma}
\begin{proof}
  Let $W\in\cW$ be an arbitrary codeword with generator matrix $$H=\begin{pmatrix} G_1 & \varphi_{G_2}(M) \\ \mathbf{0} & G_2\end{pmatrix}.$$ 
  Since $\rk(H)=\rk(G_1)+rk(G_2)=k_1+k_2$ we have $\dim(W)=k_1+k_2$. Let $W'\in\cW$ be another codeword with $W'\neq W$ with generator matrix 
  $H'=\begin{pmatrix} G_1' & \varphi_{G_2'}(M') \\ \mathbf{0} & G_2'\end{pmatrix}$. 
  Set  
  \begin{eqnarray*}
    R:=\rk\left(\begin{pmatrix}G_1&\varphi_{G_2}(M)\\\mathbf{0}& G_2\\ G_1' & \varphi_{G_2'}(M') \\ \mathbf{0} & G_2'\end{pmatrix} \right)=
    \rk\left(\begin{pmatrix}G_1&\varphi_{G_2}(M)\\ G_1'-G_1&\varphi_{G_2'}(M')-\varphi_{G_2}(M)\\\mathbf{0}&G_2\\\mathbf{0}&G_2'-G_2\end{pmatrix}\right)
  \end{eqnarray*}
  and note that
  \begin{eqnarray*}
    \rk\left(\begin{pmatrix}G_1\\G_1'-G_1\end{pmatrix}\right) = \frac{\ds(\left\langle G_1\right\rangle,\left\langle G_1'\right\rangle)}{2}+k_1\ge \frac{d}{2}+k_1\\ 
    \rk\left(\begin{pmatrix}G_2\\G_2'-G_2\end{pmatrix}\right) = \frac{\ds(\left\langle G_2\right\rangle,\left\langle G_2'\right\rangle)}{2}+k_2\ge \frac{d}{2}+k_2.
  \end{eqnarray*}
  Since $\ds(W,W')=2\cdot\left(R-k_1-k_2\right)$ it suffices to show $R\ge k_1+k_2+\tfrac{d}{2}$ in order to deduce $\ds(W,W')$.  
  
  If $G_1\neq G_1'$ we have
  $$
    R\ge \rk\left(\begin{pmatrix} G_1 & \star \\ G_1'-G_1 & \star\\ \mathbf{0} & G_2\end{pmatrix}\right)=\rk\left(\begin{pmatrix}G_1\\G_1'-G_1\end{pmatrix}\right)+\rk(G_2)
    \ge d/2+k_1+k_2.
  $$
    
  If $G_1=G_1'$ and $G_2\neq G_2'$ we have
  $$
    R\ge \rk\left(\begin{pmatrix} G_1 & \star \\ \mathbf{0} & G_2\\ \mathbf{0} & G_2'-G_2\end{pmatrix}\right) =\rk(G_1)+\rk\left(\begin{pmatrix}G_2\\G_2'-G_2\end{pmatrix}\right) 
    \ge d/2+k_1+k_2.
  $$
  
  If $G_1=G_1'$ and $G_2=G_2'$ then we have $M\neq M'$ so that $\rk(M-M')=\dr(M,M')\ge d/2$ and
  \begin{eqnarray*}
    R&\ge& \rk\left(\begin{pmatrix}G_1&\star\\ \mathbf{0}&\varphi_{G_2}(M')-\varphi_{G_2}(M)\\\mathbf{0}&G_2\end{pmatrix}\right) 
      =\rk(G_1)+\rk\left(\begin{pmatrix}\varphi_{G_2}(M')-\varphi_{G_2}(M)\\ G_2\end{pmatrix}\right)\\ 
    &=& k_1+k_2+\rk(M-M')\ge k_1+k_2+d/2.
  \end{eqnarray*}

  Thus we have $\ds(\cW)\ge d$ and the stated cardinality follows from the distance analysis.  
\end{proof}
The corresponding matrix description is denoted by 
\begin{center}
  \begin{tabular}{|C{2cm}|C{2cm}|} 
    \hline
    \texttt{C} & \texttt{R}$\uparrow$\\ 
    \hline
    \texttt{0} & \texttt{C}\\
    \hline
  \end{tabular}
\end{center}
where \texttt{R}$\uparrow$ indicates a {\rmc} whose length is increased by addition additional zero columns according to a {\cdc} sharing the same positions of the final code.

While the conditions on the components $\cC_1$, $\cC_2$, and $\cM$ in the product construction in Lemma~\ref{lemma_coset_part} are rather demanding, one advantage is that 
the three code sizes are multiplied. The other is that we can combine several such subcodes to a larger {\cdc}:
\begin{ntheorem}\textbf{(Coset construction -- \cite[Lemma 3, Lemma 4]{heinlein2017coset})}\\ 
  \label{theorem_coset}
  Let $\cC_1$ be an $\left(n_1,d_1;k_1\right)_q$--{\cdc}, $\cC_2$ be an $\left(n_2,d_2;k_2\right)_q$--{\cdc}, and $\cM$ be a $(k_1\times (n_2-k_2),d/2)_q$--{\rmc}, 
  where $d=d_1+d_2$. For a positive integer $s$ let $\cC_1^1,\dots,\cC_1^s$ be a $d$-packing of $\cC_1$ and $\cC_2^1,\dots,\cC_2^s$ be a $d$-packing of $\cC_2$. 
  For $j\in\{1,2\}$ and $1\le i\le s$ let $\cG_j^i$ be a generating set of $\cC_j^i$ and 
  $\cG_j=\cup_{i=1}^s \cG_j^i$, where $j\in\{1,2\}$. For each $G\in\cG_2$ let $\varphi_{G}$ be an embedding function $\F_q^{k_1\times(n_2-k_2)}\to\F_q^{k_1\times n_2}$  
  compatible with $G$. With this let
  $$
    \left\{\begin{pmatrix} G_1 & \varphi_{G_2}(M)\\ \mathbf{0}_{k_2\times n_1} & G_2  \end{pmatrix}\,:\,G_1\in\cG_1^i, M\in\cM, G_2\in\cG_2^i\right\}
  $$
  be a generating set of a subcode $\cW^i$ for $1\le i\le s$. Then, $\cW=\cup_{i=1}^s \cW^i$ is an $(n_1+n_2,d_1+d_2;k_1+k_2)_q$--{\cdc} with cardinality
  \begin{equation}
    \label{eq_card_coset}
    \#\cW=\sum_{i=1}^s\#\cW^i =\#\cM\cdot \sum_{i=1}^s \#\cC_1^i \cdot\#\cC_2^i.
  \end{equation}
\end{ntheorem}
\begin{proof}
  The subcodes $\cW^i$ are $\left(n_1+n_2,d_1+d_2;k_1+k_2\right)_q$--{\cdc}s for all $1\le i\le s$ by Lemma~\ref{lemma_coset_part}, which also yields the stated cardinality of 
  $\cW$. For arbitrary $G_1,G_1'\in\cG_1$, $G_2,G_2'\in\cG_2$, and $M,M'\in\cM$ let 
  $$
    H=\begin{pmatrix} G_1 & \varphi_{G_2}(M)\\ \mathbf{0} & G_2  \end{pmatrix}
    \quad\text{and}\quad
    H'=\begin{pmatrix} G_1' & \varphi_{G_2'}(M)\\ \mathbf{0} & G_2'  \end{pmatrix}
  $$
  i.e., $W=\left\langle H\right\rangle$, $W'=\left\langle H'\right\rangle$ are arbitrary codewords in $\cW$. 
  
  If $G_1=G_1'$ or $G_2=G_2'$ then there exists an index $1\le i\le s$ 
  so that $W,W'\in\cW^i$ and either $W=W'$ or $\ds(W,W')\ge\ds(\cW^i)\ge d_1+d_2$. 
  
  If $G_1\neq G_1'$ and $G_2\neq G_2'$, then we set $U_1=\langle G_1\rangle$, $U_1'=\langle G_1'\rangle$, $U_2=\langle G_2\rangle$, $U_2'=\langle G_2'\rangle$, so that 
  $$
    \rk\left(\begin{pmatrix} G_1  \\G_1'-G_1\end{pmatrix}\right)=\frac{\ds(U_1,U_1')}{2}+k_1\ge \frac{\ds(\cC_1)}{2}+k_1\ge \frac{d_1}{2}+k_1
  $$  
  and
  $$
    \rk\left(\begin{pmatrix} G_2  \\G_2'-G_2\end{pmatrix}\right)=\frac{\ds(U_2,U_2')}{2}+k_2\ge \frac{\ds(\cC_2)}{2}+k_2\ge \frac{d_2}{2}+k_2.
  $$
  Since 
  \begin{eqnarray*}
    R&:=&\rk\left(\begin{pmatrix} G_1 & \varphi_{G_2}(M)\\ \mathbf{0} & G_2 \\G_1' & \varphi_{G_2'}(M)\\ \mathbf{0} & G_2' \end{pmatrix}\right)=
    \rk\left(\begin{pmatrix} G_1 & \star \\G_1'-G_1 & \star\\ \mathbf{0} & G_2 \\ \mathbf{0} & G_2'-G_2 \end{pmatrix}\right)\\ 
    &=& \rk\left(\begin{pmatrix} G_1  \\G_1'-G_1\end{pmatrix}\right)+ \rk\left(\begin{pmatrix} G_2  \\G_2'-G_2\end{pmatrix}\right)
    \ge \frac{d_1+d_2}{2}+k_1+k_2
  \end{eqnarray*} 
  we have $\ds(W,W')=2\cdot\left(R-k_1-k_2\right)\ge d_1+d_2$.
\end{proof}  
The corresponding matrix description is denoted by 
\begin{center}
  \begin{tabular}{|C{2cm}|C{2cm}|} 
    \hline
    \texttt{C}$^i$ & \texttt{R}$\uparrow$\\ 
    \hline
    \texttt{0} & \texttt{C}$^i$\\
    \hline
  \end{tabular}
\end{center}
where \texttt{C}$^i$ indicates that the have a sequence of {\cdc}s and using the same superscript $i$ indicates how the components have to be arranged.

We remark that we may also use different {\rmc}s $\cM^i$ for the construction of the subcodes $\cW^i$ instead a single {\rmc} $\cM$ for all. However, 
since there is no obvious benefit of such a generalization we prefer the simplicity of the stated formulation and Equation~(\ref{eq_card_coset}) for 
the cardinality of the resulting code.
\begin{ndefinition}
  \label{definition_max_coset_sum}
  By $C_q(n_1,n_2,d;k_1,k_2)$ we denote that maximum possible cardinality of a {\cdc} $\cW$ obtained via the coset construction in Theorem~\ref{theorem_coset} 
  with {\rmc} $\cM=\left\{\mathbf{0}_{k_1\times(n_2-k_2)}\right\}$, where $d_1,d_2$ are arbitrary besides satisfying $d_1+d_2=d$.  
\end{ndefinition} 
In other words, $C_q(n_1,n_2,d;k_1,k_2)$ is a shorthand for the maximum possible value of $\sum_{i=1}^s \#\cC_1^i \cdot\#\cC_2^i$ in Equation~(\ref{eq_card_coset}).
\begin{nexercise}
  \label{exercise_C_q}
  Show $C_q(n_1,n_2,d;k_1,k_2)=C_q(n_2,n_1,d;k_2,k_1)$ and $C_q(n_1,n_2,$ $d;$ $k_1,k_2)=C_q(n_1,n_2,d;k_1,n_2-k_2)$.
\end{nexercise}
Since the optimal choice for the {\rmc} $\cM$ in the coset construction for a {\cdc} $\cW$ is 
an {\mrd} code, $C_q(n_1,n_2,d;k_1,k_2)$ is indeed the essential quantity to express the maximum possible cardinality $\#\cW$:  
\begin{nlemma}
  Let $\cW$ be a {\cdc} constructed via the coset construction in Theorem~\ref{theorem_coset} with parameters $\left(n_1,n_2,d;k_1,k_2\right)$ 
  of maximum possible cardinality. Then, we have 
  \begin{eqnarray}
    \#\cW &=& A_q^R(k_1\times (n_2-k_2),d/2)\cdot C_q(n_1,n_2,d;k_1,k_2) \notag\\ 
    &=& \left\lceil q^{\max\{k_1,n_2-k_2\}\cdot\left(\min\{k_1,n_2-k_2\}-d+1\right)}\right\rceil \cdot C_q(n_1,n_2,d;k_1,k_2).   
  \end{eqnarray}
\end{nlemma}
When estimating lower bounds for constant dimension codes we may also replace the term $C_q(n_1,n_2,d;k_1,k_2)$ by some lower bound. The matrix description 
underlying Definition~\ref{definition_max_coset_sum} can be written as 
\begin{center}
  \begin{tabular}{|C{2cm}|C{2cm}|} 
    \hline
    \texttt{C}$^i$ & \texttt{0}\\ 
    \hline
    \texttt{0} & \texttt{C}$^i$\\
    \hline
  \end{tabular}
\end{center}
We remark that \cite[Lemma 4.4]{cossidente2021combining} for $l=2$ can be seen as a special case of this construction.

Before we state an example for the coset construction we introduce another notion from geometry.
\begin{ndefinition} (Parallelisms)\\
  A \emph{parallelism} in $\cG_q(n,k)$ is a $2k$-partition of the $(n,2;k)_q$--{\cdc} $\cG_q(n,k)$. A $2k$-packing of $\cG_q(n,k)$ is called  
  \emph{partial parallelism} in $\cG_q(n,k)$.  
\end{ndefinition}
In other words, a parallelism is a partition of the $k$-spaces in $\F_q^n$ into $k$-spreads. 
The size of a spread in $\cG_q(n,k)$ (or a $k$-spread in $\F_q^n$) is given by $A_q(n,2k;k)=\qbin{n}{1}{q}/\qbin{k}{1}{q}=\frac{q^n-1}{q^k-1}$.

\begin{nproposition}
  \label{proposition_parallelism}
  Parallelisms in $\cG_q(n,k)$ are known to exist for:
  \begin{enumerate}
    \item[(a)] $k=2$, $q=2$, and $n$ even \cite{baker1976partitioning,baker1983preparata};
    \item[(b)] $k=2$, all $q$ and $n=2^m$ for $m\ge 2$ \cite{beutelspacher1974parallelisms};
    \item[(c)] $k=2$, $q=3$, and $n=6$ \cite{etzion2012automorphisms}; 
    \item[(d)] $k=3$, $q=2$, and $n=6$ \cite{hishida2000cyclic,sarmiento2002point}.
  \end{enumerate}
\end{nproposition}
See e.g.\ \cite[Section 4.9]{etzion2016galois} for more details. For lower bounds for partial parallelisms we refer to \cite{beutelspacher1990partial,etzion2015partial,zhang2020new}.

\begin{nexample}
  \label{example_coset_8_4_4}
  Consider the coset construction for parameters $\big(n_1,n_2,d_1,d_2,k_1,$ $k_2\big)=(4,4,2,2,$ $2,2)$. To this end, let $\cC_1=\cC_2=\cG_q(4,2)$ and $\cM$ be a $(2\times 2,2)_q$--{\mrd} 
  code. For $s=\qbin{4}{2}{q}/A_q(4,4;2)=q^2+q+1$ 
  let $\left\{ \cC_1^1,\dots,\cC_1^s\right\}$ and $\left\{ \cC_2^1,\dots,\cC_2^s\right\}$ be parallelisms in $\cG_q(4,2)$. With this we can apply the 
  coset construction in Theorem~\ref{theorem_coset} to construct an $(8,4;4)_q$--{\cdc} $\cW_2$. Since $\#\cM=q^2$ and $\#\cC_j^i=q^2+1$ for all $j\in\{1,2\}$ and  
  all $1\le i\le s$ we have 
  $$
    \#\cW_2=q^2\cdot \left(q^2+q+1\right)\cdot \left(q^2+1\right)^2=q^8 + q^7 + 3q^6 + 2q^5 + 3q^4 + q^3 + q^2\!\!.
  $$ 
\end{nexample}  
For the chosen parameters $n_i$, $k_i$, and $d_i$ the other choices are indeed optimal for the coset construction. I.e., starting from Equation~(\ref{eq_card_coset}) 
we note $\#\cM\le A_q^R(k_1\times(n_2-k_2),(d_1+d_2)/2)$ and: 
\begin{nlemma}(\cite[Corollary 1]{heinlein2017coset})
  \label{lemma_upper_bound_coset}
  %% Using the notation from Theorem~\ref{theorem_coset}, we have
  %% $$
  %%   \sum_{i=1}^s \#\cC_1^i\cdot \#\cC_2^i \le \min\left\{ \qbin{n_1}{k_1}{q}\cdot A_q(n_2,d_1+d_2;k_2), \qbin{n_2}{k_2}{q}\cdot A_q(n_1,d_1+d_2;k_1)\right\}.
  %% $$
  $$
    C_q(n_1,n_2,d;k_1,k_2)\le \min\left\{ \qbin{n_1}{k_1}{q}\cdot A_q(n_2,d;k_2), \qbin{n_2}{k_2}{q}\cdot A_q(n_1,d;k_1)\right\}
  $$
\end{nlemma}
 
Via orthogonality the existence question for a $4$-partition of $\cG_q(6,4)$ translates to the existence question for a parallelism in $\cG_q(6,2)$, which is 
known for $q\in\{2,3\}$, see Proposition~\ref{proposition_parallelism}.   
\begin{nexample}
  \label{coset_6_6_4_2}
  Consider the coset construction for parameters $\big(n_1,n_2,d_1,d_2,k_1,$ $k_2\big)=(6,6,2,2,$ $4,2)$ and assume $q\in\{2,3\}$. To this end, let $\cC_1=\cG_q(6,4)$, $\cC_2=\cG_q(6,2)$, 
  and $\cM$ be a $(4\times 4,2)_q$--{\mrd} code. For $s=\qbin{6}{2}{q}/A_q(6,4;2)=\qbin{5}{1}{q}$  let $\left\{ \cC_2^1,\dots,\cC_2^s\right\}$ be a parallelism in 
  $\cG_q(6,2)$ and set $\cC_1^i=\left(\cC_2^i\right)^\perp$ for $1\le i\le s$. Since $A_q(6,4;2)=q^4+q^2+1$ we have 
  $$
    C_q(6,6,4;4,2)\ge \sum_{i=1}^s \#\cC_1^i\cdot\#\cC_2^i =\qbin{6}{2}{q}\cdot\left(q^4+q^2+1\right),
  $$
  i.e., the upper bound from Lemma~\ref{lemma_upper_bound_coset} is attained with equality. Since $\#\cM=q^{12}$, the {\cdc} $\cW$ resulting from the corresponding coset 
  construction has cardinality $55\,996\,416$ if $q=2$ and $532\,504\,413\,441$ if $q=3$. 
\end{nexample}
 
As conjectured in \cite{etzion2012codes}, Example~\ref{example_coset_8_4_4} is just an instance of a more general result:
\begin{nproposition}(\cite[Theorem 9]{heinlein2017coset})
  If parallelisms in $\cG_q(n_1,k_1)$, $\cG_q(n_2,k_2)$ exist and $d_1=d_2=2$, then we have
  $$
    C_q(n_1,n_2,4;k_1,k_2)= \min\left\{ \qbin{n_1}{k_1}{q}\cdot A_q(n_2,d;k_2), \qbin{n_2}{k_2}{q}\cdot A_q(n_1,d;k_1)\right\}\!.
  $$
  %% i.e., the upper bound in Lemma~\ref{lemma_upper_bound_coset} can be attained. 
\end{nproposition} 

\begin{nexample}
  \label{example_coset_10_6_4}
  Consider a {\cdc} $\cW$ obtained by the coset construction in Theorem~\ref{theorem_coset} with parameters $\left(n_1,n_2,d_1,d_2,k_1,k_2\right)=(4,6,2,4,1,3)$. 
  For the components we do not have too many choices. Since $\cC_1\subseteq \cG_q(4,1)$ we have $s\le \qbin{4}{1}{q}=q^3+q^2+q+1$. The fact that $2k_1<d_1+d_2$ 
  implies $\#\cC_1^i=1$ for all $1\le i\le s$. Similarly, the $(1\times 1,3)_q$--{\rmc} $\cM$ has to be of cardinality $1$. The ambient code $\cC_2$ has to be a 
  $(6,3;4)_q$--{\cdc} and the $\cC_2^i$ have to be $(6,3;6)_q$--{\cdc}s, i.e.\ partial spreads, for all $1\le i\le s$. From Equation~(\ref{eq_card_coset}) we conclude 
  $$
    \#\cW=\#\cM\cdot\sum_{i=1}^s \#\cC_1^i\cdot \#\cC_2^i=\sum_{i=1}^s \cC_2^i\le \#\cC_2\le A_q(6,4;3).
  $$ 
  For $q=2$ we have $s\le 15$ and $A_2(6,4;3)=77$. In \cite{heinlein2017coset} a $6$-partition with cardinality $15$ of a $(6,4;3)_2$--{\cdc} of cardinality $76$ 
  was obtained via ILP computations and its optimality was shown, i.e., $C_2(4,6,6;1,3)=76$. Here indeed the maximum cardinality of $\qbin{4}{1}{2}=15$ is indeed 
  a limiting factor.
\end{nexample}

The packing problem of a given ambient {\cdc} into {\cdc}s of larger minimum subspace distance is a hard but interesting algorithmical problem. For ambient {\cdc}s with 
a specific structure we give preliminary parametric constructions in a moment. First we consider the compatibility with other subcode constructions and the 
extenability problem. 

Directly from the construction we conclude:
\begin{nlemma}
  \label{lemma_pivot_structure_coset}
  The pivot structure of a {\cdc} $\cW$ obtained via the coset construction in Theorem~\ref{theorem_coset} is a subset of 
  $\left({n_1 \choose k_1},{n_2\choose k_2}\right)$.
\end{nlemma}
So we can directly apply the generalized Echelon--Ferrers construction:

\begin{nexample}(Sequel of Example~\ref{example_coset_8_4_4})\\
  \label{example_sequel_coset_8_4_4}
  Let $\cW_2$ as in Example~\ref{example_coset_8_4_4}, so that its pivot structure is contained in  ${4 \choose 2},{4\choose 2}$. Let $\cW_1$ be the $(8,4;4)_q$--{\lmrd} 
  code of cardinality $q^{12}$ and $\cW_3=\left\{\left\langle\begin{pmatrix}\mathbf{0}_{4\times 4} & I_4\end{pmatrix}\right\rangle\right\}$ 
  be an $(8,4;4)_q$--{\cdc} of cardinality $1$. The pivot structures of these two codes are given by the unique vectors $11110000$ and $00001111$. Due to 
  $\dH({4 \choose 2},{4\choose 2},\{11110000,00001111\})$ $=4$ and $\dH(11110000,00001111)\ge 4$ we have $$\ds(\cW_1,\cW_2), \ds(\cW_1,\cW_3), \ds(\cW_2,\cW_3)\ge 4,$$ 
  so that $\cW=\cW_1\cup\cW_2\cup\cW_3$ is an 
  $(8,4;4)_q$--{\cdc} of cardinality $q^{12}+\left(q^2+q+1\right)\cdot \left(q^2+1\right)^2+1$.
\end{nexample}
We remark that corresponding lower bound
\begin{equation}
  A_q(8,4;4)\ge q^{12}+\left(q^2+q+1\right)\cdot \left(q^2+1\right)^2+1
\end{equation}
is still unsurpassed for all $q\ge 3$. For $q=2$ the corresponding code size of $4797$ was surpassed by {\cdc}s of sizes $4801$ and $4802$, see 
\cite{braun2018new} and \cite{zhou2021construction}, respectively.
\begin{nexercise}
  Show that $\left\langle\begin{pmatrix}\mathbf{0}_{4\times 4} & I_4\end{pmatrix}\right\rangle\in\cG_q(8,4)$ is the unique codeword that can be added to the 
  $(8,4;4)_q$--{\cdc} $\cW_1+\cW_2$ in Example~\ref{example_sequel_coset_8_4_4} without violating the minimum subspace distance. 
\end{nexercise}

\medskip

From Lemma~\ref{lemma_pivot_structure_coset} and Lemma~\ref{lemma_construction_d_additional_codewords}.(b)) we conclude:
\begin{nlemma}\textbf{(Construction D $+$ coset construction)}\\
  \label{lemma_combine_construction_d_and_coset}
  Let $\cW_1$ be a {\cdc} constructed via construction D in Theorem~\ref{theorem_construction_d} with parameters $(n_1,n_2,d,k)$ and $\cW_2$ be a {\cdc} constructed 
  via the coset construction in Theorem~\ref{theorem_coset} with parameters $\left(n_1,n_2,d_1,d_2,k_1,k_2\right)$, where $k_1+k_2=k$ and $d_1+d_2=d$. If $k_2\ge d/2$, then $\cW=\cW_1+\cW_2$ is an
  $(n_1+n_2,d;k)_q$--{\cdc} with cardinality $\#\cW_1+\#\cW_2$.
\end{nlemma}
The corresponding matrix description is given by: 
\begin{center}
  \begin{tabular}{|C{2cm}|C{2cm}|} 
    \hline
    \texttt{C} & \texttt{R}\\ 
    \hline
    \hline
    \texttt{C}$^i$ & \texttt{R}$\uparrow$\\ 
    \hline
    \texttt{0} & \texttt{C}$^i$\\
    \hline
  \end{tabular}
\end{center}
\begin{nexample}(Sequel of Example~\ref{example_coset_10_6_4})\\
  Let $\cW_1$ be constructed via construction~D in Theorem~\ref{theorem_construction_d} with parameters $\left(n_1,n_2,d,k\right)=(4,6,6,4)$ and $\cW_2$ be 
  constructed via the coset construction in Theorem~\ref{theorem_coset} with parameters $\left(n_1,n_2,d_1,d_2,k_1,k_2\right)=(4,6,2,4,1,3)$. 
  Since the {\lq\lq}linkage condition{\rq\rq} $k_2\ge d/2$ in Lemma~\ref{lemma_combine_construction_d_and_coset} is satisfied, $\cW_1\cup\cW_2$ is 
  a $(10,6;4)_q$--{\cdc} of cardinality $\#\cW_1+\#\cW_2$, so that
  $$
    A_q(10,6;4)\ge A_q(4,6;4)\cdot A_q^R(4\times 6,3) +C_q(4,6,6,1,3)=q^{12}+C_q(4,6,6,1,3).
  $$   
  For $q=2$, $C_2(4,6,6,1,3)=76$ was mentioned Example~\ref{example_coset_10_6_4}, so that $\#\cW_1+\cW_2=4172$ can be attained. In \cite{heinlein2017coset} it 
  was observed by an exhaustive computer search that an additional codeword can be added to $\cW$, so that $A_2(10,6;4)\ge 4173$. This is still the best known lower bound.  
\end{nexample}
We remark that Construction 1 in Theorem~\ref{theorem_construction_1} yields the same lower bound.

Also different subcodes constructed via the coset construction can be combined to yield larger codes. Here the distance analysis in the Hamming 
metric combined with Lemma~\ref{lemma_pivot_structure_coset} gives:
\begin{nlemma}\textbf{(Coset construction $+$ coset construction -- cf.\ \cite[Lemma 6]{heinlein2017coset})}\\
  \label{lemma_combine_coset_and_coset}
  Let $\cW_1$ be a {\cdc} constructed via the coset construction in Theorem~\ref{theorem_coset} with parameters $\left(n_1,n_2,d_1,d_2,k_1,k_2\right)$ and
  $\cW_2$ be a {\cdc} constructed via the coset construction in Theorem~\ref{theorem_coset} with parameters $\left(n_1,n_2,d_1',d_2',k_1',k_2'\right)$.
  If $k:=k_1+k_2=k_1'+k_2'$,  $d:=d_1+d_2=d_1'+d_2'$ and $\left|k_1-k_1'\right|+\left|k_2-k_2'\right|\ge d$, then $\cW=\cW_1\cup\cW_2$ is an 
  $(n_1+n_2,d;k)_q$-{\cdc} with cardinality $\#\cW_1+\#\cW_2$.
\end{nlemma}
The corresponding matrix description is given by: 
\begin{center}
  \begin{tabular}{|C{2cm}|C{2cm}|} 
    \hline
    \texttt{C}$^i$ & \texttt{R}$\uparrow$\\ 
    \hline
    \texttt{0} & \texttt{C}$^i$\\
    \hline
    \hline
    \texttt{C}$^i$ & \texttt{R}$\uparrow$\\ 
    \hline
    \texttt{0} & \texttt{C}$^i$\\
    \hline
  \end{tabular}
\end{center}

\begin{nexample}
  \label{example_a_12_4_6_construction_d_and_coset}
  Let $\cW_2$ and $\cW_3$ be constructed via the coset construction in Theorem~\ref{theorem_coset} with parameters $\left(n_1,n_2,d_1,d_2,k_1,k_2\right)=(6,6,2,2,4,2)$ 
  and $\left(n_1,n_2,d_1',d_2',k_1',k_2'\right)=(6,6,2,2,2,4)$, respectively. Note that the conditions of Lemma~\ref{lemma_combine_coset_and_coset} 
  for the combination of $\cW_2$ and $\cW_3$ are satisfied and $C_q(6,6,4;4,2)$ $=C_q(6,6,4;2,4)$. The maximum size of the {\rmc} for $(6,6,4,4,2)$ is 
  $A_q^R(4\times 4,2)=q^{12}$ and $A_q^R(2\times 2,2)=q^2$ for $(6,6,4,2,4)$. Since the conditions of Lemma~\ref{lemma_combine_construction_d_and_coset} 
  are satisfied for $k_2\in\{2,4\}$, we can choose $\cW_1$ as the $(6\times 6,4)_q$--{\lmrd} code of cardinality $q^{30}$, so that considering the {\cdc} 
  $\cW=\cW_1\cup\cW_2\cup\cW_3$ yields
  $$
    A_q(12,4;6)\ge q^{30}+C_q(6,6,4;4,2)\cdot\left(q^{12}+q^2\right).
  $$  
  For $q\in \{2,3\}$ we can use the exact value of $C_q(6,6,4;4,2)$ determined in Example~\ref{coset_6_6_4_2} to conclude
  $$
    A_2(12,4;6)\ge 1\,129\,792\,924 \quad\text{and}\quad A_3(12,4;6)\ge 206\,423\,645\,526\,099.
  $$   
\end{nexample}

\begin{trailer}{Mirrored coset construction}
  Of course one can easily adjust the coset construction in Theorem~\ref{theorem_coset} so that its matrix description is given by
  \begin{center}
    \begin{tabular}{|C{2cm}|C{2cm}|} 
      \hline
      \texttt{C}$^i$ & \texttt{0} \\ 
      \hline
      \texttt{R}$\uparrow$ & \texttt{C}$^i$\\
      \hline
    \end{tabular}
    \end{center}
  instead of
  \begin{center}
    \begin{tabular}{|C{2cm}|C{2cm}|} 
      \hline
      \texttt{C}$^i$ & \texttt{R}$\uparrow$\\ 
      \hline
      \texttt{0} & \texttt{C}$^i$\\
      \hline
    \end{tabular}
  \end{center}
  and call it mirrored coset construction. In Lemma~\ref{lemma_combine_construction_d_and_coset} we then have to replace the condition $k_2\ge d/2$ by $k_2-\rk(M)\ge d/2$ 
  for all $M\in \cM$ if we use a subcode obtained by the mirrored coset construction and $\cM$ is its utilized {\rmc}.
  
  In Example~\ref{example_a_12_4_6_construction_d_and_coset} the advantage of choosing the mirrored coset construction for $\cW_3$ with parameters 
  $\left(n_1,n_2,d_1,d_2,k_1,k_2\right)=(6,6,2,2,2,4)$ is that we can choose a {\rmc} of size $A_q^R(4\times 4,2;\le 2)>A_q^R(2\times 2,2)$. However, 
  in a modified version of Lemma~\ref{lemma_combine_coset_and_coset} considering the combination of a subcode from the coset construction with a 
  subcode from the mirrored coset construction we have to replace the condition $\left|k_1-k_1'\right|+\left|k_2-k_2'\right|\ge d$. The following example shows 
  that the ranks of the elements in the involved {\rmc}s have to be taken into account. 
  The generator matrix  
  $$
   H=\begin{pmatrix}
      100000 & 000000 \\
      010000 & 000000 \\
      001000 & 001000 \\
      000100 & 000100 \\[1mm]
      000000 & 100000 \\
      000000 & 010000 \\
    \end{pmatrix}
    =\begin{pmatrix} G_1 & M\!\uparrow_{G_2} \\ \mathbf{0}_{2\times 6} & G_2\end{pmatrix} 
  $$ 
  with $G_1\in\F_q^{4\times 6}$, $\rk(G_1)=4$, $G_2\in\F_q^{2\times 6}$, $\rk(G_2)=2$, $M\in\F_q^{4\times 4}$, and $\rk(M)\le 2$ fits into the shape 
  of the coset construction with parameters $\left(n_1,n_2,d_1,d_2,k_1,k_2\right)=(6,6,2,2,4,2)$. Similarly, the generator matrix
  $$
    H'=\begin{pmatrix}
      100000 & 000000 \\
      010000 & 000000 \\[1mm]
      000000 & 100000 \\
      000000 & 010000 \\
      001000 & 001000 \\ 
      000100 & 000100 \\
    \end{pmatrix}  
    =\begin{pmatrix} G_1' & \mathbf{0}_{2\times 6} \\ M'\!\uparrow_{G_1'} & G_2'\end{pmatrix}
  $$
  with $G_1'\in\F_q^{2\times 6}$, $\rk(G_1')=2$, $G_2'\in\F_q^{4\times 6}$, $\rk(G_2')=4$, $M'\in\F_q^{4\times 4}$, and $\rk(M')\le 2$ fits into the shape 
  of the mirrored coset construction with parameters $\left(n_1,n_2,d_1,d_2,k_1,k_2\right)=(6,6,2,2,2,4)$. However, as $H'$ arises from $H$ by swapping 
  row three with row five and row four with row six, we have $\langle H\rangle=\langle H'\rangle$, i.e., $\ds(\langle H\rangle,\langle H'\rangle)=0$.
  
  While it is possible to suitably modify the condition in Lemma~\ref{lemma_combine_coset_and_coset} we are not aware of a construction of a {\cdc} 
  leading to the best known lower bound that involves both a subcode obtained from the coset construction and a subcode obtained from the mirrored 
  coset construction. So, we refrain from going into more details.   
\end{trailer}

\medskip

If we want to combine the generalized linkage construction with the coset construction, then we eventually have the restrict the maximum occurring ranks in the 
{\rmc} of the coset part, as it is the case if we combine construction~D with the mirrored coset construction. 
\begin{nlemma}\textbf{(Generalized linkage construction $+$ coset construction)}\\
  \label{lemma_combine_generalized_linkage_and coset}
  Let $\cW_1$ be a {\cdc} constructed via the generalized linkage construction in Theorem~\ref{theorem_generalized_linkage} with parameters $\left(n_1,n_2,d,k\right)$ 
  and $\cW_2$ be a {\cdc} constructed via the coset construction in Theorem~\ref{theorem_coset} with parameters $\left(n_1,n_2,d_1,d_2,k_1,k_2\right)$ and {\rmc} $\cM$. 
  If $k_1+k_2=k$, $d_1+d_2=d$, $k_2\ge d/2$ and $k_1-\rk(M)\ge d/2$ for all $M\in\cM$, then $\cW=\cW_1\cup\cW_2$ is an $\left(n_1+n_2,d;k\right)_q$--{\cdc} with cardinality
  $\cW_1+\cW_2$. 
\end{nlemma}
\begin{proof}
  Let $E_1$ and $E_2$ be as in Lemma~\ref{lemma_additional_codewords_generalized_linkage_construction} for $\cW_1$. For each codeword $U\in\cW_2$ we have 
  $\dim(U\cap E_2)\ge k_2\ge d/2$ and $\dim(U\cap E_1)\ge k_1-\rk(M)\ge d/2$, where $M\in\cM$ is the matrix used in the generator matrix of $U$. 
\end{proof}
The corresponding matrix description is given by: 
\begin{center}
  \begin{tabular}{|C{2cm}|C{2cm}|} 
    \hline
    \texttt{C} & \texttt{R}\\ 
    \hline
    \hline
    \texttt{R} & \texttt{C}\\ 
    \hline
    \hline
    \texttt{C}$^i$ & \texttt{R}$\uparrow$\\ 
    \hline
    \texttt{0} & \texttt{C}$^i$\\
    \hline
  \end{tabular}
\end{center}

\begin{nexample}
  \label{example_a_10_4_5_generalized_linkage_and_coset}
  Let $\cW_1$ arise from the generalized linkage construction with parameters $\left(n_1,n_2,d,\right.$ $\left.k\right)=(5,5,4,5)$, so that we can assume 
  $\#\cW_1=q^{20}+A_q^R(5\times 5,2;\le 3)$. Let $\cW_2$ arise from the coset construction with parameters 
  $\left(n_1,n_2,d_1,d_2,k_1,k_2\right)=(5,5,2,2,3,2)$, so that we can assume $\#\cW_2=A_q^R(3\times 3,2;\le 1)\cdot C_q(5,5,4,3,2)$. 
  Due to Lemma~\ref{lemma_combine_generalized_linkage_and coset} we can consider the {\cdc} $\cW_1\cup\cW_2$ to conclude
  $$
    A_q(10,4;5)\ge q^{20}+A_q^R(5\times 5,2;\le 3) + A_q^R(3\times 3,2;\le 1)\cdot C_q(5,5,4,3,2),
  $$
  which can be refined to
  $$
    A_q(10,4;5)\ge q^{20}+A_q^R(5\times 5,2;\le 3) + \qbin{3}{1}{q}\cdot C_q(5,5,4,3,2)
  $$
  using Proposition~\ref{proposition_exact_crc_value_1}. For a lower bound for $A_q^R(5\times 5,2;\le 3)$ we refer to Example~\ref{example_A_R_5_5_2_le_3} 
  and for a lower bound for $C_q(5,5,4,3,2)$ we refer to Proposition~\ref{prop_c_q_5_5_4_2_2_two_mrd} and Exercise~\ref{exercise_C_q} noting the 
  computer result $C_q(5,5,4,3,2)\ge 1313$ mentioned in Subsection~\ref{subsec_packing_constructions}. Plugging in these lower bounds gives
  \begin{eqnarray}
    A_q(10,4;5) &\ge& q^{20} + q^{16} + q^{15} + 2q^{14} + q^{13} - q^{11} - 2q^{10} - q^9 + 2q^8  \notag \\ 
    && + 5q^7 + 4q^6 + 7q^5 + 11q^4 + 15q^3 + 12q^2 + 6q + 2
  \end{eqnarray}
  and 
  \begin{equation}
    A_2(10,4;5)\ge 1048576+130696+7\cdot 1313=1\,188\,463.
  \end{equation}
  %%$$
  %%  A_2(10,4;5)\ge 1048576++7\cdot 1313=1\,187\,503
  %%$$ 
  %%$$
  %%  A_3(10,4;5)\ge 3486784401+67934241+13\cdot 24808=3\,555\,041\,146
  %%$$
\end{nexample}
\begin{warning}{Flawed bound in the literature}
  The construction for a lower bound for $A_q(10,4;5)$ from \cite{cossidente2021combining} was flawed. 
  Applying Lemma~\ref{lemma_combine_coset_and_coset} with $(k_1,k_2)=(3,2)$ and $(k_1',k_2')=(2,3)$ 
  is possible for minimum subspace distance $2$ only. However, the lower bound from Example~\ref{example_a_10_4_5_generalized_linkage_and_coset} 
  is better anyway.
\end{warning}

\begin{nexample}
  \label{example_a_12_4_6_generalized_linkage_and_coset}
  Consider the construction from Example~\ref{example_a_12_4_6_construction_d_and_coset} again, e.g.\ we choose the parameters 
  $\left(n_1,n_2,d,k\right)=(6,6,4,6)$. This time we let $\cW_1$ arise from the generalized linkage construction, so that we can assume 
  $\#\cW_1=q^{30}+A_q^R(6\times 6,2;\le 4)$. For the {\cdc}s $\cW_2$ and $\cW_3$, obtained from the coset construction, we have to adjust 
  the corresponding {\rmc} $\cM$ so that the condition $k_1-\rk(M)\ge d/2$ from Lemma~\ref{lemma_combine_generalized_linkage_and coset} 
  is satisfied for all $M\in \cM$. For $\cW_2$ with parameters $\left(n_1,n_2,d_1,d_2,k_1,k_2\right)=(6,6,2,2,4,2)$ we can choose 
  $\cM$ as a $(4\times 4,2;\le 2)_q$--{\rmc}. For $\cW_3$ with parameters $\left(n_1,n_2,d_1,d_2,k_1,k_2\right)=(6,6,2,2,2,4)$ 
  we have to use a $(2\times 2,2;\le 0)_q$--{\rmc}, i.e., we can just use the one-element {\rmc} consisting of $\mathbf{0}_{2\times 2}$. Considering 
  the $(12,4;6)_q$--{\cdc} $\cW_1\cup\cW_2\cup\cW_3$ yields
  $$
    A_q(12,4;6)\ge q^{30}+A_q^R(6\times 6,2;\le 4)+C_q(6,6,4;4,2)\cdot\left(A_q^R(4\times 4,2;\le 2)+1\right).
  $$
  Using Lemma~\ref{lemma_lb_restricted_rank_additive_mrd} and Example~\ref{coset_6_6_4_2} we conclude
  $$
    A_2(12,4;6)\ge 1\,212\,418\,496\,+\, 7\,204\,617 = 1\,219\,623\,113
  $$
  and
  $$
    A_3(12,4;6) \ge 209\,943\,770\,460\,426\,+\, 10\,422\,814\,402 = 209\,954\,193\,274\,828.
  $$
\end{nexample}
We remark that the stated construction constitutes the best known lower bound for $(12,4;6)_q$--{\cdc}s where $q\in\{2,3\}$. For $q>3$ the existence a parallelism 
in $\cG_q(6,2)$ is unknown, so that we cannot apply the construction in Example~\ref{example_a_12_4_6_construction_d_and_coset} for $C_q(6,6,4;4,2)$ directly. 
In the subsequent Subsection~\ref{subsec_packing_constructions} we study general constructions for $d$-packings of {\cdc}s and take up the construction in 
Example~\ref{example_a_12_4_6_generalized_linkage_and_coset} again. 
\begin{nexercise}
  \label{exercise_a_12_4_6_generalized_linkage_and_coset}
  Compute a parametric lower bound for $A_q(12,4;6)$, where $q\ge 4$, based on the construction in Example~\ref{example_a_12_4_6_generalized_linkage_and_coset} and 
  the parametric lower bound for $C_q(6,6,4;4,2)$ determined in Subsection~\ref{subsec_packing_constructions}.   
\end{nexercise}

\begin{trailer}{What are sufficient conditions for a symmetric version of the coset construction?}Given the nice symmetry of the matrix description of the generalized linkage 
construction, the question arises if a generalized version of the coset construction with matrix description
%%\begin{center}
  \begin{tabular}{|C{1cm}|C{1cm}|} 
    \hline
    \texttt{C}$^i$ & \texttt{R}$\uparrow$\\ 
    \hline
    \texttt{R}$\uparrow$ & \texttt{C}$^i$\\
    \hline
  \end{tabular}
%%\end{center}
exists? 

  The following example for subspace distance $d=4$ shows that we need further, possibly quite restrictive, conditions at the very least.
  The generator matrix
  $$
    H=\begin{pmatrix}
      1000 & 0001 \\ 
      0100 & 0000 \\[1mm]
      0010 & 0100 \\
      0000 & 0010 \\
    \end{pmatrix}
    =\begin{pmatrix} G_1 & M_1\!\!\uparrow_{G_2} \\ M_2\!\!\uparrow_{G_1} & G_2\end{pmatrix} 
  $$ 
  with $G_1\in\F_q^{2\times 4}$, $\rk(G_1)=2$, $G_2\in\F_q^{2\times 4}$, $\rk(G_2)=2$, $M_1\in\F_q^{2\times 2}$, $\rk(M_1)\le 1$, $M_2\in\F_q^{2\times 2}$, and $\rk(M_2)\le 1$ 
  as well as the generator matrix
  $$
    H'=\begin{pmatrix}
      0100 & 0000 \\ 
      0010 & 1000 \\[1mm]
      0000 & 0100 \\
      0001 & 0010 \\
    \end{pmatrix}
    =\begin{pmatrix} G_1' & M_1'\!\!\uparrow_{G_2'} \\ M_2'\!\!\uparrow_{G_1'} & G_2'\end{pmatrix}
  $$   
  with $G_1'\in\F_q^{2\times 4}$, $\rk(G_1')=2$, $G_2'\in\F_q^{2\times 4}$, $\rk(G_2')=2$, $M_1'\in\F_q^{2\times 2}$, $\rk(M_1')\le 1$, $M_2'\in\F_q^{2\times 2}$, and $\rk(M_2')\le 1$ 
  fit into the shape of the desired matrix description. Setting $U_1=\langle G_1\rangle$, $U_2=\langle G_2\rangle$, $U_1'=\langle G_1'\rangle$, $U_2'=\langle G_2'\rangle$ we 
  observe $\ds(U_1,U_1')=2$ and $\ds(U_2,U_2')=2$, so that $\ds(U_1,U_1')+\ds(U_2,U_2')=4\ge d$. For 
  $$
    M_1=\begin{pmatrix}01\\00\end{pmatrix},\quad M_1'=\begin{pmatrix}00\\10\end{pmatrix},\quad 
    M_2=\begin{pmatrix}10\\00\end{pmatrix},\text{ and } M_2'=\begin{pmatrix}00\\01\end{pmatrix}
  $$   
  we have $\dr(M_1,M_1')=2\ge d/2$ and $\dr(M_2,M_2')=2\ge d/2$ (using the natural choice for $\uparrow$). However, both $W:=\langle H\rangle$ and $W':=\langle H'\rangle$ 
  contain the $3$-space generated by
  $$
    \begin{pmatrix}
      0100 & 0000\\ 
      0010 & 1000\\
      0000 & 0100\\
    \end{pmatrix}
  $$    
  as a subspace, so that $\ds(W,W')\le 2<d$. Restricting the ranks of $M_1$, $M_1'$ to be smaller than $1$ or the ranks of $M_2$, $M_2'$ to be smaller than $1$, we end 
  up with the original coset or the mirrored coset construction, respectively.
  
  We leave it as an open research problem to generalize the coset construction and refer to Theorem~\ref{theorem_block_inserting} for a possible first step into that direction.  
\end{trailer}

\section{Constructions for $d$-packings of {\cdc}s and {\rmc}s}
\label{subsec_packing_constructions}
As already mentioned, we can separate the problem of the choice of the {\rmc} in the coset construction and 
the problem of a coset construction 
with matrix description
\begin{center}
  \begin{tabular}{|C{2cm}|C{2cm}|} 
    \hline
    \texttt{C}$^i$ & \texttt{0}\\ 
    \hline
    \texttt{0} & \texttt{C}$^i$\\
    \hline
  \end{tabular}
\end{center}
where the parts \texttt{C}$^i$ correspond to $d$-packings of {\cdc}s. If parallelisms are not available or the desired minimum subspace distance is larger than $4$ 
then we need different techniques for the construction of the needed $d$-packings.

Without the relation to the coset construction the following result was obtain in \cite{cossidente2021combining} in the context of the extension problem for the generalized 
linkage construction.
\begin{nproposition}(Cf.~\cite[Corollary 4.5 with $l=2$]{cossidente2021combining})
  \label{prop_cor_4_5}
  $$
    C_q(n_1,n_2,d;k_1,k_2)\ge \min\{\alpha_1,\alpha_2\} \cdot \prod_{i=1}^2 A_q^R(k_i\times(n_i-k_i),d/2) ,
  $$  
  where $\alpha_i=A_q^R(k_i\times(n_i-k_i),d_i/2) / A_q^R(k_i\times(n_i-k_i),d/2)$ for $i=1,2$ and $d_1,d_2\in 2\N$ with $d_1+d_2=d$.
\end{nproposition}
The underlying idea can be briefly indicated by the matrix description 
\begin{center}
  \begin{tabular}{|C{1.5cm}|C{1.5cm}|C{1.5cm}|C{1.5cm}|} 
    \hline
    \texttt{I} & \texttt{R}$^i$ & \texttt{0} & \texttt{0}\\ 
    \hline
    \texttt{0} & \texttt{0} & \texttt{I} & \texttt{R}$^i$\\
    \hline
  \end{tabular}
\end{center}
and Lemma~\ref{lemma_parallel_mrd} mimicking parallelisms for {\lmrd} codes, cf.~\cite[Section 4.9]{etzion2016galois}.
 
\begin{nlemma}\textbf{(Parallel {\fdrm} codes -- C.f.~\cite[Lemma 2.5]{lao2020parameter}, \cite[proof Corollary 4.5]{cossidente2021combining})}\\
  \label{lemma_cosets}
  Let $\cF$ be a Ferrers diagram and $\cM$ be a corresponding additive {\fdrm} code with minimum rank distance $d$. If $\cM$ is a subcode of a 
  an additive {\fdrm} code $\cM'$ with minimum rank distance $d'<d$ and Ferrers diagram $\cF$, then there exist {\fdrm} codes $\cM_i$ with Ferrers 
  diagram $\cF$ for $1\le i\le \alpha:=\#\cM'/\#\cM$ satisfying 
  \begin{enumerate}
    \item[(1)] $\dr(\cM_i)\ge d$ for all $1\le i\le \alpha$;
    \item[(2)] $\dr(\cM_i,\cM_j)\ge d'$ for all $1\le i<j\le \alpha$; and
    \item[(3)] $\cM_1,\dots,\cM_\alpha$ is a partition of $\cM'$.
  \end{enumerate} 
\end{nlemma}
\begin{proof}
  For each $M'\in \cM'$ the code $\cM+M':=\{M+M'\,:\, M\in \cM\}$ is an $\fdrm$  code with Ferrers diagram $\cF$ and minimum rank distance $d$. For 
  $M',M''\in \cM'$ we have $M'+\cM=M''+\cM$ iff $M'-M''\in \cM$ and $M'+\cM\cap M''+\cM=\emptyset$ otherwise. Now let $\cM_1,\dots,\cM_\alpha$ be the 
  $\alpha=\#\cM'/\#\cM$ different codes $M+\cM$, which are cosets of $\cM$ in $\cM'$ and partition $\cM'$. Since all elements of $\cM_i$ and $\cM_j$ 
  are different elements of $\cM'$ we have $\dr(\cM_i,\cM_j)\ge d'$ for all $1\le i<j\le \alpha$. 
\end{proof}

Choosing $\cF$ as $a\times b$ rectangular Ferrers diagram, we end up with \cite[Lemma 2.5]{lao2020parameter}, see also Exercise~\ref{exercise_parallel_rmc}. 
Note that we have to choose Delsarte--Gabidulin (or some other specific class of) {\mrd} codes in order to ensure that an {\mrd} code for minimum rank distance 
$d$ contains an {\mrd} code with minimum rank distance $d+1$ as a subcode. In the proof of 
\cite[Corollary 4.5]{cossidente2021combining} this lemma is indirectly applied with $a=a_i$ and $b=n_i-a_i$. Note that for minimum rank 
distance $\delta=2$ the upper bound from \cite[Theorem 1]{etzion2009error}, cf.\ Theorem~\ref{thm_upper_bound_ef}, can always be attained by 
linear rank metric codes. Moreover, the only choice for $\delta'$ then is $\delta'=1$ and $\cM'$ consists of all matrices with Ferrers diagram $\cF$. 
Thus, $\cM'$ is automatically linear and contains $\cM$ as a subcode. 

\begin{question}{Research problem}Study the existence of {\lq\lq}large{\rq\rq} linear {\fdrm} codes that contain {\fdrm} codes of larger 
minimum rank distance as a subcode.
\end{question} 

A first approach might be to start from a linear Delsarte--Gabidulin {\mrd} code and to consider linear subcodes going in line with the support restrictions of 
a given Ferrers diagram $\cF$.

\begin{table}[htp]
\begin{center}
  \begin{tabular}{lll}
    \hline
    pivot vector & size $m(q,\cF,2)$ & $\#$ of cosets $m(q,\cF,1)/m(q,\cF,2)$ \\  
    \hline
    $11000$ & $q^3$ & $q^3$ \\ 
    $10100$ & $q^2$ & $q^3$ \\ 
    $10010$ & $q$   & $q^3$ \\
    $10001$ & $1$   & $q^3$ \\ 
    $01100$ & $q^2$ & $q^2$ \\ 
    $01010$ & $q$   & $q^2$ \\
    $01001$ & $1$   & $q^2$ \\
    $00110$ & $1$   & $q^2$ \\
    $00101$ & $1$   & $q$   \\
    $00011$ & $1$   & $1$   \\  
    \hline
  \end{tabular}
  \caption{Data for Lemma~\ref{lemma_cosets} with $\cF\in\cG_1(5,2)$.}
  \label{table_cosets_5_2}
\end{center}      
\end{table}

\begin{table}[htp]
\begin{center}
  \begin{tabular}{lll}
    \hline
    skeleton code & size & $\#$ of used cosets \\ 
    \hline
    $\{11000,00110\}$ & $q^3+1$ & $q^2$ \\
    $\{11000,00101\}$ & $q^3+1$ & $q$ \\
    $\{11000,00011\}$ & $q^3+1$ & $1$ \\
    $\{11000\}$       & $q^3$ & $q^3-q^2-q-1$ \\      
    \hline
    $\{10100,01010\}$ & $q^2+q$ & $q^2$ \\
    $\{10100,01001\}$ & $q^2+1$ & $q^2$ \\
    $\{10100\}$       & $q^2$   & $q^3-2q^2$ \\
    \hline
    $\{01100,10010\}$ & $q^2+q$   & $q^2$ \\
    $\{10010\}$       & $q$       & $q^3-q^2$ \\
    $\{10001\}$       & $1$       & $q^3$ \\
    \hline   
  \end{tabular}
  \caption{$4$-packing scheme for $\cG_q(5,2)$.}
  \label{table_cosets_5_2_packed}
\end{center}      
\end{table}

\begin{ncorollary}
  $$
    C_q(5,5,4;2,2)\ge q^9 + q^7 + q^6 + 7q^5 + 5q^4 + 3q^3 + 2q^2 + q + 1  
  $$    
\end{ncorollary}
I.e., we have $C_2(5,5,4;2,2)\ge 1043$. Proposition~\ref{prop_cor_4_5} yields $C_q(5,5,4;2,2)\ge q^9$, i.e., $C_2(5,5,4;2,2)\ge 512$. Proposition~\ref{prop_special_symmetric_C_q_lb} 
gives $C_q(5,5,4;2,2)\ge q^9 + q^7 + q^6 + q^5 + q^4 + q^3 + 2q^2 + q + 1$, i.e., $C_2(5,5,4;2,2)\ge 771$. In \cite{kurz2021interplay} the lower bound 
$C_2(5,5,4;2,2)\ge 1313$ was shown by a heuristic computer search. By an easy argument the upper bound $C_2(5,5,4;2,2)\le 1381$ was shown.

We can also use more geometric ideas.
\begin{nproposition}
  \label{prop_c_q_5_5_4_2_2_two_mrd}
  $$
    C_q(5,5,4;2,2)\ge q^9 + q^7 + 2q^6 + q^5 - q^4 + 4q^3 + 6q^2 + 4q + 2 
  $$
\end{nproposition}
\begin{proof}
  Let $\pi$ and $\pi'$ be two planes in $\F_q^5$ intersecting in a point $P$. Let $\cC$ be an {\lmrd} code disjoint to $\pi$ that 
  can be partitioned into $q^3$ partial line spreads $\cC_i$ of cardinality $q^3$. Similarly, let $\cC'$ be an {\lmrd} code disjoint to $\pi'$ that 
  can be partitioned into $q^3$ partial line spreads $\cC_i$ of cardinality $q^3$. For $1\le i\le\qbin{3}{2}{q}=q^2+q+1$ we add one of the $\qbin{3}{2}{q}$ different lines contained 
  in $\pi$ to $\cC_i$. To ensure that no line occurs twice we only keep those lines in $\cC_i'$ that intersect $\pi$ in exactly a point. Let us now determine the 
  resulting sizes $\#\cC_i'$. To this end, let $\cL$ be the set of the $q^2$ lines in $\pi$ that do not contain $P$. Since the elements of $\cL$ are pairwise intersecting in a point, 
  there are exactly $q^2$ partial line spreads $\cC_i'$ that contain one element from $\cL$. For these, exactly $\qbin{3}{1}{q}-\qbin{2}{1}{q}-1=q^2-1$ elements intersect in exactly 
  one point. For the other $q^3-q^2$ partial line spreads, $\qbin{3}{1}{q}-1=q^2+q$ of its elements intersect $\pi$ in exactly a point. Since $\pi'$ contains $\qbin{2}{1}{q}=q+1$ 
  lines intersecting $\pi$ in a point, we can add a further line to $q+1$ of the latter partial line spreads $\cC_i'$ each. This gives
  \begin{eqnarray*}
    && \sum_{i=1}^{q^3} \left(\#\cC_i\right)^2 \,+\, \sum_{i=1}^{q^3} \left(\#\cC_i'\right)^2\\ 
     &=& \left(q^2+q+1\right)\cdot\left(q^3+1\right)^2 +\left(q^3-q^2-q-1\right)\cdot \left(q^3\right)^2\\ 
    &&+q^2\cdot\left(q^2-1\right)^2+(q+1)\cdot \left(q^2+q+1\right)^2+\left(q^3-q^2-q-1\right)\cdot\left(q^2+q\right)^2\\ 
    &=& q^9 + q^7 + 2q^6 + q^5 - q^4 + 4q^3 + 6q^2 + 4q + 2.  
  \end{eqnarray*}      
\end{proof}  
\begin{nexercise}
  Improve the lower bound of Proposition~\ref{prop_c_q_5_5_4_2_2_two_mrd} by taking the unused lines into account. Conclude a similar bound assuming that the planes 
  $\pi$ and $\pi'$ intersect in a line.
\end{nexercise}  

\begin{table}[htp]
\begin{center}
  \begin{tabular}{lll}
    \hline
    pivot vector & size $m(q,\cF,2)$ & $\#$ of cosets $m(q,\cF,1)/m(q,\cF,2)$ \\  
    \hline
    $110000$ & $q^4$ & $q^4$ \\ 
    $101000$ & $q^3$ & $q^4$ \\
    $100100$ & $q^2$ & $q^4$ \\
    $100010$ & $q$ & $q^4$\\
    $100001$ & $1$ & $q^4$ \\  
    $011000$ & $q^3$ & $q^3$ \\ 
    $010100$ & $q^2$ & $q^3$ \\ 
    $010010$ & $q$   & $q^3$ \\
    $010001$ & $1$   & $q^3$ \\ 
    $001100$ & $q^2$ & $q^2$ \\ 
    $001010$ & $q$   & $q^2$ \\
    $001001$ & $1$   & $q^2$ \\
    $000110$ & $1$   & $q^2$ \\
    $000101$ & $1$   & $q$   \\
    $000011$ & $1$   & $1$   \\
    \hline
  \end{tabular}
  \caption{Data for Lemma~\ref{lemma_cosets} with $\cF\in\cG_1(6,2)$.}
  \label{table_cosets_6_2}
\end{center}      
\end{table}

\begin{table}[htp]
\begin{center}
  \begin{tabular}{lll}
    \hline
    skeleton code & size & $\#$ of used cosets \\ 
    \hline
    $\{110000,001100,000011\}$ & $q^4+q^2+1$ & $1$ \\
    $\{110000,001100\}$        & $q^4+q^2$   & $q^2-1$ \\
    $\{110000,001010,000101\}$ & $q^4+q+1$   & $q$ \\
    $\{110000,001010\}$        & $q^4+q$     & $q^2-q$\\
    $\{110000,000110,001001\}$ & $q^4+q+1$   & $q$  \\
    $\{110000,001001\}$        & $q^4+1$     & $q^2-q$  \\
    $\{110000\}$               & $q^4$       & $q^4-3q^2$ \\    
    \hline
    $\{101000,010100\}$ & $q^3+q^2$ & $q^3$ \\   
    $\{101000,010010\}$ & $q^3+q$   & $q^3$\\ 
    $\{101000\}$         & $q^3$    &$q^4-2q^3$ \\
    $\{011000,100100\}$ & $q^3+q^2$ & $q^3$\\
    \hline
    $\{100100,010001\}$ & $q^2+1$ & $q^3$\\
    $\{100100\}$        & $q^2$   & $q^4-2q^3$ \\
    \hline
    $\{100010\}$ & $q$ & $q^4$ \\    
    \hline
    $\{100001\}$ & $1$ & $q^4$ \\ 
    \hline
  \end{tabular}
  \caption{$4$-packing scheme for $\cG_q(6,2)$.}
  \label{table_cosets_6_2_packed}
\end{center}      
\end{table}

\begin{ncorollary}
  $$
    C_q(6,6,4;2,2)\ge q^{12} + q^{10} + q^9 + 7q^8 + 5q^7 + 6q^6 + 5q^5 + 4q^4 + 2q^3 + 7q^2 + q + 1  
  $$    
\end{ncorollary}
I.e., we have $C_2(6,6,4;2,2)\ge 8719$. Proposition~\ref{prop_cor_4_5} yields $C_q(6,6,4;2,2)\ge q^{12}$, i.e., $C_2(6,6,4;2,2)\ge 4096$.
The upper bound from Lemma~\ref{lemma_upper_bound_coset} is given by 
$$
  q^{12} + q^{11} + 3q^{10} + 3q^9 + 6q^8 + 5q^7 + 7q^6 + 5q^5 + 6q^4 + 3q^3 + 3q^2 + q + 1,
$$
i.e., $C_2(6,6,4;2,2)\le 13671$. Due to he existence of parallelisms in $\cG_q(6,2)$ for $q\in\{2,3\}$ the upper bound is indeed attained. So our 
packing constructions are very far from being optimal. (For $q=2$ the polynomial in Proposition~\ref{prop_c_q_6_6_4_2_2} would result in $8839$.)

\begin{nexercise}
  Improve the stated packing scheme for $\cG_q(6,2)$ for $q>2$.
\end{nexercise}

\begin{nproposition}
  \label{prop_c_q_6_6_4_2_2}
  For $q\ge 3$ we have
  $$
    C_q(6,6,4;2,2)\ge q^{12} + q^{10} + q^9 + 7q^8 + 5q^7 + 8q^6 + 4q^5 + 6q^4 + 3q^3 + 3q^2 + q + 1.
  $$
\end{nproposition}
\begin{proof}
  Let $S$ be the solid with pivot vector $001111$ in $\F_q^6$ and $\cC$ be an {\lmrd} code disjoint to $S$ that can be partitioned into $q^4$ partial 
  line spreads $\cC_i$ of cardinality $q^4$. Since $S\cong \F_q^4$ there exists a parallelism of $S$, so that we can add $q^2+1$ additional lines to $q^2+q+1$ 
  of the partial line spreads $\cC_i$. So, we have
  \begin{eqnarray*}
    \sum_{i=1}^{q^4} \left(\#\cC_i\right)^2 &=&\left(q^2+q+1\right)\cdot\left(q^4+q^2+1\right)^2+\left(q^4-q^2-q-1\right)\cdot\left(q^4\right)^2\\ 
    &=& q^{12} + 2q^8 + 2q^7 + 5q^6 + 3q^5 + 5q^4 + 2q^3 + 3q^2 + q + 1.
  \end{eqnarray*}  
  The lines used so far, all lines being either disjoint to $S$ or contained in $S$, i.e., the have pivot vector $110000$ or their pivot vector is contained in 
  $\left({2\choose 0},{4\choose 2}\right)$. For the remaining pivot vectors we consider the packing scheme
  \begin{center}
  \begin{tabular}{lll}
    \hline
    skeleton code & size & $\#$ of used cosets \\ 
    \hline
    $\{101000,010100\}$ & $q^3+q^2$ & $q^3$ \\
    $\{101000,010010\}$ & $q^3+q$   & $q^3$ \\
    $\{101000,010001\}$ & $q^3+1$   & $q^3$ \\
    $\{100100,011000\}$ & $q^3+q^2$ & $q^3$ \\
    \hline
    $\{101000\}$ & $q^3$ & $q^4-3q^3$\\ 
    $\{100100\}$ & $q^2$ & $q^4-q^3$\\
    $\{100010\}$ & $q$   & $q^4$\\
    $\{100001\}$ & $1$   & $q^4$\\
    \hline
  \end{tabular}
  \end{center}
  yielding an additional contribution of
  $$
    q^{10} + q^9 + 5q^8 + 3q^7 + 3q^6 + q^5 + q^4 + q^3\!.
  $$
\end{proof}  

\begin{nproposition}(\cite[Proposition 3.5]{kurz2021interplay})
  \label{prop_special_symmetric_C_q_lb}
  $$
    C_q(n,n,4;k,k) \ge \sum_{v\in\cG_1(n,k)} A_q(n,1;k;v)\cdot A_q(n,2;k;v) 
  $$  
\end{nproposition}

\section{Inserting constructions}
\label{subsec_inserting_constructions}

We have seen in Subsection~\ref{subsec_lifting} that the generalized linkage construction yields {\cdc}s with competitive cardinalities. In 
Lemma~\ref{lemma_combine_generalized_linkage_and coset} we have summarized sufficient conditions for the combination with subcodes obtained via the coset construction. 
In these subsection we want to study further variants of subcodes that can be used to improve the generalized linkage construction. In e.g.\ 
\cite{lao2021new,niu2020new,niu2021construction} the authors speak of inserting constructions cf.~also \cite{he2021new}.

Packings of {\rmc}s constructed in Subsection~\ref{subsec_packing_constructions} can be exploited as follows:
\begin{ntheorem}\textbf{(Block inserting construction I -- \cite[Theorem 4]{lao2021new})}\\
  \label{theorem_block_inserting}
  Let $\cC_1$ be an $(n_1,d;k)_q$--{\cdc}, $\cC_2$ be an $(n_3,d;k)_q$--{\cdc}, $\cM_3$ be a $(k_1\times n_4,d/2;k_1-d/2)_q$--{\rmc}, $\cM_4$ be a 
  $(k_2\times n_2,d/2;k_2-d/2)_q$--{\rmc}, $\cM_1$ be a $(k_1\times n_2,d_1/2)_q$--{\rmc}, and $\cM_2$ be a $(k_2\times n_3,d_2/2)_q$--{\rmc}, where $d_1+d_2=d$. 
  Let $\cM_1^1,\dots,\cM_1^s$ and $\cM_2^1,\dots,\cM_2^s$ be $\tfrac{t}{2}$-packings of cardinality $s$ of $\cM_1$ and $\cM_2$, respectively. With this let
  \begin{eqnarray*}
    \Big\{ \begin{pmatrix} G_1 & M_1 & \mathbf{0}_{k_1 \times n_3} & M_3\\ \mathbf{0}_{k_2\times n_1} & M_4 & G_2 & M_2\end{pmatrix} &:& G_1\in\cG_1, M_1\in\cM_1^i, M_3\in\cM_3, \\ 
    && 
    M_4\in\cM_4, G_2\in\cG_2, M_2\in\cM_2^i\Big\}
  \end{eqnarray*}  
  be a generating set of a subcode $\cW^i$ for $1\le i\le s$, where $\cG_1$ and $\cG_2$ are generating sets of $\cC_1$ and $\cC_2$, respectively. 
  Then, $\cW=\cup_{i=1}^s \cW^i$ is an $\left(n_1+n_2+n_3+n_4,d;k_1+k_2\right)_q$--{\cdc} with cardinality
  $$
    \#\cW=\sum_{i=1}^s \#\cW^i=  \#\cC_1\cdot \#\cC_2\cdot\#\cM_3\cdot\#\cM_4 \cdot \sum_{i=1}^s \cM_1^i\cdot \#\cM_2^i.
  $$
\end{ntheorem}
\begin{proof}
  Let
  $$
    H=\begin{pmatrix} G_1 & M_1 & \mathbf{0} & M_3\\ \mathbf{0} & M_4 & G_2 & M_2\end{pmatrix}
  $$   
  be the generator matrix of an arbitrary codeword $W\in\cW$. Since $k_1+k_2\ge \rk(H)\ge \rk(G_1)+\rk(G_2)=k_1+k_2$, every codeword is a $\left(k_1+k_2\right)$-space. 
  
  Let 
  $$
    H'=\begin{pmatrix} G_1' & M_1' & \mathbf{0} & M_3'\\ \mathbf{0} & M_4' & G_2' & M_2'\end{pmatrix}
  $$
  be the generator matrix of another codeword $W'\in \cW$,
  $$
    R:=\rk\left(\begin{pmatrix} G_1 & M_1 & \mathbf{0} & M_3\\ \mathbf{0} & M_4 & G_2 & M_2\\G_1' & M_1' & \mathbf{0} & M_3'\\ \mathbf{0} & M_4' & G_2' & M_2'\end{pmatrix}\right)=
    \rk\left(\begin{pmatrix} G_1 & \mathbf{0} & M_1& M_3\\ G_1'-G_1 & \mathbf{0} & M_1'-M_1 & M_3'-M_3\\ \mathbf{0} & G_2 & M_4 & M_2 \\ \mathbf{0} & G_2'-G_2 & M_4'-M_4 & M_2'-M_2\end{pmatrix}\right),
  $$ and $U_1:=\left\langle G_1\right\rangle$, $U_2:=\left\langle G_2\right\rangle$, $U_1':=\left\langle G_1'\right\rangle$, 
  $U_2':=\left\langle G_2'\right\rangle$. 
  
  If $G_1\neq G_1'$ or $G_2\neq G_2'$, then we have $U_1\neq U_1'$ or $U_2\neq U_2'$, so that
  \begin{eqnarray*}
    \ds(W,W') &=& 2(R-k_1-k_2)\ge 2\cdot \rk\left(\begin{pmatrix} G_1 & \mathbf{0}\\ G_1'-G_1 & \mathbf{0} \\ \mathbf{0} & G_2 \\ \mathbf{0} & G_2'-G_2 & \end{pmatrix}\right) -2k_1-2k_2 \\ 
    &=&\ds(U_1,U_1')+\ds(U_2,U_2')\ge d.
  \end{eqnarray*}
  
  If $G_1=G_1'$ and $G_2=G_2'$, then we have
  $$
    R=\rk \left(\begin{pmatrix} G_1 & \mathbf{0} & M_1& M_3\\ \mathbf{0} & G_2 & M_4 & M_2 \\ \mathbf{0} & \mathbf{0} & M_1'-M_1 & M_3'-M_3\\  \mathbf{0} & \mathbf{0} & M_4'-M_4 & M_2'-M_2\end{pmatrix}\right) 
    =k_1+k_2+\rk\Big(\overset{\widetilde{M}:=}{\overbrace{\begin{pmatrix}M_1'-M_1 & M_3'-M_3\\ M_4'-M_4 & M_2'-M_2\end{pmatrix}}}\Big),
  $$  
  so that it suffices to show $\rk(\widetilde{M})\ge d/2$ in order to deduce $\ds(W,W')\ge d$.
  
  If $M_3\neq M_3'$ or $M_4\neq M_4'$, then we have $\rk(\widetilde{M})\ge \rk(M_3-M_3')+\rk(M_4-M_4')\ge \min\left\{\dr(\cM_3),\dr(\cM_4)\right\}\ge d/2$.
  
  If $M_3=M_3'$ and $M_4=M_4'$, then we have $\rk(\widetilde{M})=\rk(M_1-M_1')+\rk(M_2-M_2')=\dr(M_1,M_1')+\dr(M_2,M_2')$. If $M_1=M_1'$, then their exists an 
  index $1\le i\le s$ with $M_2,M_2'\in\cM_2^i$ and we have $M_2\neq M_2'$, so that $\rk(\widetilde{M})\ge \dr(M_2,M_2')\ge \dr(\cM_2^i)\ge d/2$. Similarly, 
  if $M_2=M_2'$, then their exists an index $1\le i\le s$ with $M_1,M_1'\in\cM_1^i$ and we have $M_1\neq M_1'$, so that $\rk(\widetilde{M})\ge \dr(M_1,M_1')\ge 
  \dr(\cM_1^i)\ge d/2$. If $M_1\neq M_1'$ and $M_2\neq M_2'$, then we have $\rk(\widetilde{M})\ge \dr(M_1,M_1')+\dr(M_2,M_2')\ge \dr(\cM_1)+\dr(\cM_2)\ge 
  d_1/2+d_2/2=d/2$.  
\end{proof}
The matrix description of the block inserting construction I is given by 
\begin{center}
  \begin{tabular}{|C{1.5cm}|C{1.5cm}|C{1.5cm}|C{1.5cm}|} 
    \hline
    \texttt{C} & \texttt{R}$^i$ & \texttt{0} & \texttt{R}\\ 
    \hline
    \texttt{0} & \texttt{R} & \texttt{C} & \texttt{R}$^i$\\
    \hline
  \end{tabular}
\end{center}
\begin{ncorollary}
  Let $\cW$ be a {\cdc} constructed via the block inserting construction in Theorem~\ref{theorem_block_inserting} with parameters $\left(n_1,n_2,n_3,n_4,d,k_1,k_2\right)$, 
  where $d_1,d_2$ with $d_1+d_2=d$ are arbitrary, of maximum possible cardinality. Then, we have
  \begin{eqnarray*}
    \#\cW & \ge & A_q(n_1,d;k_1)\cdot A_q(n_3,d;k_2)\cdot A_q^R(k_1\times n_4,\tfrac{d}{2};\le k_1-\tfrac{d}{2}) \cdot \\ 
    && A_q^R(k_2\times n_2,\tfrac{d}{2};k_2-\tfrac{d}{2})\cdot A_q^R(k_1\times n_2,d/2)\cdot A_q^R(k_2\times n_4,d/2)\cdot \alpha, 
  \end{eqnarray*}
  where 
  $$
    \alpha=\max_{d_1,d_2\,:\, d_1+d_2=d} \min\left\{\frac{A_q^R(k_1\times n_2,d_1/2)}{A_q^R(k_1\times n_2,d/2)},\frac{A_q^R(k_2\times n_4,d_2/2)}{A_q^R(k_2\times n_4,d/2)}\right\}.
  $$
\end{ncorollary}
\begin{nexample}
  \label{example_12_4_6_inserting_construction_I_part}
  Let $\cW$ be a {\cdc} constructed via the block inserting construction in Theorem~\ref{theorem_block_inserting} with parameters $\left(n_1,n_2,n_3,n_4,d_1,d_2,k_1,k_2\right)
  =(3,3,3,3,2,2,3,3)$ of maximum possible cardinality. Then, we have
  $$
    \#\cW\ge q^{12}\cdot A_q^R(4\times 4,2;\le 2)\ge q^{20} + q^{19} + 2q^{18} + q^{17} - q^{15} - 2q^{14} - q^{13} 
  $$
  using Lemma~\ref{lemma_lb_restricted_rank_additive_mrd}.
\end{nexample}

Note that the upper rank bounds for the matrices in $\cM_3$ and $\cM_4$ are not necessary in the proof of Theorem~\ref{theorem_block_inserting}
\begin{nlemma}\textbf{(Generalized linkage construction $+$ block inserting construction)}\\
  \label{lemma_combine_generalized_linkage_block_inserting}
  Let $\cW_1$ be a {\cdc} constructed via the generalized linkage construction in Theorem~\ref{theorem_generalized_linkage} with parameters $\left(n_1',n_2',d,k\right)$
  and $\cW_2$ be a {\cdc} constructed via the block inserting construction in Theorem~\ref{theorem_block_inserting} with parameters 
  $(n_1,n_2,n_3,n_4,d_1,d_2,k_1,$ $k_2)$. If $n_1'=n_1+n_2$, $n_2'=n_3+n_4$, $d=d_1+d_2$, and $k=k_1+k_2$, then $\cW=\cW_1\cup \cW_2$ is an 
  $(n_1'+n_2',d;k)_q$--{\cdc} with cardinality $\#\cW=\#\cW_1+\#\cW_2$.
 \end{nlemma}
\begin{proof}
  Let 
  $$
    H=\begin{pmatrix} G_1 & M_1 & \mathbf{0} & M_3\\ \mathbf{0} & M_4 & G_2 & M_2\end{pmatrix}=:\begin{pmatrix}P_1\\P_2\end{pmatrix}
  $$   
  be the generator matrix of an arbitrary codeword $W_2\in\cW_2$, $U_1:=\left\langle P_1\right\rangle$, $U_2:=\left\langle P_2\right\rangle$, and $E_1,E_2$ be the special 
  subspaces for $\cW_1$ as in Lemma~\ref{lemma_additional_codewords_generalized_linkage_construction}. Since 
  $\dim(W_2\cap E_1)\ge \dim(U_1\cap E_1)\ge \rk(G_1)-\rk(M_3)\ge d/2$ and $\dim(W_2\cap E_2)\ge \dim(U_2\cap E_2)\ge \rk(G_2)-\rk(M_4)\ge d/2$ we have  
  $\ds(\cW_1,W_2)\ge d$ by  Lemma~\ref{lemma_additional_codewords_generalized_linkage_construction}.  
\end{proof}
\begin{nexample}
  The {\cdc} obtained from the block inserting construction I in Example~\ref{lemma_lb_restricted_rank_additive_mrd} 
  is compatible with a {\cdc} obtained from the generalized linkage construction with parameters $\left(n_1,n;2,d,k\right)=(6,6,4,6)$, so that
  $$
    A_q(12,4;6)\ge A_q(12,4;6)\ge q^{30}+A_q^R(6\times 6,2;\le 4)+q^{12}\cdot A_q^R(4\times 4,2;\le 2).
  $$
  However, as mentioned after Example~\ref{example_12_4_6_inserting_construction_I_part}, the effort for the more complicated coset construction  
  pays off, see Example~\ref{example_a_12_4_6_generalized_linkage_and_coset}.
\end{nexample}  

As a special case of the block inserting construction in Theorem~\ref{theorem_block_inserting} we mention:
\begin{nproposition}(\cite[Proposition 2.1]{lao2020parameter})\\
  \label{proposition_block_inserting_special_case}
  Let $\cM_3$ be a $(k_1\times n_4,d/2;k_1-d/2)_q$--{\rmc}, $\cM_4$ be a 
  $(k_2\times n_2,d/2;k_2-d/2)_q$--{\rmc}, $\cM_1$ be a $(k_1\times n_2,d_1/2)_q$--{\rmc}, and $\cM_2$ be a $(k_2\times n_3,d_2/2)_q$--{\rmc}, where $d_1+d_2=d$. 
  Let $\cM_1^1,\dots,\cM_1^s$ and $\cM_2^1,\dots,\cM_2^s$ be $\tfrac{t}{2}$-packings of cardinality $s$ of $\cM_1$ and $\cM_2$, respectively. With this let
  $$
    \left\{ \begin{pmatrix} I_{k_1} & M_1 & \mathbf{0}_{k_1 \times n_3} & M_3\\ \mathbf{0}_{k_2\times n_1} & M_4 & I_{k_2} & M_2\end{pmatrix}\,:\, M_1\in\cM_1^i, M_3\in\cM_3, 
    M_4\in\cM_4, M_2\in\cM_2^i\right\}
  $$  
  be a generating set of a subcode $\cW^i$ for $1\le i\le s$. 
  Then, $\cW=\cup_{i=1}^s \cW^i$ is an $\left(n_1+n_2+n_3+n_4,d;k_1+k_2\right)_q$--{\cdc} with cardinality
  $$
    \#\cW=\sum_{i=1}^s \#\cW^i=  \#\cM_3\cdot\#\cM_4 \cdot \sum_{i=1}^s \#\cM_1^i\cdot \#\cM_2^i.
  $$
\end{nproposition}
\begin{ncorollary}
  Let $\cW$ be a {\cdc} constructed via Proposition~\ref{proposition_block_inserting_special_case} with parameters $(n_1,n_2,n_3,$ $n_4,d,k_1,k_2)$, 
  where $d_1,d_2$ with $d_1+d_2=d$ are arbitrary, of maximum possible cardinality. Then, we have
  \begin{eqnarray*}
    \#\cW & \ge & A_q^R(k_1\times n_4,\tfrac{d}{2};\le k_1-\tfrac{d}{2}) \cdot A_q^R(k_2\times n_2,\tfrac{d}{2};k_2-\tfrac{d}{2})\cdot\\ 
    &&  A_q^R(k_1\times n_2,d/2)\cdot A_q^R(k_2\times n_4,d/2)\cdot \alpha, 
  \end{eqnarray*}
  where 
  $$
    \alpha=\max_{d_1,d_2\,:\, d_1+d_2=d} \min\left\{\frac{A_q^R(k_1\times n_2,d_1/2)}{A_q^R(k_1\times n_2,d/2)},\frac{A_q^R(k_2\times n_4,d_2/2)}{A_q^R(k_2\times n_4,d/2)}\right\}.
  $$
\end{ncorollary}
\begin{nexample}
  \label{example_12_4_6_generalized_linkage_and_simplified_inserting_construction_I}
  Let $\cW$ be a $(12,4;6)_q$--{\cdc} obtained via the block inserting construction in Theorem~\ref{theorem_block_inserting} with parameters $\left(n_1,n_2,n_3,n_4,d_1,d_2,k_1,k_2\right)
  =(4,2,2,4,2,2,4,2)$. Let $\cM_4=\left\langle\mathbf{0}_{2\times 2}\right\rangle$, so that we can assume 
  $$
    \#\cW\ge q^{12}\cdot A_q^R(4\times 4;2\le 2)\ge q^{20} + q^{19} + 2q^{18} + q^{17} - q^{15} - 2q^{14} - q^{13}.
  $$
\end{nexample}

\begin{nexample}
  \label{example_12_6_6_generalized_linkage_and_simplified_inserting_construction_I}
  Let $\cW$ be a $(12,6;6)_q$--{\cdc} obtained via the block inserting construction in Theorem~\ref{theorem_block_inserting} with parameters $\left(n_1,n_2,n_3,n_4,d_1,d_2,k_1,k_2\right)
  =(3,3,3,3,2,4,3,3)$. Let $\cM_3=M_4=\left\langle\mathbf{0}_{3\times 3}\right\rangle$ and choose $\cM_1=\cM_2$ as $(3\times 3,2)_q$--{\mrd} codes, so that we can assume 
  $\#\cW\ge q^{9}$.  
\end{nexample}

In \cite[Theorem 5]{lao2021new} another inserting construction being compatible with the generalized linkage construction and the block inserting construction I was proposed 
as \emph{block inserting construction II}. We give a slight generalization under the same name.
\begin{ntheorem}\textbf{(Block inserting construction II -- cf.\ \cite[Theorem 5]{lao2021new}, \cite[Theorem 2.7]{lao2020parameter})}\\
  \label{theorem_block_inserting_II}
  Let $\cM$ be a $\left(k_1\times n_1,k_2\times n_3,d/2;\le k_1+k_2-d/2\right)_q$--{\srmc}, $\cC_1$ be an $\left(n_2,d;k_1\right)_q$--{\cdc}, and be a $\cC_2$ be an $\left(n_4,d;k_2\right)_q$--{\cdc}. 
  With this, let
  $$
    \left\{\begin{pmatrix} M_1 & G_1 & \mathbf{0}_{k_1\times n_3} & \mathbf{0}_{k_1\times n_4}\\ \mathbf{0}_{k_2\times n_1} & \mathbf{0}_{k_2\times n_2} & M_2 & G_2\end{pmatrix}\,:\, 
    G_1\in\cG_1,G_2\in\cG_2,\left(M_1,M_2\right)\in \cM\right\}
  $$
  be a generating set of a subspace code $\cW$, where $\cG_1$ and $\cG_2$ be generating sets of $\cC_1$ and $\cC_2$, respectively. Then, $\cW$ is an $\left(n_1+n_2+n_3+n_4,d;k_1+k_2\right)_q$--{\cdc} 
  with cardinality $\#\cC_1\cdot\#\cC_2\cdot\#\cM$.
\end{ntheorem}   
\begin{proof}
  Let
  $$
    H=\begin{pmatrix} M_1 & G_1 & \mathbf{0} & \mathbf{0}\\ \mathbf{0} & \mathbf{0} & M_2 & G_2\end{pmatrix}
  $$ 
  be the generator matrix of an arbitrary codeword $W\in\cW$. Since $k_1+k_2\ge \rk(H)\ge \rk(G_1)+\rk(G_2)=k_1+k_2$, every codeword is a $\left(k_1+k_2\right)$-space.  
  Let
  $$
    H'=\begin{pmatrix} M_1' & G_1' & \mathbf{0} & \mathbf{0}\\ \mathbf{0} & \mathbf{0} & M_2' & G_2'\end{pmatrix}
  $$ 
  be the generator matrix of another codeword $W'\in\cW$,
  $$
    R:=\rk\left(\begin{pmatrix} M_1 & G_1 & \mathbf{0} & \mathbf{0}\\ \mathbf{0} & \mathbf{0} & M_2 & G_2\\M_1' & G_1' & \mathbf{0} & \mathbf{0}\\ \mathbf{0} & \mathbf{0} & M_2' & G_2'\end{pmatrix}\right)
    =\rk\left(\begin{pmatrix} G_1 & \mathbf{0} & M_1 & \mathbf{0}\\ G_1'-G_1 &\mathbf{0} & M_1'-M_1 & \mathbf{0} \\ \mathbf{0} & G_2 & \mathbf{0} & M_2\\ \mathbf{0} & G_2'-G_2 & \mathbf{0} & M_2'-M_2\end{pmatrix}\right),
  $$
  and $U_1:=\left\langle G_1\right\rangle$, $U_2:=\left\langle G_2\right\rangle$, $U_1':=\left\langle G_1'\right\rangle$, $U_2':=\left\langle G_2'\right\rangle$.
  
  If $G_1\neq G_1'$ or $G_2\neq G_2'$, then we have $U_1\neq U_1'$ or $U_2\neq U_2'$, so that
  \begin{eqnarray*}
    \ds(W,W') &=& 2(R-k_1-k_2)\ge 2\cdot \rk\left(\begin{pmatrix} G_1 & \mathbf{0}\\ G_1'-G_1 & \mathbf{0} \\ \mathbf{0} & G_2 \\ \mathbf{0} & G_2'-G_2 & \end{pmatrix}\right) -2k_1-2k_2 \\ 
    &=&\ds(U_1,U_1')+\ds(U_2,U_2')\ge d.
  \end{eqnarray*}
  
  If $G_1=G_1'$ and $G_2=G_2'$, then we have
  $$
    R=\rk \left(\begin{pmatrix} G_1 & \mathbf{0} & M_1& \mathbf{0}\\ \mathbf{0} & G_2 & \mathbf{0} & M_2 \\ \mathbf{0} & \mathbf{0} & M_1'-M_1 & \mathbf{0}\\  \mathbf{0} & \mathbf{0} & \mathbf{0} & M_2'-M_2\end{pmatrix}\right) 
    =k_1+k_2+\rk(M_1'-M_1)+\rk(M_2'-M_2),
  $$  
  so that $\rk(M_1'-M_1)+\rk(M_2'-M_2)\ge \dr(M_1,M_1')+\dr(M_2,M_2')\ge d/2$ implies $\ds(W,W')\ge d$.
 \end{proof}
%% After a suitable permutation of the columns a matrix description of the block inserting construction II is given by  
%% \begin{center}
%%   \begin{tabular}{|C{1.5cm}|C{1.5cm}|C{1.5cm}|} 
%%     \hline
%%     \begin{tabular}{c}
%%     \texttt{C}\\
%%     \hline
%%      \texttt{R} & \texttt{0}\\ 
%%     \hline
%%     \texttt{0} & \texttt{R} & \texttt{C}\\
%%     \hline
%%   \end{tabular}
%% \end{center}
%% \begin{center}
%%   The middle two boxes have to be converted in a single box covering both rows.\\ 
%%  Discuss the disadvantage of the coset construction since $d_1$ and $d_2$ are chosen in an asymmetric manner while the other parameters are symmetric.
%%\end{center}
\begin{ncorollary}
  Let $\cW$ be a {\cdc} constructed via the block inserting construction II in Theorem~\ref{theorem_block_inserting_II} with parameters $\left(n_1,n_2,n_3,n_4,d,k_1,k_2\right)$  
  of maximum possible cardinality. Then, we have
  $$
    \#\cW \ge A_q(n_2,d;k_1)\cdot A_q(n_4,d;k_2)\cdot A_q^R(k_1\times n_1,k_2\times n_3,\le k_1+k_2-d/2).
  $$
\end{ncorollary}
\begin{nexample}
  \label{example_using_pair_rmc}
  Let $\cW$ be the $(12,6;6)_q$--{\cdc} obtained via the block inserting construction II in Theorem~\ref{theorem_block_inserting_II} 
  with parameters $\left(n_1,n_2,n_3,n_4,d,k_1,k_2\right)=(3,3,3,3,6,3,3)$ of maximum possible cardinality. Since $A_q(3,6;3)=1$ we can 
  assume $\#\cM\ge A_q(3\times 3,3\times 3,3,\le 3)\ge q^5 + q^4 + 2q^3 - q^2 - q$ using Example~\ref{example_pair_rmc} for the later estimation. 
\end{nexample}
We remark that the variant of the block inserting construction II in \cite[Theorem 5]{lao2021new} gives a subcode of cardinality $q^5 + q^4 + q^3 - q^2 - q$, 
i.e., $q^3$ codewords less. 
 
Note that the upper rank bounds for the matrices in $\cM$ are not necessary in the proof of Theorem~\ref{theorem_block_inserting_II}
\begin{nlemma}\textbf{(Generalized linkage constr.\ $+$ block inserting construction I,II)}\\
  \label{lemma_combine_generalized_linkage_block_inserting_I_II}
  Let $\cW_1$ be a {\cdc} constructed via the generalized linkage construction in Theorem~\ref{theorem_generalized_linkage} with parameters $(n_1',n_2',d,k)$,  
  $\cW_2$ be a {\cdc} constructed via the block inserting construction I in Theorem~\ref{theorem_block_inserting} with parameters 
  $(n_1,n_2,n_3,n_4,d_1,d_2,k_1,$ $k_2)$, and $\cW_3$ be a {\cdc} constructed via the block inserting construction II in Theorem~\ref{theorem_block_inserting_II} with parameters 
  $\left(n_1,n_2,n_3,n_4,d,k_1,k_2\right)$. If $n_1'=n_1+n_2$, $n_2'=n_3+n_4$, $d=d_1+d_2$, $k=k_1+k_2$, $k_1\ge d/2$, and $k_2\ge d/2$, then $\cW=\cW_1\cup \cW_2\cup \cW_3$ is an
  $\left(n_1'+n_2',d;k\right)_q$--{\cdc} with cardinality $\#\cW=\#\cW_1+\#\cW_2+\cW_3$.
 \end{nlemma}
\begin{proof}
  From Lemma~\ref{lemma_combine_generalized_linkage_block_inserting} we conclude that $\cW':=\cW_1\cup\cW_2$ is an $\left(n_1'+n_2',d;k\right)_q$--{\cdc} with cardinality $\#\cW'=\#\cW_1+\#\cW_2$.
  So, let
  $$
    H_3=\begin{pmatrix} M_1 & G_1 & \mathbf{0} & \mathbf{0}\\ \mathbf{0} & \mathbf{0} & M_2 & G_2\end{pmatrix}=:\begin{pmatrix}P_1\\P_2\end{pmatrix}
  $$  
  be the generator matrix of an arbitrary codeword $W_3\in\cW_3$, $U_1:=\left\langle P_1\right\rangle$, $U_2:=\left\langle P_2\right\rangle$, and $E_1,E_2$ be the special 
  subspaces for $\cW_1$ as in Lemma~\ref{lemma_additional_codewords_generalized_linkage_construction}. Since 
  $\dim(W_3\cap E_1)\ge \dim(U_1\cap E_1)=\rk(G_1)=k_1\ge d/2$ and $\dim(W_2\cap E_2)\ge \dim(U_2\cap E_2)=\rk(G_2)=k_2\ge d/2$ we have  
  $\ds(\cW_1,W_3)\ge d$ by  Lemma~\ref{lemma_additional_codewords_generalized_linkage_construction}.  
  
  Now let   
  $$
    H_2=\begin{pmatrix} G_1' & M_1' & \mathbf{0} & M_3'\\ \mathbf{0} & M_4' & G_2' & M_2'\end{pmatrix}
  $$   
  be the generator matrix of an arbitrary codeword $W_2\in\cW_2$. Observe that the pivot vector $v(H_3)$ of $H_3$ is contained in 
  $\left({n_1 \choose k_1},{n_2\choose 0},{n_3\choose k_2},{n_4\choose 0}\right)$. Since $\rk(M_1)+\rk(M_2)\le k_1+k_2-d/2$ we have 
  $\dH\!\left(v(H_3),v(H_2)\right)\ge d$, so that $\ds(W_3,W_2)\ge d$. 
\end{proof}
\begin{nexample}
  \label{example_12_6_6_generalized_linkage_and_block_inserting_I_II}
  Let $\cW_1$ be a {\cdc} constructed via the generalized linkage construction in Theorem~\ref{theorem_generalized_linkage} with parameters $(n_1,n_2,d,k)=(6,6,6,6)$,  
  $\cW_2$ be a {\cdc} constructed via the block inserting construction I in Theorem~\ref{theorem_block_inserting} with parameters 
  $\left(n_1,n_2,n_3,n_4,d_1,d_2,k_1,k_2 \right)=(3,3,3,3,2,4,3,3)$, and $\cW_3$ be a {\cdc} constructed via the block inserting construction II in
  Theorem~\ref{theorem_block_inserting_II} with parameters $(n_1,n_2,n_3,n_4,d,k_1,k_2)=(3,3,3,3,6,3,3)$. Then, considering the $(12,6;6)_q$--{\cdc} 
  yields
  \begin{eqnarray*}
    A_q(12,6;6) &\ge& q^{24}+A_q^R(6\times 6,3;\le 3)+q^9+A_q(3\times 3,3\times 3,3;\le 3) \\ 
    &\ge& q^{24} + q^{15} + q^{14} + 2q^{13} + 3q^{12} + 3q^{11} + 3q^{10} + 3q^9 + q^8 - q^7 - 2q^6 \\ &&- 2q^5 - 2q^4 - q^3 - 3q^2 - 2q
  \end{eqnarray*}
  using $A_q^R(6\times 6,3;\le 3)\ge \qbin{6}{3}{q}\!\!\cdot\left(q^6-1\right)+1=q^{15} + q^{14} + 2q^{13} + 3q^{12} + 3q^{11} + 3q^{10} + 2q^9 + q^8 
  - q^7 - 2q^6 - 3q^5 - 3q^4 - 3q^3 - 2q^2 - q$ from Lemma~\ref{lemma_lb_restricted_rank_additive_mrd} and the lower bound for $A_q(3\times 3,3\times 3,3;\le 3)$ 
  from Example~\ref{example_pair_rmc}. For $q=2$ we e.g.\ have $A_2(12,6;6)\ge 16865672$.
\end{nexample}

\section{Combining constant dimension codes geometrically}
\label{subsec_combining_subspace_codes}
So far we have combined generating sets of {\cdc}s and matrices of {\rmc}s in order to obtain generating sets of {\cdc}s. Now we want to describe a different 
possibility how smaller {\cdc}s can be combined to larger {\cdc}s. In \cite{cossidente2020subspace} the authors combined several $(6,4;3)_q$-{\cdc}s to show 
$A_q(9,4;3)\ge q^{12} + 2q^8 + 2q^7 + q^6 + q^5 + q^4 +1$, which improves upon the previously best known lower bound $A_q(9,4;3)\ge q^{12}+ 2q^8+ 2q^7+q^6+ 1$, which 
was obtained from the improved linkage construction. In \cite{kurz2020subspaces} the mentioned lower bound was further improved to $A_q(9, 4; 3) \ge q^{12} + 2q^{8} 
+ 2q^{7} + q^6 + 2q^5 + 2q^4 - 2q^2 - 2q + 1$. Here we want to present the generalization of this approach as introduced in \cite{cossidente2021combining}. 
The idea is to use a {\cdc} $\cC\subseteq\cG_q(k+t,k)$ and an $s$-space $S$ outside of $\PG(k+t-1,k)$, i.e., we want to use $\PG(k+t-1,q)\times S\cong \PG(k+s+t-1,q)$ 
as ambient space of the resulting {\cdc}. For each codeword $U\in \cC$ we consider the $(k+s)$-space $D:=U\times S\cong\PG(k+s-1,q)$. In $D$ we can choose 
an $(k+s,d;k)_q$--{\cdc} that contains $U$ as a specific codeword and whose codewords intersect $S$ in at most a certain dimension. More precisely, we assume 
that we have a list of choices for the chosen {\cdc} in $D$.     
 
 \begin{ndefinition}
  \label{def_ndk_sequence}
  An \emph{$(n,d,k)$-sequence} of {\cdc}s is a list $\left(\cD_0,\dots,\cD_r\right)$ of $(n,d;k)_q$-{\cdc}s such that for each index $0\le i\le r$ there exists 
  a codeword $U\in\cD_i$ and a disjoint $(n-k)$-subspace $S$ such that $\dim(U'\cap S)\le i$ for all $U'\in\cD_i$, where $r=k-\tfrac{d}{2}$.
\end{ndefinition} 

We remark that an {\lmrd} code gives an example for $\cD_0$ and for $\cD_i$, with $i \ge 1$, we can take $\cD_0$. Another possibility is to start with 
an arbitrary $(n,d;k)_q$-{\cdc}, pick the special subspace $S$, and remove all codewords whose dimension of the intersection with $S$ is too large. 

Assume that $U$ and $U'$ are two different codewords of $\cC$ and $D=U\times S$ and $D'=U'\times S$ are the corresponding $(k+s)$-spaces into which we 
insert codewords from a $(k+2,d;k)_q$--{\cdc}. If $U$ and $U'$ have a relatively large dimension of their intersection, so have $D$ and $D'$. In order to 
guarantee a minimum subspace distance of at least $d$ between a codeword in $D$ and a codeword in $D'$, we can reduce the allowed dimension of the 
intersection of the codewords with $S$. To this end we introduce: 
\begin{ndefinition}
  \label{def_distance_partition}
  A list $\left(\cC_0,\dots,\cC_r\right)$ is called a \emph{distance-partition} of an $(n,d;k)_q$--{\cdc} $\cC$, where $r=k-\tfrac{d}{2}$, if 
  $\cC_0,\dots,\cC_r$ is a partition of $\cC$ and $\bigcup_{j=0}^i\cC_j$ is an $(n,2k-2i;k)_q$--{\cdc} for all $0\le i\le r$. 
\end{ndefinition}  
A trivial distance-partition of an $(n,d;k)_q$--{\cdc} $\cC$ is given by $(\emptyset,\dots,\emptyset,\cC)$. A subcode $\cC'\subseteq \cC$ with 
maximal subspace distance $d=2k$ is called a \emph{partial-spread subcode}. Given such a partial-spread subcode $\cC'$, if $d<2k$, then 
$(\cC',\emptyset,\dots,\emptyset, \cC\backslash\cC')$ is a distance-partition of $\cC$.

\begin{nlemma}(\cite[Lemma 5.3]{cossidente2021combining})
  \label{lemma_construction_4}
  Let $\left(\cC_0,\dots,\cC_r\right)$ be a distance-partition of a $(k+t,d;k)_q$--{\cdc} $\cC$ and $\left(\cD_0,\dots,\cD_r\right)$ be 
a $(k+s,d,k)$-sequence, where $r=k-\tfrac{d}{2}$. If $\cA$ is an $(s,d;k)_q$--{\cdc}, then there exists a $(k+s+t,d;k)_q$--{\cdc} $\cC'$ with cardinality
  $$
    \#\cC'=\#\cA+\sum_{i=0}^r \#\cC_i\cdot\#\cD_{r-i}.
  $$
\end{nlemma}
Here $\cA$ is a {\cdc} that we can insert into the special subspace $S$ and the combination of codewords in $\cC_i$ with {\cdc} $\cD_{r-i}$ ensures that 
the subspace distance between a codeword of the resulting {\cdc} in $D$ and a codeword in $D'$, using the notation from above, has a subspace distance 
of at least $d$. For more details we refer to the proof of \cite[Lemma 5.3]{cossidente2021combining}. 

As examples we describe the application of Lemma~\ref{lemma_construction_4} for the construction of {\cdc}s reaching the lower bound for $A_q(9,4;3)$ and  
$A_q(10,4;3)$ presented in \cite{kurz2020subspaces}. First we construct a $(6,4,3)$-sequence $\left(\cD_0,\cD_1\right)$. Here we choose $\cD_0$ as an 
{\lmrd} code of cardinality $q^6$ and $\cD_1$ as a $(6,4;3)_q$-{\cdc} with cardinality $q^6+2q^2+2q$. The latter needs a bit more explanation. Choose a $(6,4;3)_q$--{\cdc} 
$\cD_1'$ of cardinality $q^6+2q^2+2q+1$, see \cite{cossidente2016subspace,hkk77}, and assume that $U$ and $S$ are two disjoint codewords. Here $U$ and $S$ 
have the same meanings as above, i.e., $U$ is a special codeword and $S$ is the special subspace used in the construction of the $(s+k)$-space $D=U\times S$. With 
this let $\cD_1$ arise from $\cD_1'$ by removing the codeword $S$. Since $\cD_1'$, as well as $\cD_1$, is a $(6,4;3)_q$--{\cdc} every codeword of $\cD_1$ has 
an intersection of dimension at most $1$ with $S$, which is what we need according to Definition~\ref{def_ndk_sequence}.     

For $A_q(9, 4; 3)$, we choose the {\cdc} $\cC$ needed in Lemma~\ref{lemma_construction_4} as a $(6,4;3)_q$--{\cdc} with cardinality $q^6+2q^2+2q+1$, see 
\cite{cossidente2016subspace,hkk77}. In order to determine a distance-partition $\left(\cC_0,\cC_1\right)$ of $\cC$, we need to find a large partial-spread 
subcode of $\cC$. In \cite[Theorem 3.12]{cossidente2020subspace}, it is shown that we can choose $\cC_0$ of cardinality $q^3-1$ if we choose $\cC$ as constructed in 
\cite{cossidente2016subspace}. However, as shown in \cite{kurz2020subspaces}, the same can be done if we choose $\cC$ as constructed in \cite{hkk77}.\footnote{
This can be made more precise in the language of linearized polynomials.  For \cite[Lemma 12, Example 4]{hkk77} the representation $\F_q^6\cong \F_{q^3}\times \F_{q^3}$ 
is used and the planes removed from the lifted {\mrd} code correspond to $ux^q-u^qx$ for $u\in\F_{q^3}$, so that the monomials $ax$ for $a\in\F_{q^3}\backslash\{\mathbf{0}\}$ 
correspond to a partial-spread subcode of cardinality $q^3-1$.} As subcode $\cA$ we choose a single $3$-space, so that we obtain
\begin{eqnarray*}
  A_q(9,4;3)&\ge& 1+\#\cC_0\cdot\#\cD_1+\#\cC_1\cdot\#\cD_0 \\ 
  &=&1+\left(q^3-1\right)\cdot\left(q^6+2q^2+2q\right)+\left(q^6-q^3+2q^2+2q+2\right)\cdot q^6 \\ 
  &=& q^{12} + 2q^{8} + 2q^{7} + q^6 + 2q^5 + 2q^4 - 2q^2 - 2q + 1.
\end{eqnarray*}  

For $A_q(10,4;3)$ we choose $\cC$ as the $(7,4;3)_q$--{\cdc} of cardinality $q^8+q^5+q^4+q^2-q$ constructed in \cite[Theorem 4]{honold2016putative}. Again we need to find 
a large partial-spread subcode $\cC_0$ of $\cC$. Here $\#\cC_0=q^4$ can be achieved, see \cite{kurz2020subspaces}. 
Thus, we obtain
\begin{eqnarray*}
  A_q(10,4;3)&\ge& 1+\#\cC_0\cdot\#\cD_1+\#\cC_1\cdot\#\cD_0 \\
  &=& 1+q^4\cdot\left(q^6+2q^2+2q\right)+\left(q^8+q^5+q^2-q\right)\cdot q^6 \\
  &=& q^{14} + q^{11} + q^{10} + q^8 - q^7 + 2q^6 + 2q^5 + 1.  
\end{eqnarray*}    

The determination of a large partial-spread subcode is mostly the hardest part in the analytic evaluation of the construction of Lemma~\ref{lemma_construction_4}. However,  
if $\mathcal{C}$ contains an $(n,d;k)$--{\cdc} that contains an {\lmrd} code as a subcode, then it contains an $(n,2k; k)$--{\cdc}  as a subcode that is again an {\lmrd} code, 
i.e., a partial-spread subcode.

\begin{question}{Research problem}Determine large partial-spread subcodes for constructions of {\cdc}s from the literature.
\end{question}
We remark that Lemma~\ref{lemma_construction_4} was used in \cite{cossidente2021combining} to construct lower bounds for $A_q(3k,4;k)$, where $k\ge 3$, for $A_q(16,4;4)$, 
and for $A_q(6k,2k;2k)$, where $k\ge 4$ is even.

\begin{question}{Research problem}Use Lemma~\ref{lemma_construction_4} for the construction of large {\cdc}s for further parameters or improve the known constructions.
\end{question}

\section{Other constructions for constant dimension codes}
\label{subsec_other_constructions}
The list of constructions for {\cdc}s presented in the previous subsections is far from being exhaustive. There are several constructions based on geometric concepts, see 
e.g.\ \cite{cossidente2018geometrical} for an overview and e.g.\ \cite{cossidente2016veronese,cossidente2017subspace}. As examples we mention two explicit and rather general 
parametric lower bounds.
\begin{ntheorem}(\cite[Theorem 3.8]{cossidente2017subspace})\\
If $n \ge 4$ is even, then $A_q(2n,4;n) \ge$
\begin{eqnarray*}
 q^{n^2-n} + \sum_{r=2}^{n-2} \qbin{n}{r}{q} \sum_{j=2}^{r} (-1)^{(r-j)} \qbin{r}{j}{q} q^{\binom{r-j}{2}}(q^{n(j-1)}-1) \,+\, 
\qbin{\frac{n}{2}}{1}{q^2} \left( \qbin{\frac{n}{2}}{1}{q^2} - 1 \right)
\\+ (q+1) \left( \prod_{i=1}^{n-1} (q^i+1) - 2q^{\frac{n(n-1)}{2}} + q^{\frac{n(n-2)}{4}} \prod_{i=1}^{\frac{n}{2}} (q^{2i-1}-1) \right) - q \cdot |G|  + 1
\end{eqnarray*}
using $$|G| = 2 \prod_{i=1}^{n/2-1}(q^{2i}+1) - 2q^{(n(n-2)/4)}$$ if $n/2$ is odd and 
$$|G| = 2 \prod_{i=1}^{n/2-1}(q^{2i}+1) - 2q^{(n(n-2)/4)} + q^{n(n-4)/8}\prod_{i=1}^{n/4}(q^{4i-2}-1)$$ if $n/2$ is even.
\end{ntheorem}

\begin{ntheorem}(\cite[Theorem 3.11]{cossidente2017subspace})\\
If $n \ge 5$ is odd, then $A_q(2n,4;n) \ge$
\begin{eqnarray*}
q^{n^2-n} + \sum_{r=2}^{n-2} \qbin{n}{r}{q} \sum_{j=2}^{r} (-1)^{(r-j)} \qbin{r}{j}{q} q^{\binom{r-j}{2}}(q^{n(j-1)}-1)\,+\,y(y-1) + 1\\ 
 + \prod_{i=1}^{n-1} (q^i+1) - q^{\frac{n(n-1)}{2}} - \qbin{n}{1}{q} \left( q^{\frac{(n-1)(n-2)}{2}} - q^{\frac{(n-1)(n-3)}{4}} \prod_{i=1}^{\frac{n-1}{2}} (q^{2i-1}-1) \right) ,
\end{eqnarray*} 
using $y:=q^{n-2}+q^{n-4}+\dots+q^3+1$.
\end{ntheorem}

Riemann--Roch spaces can be used to construct 
{\cdc}s, see \cite{bartoli2015constant,hansen2015riemann}. Removing and replacing codewords of lifted {\mrd} codes was the basis of a few specific constructions, 
see \cite{honold2016putative,hkk77}. An entire theoretic framework for such approaches was introduced in \cite{ai2016expurgation}. For {\mrd} codes linearity plays an important 
and natural role. A variant of the concept is considered in \cite{braun2013linearity}, see also \cite{pai2015bounds}. Another well studied class are so-called 
cyclic subspace codes, see e.g.\ \cite{ben2016subspace,chen2018constructions,lehmann2021weight,niu2020several,otal2017cyclic,otal2018constructions,roth2017construction}. In principle 
one can start with any construction of a {\cdc} and check if it can be extended by further codewords. This approach was e.g.\ successful for the 
$(7,4;3)_q$--{\cdc} of cardinality $6977$ constructed in \cite{honold2016putative}. Here, an extension by an additional codeword was possible, so that 
$A_3(7,4;3)\ge 6978$, see \cite{TableSubspacecodes}. However, even for moderate parameters $\qbin{n}{k}{q}$ gets huge rather soon, so that one faces algorithmical problems. 
In \cite{zhou2021construction} the extension problem is restricted to the set $\cC_1$, $\cC_2$ of codewords of two {\cdc}s. More precisely, the problem 
of the determination of the largest {\cdc} with codewords in $\cC_1\cup\cC_2$ was formulated as a minimum point-covering problem for a bipartite graph that can be solved in 
polynomial time. As example the improved lower bounds $A_2(8, 4; 3) \ge 1331$ and $A_2(8, 4; 4)\ge 4802$ were obtained in \cite{zhou2021construction}.     

\chapter{On the existence of a binary $q$-analog of the Fano plane}
\label{sec_fano}
For the binary case $q=2$ the smallest unknown value $A_q(n,d;k)$ is $A_2(7,4;3)$. Inequality~(\ref{ie_j_o}) of the Johnson bound gives
$$
  A_2(7,4;3)\le \left\lfloor \frac{127\cdot A_2(6,4;2)}{7}\right\rfloor=381
$$
since $A_2(6,4;2)=21=3\cdot 7$ due to the existence of a $2$-spread in $\PG(6,2)$. Also the improved Johnson bound in Theorem~\ref{thm:johnson_improved} cannot give a
tighter bound since in a $(7,4;3)_2$--{\cdc} $\cC_{381}$ of cardinality $381$ every point is contained in exactly $21$ codewords. Also the anticode bound yields the 
upper bound $A_2(7,4;3)\le \qbin{7}{2}{2}/\qbin{3}{2}{2}=381$, so that any line is contained in exactly one codeword of $\cC_{381}$. If $\cC_{381}$ exists, then it 
is a so-called \emph{$q$-design} and called \emph{binary $q$-analog of the Fano plane}. 
\begin{nexercise}
  Show $\#\left\{ U\in \cC_{381}\,:\, U\le K\right\}=5$ and $\#\left\{ H\in \cC_{381}\,:\, U\le K\right\}=45$ for each $K\in\cG_2(7,5)$ and each hyperplane  
  $H\in\cG_2(7,6)$. For each point $P$ and each hyperplane $H$ with $P\le H$ show that $\#\left\{U\in\cC_{381}\,:\, P\le U\le H\right\}=5$.
\end{nexercise}

\begin{ntheorem}(\cite[Theorem 1]{kiermaier2018order})\\
  \label{thm_fano_aut}
	The automorphism group of a binary $q$-analog of the Fano plane is either trivial or of order $2$.
	In the latter case, up to conjugacy in $\operatorname{GL}(7,2)$ the automorphism group is represented by
\[
	\left\langle\left(
	\begin{smallmatrix}
	0& 1& 0& 0& 0& 0& 0 \\
	1& 0& 0& 0& 0& 0& 0 \\
	0& 0& 0& 1& 0& 0& 0 \\
	0& 0& 1& 0& 0& 0& 0 \\
	0& 0& 0& 0& 0& 1& 0 \\
	0& 0& 0& 0& 1& 0& 0 \\
	0& 0& 0& 0& 0& 0& 1 
	\end{smallmatrix}\right)\right\rangle\text{.}
\]
\end{ntheorem}

For each solid $S\in\cG_2(7,4)$ we have $\#\left\{U\in\cC_{381}\,:\,U\le S\right\}\in\{0,1\}$. For the group $G$ of order two in Theorem~\ref{thm_fano_aut} 
there are exactly $15$ fixpoints, i.e.\ points $P$ such that the $P^g=P$ for all $g\in G$, where $P^g$ denotes the application of the group element $g$ to $P$. 
These $15$ fixpoints form a special solid $\bar{S}=\left\langle \be_1+\be_2,\be_3+\be_4,\be_5+\be_6,\be_7\right\rangle$. The $\qbin{4}{2}{2}=35$ lines 
in $\bar{S}$ clearly are fixed by $G$. The other $56$ fixed lines are given by $L=\left\langle P,P^g\right\rangle$, where $P$ is a arbitrary point outside 
of $\bar{S}$, so that $L$ intersects $\bar{S}$ in a point. Let $\cB_2$ denote the $91$ fixed lines . It is a bit more tedious to check that there are exactly 
$211$ planes that are fixed by $G$. Let $\cB_3$ denote these fixed planes. Note that in $\cC_{381}$ each fixed line must be contained in a codeword $U$ that 
is fixed by $G$, i.e.\ $U\in\cB_3$.
\begin{nexercise}
  Verify
  $$
    \frac{1}{7}\cdot \sum_{L\in\cB_2\,:\, L\le \bar{S}} \,\,\sum_{U\in \cB_3\,:\, L\le U} x_U\,+\,
    \frac{3}{7}\cdot \sum_{L\in\cB_2\,:\, L\not\le \bar{S}} \,\,\sum_{U\in \cB_3\,:\, L\le U} x_U =\sum_{U\in\cB_3} x_U.
  $$
\end{nexercise}
%% /btmw06/home_laptop/Tuvi/Fano$ g++ gen_ilp_fix_points.cpp
%% Email-Ornder Codes vom 17.11.2018 bzw. 03.10.2017 
Since through each line there is at most one codeword, we have $\#\left(\cC_{381}\cap\cB_3\right)\le \tfrac{1}{7}\cdot 35+\tfrac{3}{7}\cdot 56=29$. On the other 
hand the $35$ lines in $\bar{S}$ each have to be contained in a codeword from $\cB_3$, so that there exists a codeword in $\cB_3$ that is contained in 
$\bar{S}$ and there are $28$ codewords in $\cB_3$ that intersect $\bar{S}$ in a line each. Of course, this little insight does not exclude the existence of 
a {\cdc} $\cC_{381}$ with $G$ as automorphism group.
\begin{nexercise}
  Assume that $G$ is a subgroup of the automorphism group of $\cC_{381}$. Show that the set $\cF$ of fixed points with respect to $G$ is a subspace.  
  Determine restrictions for the possible dimension of $\cF$ for $\# G\in \{2,3,5,7,31,127\}$. 
\end{nexercise}
\begin{question}{Research problem}Decide whether there exist $240$ planes in $\PG(6,2)$ and an automorphism $\pi$ of order $5$ such that all planes are disjoint 
to the $3$-space $\cF$ of points fixed by $\pi$, no two planes intersect in a line, and each point outside of $\cF$ is covered $15$ times.
\end{question}
We remark that the {\lq\lq}complementary set{\rq\rq}, admitting $\pi$ as automorphism, consisting of $\cF$ and $140$ further planes intersecting $\cF$ in a point, 
such that no line is covered twice indeed exists. 

\medskip

In \cite{paper333} $A_2(7,4;3)\ge 333$ was shown. The constructed code has an automorphism group of order $4$ isomorphic to the Klein four-group. We remark 
that the corresponding code contains a subcode of cardinality $329$ that admits an automorphism group of order $16$.    
\begin{ntheorem}(\cite[Theorem 1]{paper333})\\\label{main_thm_1_fano_param}
Let $\cC$ be a set of planes in $\PG(6,2)$ mutually intersecting in at most a point.
If $\#\cC\geq 329$, then the automorphism group of $\cC$ is conjugate to one of the $33$ subgroups of $\operatorname{GL}(7,2)$ given in \cite[Appendix B]{{paper333}}. 
The orders of these groups are $1^1 2^1 3^2 4^7 5^1 6^3 7^2 8^{11} 9^2 12^1 14^1 16^1$ denoting the number of cases as exponent.
Moreover, if $\#\cC \geq 330$ then $\# \operatorname{Aut}(\cC) \leq 14$ and if $\# \cC  \geq 334$ then $\#\operatorname{Aut}(\cC) \leq 12$.
\end{ntheorem}
Interestingly enough, it is not necessary to generate all subgroups of $\operatorname{GL}(7,2)$ of order at most $16$ up to conjugacy to obtain the stated results, 
see \cite{paper333} for the algorithmic details. In \cite[Section 10]{phd_heinlein} parametric upper bounds for {\cdc}s that admit certain automorphisms are concluded. 
The group of order $12$ mention in Theorem~\ref{main_thm_1_fano_param}, that might allow a larger $(7,4;3)_2$--{\cdc}, is given by:
$$
G_{12,1}
=
\left\langle
\left(\begin{smallmatrix}
1&0&0&0&0&1&1\\
0&0&0&1&1&0&1\\
1&1&1&1&1&0&0\\
1&1&0&0&1&1&0\\
0&0&0&0&0&0&1\\
0&0&0&0&1&1&1\\
0&0&0&0&1&0&0
\end{smallmatrix}\right)
,
\left(\begin{smallmatrix}
1&0&0&0&0&0&0\\
1&1&0&0&0&1&1\\
1&0&1&0&1&0&1\\
1&0&0&1&0&0&0\\
0&0&0&0&1&0&0\\
0&0&0&0&0&1&0\\
0&0&0&0&0&0&1
\end{smallmatrix}\right)
,
\left(\begin{smallmatrix}
1&0&0&0&0&1&1\\
0&1&0&1&1&1&1\\
1&0&1&1&1&0&0\\
1&1&0&0&0&1&1\\
1&0&0&0&1&0&0\\
1&0&0&0&0&1&0\\
0&0&0&0&0&1&0
\end{smallmatrix}\right)
\right\rangle
 \cong \Z_{3} \rtimes \Z_{4}.
$$

In \cite{nakic2016extendability} it was shown that each hypothetical $(7,4;3)_2$--{\cdc} of cardinality $380$ can be extended to a {\cdc} of cardinality $381$.
Using divisible codes it was shown that either $A_2(7,4;3)\le 378$ or $A_2(7,4;3)=381$. 

For each point $P\in\cG_2(7,1)$ the subcode $\cC_P:=\left\{U\in\cC_{381}\,:\,P\le U\right\}$ gives rise to a $2$-spread 
$\cC_P/P:=\{U/P\,:\,U\in\cC_P\}$ in $\PG(6,2)/P\cong\PG(5,2)$. In our situation it is called \emph{geometric} if for any two spread lines 
$L$ and $L'$, the restriction of the $2$-spread to the $4$-space $\left\langle L,L'\right\rangle$ is a $2$-spread itself, i.e., $5$ lines are contained. 
Call every point $P$ such that $\cC_P/P$ is geometric an \emph{$\alpha$-point}. In \cite{thomas1996designs} it was shown that, even for general 
field sizes $q$, there always exists a non-$\alpha$ point $\bar{P}$ in a $q$-analog of the Fano plane. For a binary $q$-analog of the Fano plane the result 
was tightened to the existence of at least one non-$\alpha$ point in every hyperplane \cite{heden2016existence}. Recently this result was generalized 
to all prime or even field sizes $q$ in \cite{kiermaier2021alpha}. 
Here we want to consider a relaxation. Let $\cC\subseteq\cG_2(6,3)$ such that 
\begin{itemize}
  \item every $5$-space contains exactly five elements of $\cC$;
  \item every point is contained in exactly five elements of $\cC$;
  \item each line is contained in at most one element of $\cC$; and 
  \item each solid contains at most one element of $\cC$.
\end{itemize}   
Such sets of $5$-spaces indeed exist and have cardinality $\#\cC=45$, cf.~\cite{etzion2017residual} for general field sizes and the existence of induced substructures 
of a $q$-analog of the Fano plane. We call a point $P$ an $\alpha'$ point if the five elements of $\cC$ that are incident with $P$ span a $5$-space (and not the entire 
$6$-dimensional ambient space). Using an ILP formulation of the problem one can computationally show that the maximum number of $\alpha'$ points in a fixed 
$5$-space lies between $15$ and $22$. The total number of $\alpha'$ points lies between $19$ and $44$.

\begin{question}{Research problem}Determine the maximum number of $\alpha'$ points.
\end{question}

For certain infinite fields a {\lq\lq}$q$-analog of the Fano plane{\rq\rq} indeed exists, see \cite{van2020cayley}. In $\PG(6,q)$ the existence question or the maximum 
possible size $A_q(7,4;3)$ of a {\cdc} with these parameters is still widely open.

From the improved Johnson bound we conclude
$$
  A_q(8,4;4) \le \bigllfloor \frac{\frac{[8]_q}{[4]_q}\cdot A_q(7,4;3)}{[4]_q}\bigrrfloor_{q^{3}}\!\!\!\!.
$$     
If we cannot improve upon $A_q(7,4;3)\le \qbin{7}{2}{q}/\qbin{3}{2}{q}$, then this upper bound is equivalent to $A_q(8,4;4)\le \qbin{8}{3}{q}/\qbin{4}{3}{q}$, i.e., 
the anticode bound. For $q=2$ we obtain $A_2(8,4;4)\le 6477$. However, if such a code $\cC$ of cardinality $6477$ exists, then for each point $P$ the 
set of codewords of $\cC$ that contain $P$ would be a binary $q$-analog of the Fano plane.

\chapter{Lower bounds for constant dimension codes}
\label{sec_lower_bounds_cdc}
In this section we summarize the currently best known lower bounds for constant dimension codes. For subspace distance $d=2$ we can choose 
$\cC=\cG_q(n,k)$, so that $A_q(n,d;k)=\qbin{n}{k}{q}$. In general we have $A_q(n,d;k)=A_q(n,d;n-k)$. Thus we assume $4\le d\le 2k$, $d\equiv 0\pmod 2$, 
and $2\le k\le n/2$. For the dimension of the ambient space we restrict our consideration to $4\le n \le 9$ and a few selected triples $(n,d,k)$. We also treat 
the case $d=2k$, i.e.\ the case 
of (partial) $k$-spreads separately. If $n\equiv 0 \pmod k$, then $k$-spreads indeed exist and we have $A_q(n,2k;k)=[n]_q/[k]_q$, see Theorem~\ref{thm_spread}. 
For the cases where $n\not\equiv 0\pmod k$ we have used the Echelon--Ferrers construction to conclude a general lower bound in Exercise~\ref{exercise_lb_partial_spreads}, 
see also Inequality~(\ref{eq_lb_partial_spread}):
$$
  A_q(tk+r,2k;k)\ge \sum_{s=0}^{t-1}q^{sk+r}-(q^r-1),
$$
where $k,t\ge 2$ and $0\le r\le k-1$. The only known improvement is  
$$
  A_2(3t+2,6;3)\ge \sum_{s=0}^{t-1}2^{3s+2}-(2^2-1)+1,
$$
for arbitrary $t\ge 2$, see \cite{spreadsk3}. For upper bound for partial spreads much more can be said, see Subsection~\ref{subsec_bounds_partial_spreads}. 
For small parameters the known lower and upper bounds coincide. E.g.\ we have $A_q(4,4;2)=q^2+1$, $A_q(5,4;2)=q^3+1$, $A_q(6,4;2)=q^4+q^2+1$, $A_q(6,6;3)=q^3+1$, 
$A_q(7,4;2)=q^5+q^3+1$, and $A_q(7,6;3)=q^4+1$. For $A_q(8,6;3)$ the exact value is known for $q=2$ only. In the following we will discard the partial spread case 
and assume $d<2k$.
  
For the smallest parameters we have
\begin{equation}
  A_q(6,4;3)\ge q^6+2q^2+2q+1,
\end{equation}
see \cite{hkk77,cossidente2016subspace} for constructions. We remark that the lower bound is tight for $q=2$ \cite{hkk77}. 
For $A_q(7,4;3)$ a lower bound for general $q$ was given in \cite[Theorem 4]{honold2016putative}. For $q=2$ an improved lower bound was found 
via extensive ILP computations in \cite{heinlein2019subspace} and for $q=3$ it was observed that a theoretical construction can be extended by one 
further codeword, so that we have    
\begin{equation}
  A_q(7,4;3)\ge q^8 +q^5 +q^4 +q^2-q, \quad A_2(7,4;3)\ge 333, A_3(7,4;3)\ge 6978.
\end{equation}
The constructions for $A_q(6,4;3)$ and $A_q(7,4;3)$ from \cite{hkk77} and \cite{honold2016putative} can be described within the framework of the so-called 
\emph{expurgation-augmentation method}, see \cite{ai2016expurgation}, where specially selected codewords are removed from a lifted {\mrd} code in order to 
allow the augmentation with more codewords than removed before.    

Construction 1, see Theorem~\ref{theorem_construction_1} or \cite[Chapter IV, Theorem 16]{etzion2012codes} gives
\begin{equation}
  A_q(8,4;3)\ge q^{10}+\qbin{5}{2}{q}=q^{10}+q^6+q^5+2q^4+2q^3+2q^2+q+1.
\end{equation}
For $q=2$ the improved lower bound $A_2(8,4;3)\ge 1326$ was found via the prescription of automorphisms. 

The lower bound
\begin{eqnarray} 
  A_q(8,4;4) &\ge&  q^{12}+\left(q^2+q+1\right)\cdot \left(q^2+1\right)^2+1 \notag\\ 
  &=&q^{12} + q^8 + q^7 + 3q^6 + 2q^5 + 3q^4 + q^3 + q^2 + 1
\end{eqnarray}
is attained by several constructions. One example is the coset construction of Theorem~\ref{theorem_coset}, see Example~\ref{example_sequel_coset_8_4_4} 
for the details. We remark that the stated lower bound is tight if we additionally assume that a lifted {\mrd} is contained as a subcode, see e.g.\ \cite{etzion2012codes}. 
For $q=2$ this bound gives $4797$ as the maximum possible size under this extra assumption. Nevertheless a construction showing $A_2(8,4;4)\ge 4802$ is known \cite{zhou2021construction}. 
It is obtained by extending an $(8,4;4)_2$--{\cdc} with cardinality $4801$, found in \cite{braun2018new} via the prescription of automorphisms, by a single codeword.  

For the skeleton code $\{1111000,00001111\}$ the Echelon--Ferrers construction give the lower bound 
\begin{equation}
  A_q(8,6;4)\ge q^8+1.
\end{equation}
In other words, a corresponding code consists of a lifted {\mrd} code and another codeword. For $q=2$ it was shown in \cite{heinlein2019classifying} that the lower 
bound is indeed tight and that there are exactly two isomorphism types of {\cdc}s attaining the maximum possible cardinality $257$.

The geometric combination of {\cdc}s described in Subsection~\ref{subsec_combining_subspace_codes} yields the lower bound
\begin{equation}
  A_q(9,4;3)\ge q^{12} + 2q^8 + 2q^7 + q^6 + 2q^5 + 2q^4 - 2q^2 -2q + 1, 
\end{equation}
%% [noch nicht in Tables of Subspacecodes] 
see also \cite{cossidente2021combining}. For $q=2$ the tighter bound $A_2(9,4;3)\ge 5986$ was obtained in \cite{braun2018new} via the prescription of automorphisms.

The pending dots construction gives $A_2(9,4;4)\ge 37265$ and
\begin{equation}
  A_q(9,4;4)\ge q^{15} + q^{11} + q^9 + 4 q^8 + 5q^7 +3q^6 +2q^5 +3q^4 +2q^3 +2q^2 +q+1
  %% details unklar
\end{equation} 
for $q\ge 3$. Interestingly enough, for $q\ge 5$ we get a tighter lower bound by reverting the Johnson upper bound from Theorem~\ref{thm_johnson_II}, cf.~\cite{xu2018new},
\begin{equation} 
  A_q(n,d;k)\ge \left\lceil\frac{\left(q^{k+1}-1\right)A_q(n+1,;k+1)}{q^{n+1}-1}\right\rceil.
\end{equation}

\begin{question}{Research problem}Improve the tightest known lower bound for $A_q(9,4;4)$ (and $q\ge 5$) in a constructive manner.  
\end{question}

%% (10,4;3): kurz2020lifted, braun2018new for q=2
%% (10,4;4): pending dots
%% (10,4;5): CKMP2019_Cor_45 -> besseres Packing!!!
%% (10,6;4): Construction 1!!!
%% (10,6;5): EF!! für q=2 Computersuche in Daniel Heinlein: “New LMRD bounds for constant dimension codes and improved constructions” IEEE Transactions on Information Theory 65.8 (2019): 4822-4830. 
%% (10,8;5): LMRD+1

For $A_q(10,4;5)$ an improved lower bound is described in Example~\ref{example_a_10_4_5_generalized_linkage_and_coset}. In Example~\ref{example_generalized_skeleton_code}, see 
also \cite[Proposition 3.1]{kurz2021interplay}, an improved lower bound for $A_q(11,4;4)$ is presented.  
For $A_q(12,4;6)$ improved lower bounds are obtained in Example~\ref{example_a_12_4_6_construction_d_and_coset}, Example~\ref{example_a_12_4_6_generalized_linkage_and_coset}, and Exercise~\ref{exercise_a_12_4_6_generalized_linkage_and_coset}. 
For $A_q(12,6;6)$ and especially $A_2(12,6;6)\ge 16865672$ we refer to Example~\ref{example_12_6_6_generalized_linkage_and_block_inserting_I_II}. 

\chapter{Constructions and bounds for mixed dimension subspace codes}
\label{sec_mdc}
Most parts of this chapter are devoted to lower and upper bounds for {\cdc}s. The analog questions for {\mdc} are also of interest while so far less intensively studied. Here we 
restrict ourselves to the subspace distance and refer to e.g.\ \cite{khaleghi2009projective,silva2009metrics} for the injection metric. In the 
classical situation of block codes in the Hamming metric there are back and forth relations between constant weight codes and their unrestricted versions, i.e., inequalities 
involving both $A(n,d;k)$ and $A(n,d)$ are known. A few, very easy and natural, observations on the relation between $A_q(n,d;k)$ and $A_q(n,d)$ (or $A_q(n,d;T)$ in general) are 
already known, see e.g.\ \cite{honold2016constructions}. The inequality $A_q(n,d;T)\le A_q(n,d;T')$ for $T\subseteq T'$, mentioned in the preliminaries in Section~\ref{sec_preliminaries}, 
e.g.\ directly implies $A_q(n,d;k)\le A_q(n,d)$. In the other direction we can choose $T\subseteq\{0,1,\dots,n\}$ such that the differences between the occurring dimensions 
are sufficiently large with respect to a given minimum subspace distance $d$.    
\begin{ntheorem}\textbf{(Dimension layers -- \cite[Theorem 2.5]{honold2016constructions})}
  \label{thm:bounds:mdc}
  $$
    \sum_{\substack{k=0\\k\equiv\lfloor n/2\rfloor\bmod d}}^n
      A_q\bigl(n,2\lceil d/2\rceil;k\bigr)
    \leq A_q(n,d)\leq 2+\sum_{k=\lceil d/2\rceil}^{n-\lceil d/2\rceil} A_q
    \bigl(n,2\lceil d/2\rceil;k\bigr)
  $$
\end{ntheorem}
We remark that this constitute the best bound for $A_q(n,d)$ that does not depend on
information about the cross-distance distribution between different
{\lq\lq}dimension layers{\rq\rq} $\spaces{V}{k}$ and $\spaces{V}{l}$.
\begin{nlemma}(\cite[Lemma 2.4]{honold2016constructions})\\
  \label{lma:unimodal}
  For $1\leq\delta\leq k\leq\lfloor n/2\rfloor$ the inequality
  \begin{equation*}
  \frac{A_q(n,2\delta;k)}{A_q(n,2\delta;k-1)}
  >q^{n-2k+\delta}\cdot C(q,\delta)
\end{equation*}
holds with $C(q,1)=1$ and $C(q,\delta)=1-1/q$ for $\delta\geq 2$;
in particular, $A_q(n,2\delta;k)>q\cdot
A_q(n,2\delta;k-1)$. As a consequence, the numbers $A_q(n,2\delta;k)$,
$k\in[\delta,v-\delta]$, form a strictly unimodal sequence.
\end{nlemma}
The bounds of Theorem~\ref{thm:bounds:mdc} coincide for $d=1$ where we have
\begin{equation}
  A_q(n,1)=\sum_{k=0}^n A_q(n,2;k)= \sum_{k=0}^n \qbin{n}{k}{q}.
\end{equation}
For minimum subspace distance $d=n$ we have $A_q(n,n)=2$ for odd $n$ and $A_q(n,n)=A_q(n,n;k)=q^k+1$ for $n=2k$, see \cite[Theorem 3.1]{honold2016constructions} and also 
\cite[Section 5]{gabidulin2009algebraic} or \cite{GabidulinBossert}. In the latter case of an even dimension of the ambient space the maximum number of codewords $q^k+1$ can only be attained if all 
codewords have dimension $k$, i.e., the codes are $k$-spreads.

\begin{ntheorem}\textbf{(Dimension layers are optimal for $\mathbf{d=2}$ -- \cite[Theorem 3.4]{honold2016constructions})}
  \label{thm_d_2_mdc}
  \begin{enumerate}
  \item[(i)] If $n=2k$ is even then
    \begin{equation*}
      A_q(n,2)=\sum_{\substack{0\leq i\leq n\\i\equiv 0\bmod 2}}\qbin{n}{i}{q}\!.
    \end{equation*}
    The unique (as a set of subspaces) 
    optimal code in $\PG(n-1,q)$ consists of all subspaces $X$ of $\F_q^n$
    with $\dim(X)\equiv k\bmod 2 $, and thus of all even-dimensional
    subspaces for $n\equiv0\bmod 4$ and of all odd-dimensional
    subspaces for $n\equiv2\bmod 4 $.
  \item[(ii)] $n=2k+1$ is odd then
    \begin{equation}
      \label{eq:kleitman}
      A_q(n,2)=\sum_{\substack{0\leq i\leq n\\i\equiv
          0\bmod 2}}\qbin{n}{i}{q}
      =\sum_{\substack{0\leq i\leq n\\i\equiv
          1\bmod 2}}\qbin{n}{i}{q},
    \end{equation}
    and there are precisely two distinct optimal codes in
    $\PG(n-1,q)$, containing all even-dimensional and all
    odd-dimensional subspaces of $\F_q^n$, respectively. Moreover
    these two codes are isomorphic.
  \end{enumerate}
\end{ntheorem}

If $n=2k$ is even then $A_q(n,n-1)=A_q(n,n;k)=q^k+1$ and $A_q(n,n-1)=A_q(n,n-1;k)=q^{k+1}+1$ if $n=2k+1\geq 5$ is odd, see  \cite[Theorem 3.2]{honold2016constructions}. 
Note that we have to exclude the case $A_q(3,2)=q^2+q+2$, see Theorem~\ref{thm_d_2_mdc}. The case of subspace distance $d=n-2\ge 3$ is much more involved and only 
partial results are known:
\begin{ntheorem}(\cite[Theorem 3.3]{honold2016constructions})\\
  \label{thm_d_v_minus_2_mdc}
  \begin{enumerate}
  \item[(i)]
    If $n=2k\geq 8$ is even then 
    $A_q(n,n-2)=A_q(n,n-2;k)$,
     and the known bound $q^{2k}+1\leq A_q(n,n-2;k)\leq(q^k+1)^2$
    applies. Moreover, $A_q(4,2)=q^4+q^3+2q^2+q+3$ for all $q$,
    $A_2(6,4)=77$ 
    and $q^6+2q^2+2q+1\leq A_q(6,4)\leq(q^3+1)^2$ for all
    $q\geq 3$.
  \item[(ii)] If $n=2k+1\geq 5$ is odd then
    $A_q(n,n-2)\in\{2q^{k+1}+1,2q^{k+1}+2\}$. Moreover,
    $A_q(5,3)=2q^3+2$ 
    for all $q$ and $A_2(7,5)=2\cdot 2^4+2=34$.
  \end{enumerate}
\end{ntheorem}
Note that the bounds for $A_2(n,n-2)$ with odd $n$ were already established in \cite[Theorem~5]{etzion2013problems} and $A_2(5,3)=18$ in \cite[Theorem~14]{MR2810308}. 
Further constructions for $A_q(5,3)=2q^3+2$ are discussed in \cite{cossidente2019optimal,ghatak2017optimal,ghatak2021intersection}. The subspace codes attaining the upper 
bound $A_2(7,5)=34$ were classified up to isomorphism in \cite{honold2019classification}. For $k\ge 3$ it was shown in \cite{honold2016constructions} 
that subspace codes attaining the upper bound $A_q(n,n-2)\in\{2q^{k+1}+1,2q^{k+1}+2\}$ for $n=2k+1$ have to consist of $q^{k+1}+1$ codewords of dimension $k$ and also 
$q^{k+1}+1$ codewords of dimension $k+1$. For dimension $k$ the codewords form a partial $k$-spread of maximum cardinality $A_q(2k+1,2k;k)=q^{k+1}+1$ and for dimension 
$k+1$ the codewords form the dual of such a maximum partial $k$-spread in $\PG(2k,q)$. Some authors also speak of a \emph{doubling construction}.  

\begin{question}{Research problem}Does a doubling construction exist for $k\ge 4$ or for $k=3$ and $q\ge 3$?
\end{question}

Also the proven non-existence of a doubling construction is of interest, since it would yields an improve upper bound for $A_q(2k,2k-2;k)$. 

The previous results imply that $A_q(n,d)$ is determined for all $n\le 5$:
\begin{eqnarray}
  A_q(3,2) &=& q^2+q+2, \\
  A_q(3,3) &=& 2,\\ 
  A_q(4,2) &=& q^4 + q^3 + 2q^2 + q + 3,\\ 
  A_q(4,3) &=& q^2+1,\\ 
  A_q(4,4) &=& q^2+1,\\
  A_q(5,2) &=& q^6 + q^5 + 3q^4 + 3q^3 + 3q^2 + 2q + 3,\\ 
  A_q(5,3) &=& 2q^3 + 2,\\ 
  A_q(5,4) &=& q^3 + 1, \text{ and}\\ 
  A_q(5,5) &=& 2. 
\end{eqnarray}

ILP formulations for the exact determination of $A_q(n,d)$ and bounds for $A_2(n,d)$, where $n\le 8$, are provided in \cite{heinlein2018binary}. 
In \cite{EtzionVardy} an LP upper bound for $A_q(n,3)$ was presented. Another LP upper bound for the general case $A_q(n,d)$ can be found in 
\cite{ahlswede2009error}. For upper bounds based on semidefinite programming we refer to \cite{MR3063504,heinlein2020new}. The Johnson upper bound for 
{\cdc}s from Theorem~\ref{thm_johnson_II} was adjusted to {\mdc}s in \cite{honold2019johnson}. There also the refinement using results for divisible codes 
is discussed. A few general lower bounds for {\mdc}s are surveyed in \cite{khaleghi2009subspace}. 

\begin{table}[htp!]
  \begin{center}
    \begin{tabular}{rrrrrrrrrr}
      \hline
      n/d & 1 & 2 & 3 & 4 & 5 & 6 & 7 & 8 \\
      \hline
      1 &  2 \\
      2 & 5 & 3 \\
      3 & 16 & 8 & 2 \\
      4 & 67 & 37 & 5 & 5\\
      5 & 374 & 187 & 18 & 9 & 2\\
      6 & 2825 & 1521 & 108--117 & 77 & 9 & 9\\
      7 & 29212 & 14606 & 614--776 & 334--388 & 34 & 17 & 2 \\
      8 & 417199 & 222379 & 5687--9191 & 4803--6479 & 263--326 & 257 & 17 & 17\\
      \hline
    \end{tabular}
    \caption{Exact values and bounds for $A_2(n,d)$.}
    \label{table_mdc_binary_bounds}
  \end{center}
\end{table}
  
\begin{question}{Research problem}Improve a few lower or upper bounds for $A_2(n,d)$, see Table~\ref{table_mdc_binary_bounds}.  
\end{question}  
 
\chapter{Variants of subspace codes}
\label{sec_variants} 
In this section we want to briefly discuss topics that are closely related to the concept of subspace codes. For block codes the (Hamming) weights of 
codewords as well as the minimum Hamming distance are important invariants. For linear codes one may also consider the cardinality of the support of 
the $2$-dimensional subcode spanned by two codewords (which have to be linearly independent). This idea can of course be generalized and leads to the 
notion of \emph{generalized Hamming weights} for linear codes, see e.g.\ \cite{heijnen1998generalized,helleseth1992generalized,wei1991generalized}.  
For networks and subspace codes the notion was generalized in \cite{ngai2011network} and \cite{ballico2016higher}, respectively. The latter considered 
the dimension of the span of triples of codewords.

\begin{question}{Research problem}Study the distribution of combinations of the span and the intersection for triples and quadruples of codewords 
in {\cdc}s.\end{question}   

Having a minimum subspace distance of at least $d$ for a given {\cdc} $\cC\subseteq\cG_q(n,k)$ is equivalent to the property that the dimension of the 
intersection of two different codewords is at most $k-d/2$. In other words, every $(k-d/2+1)$-space is contained in at most one codeword. A natural 
generalization of {\cdc}s is to ask for subsets $\cC\subseteq\cG_q(n,k)$ such that every $t$-space is covered at most $\lambda$ times, see e.g.\ 
\cite{ubt_eref48694,etzion2020subspace}. One may also ask for subsets $\cC\subseteq\cG_q(n,k)$ such that every $t$-space is covered at least once (or 
at least $\lambda$ times), see e.g.\ \cite{etzion2014covering}.

Instead of $\PG(n-1,q)$ as ambient space we can also consider subspace codes over different geometries over finite fields, see e.g.\ 
\cite{wan1997geometry}. For first results into this direction we refer to e.g.\ \cite{gao2021bounds,gao2014bounds,gao2015error,gao2016bounds,hakimi2019bounds}. 
For affine spaces we refer to \cite{niu2018subspace}.

\begin{question}{Research problem}For $A_q(n,d;k)$ with $d<2k$ and $2k\le n$ almost all of the tightest known upper bounds are implied by the improved Johnson bound in 
Theorem~\ref{thm:johnson_improved}, which is based on divisible codes. Develop a similar theory of divisible codes and generalize the approach of the improved 
Johnson bound to the settings of the paper mentioned above.
\end{question}

In Subsection~\ref{subsec_equidistant} we briefly consider equidistant subspace codes and flag codes in Subsection~\ref{subsec_flag_codes}.

\section{Equidistant subspace codes}
\label{subsec_equidistant}

Partial $k$-spreads or {\cdc}s minimum subspace distance $d=2k$, where $n\ge 2k$, are a special class of so-called \emph{equidistant subspace codes}. These 
are subspace codes where any two different codewords have the same distance. Another special class of equidistant codes are so-called \emph{sunflowers} where 
all codewords pairwise intersect in the same subspace, say of dimension $t$. For the classical set case {\lq\lq}$q=1${\rq\rq}, i.e.\ equidistant block codes 
in the Hamming metric, we refer the interested reader e.g.~to \cite{deza1981every,fu2003equidistant,hall1977bounds,sinha2008good,van1973theorem}. Of course, 
geometers have already studied the case $q\ge 2$, see e.g.~\cite{ballico2004families,beutelspacher1999sets,eisfeld2000sets,eisfeld2002sets}. 

By $B_q(n,t;k)$ we denote the maximum number of $k$-spaces in 
$\PG(n-1,q)$ such that the intersection of each pair of different $k$-spaces has dimension exactly $t$. We also speak of \emph{$t$-intersecting equidistant 
codes of $k$-spaces in $\PG(n-1,q)$}.
\begin{nexercise}
  Show that $B_q(n,t;k)=1$ for $t<2k-n$ and that the maximum cardinality of a sunflower is $A_q(n-t,2(k-t);k-t)$ if $t\ge 2k-n$.
\end{nexercise}  
\begin{ntheorem}(\cite[Theorem 1]{etzion2015equidistant})\\ 
  If $\cC$ in $\cG_q(n,k)$ is a $t$-intersecting equidistant code with
  $$
    \#\cC\ge \left(\frac{q^k-q^t}{q-1}\right)^2+\frac{q^k-q^t}{q-1}+1,
  $$
  then $\cC$ is a sunflower.
\end{ntheorem}
For $2k>n$ we obtain $B_q(n,t;k)=B_q(n,n-2k+t;n-k)$ by duality. So, optimal codes can also
be duals of sunflowers and it remains to restrict to the cases where $2k \le n$. 
\begin{nexercise} 
  Show $B_q(n,1;2)=[n-1]_q$, $B_2(3, 1; 3) = 1$, $B_2 (4, 1; 3) = 1$, $B_2 (5, 1; 3) = 9$, $B_2(n, 1; 3)=A_2(n-1, 4; 2)$ for $n\ge 7$,  and that all values are
  attained by sunflowers or the dual of a sunflower.
\end{nexercise}
Sunflower codes and their properties have e.g.~been investigated in \cite{barrolleta2017primitive,bartoli2021improvement,blokhuis2021sunflower,d2021families,gorla2016equidistant,lucas2019properties,lucasgeometrical}. 
Under additional assumptions sunflower codes and their duals can be shown to be the maximal possibilities \cite{mahak2025equidistant}.  
In general it seems to be easier to determine $B_q(v,t;k)$ if $q$ gets larger, see e.g.~\cite{beutelspacher1999sets}, so that we here focus on the binary case 
$q=2$. Cf.~the remark in the third paragraph of the first section in \cite{boston2010spaces} on the {\lq\lq}unusual property{\rq\rq} of $\F_2$ in our context. 
In \cite{bartoli2016note} $B_2(6,1;3) = 20 > 9$ was proven, i.e., the optimal equidistant codes for these parameters are not given by sunflowers or their dual codes. 

An $m\times n$ equidistant rank metric code over $\F_q$ with rank distance $d$ is a set $\cM$ of $m\times n$ matrices over $\F_q$ such that for each pair of different 
$M,M'\in\cM$ we have $\dr(M,M')=d$. As an example, the five matrices
$$
\begin{pmatrix}
0&0&0&0\\ 
1&0&1&1\\ 
1&1&1&1\\ 
0&0&0&1 
\end{pmatrix},
\begin{pmatrix}
0&1&0&1\\ 
1&1&1&0\\ 
0&0&0&0\\ 
0&0&1&0 
\end{pmatrix},
\begin{pmatrix}
1&0&1&1\\ 
1&0&1&1\\ 
0&1&0&1\\ 
0&1&1&0 
\end{pmatrix},
\begin{pmatrix}
1&0&0&1\\ 
1&0&1&1\\ 
1&1&0&1\\ 
1&0&1&1 
\end{pmatrix},
\begin{pmatrix}
0&1&1&0\\ 
0&1&0&1\\ 
1&0&0&1\\ 
1&0&1&0 
\end{pmatrix}
$$
span a linear $4\times 4$ equidistant rank metric code over $\F_2$ with rank distance $3$. By \cite[Theorem 6]{dumas2010subspaces} there cannot be 
six such matrices. By prepending a suitable unit matrix, i.e.\ by lifting, we obtain an equidistant subspace code in general. So, our example gives 
$B_2(8,1;4)\ge 32$. We remark that several linear $4\times 4$ equidistant rank metric codes over $\F_2$ with rank distance $3$ and cardinality $2^5$ 
exist and that their lifted versions allow the addition of further codewords. By a computer search up to $8$ additional codewords can be found easily, 
so that $B_2(8,1;4)\ge 40$.

\begin{question}{Research problem}Determine the exact value of $B_2(8,1;4)$. 
\end{question}  
  
Another example is given by the four matrices
$$
\begin{pmatrix}
1&0&0\\ 
0&0&0\\ 
1&0&1 
\end{pmatrix},
\begin{pmatrix}
0&1&1\\ 
1&1&1\\ 
0&1&1 
\end{pmatrix},
\begin{pmatrix}
1&0&0\\ 
1&1&0\\ 
1&0&0 
\end{pmatrix},
\begin{pmatrix}
0&0&0\\ 
0&1&0\\ 
0&1&1
\end{pmatrix},
$$
which span a linear $3\times 3$ equidistant rank metric code over $\F_2$ with rank distance $2$. By \cite[Theorem 6]{dumas2010subspaces} there cannot be 
five such matrices. Note that this gives $B_2(6,1;3)\ge 16$. In \cite{etzion2015equidistant} 
an equidistant code with these parameters was stated by explicitly listing $16$ codewords. There it was mentioned as a counter example to a conjecture attributed to Deza, 
i.e., if a $t$-intersection equidistant code of $k$-subspaces in $\PG(v-1,q)$ has more than $\qbin{k+1}{1}{q}$ codewords, then it is a sunflower. In \cite{boston2010spaces} 
the author determined, using an exhaustive MAGMA search, that there are exactly $1176$ binary linear $3\times 3$ equidistant rank metric codes over $\F_2$ with rank distance 
$2$ and dimension $4$. Under conjugation by $\operatorname{GL}(3,2)$ they fall into $12$~orbits, which are explicitly listed. An example of a binary linear $4\times 4$ equidistant 
rank metric code over $\F_2$ with rank distance $3$ and dimension $5$ as well as a linear $5\times 5$ equidistant rank metric code over $\F_2$ with rank distance $4$ and dimension 
$6$, found by a heuristic search using MAGMA, is also stated there. By \cite[Theorem 6]{dumas2010subspaces} the dimension is extremal in both cases. However, the resulting lower bounds 
$B_2(8,1;4)\ge 32$ and $B_2(10,1;5)\ge 64$ have not found their way into the literature on equidistant subspace codes. With respect to the two latter bounds we mention the example
\begin{eqnarray*}
&&
\begin{pmatrix}
0&0&0&0&1\\ 
0&0&1&1&0\\ 
1&1&1&0&1\\ 
1&1&0&1&0 
\end{pmatrix},
\begin{pmatrix}
1&1&1&1&0\\ 
1&0&0&0&0\\ 
1&1&1&1&0\\ 
0&0&1&0&1 
\end{pmatrix},
\begin{pmatrix}
1&0&0&1&1\\ 
1&0&1&0&0\\ 
0&0&1&1&1\\ 
1&0&0&0&1 
\end{pmatrix},\\ 
&&
\begin{pmatrix}
0&1&1&0&1\\ 
0&1&0&1&0\\ 
0&1&1&0&1\\ 
0&1&0&1&1 
\end{pmatrix},
\begin{pmatrix}
1&0&0&1&1\\ 
0&1&0&0&0\\ 
1&0&0&0&1\\ 
0&1&0&1&0 
\end{pmatrix},
\begin{pmatrix}
1&0&0&1&1\\ 
0&0&0&0&1\\ 
1&1&0&1&1\\ 
0&1&0&0&0 
\end{pmatrix},
\end{eqnarray*}
which shows $B_2(9,1;4)\ge 64$, see also \cite[Example 1]{dumas2010subspaces}. By \cite[Theorem 6]{dumas2010subspaces} there cannot be 
seven such matrices. 

According to \cite{beasley1999spaces} the problem of determining lower and upper bounds for rank-$k$-spaces in $\F_q{m\times n}$ has been studied by matrix theorists, 
group theorists, and algebraic geometers, see his list of references and \cite{boston2010spaces,gow2015dimension}.

We remark $B_2(11,1;5),B_2(12,1;6)\ge 64$ since corresponding linear equidistant rank metric codes can be found easily. However, \cite[Theorem 6]{dumas2010subspaces} 
might allow even linear equidistant rank metric codes of cardinality $2^7$. 

\begin{question}{Research problem}Study linear equidistant rank metric codes and their extendability to equidistant subspace codes.\end{question}  

Linear equidistant subspace codes have e.g.\ been studied in \cite{basu2021equidistant}. Instead of restricting the dimension of the pairwise intersection of codewords to 
a single dimension one might also allow e.g.\ two possible intersection dimensions, see \cite{longobardi2021sets}.

\section{Flag codes}
\label{subsec_flag_codes}

A \emph{full flag} in $\PG(n-1,q)$ is a sequence of nested subspaces with dimensions from $1$ to $n-1$. 
If not all of these dimensions need to occur, we speak of a \emph{flag}. (Full) \emph{flag codes} are collections 
of flags. The use of flag codes for network coding was proposed in \cite{liebhold2018network}. In \cite{liebhold2019flag} the author argues that 
subspace coding with flags can be ranged between random linear network coding, using constant dimension codes, and optimized 
routing solutions, whose computation is time-consuming. The interested reader can find more details on this e.g.\ in 
\cite{fourier2020degenerate,liebhold2019flag,liebhold2018network,liebhold2018generalizing}. For special multicast networks network coding 
solutions also lead to hard combinatorial problems, see e.g.~\cite{cai2019network, etzion2020subspace} for so-called generalized combination networks. 

The set of all subspaces in $\PG(n-1,q)$ is turned into a metric space via the \emph{injection distance}
\begin{eqnarray}
  \di(U,W) &=&\dim(U+W)-\min\{\dim(U),\dim(W)\}\notag\\ 
  &=&\max\{\dim(U),\dim(W)\}-\dim(U\cap W)
\end{eqnarray}
as it is the case for the subspace distance. Note that for $U,W\in\cG_q(n,k)$ we have $\di(U,W) = \dim(U+W)-k=k-\dim(U\cap W)$. 
\begin{ndefinition} 
   A \emph{flag} is a list of subspaces $\Lambda=\left(W_1,\dots,W_m\right)$ of $\PG(n-1,q)$ with 
   $$
     \{0\}<W_1<\dots<W_m<\mathbb{F}_q^n.
   $$
   The \emph{type} of $\Lambda=\left(W_1,\dots,W_m\right)$ is the set of dimensions
  $$
    \operatorname{type}(\Lambda):=\left\{\dim(W_i)\mid 1\le i\le m\right\}\subseteq\left\{1,\dots ,n-1\right\}.
  $$
  Let 
  $$
    \cF(n,q):=\left\{\Lambda\mid \Lambda\text{ is a flag in }\PG(n-1,q)\right\}
  $$
  denote the set of all flags in $\PG(n-1,q)$ and for $T\subseteq\{1,\dots,n-1\}$ let 
  $$
    \cF_T(n,q):=\left\{\Lambda\in\cF(n,q)\mid \operatorname{tpye}(\Lambda)=T\right\}
  $$
  be the set of all flags of $\PG(n-1,q)$ of type $T$.
\end{ndefinition} 

As noted in \cite{liebhold2018network}, the intersection of two flags is again a flag and the set of all flags in $\PG(n-1,q)$ forms 
a simplicial complex (with respect to inclusion).

\begin{ndefinition}
  Let $\Lambda=\left(W_1,\dots,W_m\right)$ and $\Lambda'=\left(W_1',\dots,W_m'\right)$ be two flags of $\PG(n-1,q)$ of the 
  same type $T=\left\{k_1,\dots,k_m\right\}$ with $k_i=\dim(W_i)=\dim(W_i')$ for all $1\le i\le m$. Then, the \emph{Grassmann distance} is defined as
  $$
    \fdist(\Lambda,\Lambda'):=\sum_{i=1}^m \di(W_i,W_i')=\sum_{i=1}^m \left(k_i-\dim(W_i\cap W_i')\right).
  $$
\end{ndefinition}

So, for $m=1$ the Grassmann distance corresponds to the injection distance, i.e., half the subspace distance, between $W_1$ and $W_1'$. 
For $U,W\in \cG_q(n,k)$ we have $0\le \di(U,W)\le \min\{k,v-k\}$, so that we set 
$$
  m(n,T)=\left(\min\{k_1,n-k_1\},\dots,\min\{k_m,n-k_m\}\right),
$$ 
where $T=\left\{k_1,\dots,k_m\right\}\subseteq\{1,\dots,n-1\}$ with $k_1<\dots<k_m$. If $T=\{1,\dots,n-1\}$ we just write $m(n)$ instead of 
$m(n,T)$. Denoting by $x_i$ the $i$th component for each vector $x\in\R^n$ we state
$$
  \fdist(\Lambda,\Lambda')\le \sum_i m(n,T)_i
$$ 
for all $\Lambda,\Lambda'\in\cF_T(n,q)$. As mentioned in \cite[Remark 4.5]{liebhold2018network} we have 
$1\le \fdist(\Lambda,\Lambda')\le \left\lfloor (n/2)^2\right\rfloor$ for two distinct flags in $\PG(n-1,q)$. A 
\emph{flag code} $\cC$ of type $T$ is a collection of flags in $\PG(n-1,q)$ of type $T$. The minimum distance 
$\fdist(\cC)$ is the minimum of $\fdist(\Lambda,\Lambda')$ over all pairs of distinct elements $\Lambda,\Lambda'\in\cC$. By $A_q^f(n,d;T)$ we denote the maximum 
possible cardinality of a flag code $\cC$ of type $T$ in $\PG(n-1,q)$ 
that has minimum Grassmann distance at least $d$. The case of full flags, i.e.\ $T=\{1,\dots,n-1\}$, is abbreviated as $A_q^f(n,d)$. The \emph{dual} 
of a flag $\Lambda=\left(W_1,\dots,W_m\right)$ in $\PG(n-1,q)$ of type $T\subseteq \{1,\dots,n-1\}$, denoted by $\Lambda^\top$, is given by 
$\left(W_{m}^\top,\dots,W_{1}^\top\right)$. Since we have $\di(U,W)=\di\!\left(U^\top,W^\top\right)$ for each $U,W\in\cG_q(n,k)$, for some arbitrary 
integer $k$, the minimum Grassmann distance $d(\cC)$ of a flag code of type $T$ in $\PG(n-1,q)$ is the same as $d\!\left(\cC^\top\right)$, 
where $\cC^\top:=\left\{\Lambda^\top\mid\Lambda\in\cC\right\}$. Moreover, we have
$$
  \operatorname{type}\!\left(\cC^\top\right)=\left\{n-t\mid t\in\operatorname{type}(\cC)\right\},
$$
so that $A_q^f(n,d;T)=A_q^f\!\left(n,d;n-t\right)$.

The arguably easiest case for the determination of $A_q^f(n,d;T)$ is minimum Grassmann distance $d=1$, where $A_q^f(n,1;T)=\#\cF_T(n,q)$. 
If $T=\left\{k_1,\dots,k_m\right\}$ with $0<k_1<\dots<k_m<n$, then we have
\begin{equation}
  \label{eq_count_flags}
  A_q^f(n,1;T)=\qbin{n}{k_1}{q}\cdot \prod_{i=2}^m \qbin{n-k_{i-1}}{k_i-k_{i-1}}{q} 
\end{equation}
and
\begin{equation}
  A_q^f(n,1)=\prod_{i=2}^{n} \frac{q^i-1}{q-1}. 
\end{equation}

For the maximum possible minimum Grassmann distance $d=\left\lfloor(n/2)^2\right\rfloor$ we have:
\begin{nproposition}(\cite[Proposition 2.4]{kurz2021bounds})\\
  \label{prop_a_max_distance}
  For each integer $k\ge 1$ we have
  $$
    A_q^f(2k,k^2)=q^k+1  
  $$
  and for each integer $k\ge 2$ we have
  $$
    A_q^f(2k+1,k^2+k)=q^{k+1}+1.  
  $$
\end{nproposition}
We remark that the case $n=2k$ of Proposition~\ref{prop_a_max_distance} was also proven in \cite{alonso2020flag}, where the authors also give a decoding algorithm and 
further details. In \cite[Proposition 2.6]{kurz2021bounds} the exact value 
\begin{equation}
  A_q^f(4,3)=\qbin{4}{1}{q}=q^3+q^2+q+1
\end{equation}
was determined. In Table~\ref{table_bounds_binary} and Table~\ref{table_upper_bounds_binary_v_7} we present the current knowledge on $A_2^f(n,d)$ from 
\cite{kurz2021bounds}. Research on bounds and constructions for flag codes currently is quite an active research field, see e.g.\ 
\cite{alonso2020consistent,alonso2021cyclic,alonso2020flag,alonso2021distance,alonso2021flag,alonso2021orbital,alonso2020optimum,kurz2021bounds,navarro2021flag}.  

\begin{table}[htp]
  \begin{center}
    %%\begin{tabular}{rrrrrrrrrr}
    \begin{tabular}{rrrrrrr}
    \hline
    $n/d$ & 1  & 2          & 3       & 4      & 5      & 6    \\ %% & 7 & 8 & 9 \\
    \hline
    2 & 3      &            &         &        &        &   \\
    3 & 21     & 7          &         &        &        &   \\ 
    4 & 315    & 105        & 15      & 5      &        &   \\
    5 & 9765   & 3120--3255 & 465     & 155    & 31     & 9 \\
    %%6 & 615195 & --205065   & --29295 & --9765 & --1953 & 224--567 & 63 & 21 & 9 \\
    \hline
    \end{tabular}
    \caption{Bounds and exact values for $A_2^f(n,d)$ for $n\le 5$.}
    \label{table_bounds_binary}
  \end{center}
\end{table}

\begin{question}{Research problem}Find improved lower and upper bounds for $A_q^f(n,d)$.\end{question}

\begin{table}[htp]
  \tabcolsep=1.8pt
  \begin{center}
    \begin{tabular}{rrrrrrrrrrrrr}
    \hline
    $n/d$ & 1  & 2          & 3       & 4      & 5      & 6     & 7 & 8 & 9 & 10 & 11 & 12\\
    \hline
    6 & \textbf{615195} & 205065   & 29295 & 9765 & 1953 & 567 & \textbf{63} & \textbf{21} & \textbf{9} \\
    7 & \textbf{78129765} & 26043255 & 3720465 & 1240155 & 248031 & 72009 & 8001 & 2667 & 1143 & 127 & 41 & \textbf{17}\\
    \hline
    \end{tabular}
    \caption{Upper bounds for $A_2^f(6,d)$ and $A_2^f(7,d)$ (tight bounds in bold).}
    \label{table_upper_bounds_binary_v_7}
  \end{center}
\end{table}

%% \bibliography{chap_subspace_codes}

\begin{thebibliography}{100}

\bibitem{agrell2001table}
Erik Agrell, Alexander Vardy, and Kenneth Zeger.
\newblock A table of upper bounds for binary codes.
\newblock {\em IEEE Transactions on Information Theory}, 47(7):3004--3006,
  2001.

\bibitem{ahlswede2009error}
Rudolf Ahlswede and Harout (Haratyun)~K. Aydinian.
\newblock On error control codes for random network coding.
\newblock In {\em 2009 Workshop on Network Coding, Theory, and Applications
  (NetCod 2009), Lausanne, Switzerland}, pages 68--73. IEEE, 2009.

\bibitem{ahlswede2001perfect}
Rudolf Ahlswede, Harout (Haratyun)~K. Aydinian, and Levon~H. Khachatrian.
\newblock On perfect codes and related concepts.
\newblock {\em Designs, Codes and Cryptography}, 22(3):221--237, 2001.

\bibitem{ai2016expurgation}
Jingmei Ai, Thomas Honold, and Haiteng Liu.
\newblock The expurgation-augmentation method for constructing good plane
  subspace codes.
\newblock {\em arXiv preprint 1601.01502}, 2016.

\bibitem{alonso2020consistent}
Clementa Alonso-Gonz{\'a}lez and Miguel~{\'A}ngel Navarro-P{\'e}rez.
\newblock Consistent flag codes.
\newblock {\em Mathematics}, 8(12):19 pp., 2020.

\bibitem{alonso2021cyclic}
Clementa Alonso-Gonz{\'a}lez and Miguel~{\'A}ngel Navarro-P{\'e}rez.
\newblock Cyclic orbit flag codes.
\newblock {\em Designs, Codes and Cryptography}, 89(19):2331--2356, 2021.

\bibitem{alonso2020flag}
Clementa Alonso-Gonz{\'a}lez, Miguel~{\'A}ngel Navarro-P{\'e}rez, and Xaro
  Soler-Escriv{\`a}.
\newblock Flag codes from planar spreads in network coding.
\newblock {\em Finite Fields and Their Applications}, 68:19 pp., 2020.

\bibitem{alonso2021distance}
Clementa Alonso-Gonz{\'a}lez, Miguel~{\'A}ngel Navarro-P{\'e}rez, and Xaro
  Soler-Escriv{\`a}.
\newblock Distance and bounds for flag codes.
\newblock {\em arXiv preprint 2111.00910}, 2021.

\bibitem{alonso2021flag}
Clementa Alonso-Gonz{\'a}lez, Miguel~{\'A}ngel Navarro-P{\'e}rez, and Xaro
  Soler-Escriv{\`a}.
\newblock Flag codes: Distance vectors and cardinality bounds.
\newblock {\em Linear Algebra and its Applications}, 656:27--62, 2023.

\bibitem{alonso2021orbital}
Clementa Alonso-Gonz{\'a}lez, Miguel~{\'A}ngel Navarro-P{\'e}rez, and Xaro
  Soler-Escriv{\`a}.
\newblock An orbital construction of optimum distance flag codes.
\newblock {\em Finite Fields and Their Applications}, 73:21 pp., 2021.

\bibitem{alonso2020optimum}
Clementa Alonso-Gonz{\'a}lez, Miguel~{\'A}ngel Navarro-P{\'e}rez, and Xaro
  Soler-Escriv{\`a}.
\newblock Optimum distance flag codes from spreads via perfect matchings in
  graphs.
\newblock {\em Journal of Algebraic Combinatorics}, 54(4):1279--1297, 2021.

\bibitem{andrews1999q}
George~Eyre Andrews.
\newblock $q$-analogs of the binomial coefficient congruences of {B}abbage,
  {W}olstenholme and {G}laisher.
\newblock {\em Discrete Mathematics}, 204(1-3):15--25, 1999.

\bibitem{antrobus2019maximal}
Jared Evan Antrobus and Heide Gluesing-Luerssen.
\newblock Maximal {F}errers diagram codes: constructions and genericity
  considerations.
\newblock {\em IEEE Transactions on Information Theory}, 65(10):6204--6223,
  2019.

\bibitem{antrobus2019state}
Jared~Evan Antrobus.
\newblock {\em The state of {L}exicodes and {F}errers diagram rank-metric
  codes}.
\newblock PhD thesis, University of Kentucky, 2019.

\bibitem{MR3063504}
Christine Bachoc, Alberto Passuello, and Frank Vallentin.
\newblock Bounds for projective codes from semidefinite programming.
\newblock {\em Advances in Mathematics of Communications}, 7(2):127--145, 2013.

\bibitem{baker1976partitioning}
Ronald~Dee Baker.
\newblock Partitioning the planes of $\operatorname{AG}_{2m}(2)$ into
  $2$-designs.
\newblock {\em Discrete Mathematics}, 15(3):205--211, 1976.

\bibitem{baker1983preparata}
Ronald~Dee Baker, Jacobus~Hendricus Van~Lint, and Richard~Michael Wilson.
\newblock On the {P}reparata and {G}oethals codes.
\newblock {\em IEEE Transactions on Information Theory}, 29(3):342--345, 1983.

\bibitem{ballico2016higher}
Edoardo Ballico.
\newblock Higher distances for constant dimensions codes: the case of
  osculating spaces to a {V}eronese variety.
\newblock {\em Afrika Matematika}, 27(5):1003--1020, 2016.

\bibitem{ballico2004families}
Edoardo Ballico, Nadia Chiarli, and Silvio Greco.
\newblock Families of projective planes meeting in codimension two.
\newblock {\em Results in Mathematics}, 46(1-2):16--23, 2004.

\bibitem{barcucci1999some}
Elena Barcucci, Alberto Del~Lungo, Elisa Pergola, and Renzo Pinzani.
\newblock Some combinatorial interpretations of $q$-analogs of {S}chr{\"o}der
  numbers.
\newblock {\em Annals of Combinatorics}, 3(2):171--190, 1999.

\bibitem{barrolleta2017primitive}
Roland~D. Barrolleta, Emilio Jos\'e Su{\'a}rez-Canedo, Leo Storme, and Peter
  Vandendriessche.
\newblock On primitive constant dimension codes and a geometrical sunflower
  bound.
\newblock {\em Advances in Mathematics of Communications}, 11(4):757--765, 2017.

\bibitem{bartoli2015constant}
Daniele Bartoli, Matteo Bonini, and Massimo Giulietti.
\newblock Constant dimension codes from {R}iemann-{R}och spaces.
\newblock {\em Advances in Mathematics of Communications}, 11(4):705--713, 2015.

\bibitem{bartoli2016note}
Daniele Bartoli and Francesco Pavese.
\newblock A note on equidistant subspace codes.
\newblock {\em Discrete Applied Mathematics}, 198:291--296, 2016.

\bibitem{bartoli2021improvement}
Daniele Bartoli, Ago-Erik Riet, Leo Storme, and Peter Vandendriessche.
\newblock Improvement to the sunflower bound for a class of equidistant constant dimension subspace codes.
\newblock {\em Journal of Geometry}, 112(1):12 pages, 2021.

\bibitem{bassoli2013network}
Riccardo Bassoli, Hugo Marques, Jonathan Rodriguez, Kenneth~Wing-Ki Shum, and Rahim
  Tafazolli.
\newblock Network coding theory: A survey.
\newblock {\em IEEE Communications Surveys \& Tutorials}, 15(4):1950--1978,
  2013.

\bibitem{basu2021equidistant}
Pranab Basu.
\newblock Equidistant linear codes in projective spaces.
\newblock {\em arXiv preprint 2107.10820}, 13 pages, 2021.

\bibitem{beasley1999spaces}
LeRoy~B. Beasley.
\newblock Spaces of rank-$2$ matrices over ${G}{F}(2)$.
\newblock {\em Electronic Journal of Linear Algebra}, 5:11--18, 1999.

\bibitem{ben2016subspace}
Eli Ben-Sasson, Tuvi Etzion, Ariel Gabizon, and Netanel Raviv.
\newblock Subspace polynomials and cyclic subspace codes.
\newblock {\em IEEE Transactions on Information Theory}, 62(3):1157--1165,
  2016.

\bibitem{beutelspacher1974parallelisms}
Albrecht Beutelspacher.
\newblock On parallelisms in finite projective spaces.
\newblock {\em Geometriae Dedicata}, 3(1):35--40, 1974.

\bibitem{beutelspacher1975partial}
Albrecht Beutelspacher.
\newblock Partial spreads in finite projective spaces and partial designs.
\newblock {\em Mathematische Zeitschrift}, 145(3):211--229, 1975.

\bibitem{beutelspacher1976correction}
Albrecht Beutelspacher.
\newblock Correction to {\lq\lq}Partial spreads in finite projective spaces and
  partial designs{\rq\rq}.
\newblock {\em Mathematische Zeitschrift}, 147(3):303--303, 1976.

\bibitem{beutelspacher1990partial}
Albrecht Beutelspacher.
\newblock Partial parallelisms in finite projective spaces.
\newblock {\em Geometriae Dedicata}, 36(2-3):273--278, 1990.

\bibitem{beutelspacher1999sets}
Albrecht Beutelspacher, J{\"o}rg Eisfeld, and J{\"o}rg M{\"u}ller.
\newblock On sets of planes in projective spaces intersecting mutually in one
  point.
\newblock {\em Geometriae Dedicata}, 78(2):143--159, 1999.

\bibitem{bierbrauer2007direct}
J{\"u}rgen Bierbrauer.
\newblock A direct approach to linear programming bounds for codes and
  tms-nets.
\newblock {\em Designs, Codes and Cryptography}, 42(2):127--143, 2007.

\bibitem{bierbrauer2016introduction}
J\"urgen Bierbrauer.
\newblock {\em Introduction to Coding Theory}.
\newblock Chapman and Hall/CRC, 2016.

\bibitem{blackburn2012asymptotic}
Simon~Robert Blackburn and Tuvi Etzion.
\newblock The asymptotic behavior of {G}rassmannian codes.
\newblock {\em IEEE Transactions on Information Theory}, 58(10):6605--6609,
  2012.

\bibitem{blokhuis2021sunflower}
Aart Blokhuis, Maarten De~Boeck, and Jozefien D’Haeseleer.
\newblock On the sunflower bound for k-spaces, pairwise intersecting in a
  point.
\newblock {\em Designs, Codes and Cryptography}, 90(9):2101--2111, 2022.

\bibitem{boston2010spaces}
Nigel Boston.
\newblock Spaces of constant rank matrices over ${G}{F}(2)$.
\newblock {\em Electronic Journal of Linear Algebra}, 20(1):1--5, 2010.

\bibitem{braun2013linearity}
Michael Braun, Tuvi Etzion, and Alexander Vardy.
\newblock Linearity and complements in projective space.
\newblock {\em Linear Algebra and its Applications}, 438(1):57--70, 2013.

\bibitem{braun2018new}
Michael Braun, Patric~R.J. {\"O}sterg{\aa}rd, and Alfred Wassermann.
\newblock New lower bounds for binary constant-dimension subspace codes.
\newblock {\em Experimental Mathematics}, 27(2):179--183, 2018.

\bibitem{brouwer2006new}
Andries~Evert Brouwer, James~Bergheim Shearer, Neil James~Alexander Sloane, and
  Warren~Douglas Smith.
\newblock A new table of constant weight codes.
\newblock {\em IEEE Transactions on Information Theory}, 36(6):1334--1380,
  2006.

\bibitem{byrne2021fundamental}
Eimear Byrne, Heide Gluesing-Luerssen, and Alberto Ravagnani.
\newblock Fundamental properties of sum-rank-metric codes.
\newblock {\em IEEE Transactions on Information Theory}, 67(10):6456--6475,
  2021.

\bibitem{byrne2017covering}
Eimear Byrne and Alberto Ravagnani.
\newblock Covering radius of matrix codes endowed with the rank metric.
\newblock {\em SIAM Journal on Discrete Mathematics}, 31(2):927--944, 2017.

\bibitem{cai2019network}
Han Cai, Tuvi Etzion, Moshe Schwartz, and Antonia Wachter-Zeh.
\newblock Network coding solutions for the combination network and its
  subgraphs.
\newblock In {\em 2019 IEEE International Symposium on Information Theory
  (ISIT)}, pages 862--866. IEEE, 2019.

\bibitem{chen2018constructions}
Bocong Chen and Hongwei Liu.
\newblock Constructions of cyclic constant dimension codes.
\newblock {\em Designs, Codes and Cryptography}, 86(6):1267--1279, 2018.

\bibitem{chen2020new}
Hao Chen, Xianmang He, Jian Weng, and Liqing Xu.
\newblock New constructions of subspace codes using subsets of {M}{R}{D} codes
  in several blocks.
\newblock {\em IEEE Transactions on Information Theory}, 66(9):5317--5321,
  2020.

\bibitem{cooperstein1998external}
Bruce~Nathan Cooperstein.
\newblock External flats to varieties in $\operatorname{PG}(\wedge^2 (v))$ over
  finite fields.
\newblock {\em Geometriae Dedicata}, 69(3):223--235, 1998.

\bibitem{cossidente2021combining}
Antonio Cossidente, Sascha Kurz, Giuseppe Marino, and Francesco Pavese.
\newblock Combining subspace codes.
\newblock {\em Advances in Mathematics of Communications}, 17(3):1--15, 2023.

\bibitem{cossidente2020subspace}
Antonio Cossidente, Giuseppe Marino, and Francesco Pavese.
\newblock Subspace code constructions.
\newblock {\em Ricerche di Matematica}, 71(1):1--16, 2022.

\bibitem{cossidente2016subspace}
Antonio Cossidente and Francesco Pavese.
\newblock On subspace codes.
\newblock {\em Designs, Codes and Cryptography}, 78(2):527--531, 2016.

\bibitem{cossidente2016veronese}
Antonio Cossidente and Francesco Pavese.
\newblock Veronese subspace codes.
\newblock {\em Designs, Codes and Cryptography}, 81(3):445--457, 2016.

\bibitem{cossidente2017subspace}
Antonio Cossidente and Francesco Pavese.
\newblock Subspace codes in $\operatorname{PG} (2 n- 1, q)$.
\newblock {\em Combinatorica}, 37(6):1073--1095, 2017.

\bibitem{cossidente2018geometrical}
Antonio Cossidente, Francesco Pavese, and Leo Storme.
\newblock Geometrical aspects of subspace codes.
\newblock In Marcus Greferath, Mario~Osvin Pav{\v{c}}evi{\'c}, Natalia
  Silberstein, and Mar{\'\i}a~{\'A}ngeles V{\'a}zquez-Castro, editors, {\em
  Network Coding and Subspace Designs}, pages 107--129. Springer, 2018.

\bibitem{cossidente2019optimal}
Antonio Cossidente, Francesco Pavese, and Leo Storme.
\newblock Optimal subspace codes in $\operatorname{PG}(4, q)$.
\newblock {\em Advances in Mathematics of Communications}, 13(3):393--404,
  2019.

\bibitem{delsarte1973algebraic}
Philippe Delsarte.
\newblock {\em An algebraic approach to the association schemes of coding
  theory}.
\newblock PhD thesis, Universit\'e Catholique de Louvain, Eindhoven, 6 1973.
\newblock Philips Research Reports Supplements, No.~10.

\bibitem{delsarte1978bilinear}
Philippe Delsarte.
\newblock Bilinear forms over a finite field, with applications to coding
  theory.
\newblock {\em Journal of Combinatorial Theory, Series A}, 25(3):226--241,
  1978.

\bibitem{delsarte1978hahn}
Philippe Delsarte.
\newblock Hahn polynomials, discrete harmonics, and $t$-designs.
\newblock {\em SIAM Journal on Applied Mathematics}, 34(1):157--166, 1978.

\bibitem{delsarte1998association}
Philippe Delsarte and Vladimir~Iosifovich Levenshtein.
\newblock Association schemes and coding theory.
\newblock {\em IEEE Transactions on Information Theory}, 44(6):2477--2504,
  1998.

\bibitem{deza1981every}
Michel Deza and Peter Frankl.
\newblock Every large set of equidistant $(0,+ 1,- 1)$-vectors forms a
  sunflower.
\newblock {\em Combinatorica}, 1(3):225--231, 1981.

\bibitem{d2021families}
Jozefien D'Haeseleer.
\newblock {\em Families of intersecting subspaces}.
\newblock PhD thesis, Ghent University, 2021.

\bibitem{nets_and_spreads}
David~Allyn Drake and James~W. Freeman.
\newblock Partial $t$-spreads and group constructible $(s,r,\mu)$-nets.
\newblock {\em Journal of Geometry}, 13(2):210--216, 1979.

\bibitem{dumas2010subspaces}
Jean-Guillaume Dumas, Rod Gow, Gary McGuire, and John Sheekey.
\newblock Subspaces of matrices with special rank properties.
\newblock {\em Linear algebra and its applications}, 433(1):191--202, 2010.

\bibitem{dunkl1978addition}
Charles~F. Dunkl.
\newblock An addition theorem for some $q$-{H}ahn polynomials.
\newblock {\em Monatshefte f{\"u}r Mathematik}, 85(1):5--37, 1978.

\bibitem{eisfeldt}
J{\"o}rg~Eisfeld and Leo~Storme.
\newblock $t$-spreads and minimal $t$-covers in finite projective spaces.
\newblock {\em Lecture notes, Universiteit Gent, 29 pages}, 2000.

\bibitem{eisfeld2000sets}
J{\"o}rg Eisfeld.
\newblock On sets of sets mutually intersecting in exactly one element.
\newblock {\em Journal of Geometry}, 67(1-2):96--104, 2000.

\bibitem{eisfeld2002sets}
J{\"o}rg Eisfeld.
\newblock On sets of $n$-dimensional subspaces of projective spaces
  intersecting mutually in an $(n- 2)$-dimensional subspace.
\newblock {\em Discrete Mathematics}, 255(1-3):81--85, 2002.

\bibitem{el2003design}
Hesham El~Gamal and Arthur~Roger Hammons.
\newblock On the design of algebraic space-time codes for {M}{I}{M}{O}
  block-fading channels.
\newblock {\em IEEE Transactions on Information Theory}, 49(1):151--163, 2003.

\bibitem{spreadsk3}
Saad~Ibrahim El-Zanati, Heather Jordon, George~Francis Seelinger, Papa~Amar
  Sissokho, and Lawrence~Edward Spence.
\newblock The maximum size of a partial $3$-spread in a finite vector space
  over ${G}{F}(2)$.
\newblock {\em Designs, Codes and Cryptography}, 54(2):101--107, 2010.

\bibitem{etzion2013problems}
Tuvi Etzion.
\newblock Problems on $q$-analogs in coding theory.
\newblock {\em arXiv preprint 1305.6126}, 2013.

\bibitem{etzion2014covering}
Tuvi Etzion.
\newblock Covering of subspaces by subspaces.
\newblock {\em Designs, Codes and Cryptography}, 72(2):405--421, 2014.

\bibitem{etzion2015partial}
Tuvi Etzion.
\newblock Partial-parallelisms in finite projective spaces.
\newblock {\em Journal of Combinatorial Designs}, 23(3):101--114, 2015.

\bibitem{etzion2016optimal}
Tuvi Etzion, Elisa Gorla, Alberto Ravagnani, and Antonia Wachter-Zeh.
\newblock Optimal {F}errers diagram rank-metric codes.
\newblock {\em IEEE Transactions on Information Theory}, 62(4):1616--1630,
  2016.

\bibitem{etzion2017residual}
Tuvi Etzion and Niv Hooker.
\newblock Residual $q$-{F}ano planes and related structures.
\newblock {\em The Electronic Journal of Combinatorics}, 25(2):25 pp., 2018.

\bibitem{ubt_eref48694}
Tuvi Etzion, Sascha Kurz, Kamil Otal, and Ferruh {\"O}zbudak.
\newblock Subspace packings.
\newblock In {\em The Eleventh International Workshop on Coding and
  Cryptography 2019 : WCC Proceedings}. IEEE, Saint-Jacut-de-la-Mer, 2019.

\bibitem{etzion2020subspace}
Tuvi Etzion, Sascha Kurz, Kamil Otal, and Ferruh {\"O}zbudak.
\newblock Subspace packings: constructions and bounds.
\newblock {\em Designs, Codes and Cryptography}, 88:1781--1810, 2020.

\bibitem{etzion2015equidistant}
Tuvi Etzion and Netanel Raviv.
\newblock Equidistant codes in the Grassmannian.
\newblock {\em Discrete Applied Mathematics}, 186:87--97, 2015.

\bibitem{etzion2009error}
Tuvi Etzion and Natalia Silberstein.
\newblock Error-correcting codes in projective spaces via rank-metric codes and
  {F}errers diagrams.
\newblock {\em IEEE Transactions on Information Theory}, 55(7):2909--2919,
  2009.

\bibitem{etzion2012codes}
Tuvi Etzion and Natalia Silberstein.
\newblock Codes and designs related to lifted {M}{R}{D} codes.
\newblock {\em IEEE Transactions on Information Theory}, 59(2):1004--1017,
  2013.

\bibitem{etzion2016galois}
Tuvi Etzion and Leo Storme.
\newblock Galois geometries and coding theory.
\newblock {\em Designs, Codes and Cryptography}, 78(1):311--350, 2016.

\bibitem{EtzionVardy}
Tuvi Etzion and Alexander Vardy.
\newblock Error-correcting codes in projective space.
\newblock In {\em Proceedings. International Symposium on Information Theory,
  2008, ISIT 2008}, pages 871--875. IEEE, 2008.

\bibitem{MR2810308}
Tuvi Etzion and Alexander Vardy.
\newblock Error-correcting codes in projective space.
\newblock {\em IEEE Transactions on Information Theory}, 57(2):1165--1173,
  2011.

\bibitem{etzion2012automorphisms}
Tuvi Etzion and Alexander Vardy.
\newblock Automorphisms of codes in the {G}rassmann scheme.
\newblock {\em arXiv preprint 1210.5724}, 2012.

\bibitem{feng2020bounds}
Tao Feng, Sascha Kurz, and Shuangqing Liu.
\newblock Bounds for the multilevel construction.
\newblock {\em arXiv preprint 2011.06937}, 2020.

\bibitem{fourier2020degenerate}
Ghislain Fourier and Gabriele Nebe.
\newblock Degenerate flag varieties in network coding.
\newblock {\em Advances in Mathematics of Communications}, 17(4):888--899, 2023.

\bibitem{MR829351}
P\'eter Frankl and Vojt\u{e}ch R\"odl.
\newblock Near perfect coverings in graphs and hypergraphs.
\newblock {\em European Journal of Combinatorics}, 6(4):317--326, 1985.

\bibitem{MR867648}
P\'eter Frankl and Richard~Michael Wilson.
\newblock The {E}rd{\H{o}}s-{K}o-{R}ado theorem for vector spaces.
\newblock {\em Journal of Combinatorial Theory, Series A}, 43(2):228--236,
  1986.

\bibitem{fu2003equidistant}
Fang-Wei Fu, Torleiv Kl{\o}ve, Yuan Luo, and Victor~K. Wei.
\newblock On equidistant constant weight codes.
\newblock {\em Discrete Applied Mathematics}, 128(1):157--164, 2003.

\bibitem{gabidulin1985theory}
Ernst~Muhamedovich Gabidulin.
\newblock Theory of codes with maximum rank distance.
\newblock {\em Problemy Peredachi Informatsii}, 21(1):3--16, 1985.

\bibitem{gabidulin2021rank}
Ernst~Muhamedovich Gabidulin.
\newblock {\em Rank codes}.
\newblock TUM. University Press, 2021.
\newblock translated by Vladimir Sidorenko.

\bibitem{GabidulinBossert}
Ernst~Muhamedovich Gabidulin and Martin Bossert.
\newblock Codes for network coding.
\newblock In {\em Proceedings. International Symposium on Information Theory,
  2008. ISIT 2008.}, pages 867--870. IEEE, 2008.

\bibitem{gabidulin2009algebraic}
Ernst~Muhamedovich Gabidulin and Martin Bossert.
\newblock Algebraic codes for network coding.
\newblock {\em Problems of Information Transmission}, 45(4):343--356, 2009.

\bibitem{gabidulin2021bounds}
Ernst~Muhamedovich Gabidulin, Nina~Ivanovna Pilipchuk, and
  Oksana~Viacheslavovna Trushina.
\newblock Bounds on the cardinality of subspace codes with non-maximum code
  distance.
\newblock {\em Problems of Information Transmission}, 57(3):241--247, 2021.

\bibitem{code_based_cryptography}
Philippe Gaborit and Jean-Christophe Deneuville.
\newblock Code-based cryptography.
\newblock In William~Cary Huffman, Jon-Lark Kim, and Patrick Sol{\'e}, editors,
  {\em Concise Encyclopedia of Coding Theory}, pages 799--822. Chapman and
  Hall/CRC, 2021.

\bibitem{constant_rank_codes}
Maximilien Gadouleau and Zhiyuan Yan.
\newblock Constant rank codes.
\newblock In {\em Proceedings. International Symposium on Information Theory,
  2008. ISIT 2008.}, pages 876--880. IEEE, 2008.

\bibitem{gadouleau2010constant}
Maximilien Gadouleau and Zhiyuan Yan.
\newblock Constant-rank codes and their connection to constant-dimension codes.
\newblock {\em IEEE Transactions on Information Theory}, 56(7):3207--3216,
  2010.

\bibitem{gao2021bounds}
You Gao, Jinru Gao, and Gang Wang.
\newblock Bounds on subspace codes based on subspaces of type $(s, 0, 0, 0)$ in
  pseudo-sympletic spaces and singular pseudo-symplectic spaces.
\newblock {\em Applied Mathematics and Computation}, 407:11 pp., 2021.

\bibitem{gao2014bounds}
You Gao and Gang Wang.
\newblock Bounds on subspace codes based on subspaces of type in singular
  linear space.
\newblock {\em Journal of Applied Mathematics}, page 9 pp., 2014.

\bibitem{gao2015error}
You Gao and Gang Wang.
\newblock Error-correcting codes in attenuated space over finite fields.
\newblock {\em Finite Fields and Their Applications}, 33:103--117, 2015.

\bibitem{gao2016bounds}
You Gao, Liyum Zhao, and Gang Wang.
\newblock Bounds on subspace codes based on totally isotropic subspaces in
  unitary spaces.
\newblock {\em Discrete Mathematics, Algorithms and Applications},
  8(04):1650056, 2016.

\bibitem{ghatak2017optimal}
Anirban Ghatak.
\newblock Optimal binary $(5, 3) $ projective space codes from maximal partial
  spreads.
\newblock {\em arXiv preprint 1701.07245}, 2017.

\bibitem{ghatak2021intersection}
Anirban Ghatak and Sumanta Mukherjee.
\newblock Intersection patterns in optimal binary $(5, 3) $ doubling subspace
  codes.
\newblock {\em arXiv preprint 2105.01584}, 2021.

\bibitem{gluesing2015cyclic}
Heide Gluesing-Luerssen, Katherine Morrison, and Carolyn Troha.
\newblock Cyclic orbit codes and stabilizer subfields.
\newblock {\em Advances in Mathematics of Communications}, 9(2):177--197, 2015.

\bibitem{gluesing2016construction}
Heide Gluesing-Luerssen and Carolyn Troha.
\newblock Construction of subspace codes through linkage.
\newblock {\em Advances in Mathematics of Communications}, 10(3):525--540,
  2016.

\bibitem{gorla2016equidistant}
Elisa Gorla and Alberto Ravagnani.
\newblock Equidistant subspace codes.
\newblock {\em Linear Algebra and its Applications}, 490:48--65, 2016.

\bibitem{gow2015dimension}
Rod Gow.
\newblock Dimension bounds for constant rank subspaces of symmetric bilinear
  forms over a finite field.
\newblock {\em arXiv preprint 1502.05547}, 2015.

\bibitem{greferath2018network}
Marcus Greferath, Mario~Osvin Pav{\v{c}}evi{\'c}, Natalia Silberstein, and
  Mar{\'\i}a~{\'A}ngeles V{\'a}zquez-Castro.
\newblock {\em Network coding and subspace designs}.
\newblock Springer, 2018.

\bibitem{hakimi2019bounds}
Mahdieh Hakimi~Poroch.
\newblock Bounds on subspace codes based on totally isotropic subspace in
  symplectic spaces and extended symplectic spaces.
\newblock {\em Asian-European Journal of Mathematics}, 12(05):1950069, 2019.

\bibitem{hall1977bounds}
Jonathan~I. Hall.
\newblock Bounds for equidistant codes and partial projective planes.
\newblock {\em Discrete Mathematics}, 17(1):85--94, 1977.

\bibitem{hansen2015riemann}
Johan~P. Hansen.
\newblock {R}iemann-{R}och spaces and linear network codes.
\newblock {\em Computer Science}, 10(1):1--11, 2015.

\bibitem{he2020construction}
Xianmang He.
\newblock Construction of constant dimension codes from two parallel versions
  of linkage construction.
\newblock {\em IEEE Communications Letters}, 24(11):2392--2395, 2020.

\bibitem{he2020improving}
Xianmang He, Yindong Chen, and Zusheng Zhang.
\newblock Improving the linkage construction with {E}chelon-{F}errers for
  constant-dimension codes.
\newblock {\em IEEE Communications Letters}, 24(9):1875--1879, 2020.

\bibitem{he2021new}
Xianmang He, Yindong Chen, Zusheng Zhang, and Kunxiao Zhou.
\newblock New construction for constant dimension subspace codes via a
  composite structure.
\newblock {\em IEEE Communications Letters}, 25(5):1422--1426, 2021.

\bibitem{he2019hierarchical}
Xianmang He, Yindong Chen, Kunxiao Zhou, and Jianguang Deng.
\newblock A hierarchical-based greedy algorithm for echelon-{F}errers
  construction.
\newblock {\em arXiv preprint 1911.00508}, 2019.

\bibitem{heden2016existence}
Olof Heden and Papa~Amar Sissokho.
\newblock On the existence of a $(2, 3)$-spread in $v(7, 2)$.
\newblock {\em Ars Combinatoria}, 124:161--164, 2016.

\bibitem{heijnen1998generalized}
Petra Heijnen and Ruud Pellikaan.
\newblock Generalized hamming weights of q-ary reed-muller codes.
\newblock {\em IEEE Transactions on Information Theory}, 44(1):181--196, 1998.

\bibitem{phd_heinlein}
Daniel Heinlein.
\newblock {\em Integer linear programming techniques for constant dimension
  codes and related structures}.
\newblock PhD thesis, Universit\"at Bayreuth (Germany), 2018.

\bibitem{heinlein2019new}
Daniel Heinlein.
\newblock New {L}{M}{R}{D} code bounds for constant dimension codes and
  improved constructions.
\newblock {\em IEEE Transactions on Information Theory}, 65(8):4822--4830,
  2019.

\bibitem{heinlein2020generalized}
Daniel Heinlein.
\newblock Generalized linkage construction for constant-dimension codes.
\newblock {\em IEEE Transactions on Information Theory}, 67(2):705--715, 2020.

\bibitem{heinlein2019generalized}
Daniel Heinlein, Thomas Honold, Michael Kiermaier, and Sascha Kurz.
\newblock Generalized vector space partitions.
\newblock {\em Australasian Journal of Combinatorics}, 73(1):162--178, 2019.

\bibitem{heinlein2019classifying}
Daniel Heinlein, Thomas Honold, Michael Kiermaier, Sascha Kurz, and Alfred
  Wassermann.
\newblock Classifying optimal binary subspace codes of length $8$, constant
  dimension $4$ and minimum distance $6$.
\newblock {\em Designs, Codes and Cryptography}, 87(2-3):375--391, 2019.

\bibitem{heinlein2019projective}
Daniel Heinlein, Thomas Honold, Michael Kiermaier, Sascha Kurz, and Alfred
  Wassermann.
\newblock On projective $q^r$-divisible codes.
\newblock {\em arXiv preprint 1912.10147}, 2019.

\bibitem{heinlein2020new}
Daniel Heinlein and Ferdinand Ihringer.
\newblock New and updated semidefinite programming bounds for subspace codes.
\newblock {\em Advances in Mathematics of Communications}, 14(4):613, 2020.

\bibitem{TableSubspacecodes}
Daniel Heinlein, Michael Kiermaier, Sascha Kurz, and Alfred Wassermann.
\newblock Tables of subspace codes.
\newblock {\em arXiv preprint 1601.02864}, 2016.

\bibitem{paper333}
Daniel Heinlein, Michael Kiermaier, Sascha Kurz, and Alfred Wassermann.
\newblock A subspace code of size $333$ in the setting of a binary $q$-analog
  of the {F}ano plane.
\newblock {\em Advances in Mathematics of Communications}, 13(3):457--475,
  2019.

\bibitem{heinlein2019subspace}
Daniel Heinlein, Michael Kiermaier, Sascha Kurz, and Alfred Wassermann.
\newblock A subspace code of size $333$ in the setting of a binary $q$-analog
  of the {F}ano plane.
\newblock {\em Advances in Mathematics of Communications}, 13(3):457--475,
  2019.

\bibitem{heinlein2017asymptotic}
Daniel Heinlein and Sascha Kurz.
\newblock Asymptotic bounds for the sizes of constant dimension codes and an
  improved lower bound.
\newblock In {\'A}ngela~Isabel Barbero, Vitaly Skachek, and {\O}yvind Ytrehus,
  editors, {\em Coding Theory and Applications: 5th International Castle
  Meeting, ICMCTA 2017, Vihula, Estonia, August 28-31, 2017, Proceedings},
  volume 10495 of {\em Lecture Notes in Computer Science}, pages 163--191,
  Cham, 2017. Springer International Publishing.
\newblock arXiv preprint 1703.08712.

\bibitem{heinlein2017coset}
Daniel Heinlein and Sascha Kurz.
\newblock Coset construction for subspace codes.
\newblock {\em IEEE Transactions on Information Theory}, 63(12):7651--7660,
  2017.

\bibitem{heinlein2017new}
Daniel Heinlein and Sascha Kurz.
\newblock An upper bound for binary subspace codes of length $8$, constant
  dimension $4$ and minimum distance $6$.
\newblock In {\em The Tenth International Workshop on Coding and Cryptography},
  2017.
\newblock arXiv preprint 1705.03835.

\bibitem{heinlein2018binary}
Daniel Heinlein and Sascha Kurz.
\newblock Binary subspace codes in small ambient spaces.
\newblock {\em Advances in Mathematics of Communications}, 12(4):817--839,
  2018.

\bibitem{helleseth1992generalized}
Tor Helleseth, Torleiv Klove, and {\O}yvind Ytrehus.
\newblock Generalized hamming weights of linear codes.
\newblock {\em IEEE Transactions on Information Theory}, 38(3):1133--1140,
  1992.

\bibitem{hishida2000cyclic}
Takaaki Hishida and Masakazu Jimbo.
\newblock Cyclic resolutions of the {B}{I}{B} design in $\operatorname{PG} (5,
  2)$.
\newblock {\em Australasian Journal of Combinatorics}, 22:73--80, 2000.

\bibitem{honold2016putative}
Thomas Honold and Michael Kiermaier.
\newblock On putative $q$-analogues of the {F}ano plane and related
  combinatorial structures.
\newblock In {\em Dynamical Systems, Number Theory and Applications: A
  Festschrift in Honor of Armin Leutbecher's 80th Birthday}, pages 141--175.
  World Scientific, 2016.

\bibitem{hkk77}
Thomas Honold, Michael Kiermaier, and Sascha Kurz.
\newblock Optimal binary subspace codes of length $6$, constant dimension $3$
  and minimum distance $4$.
\newblock {\em Contemporary Mathematics}, 632:157--176, 2015.

\bibitem{honold2016constructions}
Thomas Honold, Michael Kiermaier, and Sascha Kurz.
\newblock Constructions and bounds for mixed-dimension subspace codes.
\newblock {\em Advances in Mathematics of Communications}, 10(3):649--682,
  2016.

\bibitem{honold2018partial}
Thomas Honold, Michael Kiermaier, and Sascha Kurz.
\newblock Partial spreads and vector space partitions.
\newblock In Marcus Greferath, Mario~Osvin Pav{\v{c}}evi{\'c}, Natalia
  Silberstein, and Mar{\'\i}a~{\'A}ngeles V{\'a}zquez-Castro, editors, {\em
  Network Coding and Subspace Designs}, pages 131--170. Springer, 2018.

\bibitem{honold2019classification}
Thomas Honold, Michael Kiermaier, and Sascha Kurz.
\newblock Classification of large partial plane spreads in
  $\operatorname{PG}(6, 2)$ and related combinatorial objects.
\newblock {\em Journal of Geometry}, 110(1):1--31, 2019.

\bibitem{honold2019johnson}
Thomas Honold, Michael Kiermaier, and Sascha Kurz.
\newblock Johnson type bounds for mixed dimension subspace codes.
\newblock {\em The Electronic Journal of Combinatorics}, 26(3):21 pp., 2019.

\bibitem{honold2019lengths}
Thomas Honold, Michael Kiermaier, Sascha Kurz, and Alfred Wassermann.
\newblock The lengths of projective triply-even binary codes.
\newblock {\em IEEE Transactions on Information Theory}, 66(5):2713--2716,
  2019.

\bibitem{horlemann2018constructions}
Anna-Lena Horlemann-Trautmann and Joachim Rosenthal.
\newblock Constructions of constant dimension codes.
\newblock In Marcus Greferath, Mario~Osvin Pav{\v{c}}evi{\'c}, Natalia
  Silberstein, and Mar{\'\i}a~{\'A}ngeles V{\'a}zquez-Castro, editors, {\em
  Network Coding and Subspace Designs}, pages 25--42. Springer, 2018.

\bibitem{ihringer2018new}
Ferdinand Ihringer, Peter Sin, and Qing Xiang.
\newblock New bounds for partial spreads of $h(2d-1, q^2)$ and partial ovoids
  of the {R}ee--{T}its octagon.
\newblock {\em Journal of Combinatorial Theory, Series A}, 153:46--53, 2018.

\bibitem{johnson1962new}
Selmer Johnson.
\newblock A new upper bound for error-correcting codes.
\newblock {\em IRE Transactions on Information Theory}, 8(3):203--207, 1962.

\bibitem{khaleghi2009projective}
Azadeh Khaleghi and Frank~Robert Kschischang.
\newblock Projective space codes for the injection metric.
\newblock In {\em 2009 11th Canadian Workshop on Information Theory}, pages
  9--12. IEEE, 2009.

\bibitem{khaleghi2009subspace}
Azadeh Khaleghi, Danilo Silva, and Frank~Robert Kschischang.
\newblock Subspace codes.
\newblock In Matthew~Geoffrey Parker, editor, {\em IMA International Conference
  on Cryptography and Coding}, pages 1--21. Springer, 2009.

\bibitem{kiermaier2021alpha}
Michael Kiermaier.
\newblock On $\alpha$-points of $q$-analogs of the {F}ano plane.
\newblock {\em Designs, Codes and Cryptography}, 90(6):1335--1345, 2022.

\bibitem{kiermaier2020lengths}
Michael Kiermaier and Sascha Kurz.
\newblock On the lengths of divisible codes.
\newblock {\em IEEE Transactions on Information Theory}, 66(7):4051--4060,
  2020.

\bibitem{kiermaier2018order}
Michael Kiermaier, Sascha Kurz, and Alfred Wassermann.
\newblock The order of the automorphism group of a binary $q$-analog of the
  {F}ano plane is at most two.
\newblock {\em Designs, Codes and Cryptography}, 86(2):239--250, 2018.

\bibitem{korner2023lengths}
Theresa K{\"o}rner and Sascha Kurz.
\newblock Lengths of divisible codes with restricted column multiplicities.
\newblock {\em Advances in Mathematics of Communications}, 18(2):505--534, 2024.

\bibitem{koetter2008coding}
Ralf Koetter and Frank~Robert Kschischang.
\newblock Coding for errors and erasures in random network coding.
\newblock {\em IEEE Transactions on Information Theory}, 54(8):3579--3591,
  2008.

\bibitem{paper_axel}
Axel Kohnert and Sascha Kurz.
\newblock Construction of large constant dimension codes with a prescribed
  minimum distance.
\newblock In {\em Mathematical methods in computer science}, volume 5393 of
  {\em Lecture Notes in Computer Science}, pages 31--42. Springer, Berlin,
  2008.

\bibitem{network_codes}
Frank~Robert Kschischang.
\newblock Network codes.
\newblock In William~Cary Huffman, Jon-Lark Kim, and Patrick Sol{\'e}, editors,
  {\em Concise Encyclopedia of Coding Theory}, pages 685--714. Chapman and
  Hall/CRC, 2021.

\bibitem{kurzspreads}
Sascha Kurz.
\newblock Improved upper bounds for partial spreads.
\newblock {\em Designs, Codes and Cryptography}, 85(1):97--106, 2017.

\bibitem{kurz2017packing}
Sascha Kurz.
\newblock Packing vector spaces into vector spaces.
\newblock {\em Australasian Journal of Combinatorics}, 68:122--130, 2017.

\bibitem{kurz2019note}
Sascha Kurz.
\newblock A note on the linkage construction for constant dimension codes.
\newblock {\em arXiv preprint 1906.09780}, 2019.

\bibitem{kurz2021divisible}
Sascha Kurz.
\newblock Divisible codes.
\newblock {\em arXiv preprint 2112.11763}, 2021.

\bibitem{kurz2020generalized}
Sascha Kurz.
\newblock Generalized {L}{M}{R}{D} code bounds for constant dimension codes.
\newblock {\em IEEE Communications Letters}, 24(10):2100--2103, 2020.

\bibitem{kurz2023lengths}
Sascha Kurz.
\newblock Lengths of divisible codes -- the missing cases.
\newblock {\em Designs, Codes and Cryptography}, 92(8):2367--2378, 2024.

\bibitem{kurz2020lifted}
Sascha Kurz.
\newblock Lifted codes and the multilevel construction for constant dimension
  codes.
\newblock {\em arXiv preprint 2004.14241}, 2020.

\bibitem{kurz2020no131}
Sascha Kurz.
\newblock No projective $16$-divisible binary linear code of length $131$
  exists.
\newblock {\em IEEE Communications Letters}, 25(1):38--40, 2020.

\bibitem{kurz2020subspaces}
Sascha Kurz.
\newblock Subspaces intersecting in at most a point.
\newblock {\em Designs, Codes and Cryptography}, 88(3):595--599, 2020.

\bibitem{kurz2021bounds}
Sascha Kurz.
\newblock Bounds for flag codes.
\newblock {\em Designs, Codes and Cryptography}, 89(12):2759--2785, 2021.

\bibitem{kurz2021interplay}
Sascha Kurz.
\newblock The interplay of different metrics for the construction of constant
  dimension codes.
\newblock {\em Advances in Mathematics of Communications}, 17(1):152--171, 2023.

\bibitem{kurz2023lengths}
Theresa K{\"o}rner and Sascha Kurz.
\newblock Lengths of divisible codes with restricted column multiplicities.
\newblock {\em Advances in Mathematics of Communications}, 18(2):505--534, 2024.

\bibitem{landsberg1893ueber}
Georg Landsberg.
\newblock Ueber eine {A}nzahlbestimmung und eine damit zusammenh{\"a}ngende
  {R}eihe.
\newblock {\em Journal f\"ur die reine und angewandte Mathematik},
  1893(111):87--88, 1893.

\bibitem{lao2021new}
Huimin Lao, Hao Chen, and Xiaoqing Tan.
\newblock New constant dimension subspace codes from block inserting
  constructions.
\newblock {\em Cryptography and Communications}, pages 1--13, 2021.

\bibitem{lao2020parameter}
Huimin Lao, Hao Chen, Jian Weng, and Xiaoqing Tan.
\newblock Parameter-controlled inserting constructions of constant dimension
  subspace codes.
\newblock {\em arXiv preprint 2008.09944}, 2020.

\bibitem{lehmann2021weight}
Hunter~Ryan Lehmann.
\newblock {\em Weight Distributions, Automorphisms, and Isometries of Cyclic
  Orbit Codes}.
\newblock PhD thesis, University of Kentucky, 2021.

\bibitem{li2019construction}
Fagang Li.
\newblock Construction of constant dimension subspace codes by modifying
  linkage construction.
\newblock {\em IEEE Transactions on Information Theory}, 66(5):2760--2764,
  2019.

\bibitem{liang2020terwilliger}
Xiaoye Liang, Tatsuro Ito, and Yuta Watanabe.
\newblock The {T}erwilliger algebra of the {G}rassmann scheme $j_q(n, d)$
  revisited from the viewpoint of the quantum affine algebra
  $u_q(\widehat{\operatorname{sl}}_2)$.
\newblock {\em Linear Algebra and its Applications}, 596:117--144, 2020.

\bibitem{liebhold2019flag}
Dirk Liebhold.
\newblock {\em Flag codes with application to network coding}.
\newblock PhD thesis, RWTH Aachen, 2019.

\bibitem{liebhold2018network}
Dirk Liebhold, Gabriele Nebe, and Angeles Vazquez-Castro.
\newblock Network coding with flags.
\newblock {\em Designs, Codes and Cryptography}, 86(2):269--284, 2018.

\bibitem{liebhold2018generalizing}
Dirk Liebhold, Gabriele Nebe, and Mar{\'\i}a~{\'A}ngeles V{\'a}zquez-Castro.
\newblock Generalizing subspace codes to flag codes using group actions.
\newblock In Marcus Greferath, Mario~Osvin Pav{\v{c}}evi{\'c}, Natalia
  Silberstein, and Mar{\'\i}a~{\'A}ngeles V{\'a}zquez-Castro, editors, {\em
  Network Coding and Subspace Designs}, pages 67--89. Springer, 2018.

\bibitem{litsyn1998update}
Simon Litsyn.
\newblock An update table of the best binary codes known.
\newblock In Vera Pless, Richard~Anthony Brualdi, and William~Cary Huffman,
  editors, {\em Handbook of Coding Theory}, pages 463--498. Elsevier, 1998.

\bibitem{liu2019constructions}
Shuangqing Liu, Yanxun Chang, and Tao Feng.
\newblock Constructions for optimal {F}errers diagram rank-metric codes.
\newblock {\em IEEE Transactions on Information Theory}, 65(7):4115--4130,
  2019.

\bibitem{liu2020parallel}
Shuangqing Liu, Yanxun Chang, and Tao Feng.
\newblock Parallel multilevel constructions for constant dimension codes.
\newblock {\em IEEE Transactions on Information Theory}, 66(11):6884--6897,
  2020.

\bibitem{longobardi2021sets}
Giovanni Longobardi, Leo Storme, and Rocco Trombetti.
\newblock On sets of subspaces with two intersection dimensions and a
  geometrical junta bound.
\newblock {\em Designs, Codes and Cryptography}, 90:2081–2099, 2022.

\bibitem{lucas2019properties}
Lisa~Hernandez Lucas.
\newblock Properties of sets of subspaces with constant intersection dimension.
\newblock {\em Advances in Mathematics of Communications}, 15(1):191--206, 2021.

\bibitem{lucasgeometrical}
Lisa~Hernandez Lucas, Ivan Landjev, Leo Storme, and Peter Vandendriessche.
\newblock On the geometrical sunflower bound.
\newblock In {\em Eighth International Workshop on Optimal Codes and Related
  Topics, Sofia, Bulgaria}, pages 93--97, 2017.

\bibitem{macwilliams1977theory}
Florence~Jessie MacWilliams and Neil James~Alexander Sloane.
\newblock {\em The theory of error correcting codes}, volume~16.
\newblock Elsevier, 1977.

\bibitem{mahak2025equidistant}
Mahak (fathers name: Nirmal Kumar) and Bhaintwal, Maheshanand.
\newblock On equidistant single-orbit cyclic and quasi-cyclic subspace codes.
\newblock {\em Designs, Codes and Cryptography}, 93(6):2159--2175, 2025.

\bibitem{mounits2002improved}
Beniamin Mounits, Tuvi Etzion, and Simon Litsyn.
\newblock Improved upper bounds on sizes of codes.
\newblock {\em IEEE Transactions on Information Theory}, 48(4):880--886, 2002.

\bibitem{mounits2007new}
Beniamin Mounits, Tuvi Etzion, and Simon Litsyn.
\newblock New upper bounds on codes via association schemes and linear
  programming.
\newblock {\em Advances in Mathematics of Communications}, 1(2):173, 2007.

\bibitem{nakic2016extendability}
Anamari Naki{\'c} and Leo Storme.
\newblock On the extendability of particular classes of constant dimension
  codes.
\newblock {\em Designs, Codes and Cryptography}, 79(3):407--422, 2016.

\bibitem{nastase2016maximumII}
Esmeralda N{\u{a}}stase and Papa~Amar Sissokho.
\newblock The maximum size of a partial spread {II}: Upper bounds.
\newblock {\em Discrete Mathematics}, 340(7):1481--1487, 2017.

\bibitem{nastase2016maximum}
Esmeralda N{\u{a}}stase and Papa~Amar Sissokho.
\newblock The maximum size of a partial spread in a finite projective space.
\newblock {\em Journal of Combinatorial Theory, Series A}, 152:353--362, 2017.

\bibitem{navarro2021flag}
Miguel~{\'A}ngel Navarro-P{\'e}rez and Xaro~Soler Escriv{\`a}.
\newblock Flag codes of maximum distance and constructions using Singer groups.
\newblock {\em Finite Fields and Their Applications}, 80:102011, 2022.

\bibitem{ngai2011network}
Chi-Kin Ngai, Raymond~W. Yeung, and Zhixue Zhang.
\newblock Network generalized hamming weight.
\newblock {\em IEEE Transactions on Information Theory}, 57(2):1136--1143,
  2011.

\bibitem{niu2018subspace}
Min-Yao Niu, Gang Wang, You Gao, and Fang-Wei Fu.
\newblock Subspace code based on flats in affine space over finite fields.
\newblock {\em Discrete Mathematics, Algorithms and Applications},
  10(06):1850078, 2018.

\bibitem{niu2020new}
Yongfeng Niu, Qin Yue, and Daitao Huang.
\newblock New constant dimension subspace codes from generalized inserting
  construction.
\newblock {\em IEEE Communications Letters}, 25(4):1066--1069, 2020.

\bibitem{niu2021construction}
Yongfeng Niu, Qin Yue, and Daitao Huang.
\newblock Construction of constant dimension codes via improved inserting
  construction.
\newblock {\em Applicable Algebra in Engineering, Communication and Computing}, 34(6):1045--1062, 2023.

\bibitem{niu2020several}
Yongfeng Niu, Qin Yue, and Yansheng Wu.
\newblock Several kinds of large cyclic subspace codes via {S}idon spaces.
\newblock {\em Discrete Mathematics}, 343(5):111788, 2020.

\bibitem{space_time_coding}
Fr\'ed\'erique Oggier.
\newblock Space-time coding.
\newblock In William~Cary Huffman, Jon-Lark Kim, and Patrick Sol{\'e}, editors,
  {\em Concise Encyclopedia of Coding Theory}, pages 673--684. Chapman and
  Hall/CRC, 2021.

\bibitem{ostergard2010classification}
Patric~R.J. {\"O}sterg{\aa}rd.
\newblock Classification of binary constant weight codes.
\newblock {\em IEEE Transactions on Information Theory}, 56(8):3779--3785,
  2010.

\bibitem{otal2017cyclic}
Kamil Otal and Ferruh {\"O}zbudak.
\newblock Cyclic subspace codes via subspace polynomials.
\newblock {\em Designs, Codes and Cryptography}, 85(2):191--204, 2017.

\bibitem{otal2018constructions}
Kamil Otal and Ferruh {\"O}zbudak.
\newblock Constructions of cyclic subspace codes and maximum rank distance
  codes.
\newblock In Marcus Greferath, Mario~Osvin Pav{\v{c}}evi{\'c}, Natalia
  Silberstein, and Mar{\'\i}a~{\'A}ngeles V{\'a}zquez-Castro, editors, {\em
  Network Coding and Subspace Designs}, pages 43--66. Springer, 2018.

\bibitem{pai2015bounds}
Bantwal~Srikanth Pai and Bikash~Sundar Rajan.
\newblock On the bounds of certain maximal linear codes in a projective space.
\newblock {\em IEEE Transactions on Information Theory}, 61(9):4923--4927,
  2015.

\bibitem{polak2018semidefinite}
Sven~Carel Polak.
\newblock Semidefinite programming bounds for constant-weight codes.
\newblock {\em IEEE Transactions on Information Theory}, 65(1):28--38, 2018.

\bibitem{codes_for_distributed_storage}
Vinayak Ramkumar, Myna Vajha, Srinivasan~Babu Balaji, M.~Nikhil Krishnan,
  Birenjith Sasidharan, and P.~Vijay Kumar.
\newblock Codes for distributed storage.
\newblock In William~Cary Huffman, Jon-Lark Kim, and Patrick Sol{\'e}, editors,
  {\em Concise Encyclopedia of Coding Theory}, pages 735--762. Chapman and
  Hall/CRC, 2021.

\bibitem{raviv2017subspace}
Charlene Raviv.
\newblock {\em Subspace codes and distributed storage codes}.
\newblock PhD thesis, Computer Science Department, Technion, 2017.

\bibitem{roth1991maximum}
Ron~M. Roth.
\newblock Maximum-rank array codes and their application to crisscross error
  correction.
\newblock {\em IEEE Transactions on Information Theory}, 37(2):328--336, 1991.

\bibitem{roth2017construction}
Ron~M. Roth, Netanel Raviv, and Itzhak Tamo.
\newblock Construction of {S}idon spaces with applications to coding.
\newblock {\em IEEE Transactions on Information Theory}, 64(6):4412--4422,
  2017.

\bibitem{sarmiento2002point}
Jumela~F. Sarmiento.
\newblock On point-cyclic resolutions of the $2-(63, 7, 15)$ design associated
  with $\operatorname{PG} (5, 2)$.
\newblock {\em Graphs and Combinatorics}, 18(3):621--632, 2002.

\bibitem{schmidt2025linear}
Kai-Uwe Schmidt and Charlene Wei{\ss}.
\newblock The linear programming optimum for packings in classical association schemes.
\newblock {\em arXiv preprint 2508.12806}, 2025.

\bibitem{schrijver2005new}
Alexander Schrijver.
\newblock New code upper bounds from the terwilliger algebra and semidefinite
  programming.
\newblock {\em IEEE Transactions on Information Theory}, 51(8):2859--2866,
  2005.

\bibitem{segre1964teoria}
Beniamino Segre.
\newblock Teoria di {G}alois, fibrazioni proiettive e geometrie non
  desarguesiane.
\newblock {\em Annali di Matematica Pura ed Applicata}, 64(1):1--76, 1964.

\bibitem{sheekey2016new}
John Sheekey.
\newblock A new family of linear maximum rank distance codes.
\newblock {\em Advances in Mathematics of Communications}, 10(3):475--488,
  2016.

\bibitem{sheekey2019binary}
John Sheekey.
\newblock Binary additive {M}{R}{D} codes with minimum distance $n-1$ must
  contain a semifield spread set.
\newblock {\em Designs, Codes and Cryptography}, 87(11):2571--2583, 2019.

\bibitem{sheekey2019mrd}
John Sheekey.
\newblock {M}{R}{D} codes: Constructions and connections.
\newblock In {\em Combinatorics and Finite Fields: Difference Sets,
  Polynomials, Pseudorandomness and Applications}, volume~23 of {\em Radon
  Series on Computational and Applied Mathematics}. De Gruyter, Berlin, 2019.

\bibitem{shishkin2016new}
Alexander Shishkin.
\newblock New multicomponent network subspace codes: Construction and decoding.
\newblock In {\em 2016 International Conference on Engineering and
  Telecommunication (EnT)}, pages 123--127. IEEE, 2016.

\bibitem{shishkin2014cardinality}
Alexander Shishkin, Ernst~Muhamedovich Gabidulin, and Nina~Ivanovna Pilipchuk.
\newblock On cardinality of network subspace codes.
\newblock In {\em Proceeding of the Fourteenth International Workshop on
  Algebraic and Combinatorial Coding Theory (ACCT-XIV)}, volume~7, pages
  300--306, 2014.

\bibitem{silberstein2011large}
Natalia Silberstein and Tuvi Etzion.
\newblock Large constant dimension codes and lexicodes.
\newblock {\em Advances in Mathematics of Communications}, 5(2):177--189, 2011.

\bibitem{silberstein2013new}
Natalia Silberstein and Anna-Lena Trautmann.
\newblock New lower bounds for constant dimension codes.
\newblock In {\em 2013 IEEE International Symposium on Information Theory (ISIT
  2013)}, pages 514--518. IEEE, 2013.

\bibitem{silberstein2015error}
Natalia Silberstein and Anna-Lena Trautmann.
\newblock Subspace codes based on graph matchings, {F}errers diagrams, and
  pending blocks.
\newblock {\em IEEE Transactions on Information Theory}, 61(7):3937--3953,
  2015.

\bibitem{silva2009metrics}
Danilo Silva and Frank~Robert Kschischang.
\newblock On metrics for error correction in network coding.
\newblock {\em IEEE Transactions on Information Theory}, 55(12):5479--5490,
  2009.

\bibitem{silva2008rank}
Danilo Silva, Frank~Robert Kschischang, and Ralf Koetter.
\newblock A rank-metric approach to error control in random network coding.
\newblock {\em IEEE Transactions on Information Theory}, 54(9):3951--3967,
  2008.

\bibitem{sinha2008good}
Kishore Sinha, Z.~Wang, and Dianhua Wu.
\newblock Good equidistant codes constructed from certain combinatorial
  designs.
\newblock {\em Discrete Mathematics}, 308(18):4205--4211, 2008.

\bibitem{skachek2010recursive}
Vitaly Skachek.
\newblock Recursive code construction for random networks.
\newblock {\em IEEE Transactions on Information Theory}, 56(3):1378--1382,
  2010.

\bibitem{sloane1975introduction}
Neil James~Alexander Sloane.
\newblock An introduction to association schemes and coding theory.
\newblock In Richard~Allen Askey, editor, {\em Theory and application of
  special functions}, pages 225--260. Elsevier, 1975.
\newblock Proceedings of an Advanced Seminar Sponsored by the Mathematics
  Research Center, the University of Wisconsin--Madison, March 31--April 2,
  1975.

\bibitem{storme2021coding}
Leo Storme.
\newblock Coding theory and galois geometries.
\newblock In William~Cary Huffman, Jon-Lark Kim, and Patrick Sol{\'e}, editors,
  {\em Concise Encyclopedia of Coding Theory}, pages 285--305. Chapman and
  Hall/CRC, 2021.

\bibitem{thomas1987designs}
Simon Thomas.
\newblock Designs over finite fields.
\newblock {\em Geometriae Dedicata}, 24(2):237--242, 1987.

\bibitem{thomas1996designs}
Simon Thomas.
\newblock Designs and partial geometries over finite fields.
\newblock {\em Geometriae Dedicata}, 63(3):247--253, 1996.

\bibitem{tonchev1998codes}
Vladimir~Dimitrov Tonchev.
\newblock Codes and designs.
\newblock In Vera Pless, Richard~Anthony Brualdi, and William~Cary Huffman,
  editors, {\em Handbook of Coding Theory}, volume~2, pages 1229--1267.
  Elsevier, 1998.

\bibitem{trautmann2013constructions}
Anna-Lena Trautmann.
\newblock {\em Constructions, decoding and automorphisms of subspace codes}.
\newblock PhD thesis, University of Zurich, 2013.

\bibitem{vallentin2021semidefinite}
Frank Vallentin.
\newblock Semidefinite programming bounds for error-correcting codes.
\newblock In William~Cary Huffman, Jon-Lark Kim, and Patrick Sol{\'e}, editors,
  {\em Concise Encyclopedia of Coding Theory}, pages 267--282. Chapman and
  Hall/CRC, 2021.

\bibitem{van2020cayley}
Vincent van~der Noort.
\newblock Cayley algebras give rise to $q$-{F}ano planes over certain infinite
  fields and $q$-covering designs over others.
\newblock {\em arXiv preprint 2006.01268}, 2020.

\bibitem{van1973theorem}
Jacobus~Hendricus van Lint.
\newblock A theorem on equidistant codes.
\newblock {\em Discrete Mathematics}, 6(4):353--358, 1973.

\bibitem{van2001course}
Jacobus~Hendricus Van~Lint and Richard~Michael Wilson.
\newblock {\em A course in combinatorics}.
\newblock Cambridge University Press, 2 edition, 2001.

\bibitem{wan1997geometry}
Zhe-xian Wan.
\newblock Geometry of classical groups over finite fields and its applications.
\newblock {\em Discrete Mathematics}, 174(1-3):365--381, 1997.

\bibitem{wang2003linear}
Huaxiong Wang, Chaoping Xing, and Reihaneh Safavi-Naini.
\newblock Linear authentication codes: bounds and constructions.
\newblock {\em IEEE Transactions on Information Theory}, 49(4):866--872, 2003.

\bibitem{wei1991generalized}
Victor~K. Wei.
\newblock Generalized {H}amming weights for linear codes.
\newblock {\em IEEE Transactions on Information Theory}, 37(5):1412--1418,
  1991.

\bibitem{weiss2023linear}
Charlene Wei{\ss}.
\newblock {\em Linear programming bounds in classical association schemes}.
\newblock PhD thesis, Universit{\"a}t Paderborn, 2023.

\bibitem{xia2009johnson}
Shu-Tao Xia and Fang-Wei Fu.
\newblock Johnson type bounds on constant dimension codes.
\newblock {\em Designs, Codes and Cryptography}, 50(2):163--172, 2009.

\bibitem{xu2018new}
Liqing Xu and Hao Chen.
\newblock New constant-dimension subspace codes from maximum rank distance
  codes.
\newblock {\em IEEE Transactions on Information Theory}, 64(9):6315--6319,
  2018.

\bibitem{zhang2020new}
Tao Zhang and Yue Zhou.
\newblock New lower bounds for partial $k$-parallelisms.
\newblock {\em Journal of Combinatorial Designs}, 28(1):75--84, 2020.

\bibitem{zhang2011linear}
Zong-Ying Zhang, Yong Jiang, and Shu-Tao Xia.
\newblock On the linear programming bounds for constant dimension codes.
\newblock In {\em 2011 International Symposium on Networking Coding (NetCod
  2011)}, pages 1--4. IEEE, 2011.

\bibitem{zhou2021construction}
Kunxiao Zhou, Yindong Chen, Zusheng Zhang, Feng Shi, and Xianmang He.
\newblock A construction for constant dimension codes from the known codes.
\newblock In {\em International Conference on Wireless Algorithms, Systems, and
  Applications}, pages 253--262. Springer, 2021.

\end{thebibliography}
%% \bibliographystyle{plain}

\end{document}